\documentclass[11pt,letterpaper]{article}
\usepackage[margin=1in]{geometry}
\usepackage{amsfonts,amsmath,latexsym,amssymb,amsthm}
\usepackage[USenglish]{babel}

\usepackage{graphicx}
\usepackage{xcolor}
\usepackage{xspace}
\usepackage{enumitem}

\usepackage[utf8]{inputenc}
\usepackage{lmodern}

\usepackage[nothing]{algorithm}
\usepackage{algorithmicx}
\usepackage[noend]{algpseudocode}
\usepackage{array}
\usepackage{multirow}

\usepackage{authblk}
\usepackage{hyperref}
\hypersetup{pdftitle=A φ-Competitive Algorithm for Scheduling Packets with Deadlines}
\hypersetup{pdfauthor={Pavel Vesel\'{y}, Marek Chrobak, {\L}ukasz Je\.{z}, and Ji{\v{r}}\'{\i} Sgall}}
\hypersetup{
    colorlinks,
    linkcolor={blue!50!black},
    citecolor={blue!50!black},
    urlcolor={blue!80!black}
}

\usepackage[export]{adjustbox}

\usepackage{microtype}

\title{A $\phi$-Competitive Algorithm for Scheduling Packets with Deadlines\footnote{
		A preliminary version of this article first appeared in SODA'19~\cite{this-conference}.
		An extended and thoroughly revised version was then published in SIAM Journal on Computing~\cite{this-journal}, with Copyright transferred to SIAM.
		This version, with additional counterexamples to some simpler variants of the algorithm,
		is as it existed immediately prior to editing and production by SIAM,
		and is published with SIAM's permission.}}

\author[1]{Pavel Vesel\'{y}}
\author[2]{Marek Chrobak}
\author[3]{{\L}ukasz Je\.{z}}
\author[1]{Ji{\v{r}}\'{\i} Sgall}

\affil[1]{Computer Science Institute of Charles University,
	Faculty of Mathematics and Physics, Prague, Czech Republic\\
	\texttt{\{vesely,sgall\}@iuuk.mff.cuni.cz}}
\affil[2]{Department of Computer Science, University of California at Riverside, CA, U.S.A.\\
	\texttt{marek@cs.ucr.edu}}
\affil[3]{Institute of Computer Science, University of Wroc{\l}aw, Poland\\
	\texttt{lje@cs.uni.wroc.pl}}

\date{}

\newcommand{\myparagraph}[1]{{\smallskip\noindent{\bf #1}}}
\newcommand{\emparagraph}[1]{{\smallskip\noindent\emph{#1}}} 
\newcommand{\indentemparagraph}[1]{{\smallskip\emph{#1}}} 
\newcommand{\etal}{{\em et al.}}

\newcommand{\mycase}[1]{\mbox{{\underline{Case #1}}:\/}}

\newcommand{\bigcell}[2]{\begin{tabular}{@{}#1@{}}#2\end{tabular}}

\let\rho\varrho

\newcommand{\hatf}{{\hat{f}}}
\newcommand{\hatg}{{\hat{g}}}

\newcommand{\dotx}{{\dot{x}}}

\newcommand{\dotz}{{\dot{z}}}

\newcommand{\calA}{{\mathcal{A}}}
\newcommand{\calB}{{\mathcal{B}}}

\newcommand{\calP}{{\mathcal{P}}}

\newcommand{\calU}{{\mathcal{U}}}

\newcommand{\tildea}{{\tilde{a}}}
\newcommand{\tildeb}{{\tilde{b}}}

\newcommand{\eps}{{\epsilon}}

\newcommand{\fstar}{{f^\ast}}  
\newcommand{\gstar}{{g^\ast}}  
\newcommand{\hstar}{{h^\ast}}  

\newcommand{\e}{\mathrm{e}}

\newcommand{\bracedintext}[1]{{\{ #1 \}}}
\newcommand{\braced}[1]{{ \left\{ #1 \right\} }}

\newcommand{\suchthat}{{\;:\;}}
\newcommand{\assign}{\,{\leftarrow}\,}

\newcommand{\razy}{{\,}}
\newcommand{\phinegonef}{{ \textstyle{ \frac{1}{\phi} } }}  
\newcommand{\phinegtwof}{{ \textstyle{ \frac{1}{\phi^2} } }} 
\newcommand{\phinegthreef}{{ \textstyle{ \frac{1}{\phi^3} } }}

\newcommand{\malymax}{{\mbox{\tiny\rm max}}}

\newcommand{\argmax}{{\textrm{argmax}}}  

\newcommand{\BDPS}{{\textsf{PacketSchD}}}
\newcommand{\ComputePlan}{{\textsf{ComputePlan}}}

\newcommand{\PlanMonotonicity}{\textsf{PlanM}}

\newcommand{\weights}{{w}}
\newcommand{\minwt}{{\textsf{minwt}}}
\newcommand{\packetslack}{{\textsf{pslack}}}
\newcommand{\pslack}{{\textsf{pslack}}}
\newcommand{\nexttightslot}{\textsf{nextts}}
\newcommand{\firstlightpacket}{{\ell}}
\newcommand{\prevtightslot}{\textsf{prevts}}

\newcommand{\timehorizon}{{T}}

\newcommand{\tinyspace}{{\hspace{0.5pt}}} 
\newcommand{\gtrweight}{{\triangleright\tinyspace}}
\newcommand{\geqweight}{{\trianglerighteq\tinyspace}} 
\newcommand{\seggroup}[1]{{\langle #1\rangle}}

\newcommand{\terminala}{{\tildea}}
\newcommand{\initialb}{{\tildeb}}

\newcommand{\advcredit}{\textsf{advcredit}}
\newcommand{\substpacket}{\textsf{sub}}
\newcommand{\pendpackets}{{\calU}}
\newcommand{\barpendpackets}{{\overline{\calU}}}

\newcommand{\varplanP}{{\calP}}
\newcommand{\planP}{{P}}
\newcommand{\planPbefStep}{{\widetilde{\planP}^t}} 
\newcommand{\planPaftArrivals}{{{\planP}^t}}    

\newcommand{\planPaftInitSeg}{{{\planP}^{t+1}_{{\mathrm{ins}}}}}
\newcommand{\planPaftPartI}{{{\planP}^t_{(i)}}}

\newcommand{\planPleapAftPartI}{{{\planP}^{t+1}_{(i)}}}
\newcommand{\planPleapAftPartII}{{{\planP}^{t+1}_{(ii)}}}
\newcommand{\planPbefNextStep}{{\widetilde{\planP}^{t+1}}} 
\newcommand{\planPbefGroup}[1]{{\dot{\planP}^{t+1}_{#1}}}
\newcommand{\planPaftGroup}[1]{{\ddot{\planP}^{t+1}_{#1}}}

\newcommand{\planQ}{{Q}} 
\newcommand{\barplanQ}{{\overline{\planQ}}} 
 
\newcommand{\varbackupplan}{{\calB}}
\newcommand{\backupplan}{{B}}
\newcommand{\backupplanBefStep}{{\widetilde{\backupplan}^t}} 
\newcommand{\backupplanAftArrivals}{{{\backupplan}^t}}    
\newcommand{\backupplanAftAdvMove}{{{\backupplan}^t_{{\mathrm{adv}}}}}
\newcommand{\backupplanAftInitSeg}{{{\backupplan}^{t+1}_{{\mathrm{ins}}}}}
\newcommand{\backupplanAftPartI}{{{\backupplan}^t_{(i)}}}

\newcommand{\backupplanLeapAftPartI}{{{\backupplan}^{t+1}_{(i)}}}
\newcommand{\backupplanLeapAftPartII}{{{\backupplan}^{t+1}_{(ii)}}}
\newcommand{\backupplanBefNextStep}{{\widetilde{\backupplan}^{t+1}}}
\newcommand{\backupplanBefGroup}[1]{{\dot{\backupplan}^{t+1}_{#1}}}
\newcommand{\backupplanAftGroup}[1]{{\ddot{\backupplan}^{t+1}_{#1}}}

\newcommand{\ADV}{{A}}
\newcommand{\varADV}{{\calA}}
\newcommand{\ADVbefStep}{{\widetilde{\ADV}^t}} 
\newcommand{\ADVaftArrivals}{{{\ADV}^t}}    
\newcommand{\ADVaftAdvMove}{{{\ADV}^t_{{\mathrm{adv}}}}}
\newcommand{\ADVaftInitSeg}{{{\ADV}^{t+1}_{{\mathrm{ins}}}}}
\newcommand{\ADVaftPartI}{{{\ADV}^t_{(i)}}}

\newcommand{\ADVleapAftPartI}{{{\ADV}^{t+1}_{(i)}}}
\newcommand{\ADVleapAftPartII}{{{\ADV}^{t+1}_{(ii)}}}
\newcommand{\ADVbefNextStep}{{\widetilde{\ADV}^{t+1}}} 
\newcommand{\AdvBefGroup}[1]{{\dot{\ADV}^{t+1}_{#1}}}
\newcommand{\AdvAftGroup}[1]{{\ddot{\ADV}^{t+1}_{#1}}}

\newcommand{\potentialBefStep}{{\widetilde{\Psi}^t}} 
\newcommand{\potentialAftArrivals}{{\Psi^t}} 
\newcommand{\potentialBefNextStep}{{\widetilde{\Psi}^{t+1}}} 

\newcommand{\OPT}{{\textsf{OPT}\xspace}}
\newcommand{\ALG}{{\textsf{ALG}\xspace}}

\newcommand{\DeltaPsiADV}{\Delta_{{\mathrm{adv}}}\Psi}
\newcommand{\DeltaPsiALG}{\Delta_{{\mathrm{ALG}}}\Psi}

\newcommand{\InvariantA}{\textrm{(InvA)}}
\newcommand{\InvariantB}{\textrm{(InvB)}}

\newcounter{theorem}
\theoremstyle{plain}
\newtheorem{theorem}[theorem]{Theorem}
\newtheorem{corollary}[theorem]{Corollary}
\newtheorem{lemma}[theorem]{Lemma}

\newtheorem{claim}[theorem]{Claim}

\newtheorem{observation}[theorem]{Observation}

\numberwithin{theorem}{section}
\numberwithin{equation}{section}
\numberwithin{figure}{section}

\newenvironment{bigeqn*}{\large\begin{eqnarray*}}{\end{eqnarray*}}

\newcommand{\half}{{\textstyle\frac{1}{2}}}

\begin{document}
	
\thispagestyle{empty}

\maketitle

\begin{abstract}
	In the \emph{online packet scheduling problem with deadlines}  ({\BDPS}, for short), 
	the goal is to schedule transmissions of packets that arrive over time in a network switch and need
	to be sent across a link.                           
	Each packet has a deadline, representing its urgency, and a non-negative weight, that represents its
	priority. Only one packet can be transmitted in any time slot, so
	if the system is overloaded, some packets will inevitably miss their deadlines and be dropped. 
	In this scenario, the natural objective is to compute a transmission schedule that
	 maximizes the total weight of packets that are successfully transmitted. 
	The problem is inherently online, with the scheduling decisions made without the knowledge of future packet
	arrivals.   The central problem concerning {\BDPS}, that has been a subject of   intensive study since 2001,
	is to  determine the optimal competitive ratio of online algorithms, namely the worst-case 
	ratio between the optimum total weight of a schedule (computed by an offline algorithm) 
	and the weight of a schedule computed by a (deterministic) online algorithm.

	We solve this open problem by 
	presenting a $\phi$-competitive online algorithm for {\BDPS} (where $\phi\approx 1.618$ is the golden ratio),
	matching the previously established lower bound.
\end{abstract}

\vfill
\eject

\tableofcontents

\vfill
\eject

\section{Introduction}
\label{sec: introduction}

In the \emph{online packet scheduling problem with deadlines} (abbreviated {\BDPS}),
the goal is to schedule transmissions of packets that arrive over time in a network switch and need to be sent across a link.
The switch stores incoming packets in a buffer that is assumed to have unlimited capacity. Besides its release
time $r_p$, each arriving packet $p$ has two other attributes: a deadline $d_p$, which is the last time slot when $p$ can be transmitted, and 
a non-negative weight $w_p$. The deadline of $p$ represents its urgency, while its weight represents its importance or
priority. (These priorities can be used to implement various levels of service in networks with QoS guarantees.)
Only one packet can be transmitted in any time slot, so
if the system is overloaded, some packets will inevitably miss their deadlines and be dropped. 
In this scenario, the natural objective is to compute a transmission schedule that
 maximizes the total weight of packets that are successfully transmitted.
In the literature, this problem is also occasionally referred to as
\emph{bounded-delay buffer management}, \emph{QoS buffering},
or as a job scheduling problem for unit-length jobs with release times, deadlines, and weights, 
where the objective is to maximize the weighted throughput.
In the offline setting, where the information about all packets is available in advance, 
it is easy to find an optimal schedule by representing the set of packets as a bipartite
graph, with each packet $p$ connected to all slots in its interval $[r_p,d_p]$ by an edge of weight $w_p$,
and applying any algorithm for the maximum-weight bipartite matching to this graph.

In practice, however, scheduling of packets must be accomplished online, with the scheduling decisions made without the knowledge 
of future packet arrivals. The central problem concerning {\BDPS}, that has been a subject of   intensive study since 2001,
is to determine the optimal competitive ratio of online algorithms, namely the worst-case 
ratio between the optimum total weight of a schedule (computed by an offline algorithm) 
and the weight of a schedule computed by a (deterministic) online algorithm.
 
This paper provides the solution of this open problem by establishing an upper bound of
$\phi$ on the competitive ratio for {\BDPS} (where $\phi = (1 + \sqrt{5}) / 2\approx 1.618$ is the golden ratio),
matching the previously known lower
bound~\cite{hajek_unit_packets_01,andelman_queueing_policies_03,Zhu04,chin_partial_job_values_03}.
Our $\phi$-competitive algorithm {\PlanMonotonicity} is presented in Section~\ref{sec: online algorithm}.
The basic idea underlying our algorithm is relatively simple. It is based on the concept of the \emph{optimal plan},
which, at any given time $t$, is the maximum-weight subset of pending packets that can be feasibly scheduled
in the future (if no other packets arrive); we describe it in Section~\ref{sec: plans and their properties}.
When some packet $p$ from the plan is chosen to be transmitted at
time~$t$, it will be replaced in the plan by some other packet $\rho$.  The algorithm chooses $p$ to
maximize an appropriate linear combination of $w_p$ and $w_\rho$.
Furthermore, the algorithm makes additional changes in the plan, adjusting deadlines and weights
of some packets. These changes are necessary to achieve the competitive ratio
$\phi$ and can be viewed as a subtle way to remember necessary information about the past.
While the algorithm itself is not complicated, its
competitive analysis, given in Section~\ref{sec: competitive analysis}, is quite intricate. It relies on
showing a bound on amortized profit at each step, using
a potential function, which quantifies the advantage of the algorithm over the adversary in future steps,
and on maintaining an invariant that allows us to control decreases of the potential function.
The ideas leading to the development of
algorithm~\PlanMonotonicity{} and its analysis are also outlined in SIGACT News Online Algorithms Column~\cite{vesely21_sigact_news}.


\myparagraph{Past work.}
The {\BDPS} problem was first introduced independently by Hajek~\cite{hajek_unit_packets_01}
and~Kesselman~{\etal}~\cite{kesselman_buffer_overflow_04},
who both gave a proof that the greedy algorithm (that always transmits the heaviest packet) is $2$-competitive.
Hajek's paper also contains a proof of a lower bound of $\phi\approx 1.618$
on the competitive ratio. The same lower bound  was later discovered independently by
Andelman~{\etal}~\cite{andelman_queueing_policies_03,Zhu04} and also
by Chin~{\etal}~\cite{chin_partial_job_values_03} in a different, but equivalent setting. Improving over the greedy algorithm,
Chrobak~{\etal}~\cite{chrobak_improved_buffer_04,chrobak_improved_buffer_07} gave an online algorithm with competitive ratio $1.939$.
This ratio was subsequently improved to $1.854$ by Li~{\etal}~\cite{Li_better_online_07},
and to $1.828$ by Englert and Westermann~\cite{englert_suppressed_packets_12}. This value of $1.828$,
prior to our work, has remained the best upper bound on the competitive ratio of {\BDPS} in the published literature.

Algorithms with ratio $\phi$ have been developed for several restricted variants of {\BDPS}.
Li~{\etal}~\cite{li_optimal_agreeable_05} (see also~\cite{Jez_optimal_agreeable_12}) gave a $\phi$-competitive algorithm for
the case of \emph{agreeable} deadlines, which consists of instances where the deadline ordering
is the same as the ordering of release times, i.e., $r_p < r_q$ implies $d_p\le d_q$ for any packets $p$ and $q$. Another well-studied case is that of
\emph{$s$-bounded instances}, where each packet can be scheduled in at most $s$ slots, that is
$d_p \le r_p + s-1$ for each packet $p$. 
A $\phi$-competitive algorithm for $2$-bounded instances was given by Kesselman~{\etal}~\cite{kesselman_buffer_overflow_04}.
This bound was later extended to $3$-bounded instances by Chin~{\etal}~\cite{chin_weighted_throughput_06} and to $4$-bounded
instances by B{\"o}hm~{\etal}~\cite{bohm_bounded_lookahead_16}. The work of
Bienkowski~{\etal}~\cite{bienkowski_phi-competitive_13} provides
an upper bound of $\phi$ (in a somewhat more general setting) for the case where packet weights 
increase with respect to deadlines. All these results are tight, as the proof of the lower bound of 
$\phi$ in~\cite{hajek_unit_packets_01,andelman_queueing_policies_03,Zhu04,chin_partial_job_values_03} is based on an adversary
strategy that uses only $2$-bounded instances with increasing weights, and $2$-bounded instances have the agreeable-deadline 
property.

In \emph{$s$-uniform instances}, each packet has \emph{exactly} $s$ consecutive slots where it can be scheduled.
(Obviously, such instances also satisfy the agreeable deadlines property.) The lower bound
of $\phi$ in~\cite{hajek_unit_packets_01,andelman_queueing_policies_03,Zhu04,chin_partial_job_values_03} does not apply to $s$-uniform instances; 
in fact, as shown by Chrobak~{\etal}~\cite{chrobak_improved_buffer_07}, for $2$-uniform instances, competitive ratio $\approx 1.377$ is optimal.

Randomized online algorithms for {\BDPS} have been studied as well, although the gap between the upper and
lower bounds for the competitive ratio in the randomized case remains quite large. The best upper bound is
$\approx 1.582$~\cite{bartal_online_competitive_04,chin_weighted_throughput_06,bienkowski_randomized_algorithms_11,Jez13}, and it
applies even to the adaptive adversary model. For the adaptive adversary, the best
lower bound is $\approx 1.33$~\cite{bienkowski_randomized_algorithms_11}, while for
the oblivious adversary it is $1.25$~\cite{chin_partial_job_values_03}.

Kesselman~{\etal}~\cite{kesselman_buffer_overflow_04} originally proposed the problem 
in the setting with integer bandwidth $m \ge 1$, which means that $m$ packets are sent in each step.
For an arbitrary $m$, they proved that the greedy algorithm is $2$-competitive
and that there is a $\phi$-competitive algorithm for $2$-bounded instances~\cite{kesselman_buffer_overflow_04}.
Later, Chin~\etal~\cite{chin_weighted_throughput_06} gave an algorithm with ratio 
that tends to $\frac{e}{e-1}\approx 1.582$ for $m\rightarrow \infty$.
The best lower bound for any $m$, also due to Chin~\etal~\cite{chin_weighted_throughput_06}, equals $1.25$
and holds even for randomized algorithms against the oblivious adversary. 
Observe that any upper bound for bandwidth~$1$ implies the same upper bound for an arbitrary $m$,
by simulating an online algorithm for bandwidth~$1$ on an instance where each step is subdivided
into $m$ smaller steps. Hence, our algorithm in Section~\ref{sec: online algorithm}
is $\phi$-competitive for any $m$, which improves the current state-of-the-art for any $m < 13$.

There is a variety of other packet scheduling problems related to {\BDPS}. 
The semi-online setting with lookahead was proposed in~\cite{bohm_bounded_lookahead_16}.
A relaxed variant of {\BDPS} in which only the ordering of deadlines is known but not their exact values,
was studied in~\cite{bienkowski_collecting_weighted_13}, where a lower bound higher than $\phi$ was shown.
In the FIFO model (see, for example, \cite{aiello_competitive_queue_05,kesselman_buffer_overflow_04}), 
packets do not have deadlines, but the switch has a buffer that can only hold a bounded number of
packets, and packets must be transmitted in the first-in-first-out order. 
More information about {\BDPS} and related scheduling problems can be found in
a survey paper by Goldwasser~\cite{goldwasser_survey_10}.

\section{Preliminaries}
\label{sec: preliminaries}

This section provides a formal definition of the {\BDPS} problem, along with some
useful assumptions used throughout the paper.



\myparagraph{The online {\BDPS} problem.}
The instance of {\BDPS} is specified by a set of \emph{packets}, with each packet $p$ represented by a
triple $(r_p,d_p,w_p)$, where integers $r_p$ and $d_p\ge r_p$ denote
the \emph{release time} and \emph{deadline} (or \emph{expiration time}) of $p$,  and $w_p\ge 0$ is the \emph{weight} of $p$.
(To avoid double indexing, we sometimes use notation $w(p)$ to denote $w_p$
and $d(p)$ for $d_p$.) The time is discrete, with time units represented
by consecutive integers that we refer to as \emph{time slots} or \emph{steps}.
In a feasible transmission schedule, a subset of packets is transmitted. Only one packet
can be transmitted in each time step, and each packet $p$
can only be transmitted in a slot from the interval $[r_p,d_p]$.
(Note that this interval includes slot $d_p$.)
The objective is to 
compute a schedule whose total weight of transmitted packets (also called its \emph{profit})  is maximized.

In the online variant of {\BDPS}, which is the focus of our work, the algorithm needs to
compute the solution incrementally over time. At any time step $t$, packets with release times equal to $t$
are revealed and added to the set of pending packets (that is, those that are already released,
but not yet expired or transmitted). Then the algorithm needs to choose one pending packet to transmit in slot $t$. 
As this decision is made without the
knowledge of packets to be released in future time steps, such an online algorithm
cannot, in general, be guaranteed to compute an optimal solution. The quality
of the schedules it computes can be quantified using competitive analysis.
We say that an online algorithm is \emph{$c$-competitive} if, for each instance, the
optimal profit (computed offline) is at most $c$ times the profit of the schedule computed by the online algorithm.


\myparagraph{Useful assumptions.}
The range of integers used as release times and deadlines is not important, as the whole instance can be
shifted in time without affecting the competitive ratio. However, for the sake of concreteness we can assume
that the whole instance fits in the interval $[0,\timehorizon]$, where $\timehorizon$ is some large time horizon,
and that the first step of the computation takes place at time $0$.
Throughout the paper, all time slots are tacitly assumed to be from this finite range.
(For notational reasons, we will occasionally also use slot $-1$.)

Without loss of generality, we make two simplifying assumptions:

\begin{enumerate}[label=(A\arabic*),leftmargin=3em]
	
\item \label{ass:virtualPackets}
We assume that at each step $t$ and for each $\tau\in [t, \timehorizon]$,
there is a pending packet with deadline $\tau$ (or more such packets, if needed). 
This can be achieved by releasing, at time $t$, virtual $0$-weight packets
with deadline $\tau$, for each $\tau\in [t, \timehorizon]$.

\item \label{ass:diffWeights}
We also assume that all packets have different
weights. Any instance can be transformed into an instance with distinct weights
through infinitesimal perturbation of the weights, without affecting the
competitive ratio. Thus the virtual packets from the previous assumption
have, in fact, infinitesimal positive weights, although in the calculations we treat these
values as being equal $0$. The only purpose of this assumption is to facilitate consistent tie-breaking,
in particular uniqueness of plans (to be defined shortly). 
\end{enumerate}
A careful reader may notice that assumption~\ref{ass:virtualPackets} as written reveals the value of $\timehorizon$ to the online algorithm. 
This does not affect our analysis because, for any instance, if the maximum packet deadline in this instance is $d_{\malymax}$,
then the computation of our algorithm is independent of the choice of the time horizon $\timehorizon\ge d_{\malymax}$, provided that
the ties between weights of virtual $0$-weight packets are broken in favor of earlier-deadline packets, in the sense that
those with larger deadlines are considered ``lighter''.%
\footnote{An alternative and equivalent approach would be to change the time horizon dynamically, so that it is always equal to the 
maximum deadline of already released packets. This, however, creates some minor but distracting technical difficulties in the analysis.
As yet another option, one can show that, in fact,  modifying the definition of {\BDPS} by having the  time horizon  revealed 
when the computation starts does not affect the optimal competitive ratio. 
We don't need such a strong statement, however, for our analysis.
}


\myparagraph{The golden ratio.}
The competitive ratio of our algorithm is $\phi = \half (\sqrt{5}+1) \approx 1.618$, the famous number known as the golden ratio.
Its most important property is that it satisfies the equation $\phi^2 = \phi+1$. This identity will
be frequently used in calculations, usually in the form $\frac{1}{\phi^2} + \frac{1}{\phi} = 1$. 
Another useful property is that $\sum_{i=0}^\infty \phi^{-2i} = \phi$, which follows easily from the
formula for summing a geometric sequence and the earlier-mentioned properties of $\phi$.

\section{Plans and their Properties}
\label{sec: plans and their properties}

Consider an execution of some online algorithm. At any time $t$, the algorithm will
have a set of pending packets. We now discuss properties of these pending packets
and introduce the concept of plans that will play a crucial role in our algorithm.

The set of packets pending at a time $t$ has a natural ordering, 
called the \emph{canonical ordering} and denoted $\prec$, 
which orders packets in  non-decreasing order of deadlines, breaking ties in favor of heavier packets.
(By assumption~\ref{ass:diffWeights} the weights are distinct.)
Formally, for two pending packets $p$ and $q$, define $p \prec q$ iff $d_p < d_q$ or $d_p = d_q$ and $w_p > w_q$.
The \emph{earliest-deadline packet} in some subset $X$ of pending packets is the packet
that is first in the canonical ordering of $X$.
Similarly, the \emph{latest-deadline packet} in $X$ is the last packet in the canonical ordering of $X$.


\myparagraph{Plans.}
Consider a set $X$ of packets pending at a time step $t$. $X$ is called a \emph{plan} 
if the packets in $X$ can be feasibly scheduled in future time slots $t,t+1,...$, where ``feasibly''
means that all packets in $X$ meet their deadlines.  We will typically use symbols $\planP, \planQ,...$ to denote plans.
We emphasize that in a plan we do not assign packets to time slots; that is, a plan is \emph{not} a schedule.
A plan has at least one schedule, but in general it may have many.
(In the literature, such scheduled plans are sometimes called \emph{provisional schedules}.) Using a standard
exchange argument, if $\planP$ is a plan then any schedule of $\planP$ can be converted into its \emph{canonical schedule}, 
in which the packets from $\planP$ are assigned to the slots $t,t+1,...$ in the canonical order.

For any two time slots $\sigma \ge \tau \ge t$, by 
$X[\tau, \sigma] = \braced{ j\in X\suchthat d_j \in [\tau, \sigma] }$ we denote the subset of $X$ consisting of
packets whose deadlines are in $[\tau, \sigma]$. In particular, $X[t,\tau]$ contains packets with
deadlines at most $\tau$. In a similar way, we define $X[\tau, \sigma)$, $X(\tau, \sigma]$, and $X(\tau, \sigma)$.
We also define
\begin{equation*} 
	\packetslack(X, \tau) = (\tau-t+1)  - |X[t,\tau]|.
\end{equation*}
Note that $\tau-t+1$ is the number of slots in interval $[t,\tau]$. 
For convenience, we also allow $\tau = t-1$, in which case the above formula gives us
that $\packetslack(X, t-1) = 0$. We stress that the formula for $\packetslack(X, \tau)$ also depends 
on the value of $t$. This $t$ will always be the ``current step'' under consideration and it will be
uniquely determined from context, so we do not include it as a parameter of this function, to reduce clutter.

The values of $\packetslack(X, \tau)$ are useful for determining feasibility of $X$:


\begin{observation}\label{obs: feasible = nonnegative slack}
Let $X$ be a subset of packets pending at some step $t$.
$X$ is a plan if and only if $\packetslack(X, \tau)\ge 0$ for each $\tau\ge t$.
\end{observation}


This observation is quite straightforward: If $X$ is a plan then, in any (feasible) schedule of $X$, for each $\tau\ge t$
all packets in $X[t,\tau]$ are scheduled in different slots of the interval $[t,\tau]$; 
thus $|X[t,\tau]| \le \tau-t+1$. This shows the necessity of the condition in the observation.
To show sufficiency,  assume that $\packetslack(X, \tau)\ge 0$ for each $\tau\ge t$. This implies, by simple
induction, that assigning the packets in $X$ to slots $t,t+1,...$ in the canonical order
will produce a schedule (the canonical schedule of $X$, in fact) in which all packets will meet their deadlines.

\medskip

Throughout the paper, we will tacitly assume that any plan $P$ we consider is \emph{full}, in the sense that
it contains enough packets to fill all slots between the current time $t$ and the time horizon $\timehorizon$; that is $|P[t,\timehorizon]| = \timehorizon-t+1$.
This assumption can be made without loss of generality because
there are always sufficiently many infinitesimal-weight packets pending, by Assumption~\ref{ass:virtualPackets}.
With this assumption, the concepts of tight slots, segments, etc., to be introduced below will be always well defined.

Let $\planP$ be a plan for the set of packets pending at some time $t$. We now discuss the structure of $\planP$.
(See Figure~\ref{fig:packets-plans} for an illustration.)
Slot $\tau\ge t$ is called \textit{tight} in $\planP$ if $\packetslack(\planP, \tau) = 0$.
In particular, according to our conventions above, slots $t-1$ and $T$ are tight. If the tight slots of
$\planP$ are $t_0 = t-1 < t_1 < t_2 < \cdots < t_c = T$, then for each 
$i = 1,2,...,c$
the time interval $(t_{i-1},t_i] = \braced{t_{i-1}+1, t_{i-1}+2, \dots, t_i}$
is called a \emph{segment} of $\planP$. In other words, the tight slots divide the time range
into segments, each starting right after a tight slot and ending at (and including) the next tight slot.
The significance of a segment $(t_{i-1},t_i]$ is that in any schedule of $\planP$
all packets in $\planP(t_{i-1},t_i]$ must be scheduled in this segment.
Thus, slightly abusing terminology, we occasionally think of each segment as a 
set of packets to be scheduled in this segment, namely the set $\planP(t_{i-1},t_i]$.
Within a segment, packets from $\planP$ can be permuted, although only in some restricted ways.
In particular, the first slot of a segment may contain any packet from that segment (see Observation~\ref{obs: plan properties}).  
Let $\alpha = t_1$ be the first tight slot of plan $\planP$ (we regard tight slot $t_0 = t-1$ as the ``0-th tight slot'').
The first tight slot $\alpha$ will play a special role in our algorithm, together with
the segment $[t, \alpha] = (t_0,t_1]$ of $\planP$, called the \emph{initial segment}.
Whenever we consider several plans at a time $t$ and some ambiguity arises, we will write $\alpha = \alpha_{\planP}$, to
indicate that $\alpha$ is the first tight slot for this specific plan $\planP$.

\begin{figure}[!t]
	\centering
	\includegraphics[width = 2in]{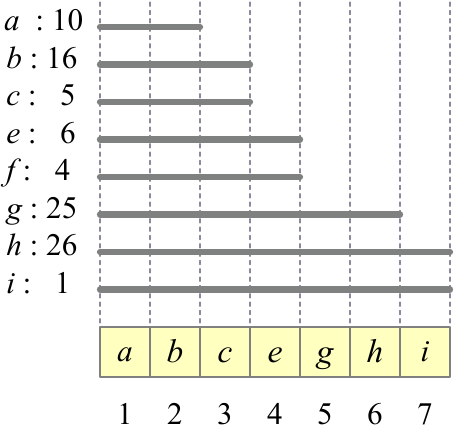} 
	\caption{An example of an instance with optimal plan $\planP = \braced{a,b,c,e,g,h,i}$ and its canonical schedule in time
		slots $1,2,...,7$. Each packet is represented by one row, that shows this packet's identifier
		and weight (separated by colon), followed by a line segment stretching from
		the current time slot $t$ (assumed to be $1$ here) to this packet's deadline. 
		Note that packet $f$ is not in $\planP$, even though it is heavier than $i$.
		The segments of $\planP$ are $(0,3]$, $(3,4]$, $(4,7]$.
		(Alternatively, in terms of packets from $\planP$, these segments are
		$\braced{a,b,c}$, $\braced{e}$, $\braced{g,h,i}$.) The values of $\minwt(\planP,\tau)$ are: 
		$5$ for $\tau\in\braced{1,2,3,4}$, and $1$ for $\tau\in\braced{5,6,7}$. 
	}
	\label{fig:packets-plans}
\end{figure} 

For a plan $\planP$ and a slot $\tau\ge t$, let $\nexttightslot(\planP, \tau)$
be the earliest tight slot $\tau'\ge \tau$ and let $\prevtightslot(\planP, \tau)$ be the latest tight slot $\tau' < \tau$.
Both notations are well-defined: that
$\nexttightslot(\planP, \tau)$ is well-defined follows from slot $T$ being tight, and 
$\prevtightslot(\planP, \tau)$ is well-defined because $t-1$ is a tight slot, according to our convention.
(Both functions also depend on $t$, but $t$ is implied, since $t$ is the time for which $P$ is defined. This
convention will apply to other notations involving plans.)

The notion that will be crucial in the design of our $\phi$-competitive algorithm
is the minimum weight of a packet that can appear in a schedule of the plan in some slot between the current time $t$ and 
a given slot $\tau$.
Naturally, the packets that are candidates for this minimum include all packets in segments ending before $\tau$, 
but also we need to include all packets in the segment of $\tau$, even those with deadlines larger than $\tau$. 
(This is because, as explained earlier, each packet in a segment can be scheduled at the beginning of that segment.)
Formally, for a plan $\planP$ at time $t$ and a slot $\tau\ge t$, define
\begin{equation*} 
	\minwt(\planP,\tau) = \min\braced{\,  w_j \suchthat j \in \planP[t,\nexttightslot(\planP,\tau)] \,}.
\end{equation*}
We will occasionally omit $\planP$ in this notation if it is understood from context.
By definition, if we fix $t$ and $\planP$ and think of $\minwt(\planP,\tau)$ as a function of $\tau$, 
then this function is monotonely non-increasing over the whole range of $\tau = t,t+1,...,T$, and is
constant in each segment (that is, all slots $\tau$ in any given segment have the same value of $\minwt(\planP,\tau)$).


\myparagraph{Optimal plans.}
Given a set of packets $\pendpackets$ pending at a given time $t$,
their plan $\planP$ is called \emph{optimal} if it has the maximum weight among all plans of these packets.
Since the collection of such plans forms a matroid (which is easy to verify using Observation~\ref{obs: feasible = nonnegative slack};
see also e.g.~\cite{kesselman_buffer_overflow_04}), the optimal plan at step $t$
can be computed by the following greedy algorithm:


\begin{algorithm}[H]
\caption{Algorithm~$\ComputePlan(\pendpackets,t)$\label{alg:greedy}}
\begin{algorithmic}[1]
	\Require{$\pendpackets$ is a set of packets pending at time $t$}     
	\State{$X \assign \emptyset$}  
	\For{each packet $j\in \pendpackets$ in order of decreasing weights}           
	      \If{$\packetslack(X\cup\braced{j},\tau)\ge 0$ for all $\tau\ge t$}
	      \State{$X \assign X\cup \braced{j}$}  
		  \EndIf
	\EndFor
	\State{$\planP\assign X$}
	\Comment{$\planP$ is the optimal plan for $\pendpackets$}
\end{algorithmic}
\end{algorithm}


Assumption \ref{ass:diffWeights} about different weights implies that the optimal plan $\planP$ computed above is unique. 
See Figure~\ref{fig:packets-plans} for an example of an optimal plan for a given set of pending packets.


\smallskip

During a computation, the set of pending packets can change. There are two
types of events that can cause this change: an arrival of a new packet and a transmission of a packet, which also involves
incrementing the current time and dropping all newly expired packets. This will also cause the optimal plan to change.
The matroid property implies that at most one packet in the optimal plan changes after each event (not counting
the transmitted packet in a transmission event).  These changes are fairly intuitive and we outline them briefly below; 
for a formal description of these changes, correctness proofs, and
figures with examples, see Appendix~\ref{app: more on properties of plans}.

In the discussion below, we assume that $t$ is the current time step and $\planP$ is the current optimal plan, right before
an arrival or transmission event. 
By $\planQ$ we will denote the optimal plan right after the event. For a slot $\tau\ge t$, whenever we refer
to the change of $\pslack(\tau)$ (that it increases, decreases, or remains the same), without specifying the plan,
we mean the change from $\pslack(\planP,\tau)$ to $\pslack(\planQ,\tau)$. We use the same convention for function $\minwt(\tau)$.


\myparagraph{Packet arrivals.} 
We first consider the event of a new packet $s$ arriving at time $t$. 
As $s$ is added to the set of pending packets, the optimal plan needs to be updated accordingly.
Define $f\in \planP$ to be the packet with $w_f = \minwt(\planP,d_s)$, that is the
lightest packet in $\planP$ with $d_f\le \nexttightslot(\planP, d_s)$.
If $w_s < w_f$, then $s$ is not added to the optimal plan, which thus stays the same.
On the other hand, if $w_s > w_f$ then $s$ is added to the optimal plan and $f$ is forced out, i.e.,
the new optimal plan is $\planQ = \planP\cup \{s\}\setminus \{f\}$.
In the latter case, it is interesting to see how the values of $\pslack(\tau)$ and the segments change:
\begin{itemize}
\item If $d_s \ge d_f$, then $\pslack(\planQ, \tau) =
\pslack(\planP,\tau) + 1$ for $\tau\in [d_f, d_s)$. Therefore, all
tight slots in $[d_f, d_s)$ in $P$ are no longer tight in $Q$ and the segments 
containing $d_f$ and $d_s$ and all segments in-between get merged into
one segment of $\planQ$. 
\item If $d_s < d_f$, then $d_f$ and $d_s$ must be in the same segment of $\planP$
(because $d_f\le \nexttightslot(\planP, d_s)$)
and $\pslack(\planQ, \tau) = \pslack(\planP, \tau) - 1$ for $\tau\in [d_s, d_f)$.
Thus, there may be new tight slots in $[d_s, d_f)$, resulting in new segments.
\end{itemize}
In both cases, the values of $\pslack(\tau)$ remain the same for other slots $\tau\ge t$;
thus other tight slots and segments do not change.
Moreover, $\minwt(\planQ, \tau) \ge \minwt(\planP,\tau)$ holds for any slot $\tau\ge t$;
this property will play a significant role in our algorithm.


\myparagraph{Transmitting a packet.}
Next, we discuss the changes of the optimal plan resulting from a
transmission event.
Throughout the paper, only packets in the optimal plan will be
considered for transmission.  The scenario is this: We consider the
set $\pendpackets$ of all packets pending at time $t$, and the optimal
plan $\planP$ for $\pendpackets$.
We choose some 
packet $p\in \planP$, transmit it at time $t$, and increment the current time to $t+1$.
The new set $\pendpackets'$ of pending packets (now at time $t+1$) is obtained from $\pendpackets$
by removing $p$ and all packets with deadline $t$.
As before, we will use $\planQ$ to denote the new optimal plan after this change, 
namely the optimal plan for $\pendpackets'$ with respect to time $t+1$.

If $p$ is from the initial segment $[t,\alpha]$ of $\planP$, then $\planQ = \planP\setminus \braced{p}$. 
In this case $\pslack(\planQ,\tau) = \pslack(\planP,\tau)-1$ for 
$\tau \in [t+1,d_p)$ and $\pslack(\tau)$ remains unchanged for $t\ge d_p$.
This implies that new tight slots may appear before $d_p$,
i.e., the initial segment may get divided into more segments. 
Furthermore, $\minwt(\tau)$ does not decrease for any $\tau\ge t+1$.

The more interesting case is when $p$ is from a later segment.
Let $\firstlightpacket$ be the lightest packet in $\planP[t,\alpha]$, i.e., in the initial segment,
and let $\rho$ be the heaviest pending packet not in $\planP$ that satisfies $d_\rho > \prevtightslot(\planP, d_p)$
(such a packet $\rho$ exists by Assumption~\ref{ass:virtualPackets}).
Using the matroid property of the feasible sets of packets at time $t+1$ and the structure of the plan,
in Appendix~\ref{app: more on properties of plans} we prove that $\planQ = \planP\setminus \{p, \firstlightpacket\} \cup \{\rho\}$.
Furthermore, in this case we have:    
\begin{itemize}
\item $\pslack(\planQ,\tau) = \pslack(\planP,\tau) - 1$ for $\tau\in [t+1, d_\firstlightpacket)$.
	There may be new tight slots in the interval $[t+1, d_\firstlightpacket)$, resulting in new segments.
\item If $d_\rho \ge d_p$, then $\pslack(\planQ,\tau) = \pslack(\planP,\tau) + 1$
for $\tau\in [d_p, d_\rho)$. Here, all segments that overlap $[d_p, d_\rho]$ are merged into one segment of $\planQ$.
\item If $d_\rho < d_p$, then
$\pslack(\planQ,\tau) = \pslack(\planP,\tau) - 1$ for $\tau\in [d_\rho, d_p)$.  (In this case, $d_p$ and $d_\rho$ must be in
the same segment of $\planP$, because $d_\rho > \prevtightslot(\planP, d_p)$.) As a result of decreasing some
$\pslack(\tau)$ values, new tight slots may appear in $[d_\rho, d_p)$, creating new segments.
\end{itemize}
For slots $\tau\ge t+1$ not covered by the cases above, the value of $\pslack(\tau)$ does not change.
Unlike for packet arrivals, after a packet transmission event some values of
$\minwt(\tau)$ may decrease, either due to $\rho$ being included in $\planQ$ or
as a side-effect of segments being merged.

\smallskip

\emparagraph{Substitute packets.}
The aforementioned updates of the plan motivate the following definition:
Let $\pendpackets$ be the set of all packets pending at time $t$ and
$\planP \subseteq \pendpackets$ be the optimal plan at time $t$. For each $j\in \planP$ we define the \emph{substitute packet of $j$}, 
denoted $\substpacket(\planP,j)$, as follows.
If $j\in \planP[t,\alpha]$, then $\substpacket(\planP,j) = \firstlightpacket$,
where $\firstlightpacket$ is the lightest packet in $\planP[t,\alpha]$.
If $j\notin \planP[t,\alpha]$, then
$\substpacket(\planP,j)$ is the heaviest pending packet $\rho \in \pendpackets \setminus \planP$
that satisfies $d_\rho > \prevtightslot(\planP, d_j)$, which exists by assumption~\ref{ass:virtualPackets}.

By definition, all packets in a segment of $\planP$ have the same substitute packet.
Also, for any $j\in \planP$ it holds that $w_j \ge w(\substpacket(\planP,j))$. This is because for
$j\in \planP[t,\alpha]$ we have $\substpacket(\planP,j) = \firstlightpacket$ and $w_j\ge w_\firstlightpacket$ (the equality will
hold only when $\firstlightpacket = j$), 
while for $j\in \planP(\alpha,\timehorizon]$ we have $d(\substpacket(\planP,j)) > \prevtightslot(\planP,d_j)$;
thus in this case the set $\planP - \braced{j} \cup \braced{ \substpacket(\planP,j) }$ is
feasible, and the optimality of $\planP$ implies that $w_j > w(\substpacket(\planP,j))$.

\section{Online Algorithm}
\label{sec: online algorithm}

This section presents our online algorithm~\PlanMonotonicity{} for {\BDPS}, starting with some intuitions and
additional notation, and followed by two examples where its competitive ratio is exactly $\phi$.

\myparagraph{Intuitions.}  
Similar to other online profit maximization problems, the main challenge in achieving a small competitive ratio for
{\BDPS} is in finding the right balance between the immediate profit and future profits.
Let $\planP$ be the optimal plan at a step $t$. Consider the greedy algorithm for $\BDPS$, which at time $t$ transmits 
the heaviest pending packet $h$, which is always in $\planP$. We focus on the case when $h$ is not in the
initial segment. (The case when $h$ is in the initial segment is different, but similar intuitions apply to it as well.)
In the next step, after transmitting $h$ but before new packets are released,
$h$ will be replaced in the optimal plan by its substitute packet $\rho_h = \substpacket(P,h)$. This packet could be very light,
possibly $w(\rho_h)\approx 0$. Suppose that there is another packet $g$ in $\planP$
with $d_g<d_h$ and $w_g\approx w_h$ whose substitute packet $\rho_g = \substpacket(P,g)$  
is quite heavy, say $w(\rho_g)\approx w_g$. Thus, instead of $h$ we can transmit $g$ at time $t$, achieving roughly
the same immediate profit as from transmitting $h$, but with essentially no decrease in the weight of the optimal plan (which
serves as a rough estimate of future profits). 
This example indicates that a reasonable strategy would be to choose a packet $p$
based both on its weight and the weight of its substitute packet.     
Following this intuition, our algorithm chooses $p$ that maximizes $w_p + \phi\cdot  w(\substpacket(P,p))$,
breaking ties arbitrarily.
The choice of the coefficients in the objective function follows the
intuition from analyzing the $2$-bounded case; see the discussion of
examples in Figure~\ref{fig:two tight examples}.

As it turns out, the above strategy for choosing $p$ 
does not, by itself, guarantee $\phi$-competitiveness. The analysis of special cases and an example
where this simple approach fails lead to the second idea behind our algorithm (we describe the example in Appendix~\ref{app: simpler_algs}). The difficulty
is related to how the values of $\minwt(\tau)$, for a fixed $\tau$, vary while the current time $t$ increases. 
We were able to show $\phi$-competitiveness of the above strategy for certain instances where $\minwt(\tau)$ monotonely
increases as $t$ grows from $0$ to $\tau$. We call this property \emph{slot-monotonicity}. To extend slot-monotonicity to
instances where it does not hold, the idea is then to simply \emph{force} it to hold
by decreasing deadlines and increasing weights of some packets in the new optimal plan. (To avoid unfairly benefiting the algorithm
from these increased weights, we will need to account for them appropriately in the analysis.)
From this point on, the algorithm proceeds using these new weights and deadlines when computing
the optimal plan and choosing a packet for transmission.


\myparagraph{Notation.} To avoid ambiguity, we will index various quantities used by the algorithm with
the superscript $t$ that represents the current time. This includes weights and deadlines of some packets, since these might change over time.
\begin{itemize}
\item We use notations $w^t_j$ and $d^t_j$ for the weight and the deadline of a packet $j$ in step $t$,
before the transmission of this step is implemented. (Our algorithm only changes weights and
deadlines when transmitting a packet, so they are not affected by packet arrivals.)
To avoid double subscripts, we occasionally write $w^t(j)$ and $d^t(j)$ instead of $w^t_j$ and $d^t_j$.
We will sometimes use notation $w^0_j$ for the original weight of packet $j$, when notation $w_j$ may lead to some ambiguity. 
We will also omit $t$ in these notations whenever $t$ is unambiguously implied from context.
\item $\planPaftArrivals$ is the optimal plan at time $t$ after all packets $j$ with $r_j = t$ arrive and
before a packet is transmitted. 
\item We write $\substpacket^t(p)$ to denote $\substpacket(\planPaftArrivals, p)$ and we adopt similar conventions
for $\minwt^t(\tau)$, $\nexttightslot^t(\tau)$, and $\prevtightslot^t(\tau)$.
\item $\alpha = \nexttightslot^t(t)$ is the first tight slot of $\planPaftArrivals$,
so $[t,\alpha]$ is the initial segment.
\item $\firstlightpacket$ is the lightest packet in $\planPaftArrivals[t,\alpha]$.
\item $\rho$, $h_i$, $\gamma$, $\tau_i$, and $k$ will refer to the packets and values chosen in Algorithm~$\PlanMonotonicity(t)$.
\end{itemize}


\begin{algorithm}[ht]
\caption{Algorithm~$\PlanMonotonicity(t)$}
\label{alg:planmonotonicity}
\begin{algorithmic}[1]
\State{transmit packet $p\in \planPaftArrivals$ that maximizes  $w^t_p + \phi\cdot w^t(\substpacket^t(p))$} \label{algLn:transmit} 
\If{$d^t_p > \alpha$} 
\Comment{``leap step''}    
	\State{$\rho \assign \substpacket^t(p)$} \label{algLn:rhoDef}
	\State{$w^{t+1}_{\rho} \assign \minwt^t(d^t_\rho)$} \label{algLn:incrWeightRho} 
	\Comment{increase $w_\rho$}
    \State{$\gamma \assign \nexttightslot^t(d^t_\rho)$ and $\tau_0 \assign \nexttightslot^t(d^t_p)$}  
    \State{$i \assign 0$ and $h_0 \assign p$}   \label{algLn:whileCycleInit}
	\While{$\tau_i < \gamma$}         \label{algLn:whileCycle}      
	      \State{$i \assign i + 1$}  
	      \State{$h_i \assign$ heaviest packet in $\planPaftArrivals(\tau_{i-1}, \gamma]$} \label{algLn:choosing_h_i}
	 	  \State{$\tau_i \assign \nexttightslot^t(d^t_{h_i})$}  
		  \State{$d^{t+1}_{h_i} \assign \tau_{i-1}$ and $w^{t+1}_{h_i} \assign \max(w^t_{h_i}, \minwt^t(\tau_{i-1}))$} \label{algLn:h_i-setWeightAndDdln}
                    \Comment{adjusting packet $h_i$}
	\EndWhile
	\State{$k\assign i$}
	\Comment{final value of $i$}
\EndIf
\end{algorithmic}
\end{algorithm}

The pseudo-code of our algorithm, called~$\PlanMonotonicity(t)$, is given in Algorithm~\ref{alg:planmonotonicity}.
For a pending packet $j$, if $w^{t+1}_j$ (resp.~$d^{t+1}_j$) is not explicitly set in the algorithm,
then its value remains the same by default, that is $w^{t+1}_j \assign w^t_j$ (resp.~$d^{t+1}_j \assign d^t_j$).

\smallskip

Let $p$ be the packet sent by $\PlanMonotonicity$ in step $t$. 
If $p$ is in the initial segment $[t,\alpha]$ of $\planPaftArrivals$, the step is called an \emph{ordinary step}. 
Otherwise (if $d^t_p > \alpha$),
the step is called a \emph{leap step}, and then $\rho = \substpacket^t(p)$ is the heaviest 
pending packet $\rho\not\in \planPaftArrivals$ with $d^t_\rho > \prevtightslot^t(d^t_p)$.
We will further consider two types of leap steps.
If $d_p$ and $d_\rho$ are in the same segment, then this step is called a \emph{simple leap step};
in that case $\gamma = \tau_0$, the while loop is not executed, and $k=0$.
If $d_\rho$ is in a later segment than $d_p$, then this step is called an \emph{iterated leap step};
in that case $\gamma > \tau_0$, the while loop is executed at least once, and $k>0$.

As all packets in the segment of $\planPaftArrivals$ that contains $p$ have the same substitute packet $\substpacket^t(p)$,
$p$ must be the heaviest packet in its segment. Furthermore, $p$ is not too light
compared to the heaviest pending packet $h$; specifically, we have that
$w_p \ge w_h/\phi^2$. Indeed, as mentioned earlier, we have $w_p \ge w(\substpacket^t(p))$.
It follows that $\phi^2 w_p = w_p + \phi w_p 
			\ge w_p + \phi\cdot w(\substpacket^t(p))
			\ge w_h + \phi\cdot w(\substpacket^t(h))
			\ge w_h$, where the second inequality follows by the choice
			of $p$ in line~\ref{algLn:transmit} of Algorithm~$\PlanMonotonicity(t)$.


\myparagraph{Slot-monotonicity.}
Our goal is to maintain \emph{the slot-monotonicity property}, i.e., to ensure that for any fixed slot $\tau$ the value of
$\minwt^t(\tau)$ does not decrease as the current time $t$ progresses from $0$ to $\tau$. 
For this reason, we need to increase the weight of the substitute packet $\rho$ in each leap step (as $w^t_\rho < \minwt^t(d^t_\rho)$),
which is done in line~\ref{algLn:incrWeightRho}.
(To maintain Assumption~\ref{ass:diffWeights}, we add an infinitesimal to the new weight of $\rho$.)
For the same reason, in the iterated leap step, we also need to adjust the deadlines and weights of the
packets $h_i$, which is done in line~\ref{algLn:h_i-setWeightAndDdln}.
The deadlines of $h_i$'s are decreased to make sure that the segments between 
$\beta = \prevtightslot^t(d^t_p)$ and $\gamma$ do not merge (as merging could cause a decrease of
some values of $\minwt^t(\tau)$). These deadline changes can be thought of as a sequence of substitutions, where
$h_1$ replaces $p$ in the segment of $\planP$ ending at $\tau_0$, $h_2$ replaces $h_1$, etc., and
finally, $\rho$ replaces $h_k$ in the segment ending at $\gamma$.  
We sometimes refer to this process as a~``shift'' of the $h_i$'s\footnote{
	While this shift of the $h_i$'s may seem unnecessarily involved,
	we show in Appendix~\ref{app: simpler_algs} that two natural simpler variants of $\PlanMonotonicity$
	are \emph{not} $\phi$-competitive, even though they maintain the slot-monotonicity property.
}. See Figure~\ref{fig:shift} for an illustration.
Then, if the weight of some $h_i$ is too low for its new segment, it is increased to 
match the earlier minimum of that segment, that is $\minwt^t(\tau_{i-1})$.
(Again, to maintain Assumption~\ref{ass:diffWeights}, we add an infinitesimal to the new weight of $h_i$.)

In ordinary steps, the algorithm does not make any weight and deadline changes. Thus, in such steps the optimal
plan changes as described earlier in Section~\ref{sec: plans and their properties} and the slot monotonicity
property is preserved. (See Appendix~\ref{app: more on properties of plans} for formal proofs.)

In leap steps, the algorithm modifies weights and deadlines of some packets, and thus the
discussion from Section~\ref{sec: plans and their properties} does not apply directly.  
The changes in the optimal plan after a leap step are elaborated in detail in 
Lemma~\ref{lem:leapStep} in Appendix~\ref{app: leap step of algorithm planmonotonicity}. 
We briefly summarize them here.
Recall that in a leap step the packet $p$ is in some non-initial segment.
Let $\planP = \planPaftArrivals$ and let $\planQ$ be the optimal plan after $p$ is transmitted, the time is incremented to $t+1$, and
weights and deadlines are changed (according to lines~\ref{algLn:rhoDef}-\ref{algLn:h_i-setWeightAndDdln} in the algorithm).
Let also $\barplanQ$ be the intermediate optimal plan after $p$ is transmitted and the time is incremented, 
but before the algorithm adjusts weights and deadlines. As discussed in Section~\ref{sec: plans and their properties}, 
this plan is $\barplanQ = \planP \setminus \braced{p, \firstlightpacket} \cup \braced{\rho}$, where $\rho = \substpacket^t(p)$.

In a simple leap step, only the weight of $\rho$ is modified.
Increasing the weight of a packet in the optimal plan does not affect its
optimality and thus $\planQ = \barplanQ$. Furthermore, no segments are
merged, i.e., any tight slot of $\planP$ is tight in $\planQ$ as well.

In an iterated leap step, by the choice of $p$, the definition of
$h_i$'s in line~\ref{algLn:choosing_h_i}, and the while loop condition
in line~\ref{algLn:whileCycle}, we have that $w^t_p = w^t_{h_0} >
w^t_{h_1} > w^t_{h_2} > \cdots > w^t_{h_k} > w^t_\rho$ and that $h_k$
is in the segment of $\planP$ ending at $\gamma$, that is
$\prevtightslot^t(d^t_\rho) < d^t_{h_k} \le \gamma$.  As in the simple
leap step, increasing the weight of a packet does not affect the
optimality of a plan.  Moreover, a careful analysis of the changes of
$\pslack()$ values yields that decreasing the deadlines of $h_1, h_2,
\ldots, h_k$ (in line \ref{algLn:h_i-setWeightAndDdln}) does not
change the optimal plan, so we can conclude that $\planQ = \barplanQ$
holds in an iterated leap step as well.
The decrease of the deadlines of $h_i$'s also ensures that
no segments are merged.  (See Appendix~\ref{app: leap step of
algorithm planmonotonicity} for complete proofs.)

The property that no segments are merged, together
with the increase of the weights, allows us to prove that
$\minwt^t(\tau)$ does not decrease for any $\tau$ even in a leap
step. This slot monotonicity property is summarized in the lemma
below, whose proof follows directly from
Lemma~\ref{lem:plan_update-arrival}(c),
Lemma~\ref{lem:plan-update_transmission-1stSegment}(c), and
Lemma~\ref{lem:monotonicity-LeapStep}.


\begin{figure}[!t]
\centering
\includegraphics[width=\textwidth]{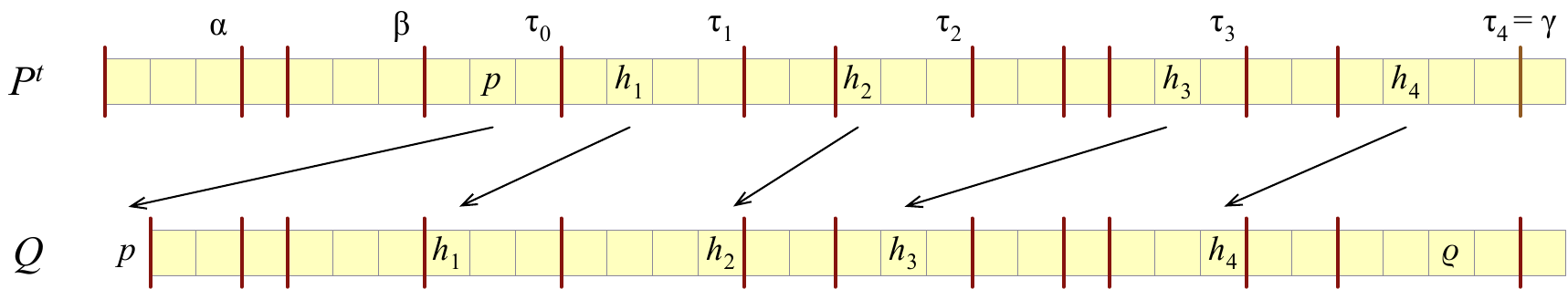}
\caption{An illustration of the shift of packets $h_1, \dots, h_k$ (for $k=4$)
in lines \ref{algLn:whileCycleInit}-\ref{algLn:h_i-setWeightAndDdln} in an iterated leap step.
In this figure, $\beta = \prevtightslot^t(d^t_p)$ and
$\planQ$ is the plan right after $p$ is transmitted and the current time is incremented to $t+1$
(but before packets released at time $t+1$ are taken into account). 
Both plans are represented by their canonical schedules. Vertical dark-red lines separate the segments of the plans.
}
\label{fig:shift}
\end{figure}


\begin{lemma}\label{lem:monotonicityForALG}
Let $\planP$ be the current optimal plan in step $t$ just before an event of either arrival
of a new packet or transmitting a packet (and incrementing the current time),
and let $\planQ$ be the plan after the event. Then $\minwt(\planQ, \tau)\ge \minwt(\planP, \tau)$ for any $\tau > t$,
and also for $\tau=t$ in case of a packet arrival. Hence, in the computation of Algorithm~$\PlanMonotonicity$,
for any fixed $\tau$, the function $\minwt^t(\tau)$ is non-decreasing in $t \in [0,\tau]$.
\end{lemma}


\begin{figure}[!ht]
\centering
\includegraphics[width = 5.5in]{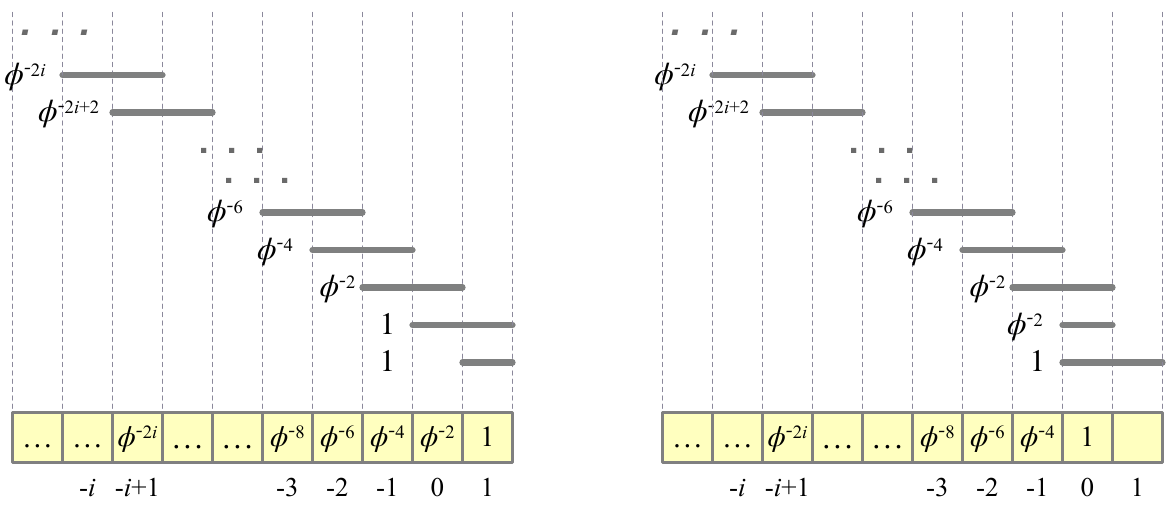}
\caption{Two examples where the competitive ratio of Algorithm~$\PlanMonotonicity$ is $\phi$. Packets
are identified by their weight. We show the schedules of~$\PlanMonotonicity$ at the bottom and time slot indices below the schedules.
}
\label{fig:two tight examples}
\end{figure}


\myparagraph{Tight examples.}
To illustrate the computation of Algorithm~$\PlanMonotonicity$, we now
give two examples of 2-bounded instances where the competitive ratio is exactly $\phi$. 
They also give some intuition for the
choice of coefficients in line~1 of Algorithm~$\PlanMonotonicity$, as the
chosen coefficients exactly balance the competitive ratio for these
two examples. We note that, interestingly, there is another
combination of the coefficients that gives ratio $\phi$ for the $2$-bounded case, namely
$\phi \cdot w_p + w(\substpacket(P,p))$; we do not know if these
coefficients can yield an optimal algorithm for the general
case.

Before we examine these examples, consider the instance where we have two packets pending at some step $t$, one tight
packet $q$ with $d_q = t$ and $w_q = c$, and one packet $r$ with $d_r = t+1$ and $w_r = \phi^2c$. The optimal plan is $\planP^t = \braced{q,r}$
and $\alpha = t$.
In this case we have a tie in line~1 of Algorithm~$\PlanMonotonicity$: For $p=q$ the substitute packet is $\ell = q$ and the
objective value is $w_q + \phi\razy w_q = \phi^2\razy c$, and for $p=r$ the substitute packet has weight $0$ and the 
objective value is $w_r + \phi\razy 0 = \phi^2\razy c$ as well. In the examples below, we will assume that we can
break this tie either way. This can be accomplished by infinitesimally perturbing the weights, without affecting the
competitive ratio.

The two examples are shown in Figure~\ref{fig:two tight examples}.
For simplicity, in these examples packets are identified by their weight and we allow negative-valued time steps. 
In both examples, the instance involves a sequence of packets, where at time $-i$, for $i\ge 0$, a packet of weight
$\phi^{-2i}$ and deadline $-i+1$ is released. For simplicity, we will think of this sequence as starting at $-\infty$.

In the instance on the left, we have one additional tight packet of weight $1$ released at time $1$.
In this instance we break the ties so that the algorithm schedules each packet $\phi^{-2i}$ 
at time $-i+1$ (thus, at each step $-i\le 0$, we transmit the tight pending packet).
As a result, the algorithm will be able to transmit only one packet of weight $1$, and
its profit will be $\sum_{i=0}^\infty \phi^{-2i} = 1/(1-\phi^{-2}) = \phi$.
The optimum solution schedules all packets, so its profit is $1+ \sum_{i=0}^\infty \phi^{-2i} = \phi^2$.

In the instance on the right, we have one additional tight packet of weight $\phi^{-2}$ released at time $0$.
Here, we also break the ties in favor of the tight pending packet until time $-1$, but at time $0$ the algorithm transmits the
packet of weight $1$. The algorithm's profit will be $\sum_{i=0}^\infty \phi^{-2i} - \phi^{-2} = \phi - \phi^{-2} = 1 + \phi^{-1} - \phi^{-2} = 2/\phi$.
The optimum solution schedules all packets, so its profit is $\sum_{i=0}^\infty \phi^{-2i} + \phi^{-2} = 2$.


\myparagraph{Comparison to previous algorithms.} 
Our algorithm shares some broad features with known algorithms in the literature.
In fact, any online algorithm (with competitive ratio below $2$) needs to capture the tradeoff between weight and urgency when
transmitting a packet, so some similarities between these algorithms are inevitable. As we found
in our earlier attempts, however, the exact mechanism of formalizing this tradeoff is critical, and minor tweaks 
can dramatically affect the competitive ratio.

Some prior algorithms used the notion of the \emph{optimal provisional schedule},
which coincides with our concept of canonically ordered optimal plan; recall that in our paper
a plan is defined as a feasible \emph{set} of pending packets rather than their particular schedule. 
For example, the $\phi$-competitive algorithm \textsc{MG} for instances with agreeable deadlines
by Li \etal~\cite{li_optimal_agreeable_05} (see also~\cite{Jez_optimal_agreeable_12}) transmits
packets from the optimal plan only, either the heaviest packet $h$ or the earliest-deadline
packet $e$. (Strictly speaking, this is true for the simplified variant of \textsc{MG}~\cite{Jez_optimal_agreeable_12}, 
whereas the original version of \textsc{MG} transmitted either $e$ or another sufficiently heavy packet from the optimal plan~\cite{li_optimal_agreeable_05}.)
The same authors~\cite{Li_better_online_07} later designed a modified algorithm called \textsc{DP} 
that achieves competitive ratio $3/\phi\approx 1.854$ for arbitrary instances.

Our approach is more similar to that of Englert and Westermann~\cite{englert_suppressed_packets_12},
who designed a $1.893$-competitive memoryless algorithm and an improved 
$1.828$-competitive variant with memory.
Both their algorithms are based on the notion of \emph{suppressed packet} \textsf{supp}($\planP,p$),
for a packet $p$ in the plan $\planP$, which, in our terminology, 
is the same as the substitute packet $\substpacket(\planP,p)$ \emph{if} $p$ is not in the initial segment.
However, the two concepts differ for packets $p$ in the initial segment.
The memoryless algorithm in~\cite{englert_suppressed_packets_12} identifies a packet $m$ of maximum
``benefit'', which is measured by an appropriate linear combination of $w_m$ and $w(\textsf{supp}(\planP,m))$,
and sends either $m$ or $e$ (the earliest-deadline packet in the optimal plan), based on the relation between
$w_e$ and the benefit of $m$.  The algorithm with memory in~\cite{englert_suppressed_packets_12} extends
this approach by comparing $m$'s benefit to $e$'s ``boosted weight''
$\max(w_e, \delta(t))$, where $t$ is the current step and $\delta(\tau)$
is the maximum value of $\minwt(\planP^{t'},\tau)$ over $t'<t$.  
We remark that considering this value of $\delta(\tau)$
takes into account the previous maximal value of $\minwt(\tau)$,
however, it does not prevent actual decreases of $\minwt(\tau)$.

Our algorithm involves several new ingredients that are critical to establishing
the competitive ratio of $\phi$. First, our analysis relies on full characterization of the
evolution of the optimal plan over time, in response to packet arrivals and transmission events.
This characterization is sketched 
in Section~\ref{sec: plans and their properties} and formally treated in Appendix~\ref{app: more on properties of plans}.
Second, we introduce a new objective function  $w^t_p + \phi\cdot w^t(\substpacket^t(p))$
for selecting a packet $p$ for transmission. This function is based on                                       
a definition of substitute packets, $\substpacket(\planP,p)$, that accurately reflects the
changes in the optimal plan following transmission events, including the case when $p$ is in the initial segment.   
Third, we introduce the concept of slot monotonicity, and devise a way for the algorithm to
maintain it over time, using adjustments of weights and deadlines of
the packets in the optimal plan.
This property is very helpful in keeping track of the optimal profit. Last but not least, 
we develop several tools used in the amortized analysis of Algorithm~$\PlanMonotonicity(t)$,
including the concepts of a backup plan, an adversary's stash, and
a novel potential function that captures the relative ``advantage'' of the algorithm over the adversary 
in terms of procuring future profits. 

\section{Competitive Analysis}
\label{sec: competitive analysis}

Let \ALG{} be the schedule computed by $\PlanMonotonicity$ for the instance of {\BDPS} under consideration,
 and let \OPT{} be a  fixed optimal schedule for this instance. (Actually,
\OPT{} can be any schedule for this instance.) For any time step $t$, by $\ALG[t]$ and $\OPT[t]$ we
will denote packets scheduled by \ALG{} and \OPT, respectively, in slot $t$.
(By assumption~\ref{ass:virtualPackets}, we can assume that $\ALG[t]$ and $\OPT[t]$ are defined for all steps $t$.)
Our overall goal is to show that $\phi\cdot w^0(\ALG) \ge w^0(\OPT)$, where, as defined earlier,
we use notation $w^0_j$ for the original weight of packet $j$.


\myparagraph{Notational convention.}
The optimal plan changes in the course of the algorithm's run, as a result of new packets arriving or of packets being transmitted.
In some contexts, it is convenient to think of the current optimal plan as
a dynamic set, which we will denote by $\varplanP$. When more formal treatment is needed, we will use
letters $\planP, \planQ, \dots$, often with appropriate subscripts or superscripts,
to denote the ``snapshot'', i.e., the current contents, of the optimal plan before or after a particular change.
Namely, as already defined in Section~\ref{sec: online algorithm}, $\planPaftArrivals$
is the snapshot of the optimal plan after all packets arrive and before a packet is transmitted in step $t$.
Furthermore, by $\planPbefStep$ we denote the optimal plan before any packet arrives
in step $t$. Thus, as a result of transmitting a packet in step $t$,
the optimal plan changes from $\planPaftArrivals$ to $\planPbefNextStep$. (See Figure~\ref{fig:notation for plans}.)
We define more snapshots in the analysis of a transmission event, which will be split into more substeps.
For clarity, the superscript of a particular snapshot contains the time index with respect to which
the optimal plan is computed (unless it is clear from the context).


\begin{figure}[!ht]
\centering
\includegraphics[width=3in]{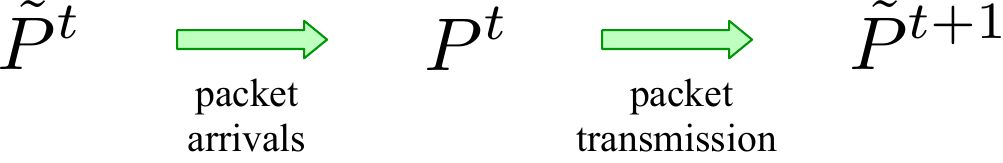}
\caption{Notation for snapshots of the optimal plan $\varplanP$.
}
\label{fig:notation for plans}
\end{figure}


The same conventions apply to other subsets of pending packets used in the analysis: $\varADV$ and $\varbackupplan$, which
we will define shortly. In general,
we think of such sets of packets as dynamic sets that change over time as
new packets arrive, the algorithm transmits packets, and as we adjust the contents
of the sets in the analysis.  As such, the dynamic sets are denoted with calligraphic letters.
Formal analysis requires that we refer to appropriate snapshots of these sets before and after the change under consideration;
such snapshots are denoted with italic letters, typically with subscripts or superscripts.


\myparagraph{Amortized analysis.}
We bound the competitive ratio via amortized analysis, using a combination of three accounting mechanisms:
\begin{itemize}
\item 
We use a potential function, which quantifies the advantage of the algorithm over the adversary in future steps. 
This potential function is defined in Section~\ref{subsec: backup plan and the potential function}.
\item 
In leap steps, when the algorithm increases weights of some packets
(the substitute packet and possibly some $h_i$'s), we charge it a ``penalty'', by subtracting
the total weight increase from its credit for the step.
\item 
The optimal profit of $ w^0(\OPT)$ is amortized over all steps using two
accounting tools: an ``adversary stash'' $\varADV$ and values of the $\minwt()$ function.
These techniques are introduced in Section~\ref{subsec: adversary stash}.
\end{itemize}
See Section~\ref{subsec: overview of the analysis} for an overview of our
analysis, showing how these three techniques can be combined into
a formal proof of $\phi$-competitiveness of Algorithm~$\PlanMonotonicity$.
Then we give the analysis of packet arrival events in Section~\ref{subsec: arrival of a packet}
and of packet transmission events in Section~\ref{subsec: transmitting a packet}.


\subsection{Adversary Stash}
\label{subsec: adversary stash} 

In our analysis, we need a mechanism for keeping track of the
adversary's future profit associated with the packets that have already been released. 
A natural candidate for such mechanism would be the set of packets that are ``pending'' for $\OPT{}$, namely the packets 
scheduled in $\OPT{}$ that have already been released but not yet transmitted by $\OPT{}$. This simple definition, however,
does not quite work for our purpose, in part because Algorithm~$\PlanMonotonicity$ 
modifies the weights and deadlines of packets that are pending for the adversary.

Instead of $\OPT{}$, we will use two other concepts which can be defined in terms of packets that are
pending for the algorithm. 
The first one, called the \emph{adversary stash} and denoted $\varADV$, is used to keep
track of the adversary's packets that are in the current optimal plan $\varplanP$; that is,
$\varADV$ is a subset of $\varplanP$ (see invariant~{\InvariantA} below). $\varADV$ is a dynamic set of packets 
scheduled in some slots in $\braced{t,t+1,\dots }$, where $t$ is the current time step. For $\tau\ge t$,
by $\varADV[\tau]$ we denote the packet scheduled in $\varADV$ in slot $\tau$; if there is no packet, $\varADV[\tau]$ is undefined.
(Abusing notation, we will use $\varADV$ to denote both the set of packets in the adversary stash and their schedule.) 

The adversary stash evolves over time, partly in response to new events and partly as a result of modifications
performed in the course of our analysis. We very briefly outline this process now; a more comprehensive summary 
is presented later in this section, with all details given during the analysis of packet arrival events 
in Section~\ref{subsec: arrival of a packet} and packet transmission events in Section~\ref{subsec: transmitting a packet}.
Initially, $\varADV$ is empty, and whenever a packet $j$ such that $j\in \OPT$ arrives, 
we add $j$ to $\varADV$ at the same slot that $j$ occupies in $\OPT{}$ provided that $j$ is added to $\varplanP$.
Sometimes in our analysis of packet arrival and transmission events, as  $\varplanP$ changes,
we may have to modify $\varADV$ by removing or replacing packets, in which case the adversary is appropriately compensated.
(As a result of such changes, for some $\tau\ge t$, $\varADV[\tau]$ may differ from $\OPT[\tau]$, i.e.,
the former may be empty or contain a different packet than the latter, even if the latter contains a packet released at or before time $t$.)
Finally, when we analyze the transmission event at a time step $t$ and $\varADV[t]$ is defined (i.e., non-empty), we will 
remove packet $\varADV[t]$ from $\varADV$ and credit the adversary with the (current) weight of $\varADV[t]$.

The second accounting mechanism deals with the packets in $\OPT{}$ that are not in $\varplanP$. It turns out that
this can be done without explicitly keeping track of such packets. Consider 
a \emph{pending} packet $q$ that is scheduled in $\OPT$ in a future step $\tau > t$. 
If  $q$ is not in $\varplanP$ then its weight is upper bounded by $\minwt(\varplanP, \tau)$.
Since $\minwt(\varplanP, \tau)$ for a fixed $\tau$ does not decrease (by Lemma~\ref{lem:monotonicityForALG}), 
its weight will be bounded by $\minwt(\planP^\tau, \tau)$ when the time reaches $\tau$.
When it happens, we will allow the adversary to obtain a profit of $\minwt(\planP^\tau, \tau) \ge w^\tau_q$. While this may
seem generous, it does not affect the competitive ratio.
The intuition is that  in each step $\tau$ the adversary can always issue a tight packet of weight just below
$\minwt(\planP^\tau, \tau)$, and this does not change the behavior
of the algorithm as such a packet is not in $\planP^\tau$ and cannot be used as a substitute packet
due to having deadline at $\tau$.  (Section~\ref{subsubsec: adversary step} gives a complete analysis.)

Overall, at any time $t$, the adversary can receive amortized profit, also called her \emph{credit}, in three ways. 
Her credit for transmitting a packet is either $w( \varADV[t] )$,
if $\varADV[t]$ is defined, or $\minwt(\planP^t,t)$ otherwise. In addition, the adversary
receives an appropriate compensation when we decrease the total weight of packets in $\varADV$.
We describe the adversary profit amortization more precisely at the end of this section.


\myparagraph{Adversary stash invariant.}
The following invariant, 
maintained throughout the analysis,
captures properties of the adversary stash $\varADV$ that will be
crucial for our argument (see also Figure~\ref{fig: sets in analysis}): 

\begin{description}
\item{\textbf{\InvariantA}}
For any time $t$ and any snapshot $\ADV$ of the adversary stash $\varADV$ at time $t$,
$\ADV$ contains only packets from the current optimal plan $P$, i.e., $A\subseteq P$, and each 
packet $g\in A$ is scheduled in $\ADV$ in a slot in interval $[t, d^t_g]$.   
\end{description}

For $\varADV$ we adopt the same notation of snapshots as for $\varplanP$, namely,
$\ADVbefStep$ is the adversary stash before any packet arrives in step $t$
and $\ADVaftArrivals$ is the adversary stash after all packet arrivals and before a packet is transmitted in step $t$.
We ensure that invariant~{\InvariantA} is preserved after each packet arrival
and after each transmission event, possibly by changing the adversary stash.


\begin{figure}[!ht]
\begin{center}
	\includegraphics[width=4.1in]{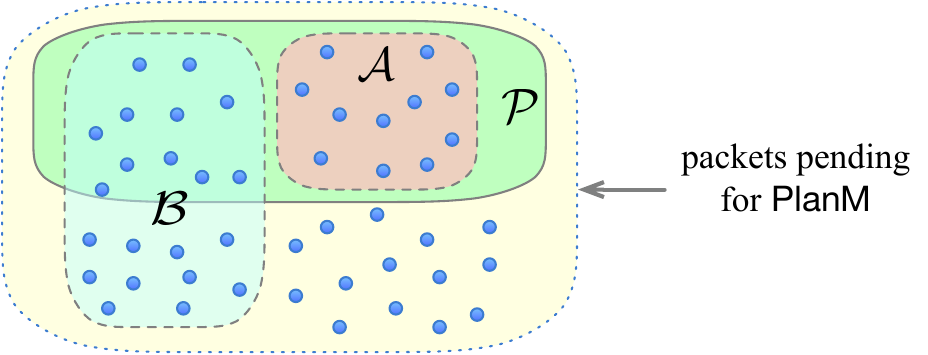}
	\caption{The sets of packets in the competitive analysis.
	The backup plan $\varbackupplan$ is introduced in Section~\ref{subsec: backup plan and the potential function}.
	It has the property that  $\varplanP \cap \varbackupplan = \varplanP\setminus \varADV$.
	}
	\label{fig: sets in analysis}
\end{center}
\end{figure}


\myparagraph{Modifications of the adversary stash.}
We now overview the principles guiding the maintenance of the adversary stash $\varADV$. These
principles are important in understanding the details of the analysis given in 
Sections~\ref{subsec: overview of the analysis}-\ref{subsec: transmitting a packet}.
Let us fix some slot $\tau$ of $\varADV$. We describe all possible changes that the packet in slot $\tau$ of $\varADV$ can
undergo in the course of the analysis and explain
how we compensate the adversary for any such change, so that the total adversary credit from
slot $\tau$ is at least $w^0(\OPT[\tau])$, as needed.


\emparagraph{Adding packet $\OPT[\tau]$ to $\varADV$.}
When packet $j = \OPT[\tau]$ arrives at a time $t$, if $j$ is added to the optimal plan $\varplanP$ then
	we also add it to $\varADV$ in slot $\tau$. Otherwise, the slot $\tau$ remains empty in $\varADV$ all the time. 
	In either case, the adversary does not get any credit for this packet at this time.
	(The credit of at least $w^0_j$ from slot $\tau$ will be awarded to the adversary later in the course of computation, possibly in smaller chunks.
	The strategy for amortizing this credit will be described shortly.)


\emparagraph{Replacing packet $\varADV[\tau]$.}
Replacement of packets in $\varADV$ may occur in an iterated leap step when, under 
some circumstances, we  replace a packet $h_i\in \varADV$ by $h_{i+1}$,
which is always lighter than $h_i$ (see Lemma~\ref{lem:leapStep}(b)). To compensate her
for this replacement, the adversary obtains credit equal to $w^t_{h_i} - w^{t+1}_{h_{i+1}}$,
which is always non-negative.

To preserve invariant~{\InvariantA}, when we make packet replacements, we need to make sure
that no packet in $\varADV$ is scheduled after its current deadline, which requires some care for
those packets whose deadlines were decreased. 
We also need to avoid adding into $\varADV$ a packet that is already in $\varADV$.
Thus, before a packet $h_i\in \varADV$ gets replaced by $h_{i+1}$,
we remove $h_{i+1}$ from its slot in $\varADV$ if it already belongs to this set.


\emparagraph{Removing packet $\varADV[\tau]$ before step $\tau$.}
As mentioned above, in some cases we remove a packet $q = \varADV[\tau]$ from $\varADV$
even though the current time has not reached $\tau$ yet.  This is done in particular
if $q$ is no longer in $\varplanP$, due to being ousted or transmitted by the algorithm.
However, in order to preserve certain invariants, we may also remove $q$ from $\varADV$ even if it remains in $\varplanP$. 
If we remove packet $q = \varADV[\tau]$ from $\varADV$ before step $\tau$, 
to compensate the adversary for this change, we give her the credit 
whose value is at least the difference (if positive) between the current weight of $q$ and the current value of
$\minwt(\varplanP, \tau)$; the precise formula for this credit will be given below.
Importantly, once we remove $\varADV[\tau]$ from $\varADV$, $\varADV[\tau]$ will remain empty forever.


\emparagraph{Removing packet $\varADV[\tau]$ at time $\tau$.}
When processing the transmission event at time $\tau$, if there is a packet $j = \ADV^\tau[\tau]$
in slot $\tau$, we remove $j$ from $\varADV$ and the adversary gains
its \emph{current} weight, i.e., it obtains credit of $w^\tau_j$.
It follows that $\varADV$ is empty after the last (transmission) event.


\emparagraph{Weight and deadline changes.}
Increasing the weight of the substitute packet $\rho$ in a leap step (line~4 of Algorithm~\PlanMonotonicity)
does not affect $\varADV$, as $\rho$ is not in $\varplanP$ and thus not in $\varADV$, by invariant~{\InvariantA}.
However, the algorithm also changes the weights and/or deadlines of some packets $h_i$ in an iterated
leap step (line~11 of Algorithm~\PlanMonotonicity), and these packets are in $\varplanP$, so $\varADV$ might be affected too.
To address this, in the analysis of an iterated leap step, we remove or replace \emph{all} packets $h_i$
that are in $\varADV$ and the adversary gets credit based on their \emph{old} weights.
Some of the packets $h_i$ may be reinserted later in the analysis of the same step, but 
they are always reinserted with their new weights and deadlines.
It follows that the weights of all packets in $\varADV$ are always current.


\myparagraph{Amortization of the adversary profit.}
We now describe how the adversary profit of $w^0(\OPT)$ is amortized. As already mentioned earlier, 
when a packet $j = \OPT[\tau]$ arrives at a time $t'$, the adversary credit for $j$ may be awarded 
in smaller payments in steps $t',t'+1,...,\tau$. All the payments except the last one are called \emph{adjustment credits}
and the last one, in step $\tau$, is called the \emph{transmission credit}. The formulas for these credits are given below.

We note that the adversary credit when processing the arrival of a packet $j$ is always $0$.
This follows from two properties of the changes in $\varADV$: First, as explained above, the adversary does not receive credit for the
arrival of $j$, whether it is added to $\varADV$ or not. Second, the only other change of $\varADV$ associated with
the arrival of $j$ and the resulting modification of $\varplanP$
 may be a removal from $\varADV$ of some packet $\varADV[\tau]$ of weight at most $\minwt(\varplanP, \tau)$,
for which the adversary does not get any credit (see below and Section~\ref{subsec: arrival of a packet}).

\smallskip

In each step $t$, we define \emph{the adversary credit for step $t$}, denoted $\advcredit^t$, as the sum of the following two values:
\begin{itemize}
\item {\emph{Transmission credit:}} This is the credit given to the adversary for her packet in slot $t$.
	More precisely, it is defined by
	\begin{equation*}
	\advcredit^t_t \;=\; \left\{\begin{array}{lcl}
					w^t(\ADVaftArrivals[t])
							& \quad& \text{if $\ADVaftArrivals[t]$ is defined (not empty)}
							\\
					\minwt(\planPaftArrivals, t) 	& \quad& \text{otherwise}
						\end{array}
						\right.
	\end{equation*}
	Recall that $w^t(\ADVaftArrivals[t])$ represents the weight of $\ADVaftArrivals[t]$ at time $t$, if $\ADVaftArrivals[t]$ is defined.
\item {\emph{Adjustment credit:}} 
This is the credit that the adversary receives as compensation for modifications in $\varADV$.
	Namely, it is the sum of the following credits, one for each adjustment of $\varADV$ performed in the analysis of step $t$:
	\begin{itemize}
	\item{\emph{Adjustment credit for weight decreases:}} The difference between previous and new weight of $\varADV[\tau]$, for each slot $\tau\ge t$
		where packet $\varADV[\tau]$ is replaced by a lighter packet.
	\item{\emph{Adjustment credit for packet removal:}} The value of $w^t_q - \minwt(\planPaftArrivals, d^t_q)$ for each packet $q = \varADV[\tau]$ removed from $\varADV$.
	\end{itemize}
		Before proceeding, we make two simple but important observations. First,
		all adjustment credits are non-negative. This is obvious for adjustment credits for weight decreases.
		As for the second type of adjustments, consider a packet
		$q = \varADV[\tau]$ getting removed from $\varADV$ at time $t$. By invariant~{\InvariantA}, we have that $q\in\planPaftArrivals$.
		Then the definition of function $\minwt()$ implies that  $w^t_q \ge \minwt(\planPaftArrivals, d^t_q)$; in other words
		the adjustment credits for packet removals are also non-negative.
		
		Second, observe that the formula above for the adjustment credit for packet removals is an upper bound on
		the actual loss of adversary's future profit due to the removal of $q=\varADV[\tau]$ from $\varADV$, which
		equals $\max\bracedintext{w^t_q - \minwt(\planP^\tau, \tau) ,0}$ (as the transmission credit in step $\tau$ is $\minwt(\planP^\tau, \tau)$).
		Indeed, we have $\minwt(\planPaftArrivals, d^t_q)\le w^t_q$ and, since $\tau \le d^t_q$, also $\minwt(\planPaftArrivals, d^t_q) \le \minwt(\planPaftArrivals, \tau)\le \minwt(\planP^\tau, \tau)$ by the slot-monotonicity property (Lemma~\ref{lem:monotonicityForALG}).
		Therefore $w^t_q - \minwt(\planPaftArrivals, d^t_q) \ge \max\bracedintext{w^t_q -  \minwt(\planP^\tau, \tau) ,0}$.
		We use the value of $w^t_q - \minwt(\planPaftArrivals, d^t_q)$ here because it is sufficient for the analysis and
		easier to work with.

	At each step $t$, these credit adjustments (if any) may be performed in multiple substeps, with each substep involving
	processing some time interval $(\zeta,\eta]$. We will use notation $\advcredit^t_{(\zeta,\eta]}$ for the
	adjustment credit resulting from modifications in the substep corresponding to time interval $(\zeta,\eta]$.
\end{itemize}


\begin{claim}\label{clm: adversary credits}
The total adversary credit for all steps covers $w^0(\OPT)$, that is,
\begin{equation}\label{eqn: adversary credits}
w^0(\OPT) \;\le\; \sum_{t} \advcredit^t\,,
\end{equation}
where the sum is over all steps $t$. 
\end{claim}

\begin{proof}
To prove~(\ref{eqn: adversary credits}), consider again a packet $j = \OPT[\tau]$ that arrived at some step $t'$.
It is sufficient to show that the sum of all adversary credits associated with slot $\tau$ for steps $t = t',t'+1,...,\tau$ is at least $w^0_j = w_j$. 

At time $t'$, if $j$ is added to $\varplanP$, then $j$ is also added to $\varADV$; otherwise $\varADV[\tau]$ remains undefined.
To streamline the proof, we think of the second case as adding $j$ to $\varADV$ in slot $\tau$ and then removing it immediately in the same step.
 
At any step $t$ when $\varADV[\tau]$ is defined, if the weight of $\varADV[\tau]$ is decreased, then, according to the
definitions above, this decrease of weight contributes to the adjustment credit at step $t$.
Consider the last step $t''$ when $\varADV[\tau]$ is defined, and let $q = \varADV[\tau]$ be the packet in slot $\tau$ of $\varADV$ at that time.
(It could be that $q\neq j$, as the packet in $\varADV[\tau]$ may change over time.)
Then the total adjustment credit in steps $t', t'+1, \dots, t''-1$ for changing $\varADV[\tau]$ equals $w^0_j - w^{t''}_q$.

There are two cases. If $t'' = \tau$, then the adversary will receive transmission credit of $\advcredit^\tau_\tau = w^{\tau}_q$ 
in step $\tau$, and the right-hand side of~(\ref{eqn: adversary credits}) associated with slot $\tau$ (the sum of adjustment credits and the transmission credit for slot $\tau$)
is $w^0_j - w^{\tau}_q +  w^{\tau}_q = w^0_j$. If $t'' < \tau$, then $q$ is removed from $\varADV$ in the analysis of step $t''$.
The adjustment credits for steps $t = t', t'+1, \dots, t''-1$ add up to $w^0_j - w^{t''}_q$, and
the adjustment credit in step $t''$ for removing $q$ from $\varADV$ is $ w^{t''}_q - \minwt(\planP'', d^{t''}_q)$, where $\planP''$ is the snapshot of $\varplanP$
when this removal of $q$ is taking place. 
Then later at step $\tau$ the adversary will receive transmission credit $\minwt(\planP^\tau, \tau)$, by the definition above.
Thus the total of adjustment and transmission credits associated with slot $\tau$ is at least
\begin{align}
(\,w^0_j - w^{t''}_q \,)  + (\, w^{t''}_q - & \minwt(\planP'', d^{t''}_q)\,)  + \minwt(\planP^\tau, \tau)
		\nonumber
		\\
	\;&\ge\;
			w^0_j - w^{t''}_q + (\,w^{t''}_q - \minwt(\planP'', \tau)\,)  + \minwt(\planP'', \tau) 
			\label{eqn: amortization inequality 1}
	\\
	\;&=\; w^0_j,
			\nonumber
\end{align}
where inequality~(\ref{eqn: amortization inequality 1}) holds because
$\minwt(\planP'', d^{t''}_q) \le \minwt(\planP'', \tau)$, by the definition of function $\minwt()$
(see also the comments on adjustments credits for packet removals before Claim~\ref{clm: adversary credits})
and because $\minwt(\planP^\tau, \tau) \ge \minwt(\planP'', \tau)$, which follows from
the slot monotonicity property, as summarized in Lemma~\ref{lem:monotonicityForALG}.
This concludes the proof of~\eqref{eqn: adversary credits}.
\end{proof}

\subsection{Backup Plan and the Potential Function}
\label{subsec: backup plan and the potential function}

In our analysis, we will maintain a set $\varbackupplan$ of pending packets called a  \emph{backup plan}.
$\varbackupplan$ contains two types of packets: all packets in $\varplanP\setminus \varADV$, and some pending packets not in $\varplanP$. 
(The relation between $\varplanP$, $\varbackupplan$ and $\varADV$ is illustrated in Figure~\ref{fig: sets in analysis}.)
The packets in $\varbackupplan\setminus\varplanP$ will typically be packets that were earlier ousted from $\varplanP$,
either as a result of arrivals of other packets or in a leap step.   These packets can also  be thought of as candidates for the
substitute packet $\substpacket^t(p)$ when the algorithm chooses a packet $p$ for transmission. 

The following invariant summarizes the essential properties of $\varbackupplan$, and it will be
maintained throughout the analysis:
\begin{description}
\item{{\textbf{\InvariantB}}}
For any snapshot $\backupplan$ of $\varbackupplan$ at any time $t$, 
(i) $\backupplan$ is a plan, i.e., a feasible set of packets pending in step $t$, and
(ii) $\backupplan \cap\planP = \planP\setminus\ADV$, where $\planP$ and $\ADV$
are the current snapshots of $\varplanP$ and $\varADV$, respectively. 
\end{description}
By Observation~\ref{obs: feasible = nonnegative slack},
invariant~{\InvariantB}(i) is equivalent to the condition that for any slot $\tau\ge t$, we have $\pslack(\backupplan, \tau) \ge 0$.
Invariant~{\InvariantB}(ii) and invariant~{\InvariantA} together imply that $\backupplan\cap \ADV = \emptyset$ and
$\varplanP = (\varplanP\cap\varbackupplan) \cup \varADV$, i.e., that
any packet in $\planP$ is either in $\ADV$ or in $\backupplan$, but not in both.
Similarly as for invariant~{\InvariantA}, preserving invariant~{\InvariantB}
in the course of the analysis will require making suitable changes in the
adversary stash $\varADV$ and the backup plan $\varbackupplan$.

The following observation is quite straightforward, but we state it explicitly here, as it is useful
later in some proofs.

\begin{observation}\label{obs: A not empty vs B-P not empty}
Consider the current optimal plan $\planP$, the backup plan $\backupplan$, and
the adversary stash $\ADV$ at time $t$. Assume that invariants~{\InvariantA} and~{\InvariantB} hold.
Let $\eta$ be a tight slot of $\planP$.
Then:
\begin{enumerate}[label=(\alph*),itemsep=2pt] 
\item $\ADV(\eta, \timehorizon]\neq\emptyset$ implies that $\backupplan(\eta, \timehorizon]\setminus \planP\neq\emptyset$, and
\item $\backupplan[t, \eta]\setminus \planP\neq\emptyset$ implies that $\ADV[t, \eta]\neq\emptyset$.
\end{enumerate}
\end{observation}

\begin{proof}
(a) From the property that all plans are full and that $\eta$ is a tight slot of $\planP$,
we have $|\planP(\eta, \timehorizon]| = \timehorizon-\eta \le |\backupplan(\eta, \timehorizon]|$,
where $\timehorizon$ is the time horizon (see Section~\ref{sec: plans and their properties}). Then, using
invariants~{\InvariantA} and~{\InvariantB}, we obtain that
$|\ADV(\eta, \timehorizon]| = |\planP(\eta, \timehorizon]\setminus\backupplan| \le |\backupplan(\eta, \timehorizon]\setminus \planP|$.

\smallskip

\noindent (b) Since $\eta$ is a tight slot of $\planP$ and $\backupplan$ is feasible,
we have $|\planP[t, \eta]| = \eta - t + 1 \ge |\backupplan[t, \eta]|$.
Therefore, if $\backupplan[t, \eta]\setminus \planP\neq\emptyset$, then 
$\planP[t, \eta]\setminus\backupplan\neq\emptyset$, which implies that $\ADV[t, \eta]\neq\emptyset$, by invariant~{\InvariantB}.
\end{proof}

The definitions of tight slots and segments apply to $\varbackupplan$ (as well as to any plan), and they will be helpful in our proofs.
We remark that $\varbackupplan$ may have a different segment structure than the current optimal plan $\varplanP$.

Before defining our potential function, we introduce a few lemmas that will be useful in showing that invariant~{\InvariantB} 
is  preserved after each step.


\begin{lemma}\label{lem:p relative changes of pslacks of P and B}
Consider the current optimal plan $\planP$, the backup plan $\backupplan$, and
the adversary stash $\ADV$ at time $t$. Assume that invariants~{\InvariantA} and~{\InvariantB} hold, and let
$\eta,\eta'$ be two time slots such that $t \le \eta \le \eta'$. Then
\begin{equation}\label{eqn:p relative changes of pslacks of P and B}
 \pslack(\planP,\eta')  - \pslack(\planP,\eta) + |\ADV(\eta,\eta']|
 	\;=\;
\pslack(\backupplan,\eta')  - \pslack(\backupplan,\eta) + |\backupplan(\eta,\eta']\setminus\planP|\,.
\end{equation}
\end{lemma}

Equation~(\ref{eqn:p relative changes of pslacks of P and B}) may appear complicated, but it's
actually quite straightforward. It compares the contributions of the interval $(\eta,\eta']$ to $\pslack()$ values of $\planP$ and $\backupplan$.
These contributions differ by $-|\planP(\eta,\eta']| + |\backupplan(\eta,\eta']|$,
and canceling out the common contribution $|\planP(\eta,\eta']\cap\backupplan|$
yields~(\ref{eqn:p relative changes of pslacks of P and B}). A formal proof follows.

\begin{proof}
{\InvariantB} implies that $\ADV(\eta,\eta'] = \planP(\eta,\eta']\setminus\backupplan$.
Using the definition of $\pslack()$, and canceling out the contributions of packets with deadline at most $\eta$, we get
\begin{align*}
 \pslack(\planP,\eta')  - \pslack(\planP,\eta) + |\ADV(\eta,\eta']|
 	&=\; \eta' - \eta - |\planP(\eta,\eta']|  + |\planP(\eta,\eta']\setminus\backupplan|
	\\
	&=\; \eta' - \eta - |\planP(\eta,\eta']\cap\backupplan|
	\\
	&=\; \eta' - \eta - |\backupplan(\eta,\eta']| + |\backupplan(\eta,\eta']\setminus\planP|
	\\
	&=\;\pslack(\backupplan,\eta')  - \pslack(\backupplan,\eta) + |\backupplan(\eta,\eta']\setminus\planP|,
\end{align*}
completing the proof.
\end{proof}


\begin{lemma}\label{lem:pslackRelation}
Consider the current optimal plan $\planP$, the backup plan $\backupplan$, and
the adversary stash $\ADV$ at time $t$. Assume that invariants~{\InvariantA} and~{\InvariantB} hold, and
let $\zeta$ be a tight slot of $\planP$ (possibly, $\zeta = t-1$).  
Let $\fstar$ be the earliest-deadline packet in $\backupplan(\zeta,\timehorizon] \setminus\planP$ 
and let $\gstar$ be the latest-deadline packet in $\ADV[t,\zeta]$. (We allow here the
possibility that $\fstar$ or $\gstar$ is undefined). Then: 
\begin{enumerate}[label=(\alph*),itemsep=2pt] 
\item  If $\fstar$ is defined, then 
$\pslack(\backupplan,\tau)\ge \pslack(P,\tau) + |\ADV(\zeta, \tau]|$ for any $\tau\in (\zeta, d_\fstar)$.
Otherwise this inequality holds for any $\tau > \zeta$.

\item If $\gstar$ is defined, then 
	$\pslack(\backupplan,\tau)\ge \pslack(P,\tau) + |\backupplan(\tau, \zeta]\setminus\planP|$ for any $\tau\in [d_\gstar, \zeta]$.
	Otherwise this inequality holds for any $\tau\in [t, \zeta]$.
\end{enumerate}
\end{lemma}

\begin{proof}
We first observe that, since $\backupplan$ is feasible and $\zeta$ is a tight slot for $\planP$, 
we have $\packetslack(\backupplan,\zeta) \ge 0 = \packetslack(\planP,\zeta)$.

\smallskip

\noindent (a)
If packet $\fstar$ exists, let $\theta = d_\fstar$, otherwise let $\theta = \timehorizon+1$, where $\timehorizon$ is the time horizon, 
as defined in Section~\ref{sec: preliminaries}. Let $\tau\in (\zeta,\theta)$.
By the definition of $\fstar$, there is no packet in $\backupplan\setminus\planP$ with deadline
in $(\zeta, \theta)$; in particular $\backupplan(\zeta, \tau]\setminus\planP = \emptyset$. Using this, equation $\packetslack(\planP,\zeta) = 0$, and
Lemma~\ref{lem:p relative changes of pslacks of P and B} (with $\eta = \zeta$ and $\eta' = \tau$), we obtain 
\begin{equation*}
\pslack(\planP,\tau)   + |\ADV(\zeta,\tau]|
	\;=\;
  \pslack(\backupplan,\tau) - \pslack(\backupplan,\zeta)
  	\;\le\;
  \pslack(\backupplan,\tau) \,,
\end{equation*}
which implies claim~(a).

\smallskip

\noindent (b)
If packet $\gstar$ exists, let $\lambda = d_\gstar$, otherwise let $\lambda = t$.  Let $\tau\in [\lambda,\zeta]$.
By the definition of $\lambda$, there is no packet in $\ADV$ with deadline in $(\lambda, \zeta]$, which implies that 
$\ADV(\tau, \zeta] = \emptyset$.
Using this, equation $\packetslack(\planP,\zeta) = 0$, and
Lemma~\ref{lem:p relative changes of pslacks of P and B} (with $\eta = \tau$ and $\eta' = \zeta$),
we obtain 
\begin{equation*}
 - \pslack(\planP,\tau) 
	\;=\;   \pslack(\backupplan,\zeta) - \pslack(\backupplan,\tau) + |\backupplan(\tau,\zeta]\setminus\planP| \,.
\end{equation*}
Using $\packetslack(\backupplan,\zeta) \ge 0$, we get 
$ - \pslack(\planP,\tau) \ge - \pslack(\backupplan,\tau) + |\backupplan(\tau,\zeta]\setminus\planP|$
and claim~(b) follows.
\end{proof}

In some cases of the analysis, we will have situations when a packet $g\in \varADV$ needs to be removed from $\varADV$.  
By the definition of the backup plan, this causes $g$ to be added to $\varbackupplan$, making $\varbackupplan$ infeasible.
The lemma below shows that we can restore the feasibility of $\varbackupplan$ after such
a change by removing from $\varbackupplan$ a suitably chosen packet that is not in $\varplanP$.


\begin{lemma}\label{lem: restore feasibility of B}
Consider the current optimal plan $\planP$, the backup plan $\backupplan$, and the adversary stash $\ADV$ at time $t$.
Assume that invariants~{\InvariantA} and~{\InvariantB} hold. Let $g\in \ADV$,
$\eta = \prevtightslot(\planP,d_g)$, and $\eta' = \nexttightslot(\backupplan,d_g)$. Then 
\begin{description}
	\item{(a)} $\backupplan(\eta,\eta']\setminus\planP\neq\emptyset$.
	\item{(b)} For any $f\in \backupplan(\eta,\eta']\setminus\planP$,
the set $C  = \backupplan \setminus \braced{f} \cup \braced{g}$ is feasible and $w_f\le w_g$ (thus $w(C)\ge w(\backupplan)$).
\end{description}
\end{lemma}

\begin{proof}
Note that the choice of $g\in \ADV$ implies that $g \notin \backupplan$.
We use Lemma~\ref{lem:p relative changes of pslacks of P and B} with $\eta$ and $\eta'$ as defined here. 
We can apply this lemma because $\backupplan$ satisfies invariant~{\InvariantB}. From that lemma, 
substituting $\pslack(\planP,\eta) = \pslack(\backupplan,\eta') = 0$, we get
\begin{equation*}
\pslack(\planP,\eta') + |\ADV(\eta,\eta']|
	\;=\;- \pslack(\backupplan,\eta) + |\backupplan(\eta,\eta']\setminus\planP|
	\;\le\; |\backupplan(\eta,\eta']\setminus\planP|\,,
\end{equation*}
where the last inequality follows from $\pslack(\backupplan,\eta)\ge 0$, which is a consequence of the feasibility of $\backupplan$. 
The existence of $g$ implies $\ADV(\eta,\eta']\neq\emptyset$, and hence, using the feasibility of $\planP$, we obtain that
$\backupplan(\eta,\eta']\setminus\planP\neq\emptyset$ holds as well, proving claim~(a). 

Pick any $f\in \backupplan(\eta,\eta']\setminus\planP$. From $d_f \le \eta' = \nexttightslot(\backupplan,d_g)$ we obtain that
$C$ is feasible. Inequality $d_f > \eta = \prevtightslot(\planP,d_g)$ and $f\not\in \planP$ imply that $w_f\le \minwt(\planP,d_g) \le w_g$,
and therefore $w(C)\ge w(\backupplan)$, completing the proof of claim~(b).
\end{proof}


As a corollary, we obtain that Lemma~\ref{lem: restore feasibility of B}(b) holds for $f$ equal to
the earliest-deadline packet $\fstar$ in $\backupplan\setminus\planP$ with $d_\fstar > \prevtightslot(\planP,d_g)$.

\begin{corollary}\label{cor: restore feasibility of B part two}
Consider the current optimal plan $\planP$, the backup plan $\backupplan$, and the adversary stash $\ADV$ at time $t$.
Assume that invariants~{\InvariantA} and~{\InvariantB} hold. Let $g\in \ADV$, $\eta = \prevtightslot(\planP,d_g)$, and 
let $\fstar$ be the earliest-deadline packet in $\backupplan(\eta,\timehorizon]\setminus\planP$.
Then packet $\fstar$ exists,
the set $C  = \backupplan \setminus \braced{\fstar} \cup \braced{g}$ is feasible, and $w_\fstar \le w_g$ (thus $w(C)\ge w(\backupplan)$).
\end{corollary}

\begin{proof}
Applying Lemma~\ref{lem: restore feasibility of B}(a) to packet $g$ and $\eta' = \nexttightslot(\backupplan,d_g)$ gives
us that $\backupplan(\eta,\eta']\setminus\planP\neq\emptyset$, which in turn implies that $d_\fstar\in (\eta,\eta']$.
Then, using Lemma~\ref{lem: restore feasibility of B}(b) for $f = \fstar$ implies the corollary.
\end{proof}


\myparagraph{Potential function.} 
We use the backup plan $\varbackupplan$ to define a potential function needed for the amortized analysis of Algorithm $\PlanMonotonicity$.
If $\backupplan$ is the current snapshot of $\varbackupplan$ then the potential value at time $t$ is     
\begin{equation}\label{eq:potential}   
\textstyle
\Psi(\backupplan) \;=\; 
		\frac{1}{\phi} \razy w^t(\backupplan) \,.
\end{equation} 
For brevity, we will use notation $\potentialBefStep = \Psi(\backupplanBefStep)$
for the potential at the beginning of step $t$, before any packet arrives,
and $\potentialAftArrivals = \Psi(\backupplanAftArrivals)$ for the potential just before a packet is transmitted in step $t$.

The intuition behind this definition is as follows. In order to be $\phi$-competitive, the average (per step) 
profit of Algorithm $\PlanMonotonicity$ should be at least $1/\phi$ times the
adversary's average profit. However, due to the choice of coefficients in line~1
of the algorithm, Algorithm $\PlanMonotonicity$ tends to postpone transmitting heavy packets with large deadlines.
(For example, given just two packets, a tight packet with weight $1$ and a non-tight packet
with weight $2.6 < \phi^2$, the algorithm will transmit the tight packet in the current
step.) As a result, in tight instances, $\PlanMonotonicity$'s actual profit per step, throughout most of the game, 
is often smaller than $1/\phi$ times the adversary's, and only near the end of the
instance, when delayed heavy packets are transmitted, the algorithm can make up for this deficit. 

In our amortized analysis, if there is a deficit in a given step, we pay for it with a ``loan''
that is represented by an appropriate increase of the potential function. 
Eventually, of course, these loans need to be repaid -- the potential eventually decreases
to $0$ and this decrease must be covered by excess profit.
The formula for the potential $\Psi$ is designed to guarantee such future excess profits.
To see this, imagine that no more packets arrive. Since $\varbackupplan\cap\varADV = \emptyset$, 
the packets in $\varbackupplan$ will not be transmitted in the future by the adversary.
If the algorithm does not execute any more leap steps
then it will collect all the packets in $\varplanP$, which are in total heavier than the packets in $\varbackupplan$
(as $\varbackupplan$ is feasible by invariant~{\InvariantB}(i) and
each packet $f\in \varbackupplan\setminus \varplanP$ satisfies $w^t_f\le \minwt(\varplanP, d^t_f)$ at each time step $t$).
On the other hand, if the algorithm executes a leap step, instead of a packet ousted from $\varplanP$ in this step,
the algorithm will collect a packet from $\varbackupplan\setminus\varplanP$ (or a better
packet) after it is added to the optimal plan as a substitute packet.

\subsection{Overview of the Analysis}
\label{subsec: overview of the analysis}

This section gives an overview of the analysis, states the main theorem,
and shows how it follows from results that will be established in the sections that follow.


\myparagraph{Initial and final state.}
At the beginning, per assumption~\ref{ass:virtualPackets},
we assume that the optimal plan is pre-filled with virtual $0$-weight packets, 
each in a slot equal to its deadline, and none of them scheduled by the adversary for transmission.  The adversary stash $\varADV$ is empty,
i.e., before the first step (at time $0$) we have $\widetilde{\ADV}^0 = \emptyset$, and
the backup plan is the same as the optimal plan, i.e., $\widetilde{\backupplan}^0 = \widetilde{\planP}^0$.
Thus invariants~{\InvariantA} and~{\InvariantB} clearly hold, and $\widetilde{\Psi}^0 = 0$.
At the end, after all (non-virtual) packets expire, the potential equals $0$ as well, i.e.,
$\widetilde{\Psi}^{\timehorizon+1}= 0$, where $\timehorizon$ is the time horizon (the last step).


\myparagraph{Amortized analysis.}
At the core of our analysis are bounds relating amortized profits of the algorithm and the adversary in each step $t$. 
For packet arrivals in step $t$, we will show the following \emph{packet-arrivals inequality}:
\begin{equation}\label{eq:deltaPotential-Arrival}
\potentialAftArrivals - \potentialBefStep\ge 0\,.
\end{equation} 
For the transmission event in a step $t$, we will show that the following \emph{packet-trans\-mission inequality} holds:
\begin{equation}\label{eq:oneStep}
\phi\razy w^t(\ALG[t]) - \phi\razy (\Delta^t \weights) + (\potentialBefNextStep - \potentialAftArrivals) \;\ge\; \advcredit^t,
\end{equation}
where $\ALG[t]$ is the packet in slot $t$ in the algorithm's schedule \ALG{} (thus $w^t(\ALG[t])$ is the algorithm's profit),
and $\Delta^t \weights$ is the total amount by which the algorithm increases the weights of its pending packets in step $t$.

We prove the packet-arrivals inequality~(\ref{eq:deltaPotential-Arrival}) in Section~\ref{subsec: arrival of a packet}
and the packet-transmission inequality~(\ref{eq:oneStep}) in Section~\ref{subsec: transmitting a packet}.
Assuming that these two inequalities hold, we now show our main result.


\begin{theorem}
Algorithm \PlanMonotonicity{} is $\phi$-competitive.
\end{theorem}

\begin{proof}
We show that $\phi\razy w^0(\ALG) \ge w^0(\OPT)$, which implies the theorem.
First, by Claim~\ref{clm: adversary credits}, we have $w^0(\OPT) \le \sum_{t} \advcredit^t$. Second, note that
\begin{equation}\label{eq:sumPotentialsBefStep}
\sum_t \,(\, \potentialBefNextStep - \potentialBefStep \,)
	\;=\; \widetilde{\Psi}^{T+1} - \widetilde{\Psi}^0
	\;=\; 0\,,
\end{equation}
where $\widetilde{\Psi}^0 = 0$ is the initial potential and $\widetilde{\Psi}^{\timehorizon+1}= 0$ is the final potential after the last step $\timehorizon$.
Observe also that 
\begin{equation}\label{eq:sumALGsProfits}
\sum_t [\, w^t(\textsf{ALG}[t]) - \Delta^t \weights \,] \;\le\; w^0(\textsf{ALG})\,.
\end{equation}
This follows from the observation that if the weight of $\textsf{ALG}[\tau]$ was
increased by some value $\delta > 0$ at some step $t' < \tau$,
then $\delta$ also contributes to $\Delta^{t'} \weights$, so such contributions cancel out in~(\ref{eq:sumALGsProfits}).
(There may be several such $\delta$'s, as the weight of a packet may have been increased multiple times.
Note that the  bound (\ref{eq:sumALGsProfits}) may not be tight if some packets with increased weights are later dropped.)

Summarizing, using the above bounds yields
\begin{align}
w^0(\OPT) \;&\le\; \sum_{t} \advcredit^t 
			\label{eqn: amortization ineq 1}
			\\
		 \;&\le\; \sum_t \Big[ \phi\razy w^t(\textsf{ALG}[t]) - \phi\razy (\Delta^t \weights) + (\potentialBefNextStep - \potentialAftArrivals) \Big] 
			+ \sum_t (\potentialAftArrivals - \potentialBefStep)
			\label{eqn: amortization ineq 2}
		\\
		&=\; \sum_t \Big[ \phi\razy w^t(\textsf{ALG}[t]) - \phi\razy (\Delta^t \weights)\Big]
			+ \sum_t (\potentialBefNextStep - \potentialBefStep)
			\nonumber
		\\
		&\le\; \phi\razy w^0(\textsf{ALG})\,,
		\label{eqn: amortization ineq 3}
\end{align} 
where inequality~\eqref{eqn: amortization ineq 1} is~\eqref{eqn: adversary credits}, inequality~\eqref{eqn: amortization ineq 2}
follows by applying~\eqref{eq:deltaPotential-Arrival} and~\eqref{eq:oneStep} for each step $t$, and inequality~\eqref{eqn: amortization ineq 3}
 holds by~\eqref{eq:sumPotentialsBefStep} and~\eqref{eq:sumALGsProfits}.
\end{proof}

\subsection{Packet Arrivals}
\label{subsec: arrival of a packet}

Let $t$ be the current time step. Our aim in this section is to prove that invariants~{\InvariantA} and~{\InvariantB}
can be preserved as packets arrive in step $t$, using appropriate modifications of sets $\varADV$ and $\varbackupplan$.
We also prove that the packet-arrivals inequality~\eqref{eq:deltaPotential-Arrival} holds for step $t$.  
To this end, it is sufficient to show how to preserve both invariants, without decreasing the value of $w(\varbackupplan)$
(and thus $\Psi$ as well), in response to an arrival of each individual packet.

Thus, consider the arrival event of a packet $s$ at time $t$. Let $\planP$ be the optimal plan just before $s$ arrives and
let $\planQ$ be the optimal plan just after $s$ arrives.
Furthermore, let $\ADV$ and $\backupplan$ be the snapshots of $\varADV$ and $\varbackupplan$ just before $s$ arrives,
and let $\ADV_{+s}$ and $\backupplan_{+s}$ denote, respectively, the snapshots of $\varADV$ and $\varbackupplan$ 
just after $s$ arrives and the changes described below are applied.
The algorithm does not change the weights and deadlines after packet arrivals, so we will omit the
superscript $t$ in the notation for weights and deadlines, that is $w_q = w^t_q$ and $d_q = d^t_q$, for
each packet $q$. There are two cases, depending on whether or not $s\in \planQ$.


\smallskip
\noindent
\mycase{A.1}
$s$ is not added to the optimal plan $\varplanP$, i.e., $\planQ=\planP$.
Lemma~\ref{lem:plan_update-arrival} implies that $w_s<\minwt(\planP,d_s)=\minwt(\planQ,d_s)$. 
In this case, we do not change $\varADV$, i.e., $\ADV_{+s} = \ADV$,
so invariant~{\InvariantA} is preserved. The backup plan $\varbackupplan$ remains the same as well, so
$w(\varbackupplan)$ does not change and invariant~{\InvariantB} is preserved. 


\smallskip
\noindent
\mycase{A.2}
$s$ is added to the optimal plan $\varplanP$. Let $u$ be the lightest packet in $\planP$ with $d_u\le \nexttightslot(\planP,d_s)$;
by assumption \ref{ass:virtualPackets} such $u$ exists. By the choice of $u$ and Lemma~\ref{lem:plan_update-arrival}, we have 
$\planQ=\planP\cup\braced{s}\setminus\braced{u}$, $\nexttightslot(\planP,d_u) \le \nexttightslot(\planP,d_s)$,
and $w_s > w_u = \minwt(\planP, d_s) = \minwt(\planP, d_u)$.

Replacing $u$ by $s$ in $\varplanP$ can also trigger changes in $\varADV$ and $\varbackupplan$.
We describe these changes in two parts: 
\begin{description}
	\item{(i)} First, we show that if $u\in\ADV$, then we can remove $u$ from $\varADV$, 
preserving the invariants and not decreasing the potential.
	 \item{(ii)} Then, assuming that $u\notin \varADV$, we describe and analyze the remaining changes.
\end{description}
Let $\ADV_\mathrm{(i)}$ denote the intermediate adversary stash, after  the change in (i) and before the change in (ii), where we let 
$\ADV_\mathrm{(i)} = \ADV$ if $u\notin \ADV$, that is when change~(i) does not apply.
We adopt the same notation for snapshots of set $\varbackupplan$.


\smallskip
\emparagraph{(i) Dealing with $u\in\ADV$.}
In this case we need to remove $u$ from $\varADV$, to preserve invariant~{\InvariantA}. But this also forces us
to add $u$ to $\varbackupplan$, in order to preserve~{\InvariantB}(ii), which in turn requires removing some packet from $\varbackupplan$
to preserve its feasibility. 

To implement these changes, we use Lemma~\ref{lem: restore feasibility of B} with $g= u$, which gives us that there
is a packet $f_u \in \backupplan\setminus\planP$ such that the set $\backupplan\setminus\braced{f_u}\cup\braced{u}$
is feasible and $w(\backupplan\setminus\braced{f_u}\cup\braced{u})\ge w(\backupplan)$. We then let
\begin{equation*}
	 \ADV_\mathrm{(i)} \;=\; \ADV\setminus\braced{u} 
	 \quad\textrm{and}\quad
	 \backupplan_\mathrm{(i)} \;=\; \backupplan\setminus\braced{f_u}\cup\braced{u}.
\end{equation*}
Since $u$ is scheduled in $A$ in interval $[t,d_u]$ and
$w_u = \minwt(\planP,d_u)$, the adversary adjustment credit associated with removing $u$ from $\varADV$ is
$w_u - \minwt(\planP,d_u) = 0$ (see Section~\ref{subsec: adversary stash}).
As explained above, we also have that $\backupplan_\mathrm{(i)}$ is feasible and
$w(\backupplan_\mathrm{(i)}) \ge w(\backupplan)$; thus
invariant~{\InvariantB} is preserved and the change of $w(\varbackupplan)$ is non-negative. 

\smallskip
\emparagraph{(ii) Analysis of other changes.}
By (i) we can now assume that $u\notin \ADV_\mathrm{(i)}$, and invariant~{\InvariantB}(ii) together with $u\in \planP$ imply that $u\in \backupplan_\mathrm{(i)}$.
We now consider the changes of $\varbackupplan$ resulting from including $s$ in $\varplanP$.
We analyze two subcases, depending on whether or not $s\in \OPT$.

\smallskip
\noindent
\mycase{A.2.a}
$s\in \OPT$. We add $s$ to $\varADV$ in the same slot as in \OPT. Specifically, we take
$\ADV_{+s} = \ADV_\mathrm{(i)}\cup \braced{s}$ and $\backupplan_{+s} = \backupplan_\mathrm{(i)}$.
Invariant~{\InvariantA} is preserved as $s\in \ADV_{+s}$,
invariant~{\InvariantB} holds as $\varbackupplan$ remains unaffected in (ii), and the value of $w(\varbackupplan)$ does not change.

\smallskip
\noindent
\mycase{A.2.b}
$s\notin \OPT$. In this case, we simply let $\ADV_{+s} = \ADV_\mathrm{(i)}$,
so invariant~{\InvariantA} continues to hold. However, $s$ needs to get added to $\varbackupplan$ to preserve invariant~{\InvariantB}(ii), so
we need to remove some packet from  $\varbackupplan$ to maintain its feasibility, and
this packet must be lighter than $s$ to assure that the potential does not decrease.
Define $\xi = \nexttightslot(\planP,d_s)$ and $\xi_{\backupplan} = \nexttightslot(\backupplan_\mathrm{(i)},d_s)$.
We have two cases.

If $d_u\le \xi_{\backupplan}$, we replace $u$ by $s$ in $\varbackupplan$. That is, we let   
$\backupplan_{+s} = \backupplan_\mathrm{(i)}\cup \braced{s} \setminus \braced{u}$, and this satisfies invariant~{\InvariantB}(ii).
The case condition implies that invariant~{\InvariantB}(i) continues to hold, and,
as $w_u < w_s$ in Case~A.2, the value of $w(\varbackupplan)$ does not decrease.

Next, assume that $d_u > \xi_{\backupplan}$. Note that $d_u \le \xi$, by the definition of $u$.
Let $\lambda = \prevtightslot(\planP,d_s)$. In this case we  have that $\lambda < d_s \le \xi_{\backupplan} < d_u \le \xi$, 
so $d_u$ and $d_s$ are in the same segment $(\lambda,\xi]$ of $\planP$.
Since $\lambda$ is a tight slot for $\planP$, $\xi_{\backupplan}$ is a tight slot for $\backupplan_\mathrm{(i)}$, and $\pslack(\planP, \xi_{\backupplan}) > 0$,
from Lemma~\ref{lem:p relative changes of pslacks of P and B} with $\eta = \lambda$ and $\eta' = \xi_{\backupplan}$ we get that 
\begin{equation*}
	|\backupplan_\mathrm{(i)}(\lambda,\xi_{\backupplan}]\setminus\planP| 
	\;=\;
		\pslack(\backupplan_\mathrm{(i)},\lambda) + \pslack(\planP, \xi_{\backupplan})
		 + |\ADV_\mathrm{(i)}(\lambda,\xi_{\backupplan}]| 
		 \;>\; 0\,,
\end{equation*}
implying $ \backupplan_\mathrm{(i)}(\lambda,\xi_{\backupplan}]\setminus\planP \neq \emptyset$.
Choose any $f\in \backupplan_\mathrm{(i)}(\lambda,\xi_{\backupplan}]\setminus\planP$ and let
$\backupplan_{+s} = \backupplan_\mathrm{(i)}\cup \braced{s} \setminus \braced{f}$, preserving invariant~{\InvariantB}(ii).   
Invariant~{\InvariantB}(i) holds for $\backupplan_{+s}$ because
$d_f \le \xi_{\backupplan} = \nexttightslot(\backupplan_\mathrm{(i)},d_s)$.
The optimality of $\planP$ and the choice of $u$ and $f$ imply
that $w_f \le \minwt(\planP,d_s) = w_u < w_s$, thus $w(\varbackupplan)$ cannot decrease.

\medskip

Summarizing, we showed that in response to the arrival of any packet $s$ in step $t$
we can modify $\varADV$ and $\varbackupplan$ in such a way that the invariants~{\InvariantA} and~{\InvariantB}
will be preserved and the value of $\Psi$ will not decrease. 
This gives us that the packet-arrivals inequality~\eqref{eq:deltaPotential-Arrival} holds for step $t$
and that after packet arrivals both invariants~{\InvariantA} and~{\InvariantB} will hold,
concluding the analysis of packet arrivals in step $t$.

\subsection{Transmitting a Packet}
\label{subsec: transmitting a packet}

Let $t$ be the current step of the computation.
After all packets with release time equal to $t$ arrive, the algorithm transmits its packet $p = \ALG[t]$. Recall that
$\planPaftArrivals$ is the optimal plan just before transmitting $p$ and $\planPbefNextStep$ is the optimal plan
after the algorithm transmits $p$, possibly adjusts weights and deadlines, and after the time is incremented to $t+1$.

We split the analysis of the 
transmission step into two parts, called the adversary step and the algorithm's step, defined as follows:
\begin{description}
\item{\emph{Adversary step:}}
If packet $\ADVaftArrivals[t]$ is defined, say $\ADVaftArrivals[t]= j$, then we need to remove it from $\varADV$. (Recall that packets in $\varADV$ may have been removed
or replaced during the analysis, so $j$ may not be equal to $\OPT[t]$, the packet scheduled in $\OPT$ at time $t$.)
Removing $j$ from $\varADV$ could trigger a change in $\varbackupplan$, but the optimal plan $\varplanP$ remains unchanged.
We show that these changes preserve both invariants~{\InvariantA} and~{\InvariantB}. We should stress that
these changes are made without advancing the current time, which will be done in algorithm's step.

In the analysis of the adversary step, given in Section~\ref{subsubsec: adversary step}, we establish a relation (inequality~\eqref{eq:advStepCost}) 
between the adversary transmission credit and the change of the potential. By $\ADVaftAdvMove$ and $\backupplanAftAdvMove$ we will denote the snapshots of sets
$\varADV$ and $\varbackupplan$, respectively, right after the adversary step.
\item{\emph{Algorithm's step:}}
In the algorithm's step, the algorithm transmits $p$, the time is incremented to $t+1$, and the optimal plan
changes from $\planPaftArrivals$ to $\planPbefNextStep$. 

The analysis of this step assumes that the changes
described in the adversary step have already been implemented.
Using the bound~\eqref{eq:advStepCost}, invariants~{\InvariantA} and~{\InvariantB},
and other properties, we then show that the packet-transmission inequality (\ref{eq:oneStep}) holds after
the sets $\varplanP$, $\varADV$, and $\varbackupplan$ are updated to reflect the changes triggered by the packet's transmission.
We also ensure that invariants~{\InvariantA} and~{\InvariantB} are preserved.

The analysis of the algorithm's step is given in 
Sections~\ref{subsubsec: ordinary step}-\ref{subsubsec: processing (beta,gamma] in an iterated leap step}.
We first analyze the ordinary step in Section~\ref{subsubsec: ordinary step}.
We then give a roadmap for the analysis of a leap step in Section~\ref{subsubsec: leap step: a road map},
as it will be divided into substeps. We analyze the particular substeps 
in Section~\ref{subsubsec: leap step: processing S1}, which
describes the changes in the initial segment $[t,\alpha]$ of $\planPaftArrivals$,
and in Sections~\ref{subsubsec: processing (beta,gamma] in a simple leap step}-\ref{subsubsec: processing (beta,gamma] in an iterated leap step},
which contain the analysis of other changes resulting from a leap step.  
\end{description}
In Table~\ref{tab:notationSnapshots}, we summarize the notation of snapshots that we use in this section.

{\renewcommand{\arraystretch}{1.3} 
\begin{center}
\begin{table}
\begin{tabular}{|l|l|l|r|c|c|c|} \hline
\multicolumn{4}{|r|}{Set description} & \bigcell{c}{Optimal\\ plan} & \bigcell{c}{Backup\\ plan} & \bigcell{c}{Adversary\\ stash}
\\ \hline          
\multicolumn{4}{|r|}{Dynamic set} & $\varplanP$ & $\varbackupplan$ & $\varADV$ 
\\ \hline          
\multicolumn{4}{|r|}{Beginning of step $t$ (before packet arrivals)} & $\planPbefStep$ & $\backupplanBefStep$ & $\ADVbefStep$
\\ \hline       
\multicolumn{4}{|r|}{Before transmission in step $t$ (after packet arrivals) } & $\planPaftArrivals$ & $\backupplanAftArrivals$ & $\ADVaftArrivals$
\\  \hline 
\multirow{5}{*}{\bigcell{c}{Transmission\\ event\\ in step $t$}}  
    & \multicolumn{3}{r|}{After the adversary step (Sec.~\ref{subsubsec: adversary step})}	& $\planPaftArrivals$ 
													& $\backupplanAftAdvMove$ 
													& $\ADVaftAdvMove$
\\   \cline{2-7}
   &  \multirow{3}{*}{\bigcell{c}{Leap\\ step\\ (Sec.~\ref{subsubsec: leap step: a road map})}}
	& \multicolumn{2}{r|}{\bigcell{r}{After processing $[t,\alpha]$ and \\ incrementing time (Sec.~\ref{subsubsec: leap step: processing S1})}} 	& $\planPaftInitSeg$ 
	 															& $\backupplanAftInitSeg$ 
																& $\ADVaftInitSeg$
\\       \cline{3-7}    
   &&  \multirow{2}{*}{\bigcell{c}{Iterated\\ leap step\\ (Sec.~\ref{subsubsec: processing (beta,gamma] in an iterated leap step})}}
	& \bigcell{r}{Before stage\\ for group $\seggroup{a,b}$} & $\planPbefGroup{\seggroup{a,b}}$ 
										& $\backupplanBefGroup{\seggroup{a,b}}$ 
										& $\AdvBefGroup{\seggroup{a,b}}$ 
\\        \cline{4-7} 
	&&& \bigcell{r}{After stage\\ for group $\seggroup{a,b}$} & $\planPaftGroup{\seggroup{a,b}}$ 
										& $\backupplanAftGroup{\seggroup{a,b}}$ 
										& $\AdvAftGroup{\seggroup{a,b}}$ 
	\\  \hline  
\multicolumn{4}{|r|}{Local notation (after part (i) of changes)}	& $\planPaftPartI$
																& $\backupplanAftPartI$ 
																& $\ADVaftPartI$ 
\\  \hline  
\end{tabular} 
\caption{Notation of snapshots of sets $\varplanP, \varbackupplan$, and $\varADV$ used in the analysis.
The particular snapshots for each substep of the analysis of a leap step will be defined later, in the respective sections.}
\label{tab:notationSnapshots}
\end{table}
\end{center}
}

\subsubsection{Adversary Step} 
\label{subsubsec: adversary step} 

As defined in Section~\ref{subsec: adversary stash},
if $\ADVaftArrivals$ contains a packet in slot $t$, then the adversary gains the transmission credit of $\advcredit^t_t = w^t(\ADVaftArrivals[t])$.
Otherwise, the adversary's transmission credit equals $\advcredit^t_t = \minwt(\planPaftArrivals, t)$.
(The overall adversary credit for this step also includes the adjustment credit,
but in this section we focus only on the relation between the transmission credit and the potential change in the adversary step.)

Except for possibly removing packet $\ADVaftArrivals[t]$ from $\varADV$, if it is defined,
we will not make other changes to $\varADV$, so invariant~{\InvariantA} will be preserved.
Below we show that with appropriate changes invariant~{\InvariantB} will also be preserved after the adversary step.
Further, denoting by $\DeltaPsiADV$ the change of the potential in the adversary step, we  prove the following auxiliary inequality:
\begin{equation}\label{eq:advStepCost}
\DeltaPsiADV - \advcredit^t_t \;\ge\; - \phinegtwof \razy {w^t_p} - \phinegonef \razy {w^t(\substpacket^t(p))} \,.
\end{equation}
Recall that throughout Section~\ref{subsec: transmitting a packet} we denote the packet scheduled by the algorithm in step $t$ by $p$.

The proof of inequality~\eqref{eq:advStepCost} is divided into two cases, depending on whether or not $\ADVaftArrivals[t]$ is defined.
As packet weights are not changed in the adversary step, below we omit the superscript $t$ in the notation for weights.


\smallskip
\noindent
\mycase{ADV.1} $\ADVaftArrivals$ contains a packet in slot $t$. Let $j = \ADVaftArrivals[t]$.
We will now remove $j$ from the adversary stash $\varADV$ and, since by invariant~{\InvariantA} we have $j\in \planPaftArrivals$, 
we need to add packet $j$ to the backup plan $\varbackupplan$ to maintain invariant~{\InvariantB}(ii),
which in turn requires removing a packet from $\varbackupplan$ to preserve its feasibility.
To this end, we apply Lemma~\ref{lem: restore feasibility of B} (with $g= j$), which implies that there is 
a packet $f_j\in \backupplanAftArrivals\setminus\planPaftArrivals$  such that $d_{f_j} > \prevtightslot(\planPaftArrivals,d_j)$, $w_{f_j} \le w_j$, and
for which set $\backupplanAftArrivals \setminus \braced{f_j}\cup \braced{j}$ is feasible.

We thus set $\ADVaftAdvMove = \ADVaftArrivals\setminus \braced{j}$ and $\backupplanAftAdvMove = \backupplanAftArrivals \setminus \braced{f_j}\cup \braced{j}$. 
Invariant~{\InvariantA} is clearly preserved,
and by Lemma~\ref{lem: restore feasibility of B}, invariant~{\InvariantB} continues to hold as well. The potential change is
\begin{equation}
\DeltaPsiADV \;=\; \phinegonef \razy \left( w(\backupplanAftAdvMove) - w(\backupplanAftArrivals) \right) 
				\;=\; \phinegonef \razy ( - w_{f_j} + w_j ).
\label{eqn: adv step potential change}
\end{equation}
From the definition of $\substpacket^t(j)$, $d_{f_j} > \prevtightslot(\planPaftArrivals,d_j)$, and $f_j\notin \planPaftArrivals$,
we have that $f_j$ is a candidate for $\substpacket^t(j)$, and thus $w(\substpacket^t(j))\ge w_{f_j}$. 
Using~\eqref{eqn: adv step potential change} and $\advcredit^t_t = w_j$, it follows that
\begin{align}
\phi \razy (\,\DeltaPsiADV - \advcredit^t_t \,)
	\;&=\; - w_{f_j} + w_j - \phi\razy w_j
	\nonumber
	\\
	&=\;  - \phinegonef \razy {w_j} - w_{f_j}
		\nonumber
	\\
	&\ge\; - \phinegonef \razy {w_j} - {w(\substpacket^t(j))}
		\nonumber
	\\
    &\ge\; - \phinegonef \razy {w_p} - {w(\substpacket^t(p))} \,,
		\label{eqn: amortization case adv.1}
\end{align}
where inequality \eqref{eqn: amortization case adv.1} follows from the choice of $p$
in line~\ref{algLn:transmit} of the algorithm's description, using also that $j\in \planPaftArrivals$. This implies~\eqref{eq:advStepCost}.


\smallskip
\noindent
\mycase{ADV.2} Slot $t$ is empty in $\ADVaftArrivals$. In this case, we do not change $\varADV$ and $\varbackupplan$,
so $\ADVaftAdvMove = \ADVaftArrivals$ and $\backupplanAftAdvMove = \backupplanAftArrivals$.
Invariants~{\InvariantA} and~{\InvariantB} are trivially preserved. Recall that 
$\advcredit^t_t = \minwt(\planPaftArrivals, t) = w_\firstlightpacket$, where $\firstlightpacket$ denotes the lightest packet in the initial segment of $\planPaftArrivals$.
Note that $\substpacket^t(\firstlightpacket) = \firstlightpacket$ and that $\DeltaPsiADV = 0$. Then 
\begin{align}
\phi \razy (\,\DeltaPsiADV - \advcredit^t_t\,)
	\;&=\; - \phi\razy w_\firstlightpacket
	\nonumber \\
	&=\; - \phinegonef \razy {w_\firstlightpacket} - w(\substpacket^t(\firstlightpacket))
	\nonumber
	\\
	&\ge\; - \phinegonef \razy {w_p} - w(\substpacket^t(p)) \,,
	\label{eqn: adv step adv.2 last}
\end{align}
where inequality~\eqref{eqn: adv step adv.2 last} follows from the choice of $p$ again.
Thus~\eqref{eq:advStepCost} holds.
                  
\smallskip

This concludes the analysis of the adversary step. In particular, we have determined the snapshots 
$\ADVaftAdvMove$ and $\backupplanAftAdvMove$, of the adversary stash
$\varADV$ and the backup plan $\varbackupplan$, respectively, resulting from the adversary step.  
Next, in the following sections, we analyze the algorithm's step.

\subsubsection{Ordinary Step} 
\label{subsubsec: ordinary step}

We now assume that the adversary step, as described in the previous section, has already been implemented. The current optimal
plan $\planPaftArrivals$ remains unchanged in the adversary step, and the current snapshots of $\varADV$ and $\varbackupplan$ are
  $\ADVaftAdvMove$ and $\backupplanAftAdvMove$.
In this section we analyze the algorithm's move at step $t$, assuming this is an ordinary step, as defined in Section~\ref{sec: online algorithm}.

In an ordinary step a packet $p\in \planPaftArrivals[t,\alpha]$ is transmitted, 
where $\alpha =\nexttightslot^t(t)$ is the first tight slot in $\planPaftArrivals$. The algorithm makes no changes in packet weights
and deadlines, so $\Delta^t\weights=0$. Thus, to simplify notation,
for any packet $q$ we can write $w_q = w^t_q$ and $d_q = d^t_q$, omitting the superscript $t$.
As usual, $\firstlightpacket$ denotes the lightest packet in the initial segment $\planPaftArrivals[t,\alpha]$.
 
By the algorithm, $p$ is the heaviest packet in $\planPaftArrivals[t,\alpha]$.
Since $\substpacket^t(p) = \firstlightpacket$,  inequality~\eqref{eq:advStepCost} gives us that
\begin{equation}
\DeltaPsiADV - \advcredit^t_t \;\ge\; - \phinegtwof \razy {w_p} - \phinegonef \razy w_\firstlightpacket \,.
\label{eq:advStepCost ordinary}
\end{equation}
According to Lemma~\ref{lem:plan-update_transmission-1stSegment},
the new optimal plan (starting at time slot $t+1$) is $\planPbefNextStep = \planPaftArrivals \setminus\braced{p}$. 

We have two cases, depending on whether or not the adversary stash contains a packet
in the initial segment $[t,\alpha]$ of $\planPaftArrivals$.


\smallskip
\noindent
\mycase{O.1}
$\ADVaftAdvMove[t,\alpha] = \emptyset$.
In this case, $p\not\in \ADVaftAdvMove$ (as $p\in \planPaftArrivals[t,\alpha]$) and we do not further change set $\varADV$,
i.e., $\ADVbefNextStep = \ADVaftAdvMove$. So invariant~{\InvariantA} is preserved and $\advcredit^t = \advcredit^t_t$.

We now show that we can preserve invariant~{\InvariantB}. As $p\not\in \ADVaftAdvMove$ and~{\InvariantB}(ii), 
we have that $p\in \backupplanAftAdvMove$. We modify the backup plan $\varbackupplan$ by removing $p$, that is 
$\backupplanBefNextStep = \backupplanAftAdvMove \setminus \braced{ p }$. This immediately gives us that invariant~{\InvariantB}(ii) continues to hold.

Next, we show that invariant~{\InvariantB}(i) holds as well.
Since $\backupplanBefNextStep$ is with respect to time $t+1$ and $\backupplanAftAdvMove$ is feasible, 
we have that $\pslack(\backupplanBefNextStep,\tau) = \pslack(\backupplanAftAdvMove,\tau)\ge 0$ for $\tau\ge d_p$.
So it remains to consider slots $\tau\in [t+1,d_p)$.
The case condition and invariant~{\InvariantB}(ii) imply that $\planPaftArrivals[t,\alpha] \subseteq \backupplanAftAdvMove$.
This, together with the feasibility of $\backupplanAftAdvMove$ and $\alpha$ being a tight slot in $\planPaftArrivals$, 
gives us that in fact $\planPaftArrivals[t,\alpha] = \backupplanAftAdvMove[t,\alpha]$.
From this and $\pslack(\planPaftArrivals,\tau)\ge 1$ for $\tau\in [t+1,d_p)$
(as $\alpha \ge d_p$ is the first tight slot), we get that
$\pslack(\backupplanBefNextStep,\tau) = \pslack(\backupplanAftAdvMove,\tau) - 1 
							=  \pslack(\planPaftArrivals,\tau) - 1 
							\ge 0$ for all $\tau\in [t+1,d_p)$,
completing the proof that invariant~{\InvariantB}(i) continues to hold.

\smallskip

The calculation showing the packet-transmission inequality~\eqref{eq:oneStep} is quite simple. Taking into account that
the adversary credit is $\advcredit^t = \advcredit^t_t$, and that the change of the potential in the algorithm's step
is $\DeltaPsiALG = \phinegonef( w(\backupplanBefNextStep) - w(\backupplanAftAdvMove) )    = -\phinegonef w_p$, we obtain
\begin{align}
\phi\razy w^t(\textsf{ALG}[t]) &- \phi\razy (\Delta^t \weights) + (\potentialBefNextStep - \potentialAftArrivals) - \advcredit^t 
\nonumber\\
&=\; \phi \razy w_p - \phi\cdot 0 + [\, \DeltaPsiALG + \DeltaPsiADV \,] - \advcredit^t_t
\nonumber\\   
&=\; \phi \razy w_p + \DeltaPsiALG + [\, \DeltaPsiADV - \advcredit^t_t  \,]               
\nonumber\\
&\ge\; \phi \razy w_p - \phinegonef\razy {w_p} 
	+ \left[\,- \phinegtwof\razy {w_p} - \phinegonef\razy {w_\firstlightpacket}\,\right]
	\label{eqn: ordinary case O.1}	\\
	& =\; \phinegonef\razy {w_p}  - \phinegonef\razy {w_\firstlightpacket} 
	\;\ge\; 0\,,
	\label{eqn: ordinary case O.2}	
\end{align} 
where inequality~\eqref{eqn: ordinary case O.1} uses~\eqref{eq:advStepCost ordinary}, and
the inequality in the last step in line~\eqref{eqn: ordinary case O.2} uses the definition of $\firstlightpacket$,
namely that $w_p\ge w_\firstlightpacket$.


\smallskip
\noindent
\mycase{O.2}
$\ADVaftAdvMove[t,\alpha] \neq \emptyset$. (This includes the case when $p\in \ADVaftAdvMove$.)
We first describe our modifications of sets $\varADV$ and $\varbackupplan$, and then argue that with these
modifications our invariants are preserved and inequality~\eqref{eq:oneStep} is satisfied.


\indentemparagraph{Changing sets $\varADV$ and $\varbackupplan$.}
Let $\gstar$ be the latest-deadline packet in $\ADVaftAdvMove[t,\alpha]$. Packet $\gstar$ exists, by the case condition.  
Note that $d_{\gstar} \ge t+1$, because $\ADVaftAdvMove$ cannot contain a packet with deadline $t$
(such a packet would be removed from $\varADV$ when we analyze the adversary step).
Let $\fstar$ be the earliest-deadline packet in $\backupplanAftAdvMove\setminus\planPaftArrivals$.
Observation~\ref{obs: A not empty vs B-P not empty}(a) with $\eta = t-1$ implies that $\fstar$ is well-defined.
Note that possibly $d_{\fstar} = t$, in which case $\fstar$ cannot remain in $\varbackupplan$ in the next step.

If $p\in \ADVaftAdvMove$, let $g = p$; otherwise let $g = \gstar$. (In either case we have $g\in\ADVaftAdvMove$.)
We remove $g$ from $\varADV$, i.e., we set $\ADVbefNextStep = \ADVaftAdvMove\setminus \braced{g}$.
If we have $g = \gstar \neq p$, then $p\in \backupplanAftAdvMove$, so $p$ will need to be removed from $\backupplanAftAdvMove$,
and due to removing $g$ from $\varADV$ we will also need to add it to $\varbackupplan$ to satisfy invariant~{\InvariantB}(ii).
In either case, we remove $\fstar$ from $\varbackupplan$. Thus the new backup plan will be
\begin{equation*}
\backupplanBefNextStep \;=\; \left\{
								\begin{array}{lcl}
										 \backupplanAftAdvMove \setminus \braced{\fstar} &\quad& \textrm{if\ } g=p
										 \\
										  \backupplanAftAdvMove \cup\braced{\gstar} \setminus \braced{\fstar,p} && \textrm{if\ } g = \gstar \neq p
								\end{array}
								\right. \,.
\end{equation*}
Note that in both cases it holds that $\backupplanBefNextStep = \backupplanAftAdvMove \cup\braced{g} \setminus \braced{\fstar,p}$.

For a warm-up, before proving that our invariants hold, let's verify that all packets in $\backupplanBefNextStep$ have deadlines at least $t+1$.
Indeed, as already mentioned earlier, we have $d_{\gstar}\ge t+1$, and if some $q\in \backupplanBefNextStep\setminus\braced{\gstar}$ had $d_q = t$ then,
by the definitions of $\fstar$ and $\backupplanBefNextStep$, this $q$ would also be in $\planPaftArrivals$, but this implies that
$\alpha = t$ and $\planPaftArrivals[t,t] = \braced{q}$, contradicting the case condition because $q\notin \ADVaftAdvMove$, by invariant~{\InvariantB}(ii).


\indentemparagraph{Preserving the invariants.}
We now have $p\notin \ADVbefNextStep$, so invariant~{\InvariantA} is preserved.
We next show that invariant~{\InvariantB} holds after the step.
For part {\InvariantB}(ii), we need to show that $\planPbefNextStep\cap \backupplanBefNextStep  = \planPbefNextStep\setminus \ADVbefNextStep$.
This can be verified quite easily by considering how $\varplanP \cap \varbackupplan$ and $\varplanP\setminus \varADV$ change:
If $g = p$ then both sets do not change, and if $g = \gstar\neq p$ then $p$ is replaced by $\gstar$ in both sets.

To streamline the argument for part {\InvariantB}(i) (feasibility), we divide the process of updating and analyzing $\varbackupplan$ into
two parts: first we analyze the effects of replacing $\fstar$ by $g$ in $\varbackupplan$, and then we show that we can remove $p$
from  $\varbackupplan$ (and increment $t$), preserving its feasibility.

So in the first part we consider the auxiliary set $\backupplanAftPartI = \backupplanAftAdvMove\cup\braced{g}\setminus\braced{\fstar}$.
Directly from Corollary~\ref{cor: restore feasibility of B part two}, 
we obtain that $\backupplanAftPartI$ is a feasible set of packets at time $t$.
Further, we also show that it has the following property:


\begin{claim}\label{cla: slack positive in B[t,dp] in case O2}
$\pslack(\backupplanAftPartI,\tau)\ge 1$ for any $\tau \in [t,d_p)$. 
\end{claim}

\begin{proof}
We first observe that for $\tau\in [t,d_\fstar)$, by invoking Lemma~\ref{lem:pslackRelation}(a) with $\zeta = t-1$,
$B = \backupplanAftAdvMove$, $A = \ADVaftAdvMove$, and with the current plan $\planPaftArrivals$, we can obtain that
\begin{equation}\label{eqn:caseO2_step1 - fstar}
\pslack(\backupplanAftAdvMove,\tau) \;\ge\; \pslack(\planPaftArrivals,\tau) + |\ADVaftAdvMove[t, \tau]|\,.
\end{equation}
Now we consider three cases depending on the value of $\tau$. The first case is for $\tau\in [t, \min(d_{\fstar}, d_g))$.
Then
\begin{equation*}
\pslack(\backupplanAftPartI,\tau) \;=\; \pslack(\backupplanAftAdvMove,\tau) \;\ge\; \pslack(\planPaftArrivals,\tau)\ge 1\,,
\end{equation*}
where the equality uses that $\backupplanAftPartI[t, \tau] = \backupplanAftAdvMove[t, \tau]$
(as $\tau < \min(d_{\fstar}, d_g)$), the first inequality uses~\eqref{eqn:caseO2_step1 - fstar}, and
the last inequality follows from the fact that $\tau$ is not a tight slot of $\planPaftArrivals$, as $\tau \in [t, \alpha)$.

In the second case, for slots $\tau\in [d_g, \min(d_{\fstar},d_p))$, we note that $g\in \ADVaftAdvMove[t, \tau]$, so
equation~\eqref{eqn:caseO2_step1 - fstar} and $\pslack(\planPaftArrivals,\tau)\ge 1$ imply
\begin{equation*}
\pslack(\backupplanAftPartI,\tau) \;=\; \pslack(\backupplanAftAdvMove,\tau) - 1 
								\;\ge\; \pslack(\planPaftArrivals,\tau) + |\ADVaftAdvMove[t, \tau]| - 1 
								\;\ge\; 1\,.
\end{equation*}

In the third case, we deal with $\tau\in [d_{\fstar}, d_p)$. For $\tau\in [d_{\fstar}, d_g)$,
replacing $\fstar$ by $g$ increases the value of $\pslack(\varbackupplan, \tau)$ by $1$,
implying $\pslack(\backupplanAftPartI,\tau) = \pslack(\backupplanAftAdvMove,\tau) + 1\ge 1$, as $\backupplanAftAdvMove$ is feasible.
In particular, we are done if $g = p$. It thus remains to consider the sub-case when $g = \gstar\neq p$ and $\tau\in [\max(d_{\fstar}, d_{\gstar}), d_p)$,
for which we use  Lemma~\ref{lem:pslackRelation}(b) with $\zeta = \alpha$, $B = \backupplanAftAdvMove$,
$A = \ADVaftAdvMove$, and with the current plan $\planPaftArrivals$, to get
\begin{equation*}
\pslack(\backupplanAftPartI,\tau) \;=\; \pslack(\backupplanAftAdvMove,\tau) \;\ge\; \pslack(\planPaftArrivals,\tau) \;\ge\; 1\,,
\end{equation*}
where the equality uses that $\backupplanAftPartI[t, \tau] = \backupplanAftAdvMove[t, \tau]\cup\braced{g}\setminus\braced{\fstar}$
(as $\tau\ge \max(d_{\fstar}, d_g)$) and
the last inequality follows from the fact that $\tau$ is not a tight slot of $\planPaftArrivals$, as $\tau \in [t, \alpha)$. 
\end{proof}

Next, note that we can express $\backupplanBefNextStep$ as $\backupplanBefNextStep = \backupplanAftPartI\setminus\braced{p}$ 
--- indeed, this works no matter whether $g=p$ or $g = \gstar$. 
It remains to show that removing $p$ from $\backupplanAftPartI$ and incrementing the current time to $t+1$ preserves feasibility.
These changes decrease the value of $\pslack(\varbackupplan, \tau)$ for $\tau \in [t+1, d_p)$, but 
Claim~\ref{cla: slack positive in B[t,dp] in case O2} ensures that we still have $\pslack(\backupplanBefNextStep, \tau)\ge 0$ for such $\tau$.
For $\tau \ge d_p$, the value of $\pslack(\varbackupplan, \tau)$ is not affected.
Hence, invariant~{\InvariantB}(i) holds after the step.


\indentemparagraph{Deriving inequality~(\ref{eq:oneStep}).}
Let $\DeltaPsiALG$ be the change of $\Psi$ in the algorithm's step, so
$\DeltaPsiALG = \frac{1}{\phi} (w(\backupplanBefNextStep) - w( \backupplanAftAdvMove ) )  =  \frac{1}{\phi}(- w_p - w_{\fstar} + w_g)$.
The adversary adjustment credit is $\advcredit^t_{(t,\alpha]} = w_g - \minwt(\planPaftArrivals, d^t_g) = w_g - w_\firstlightpacket$,
since we removed $g$ from $\varADV$ and $\minwt(\planPaftArrivals, d^t_g) = w_\firstlightpacket \le w_g$.
Thus, the total adversary credit for step $t$ is $\advcredit^t = \advcredit^t_t + \advcredit^t_{(t,\alpha]} = \advcredit^t_t + w_g - w_\firstlightpacket$.
(See Section~\ref{subsec: adversary stash} for the definition of $\advcredit^t$.)
Note that $w_{\fstar} \le w_\firstlightpacket$ as $\fstar\not\in \planPaftArrivals$ and
that $w_g\le w_p$ as $p$ is the heaviest packet in $ \planPaftArrivals[t,\alpha]$ and $g\in \planPaftArrivals[t,\alpha]$.
We show the packet-transmission inequality~(\ref{eq:oneStep}) by summing these changes:   
\begin{align}
\phi\razy w^t(\textsf{ALG}[t]) &- \phi\razy (\Delta^t \weights) + (\potentialBefNextStep - \potentialAftArrivals) - \advcredit^t 
\nonumber\\
&=\; \phi \razy w_p - \phi\cdot 0 + [\,\DeltaPsiALG + \DeltaPsiADV \,]
				- [\, \advcredit^t_t + w_g - w_\firstlightpacket \,]
\nonumber\\ 
&= \;\phi \razy w_p + \DeltaPsiALG - w_g + w_\firstlightpacket + [\, \DeltaPsiADV - \advcredit^t_t \,]
\nonumber\\ 
&\ge\; \phi \razy w_p 
	+  \phinegonef \razy \left[\,- {w_p} - {w_{\fstar}} + {w_g}  \,\right] - w_g + w_\firstlightpacket 
	 	+ \left[\, - \phinegtwof\razy {w_p} - \phinegonef \razy{w_\firstlightpacket} \,\right] 
\label{eqn: ordinary case O.2 third}\\
&=\; \phinegonef \razy {w_p} - \phinegonef \razy {w_{\fstar}} - \phinegtwof\razy {w_g} +  \phinegtwof\razy {w_\firstlightpacket} 
\nonumber\\
&\ge\; \phinegonef \razy {w_p}  - \phinegonef \razy {w_\firstlightpacket}  - \phinegtwof\razy {w_p} +  \phinegtwof\razy {w_\firstlightpacket}
\label{eqn: ordinary case O.2 penultimate}
\\
			&=\; \phinegthreef \razy (\,{w_p} - {w_\firstlightpacket}\,) \;\ge\; 0\,,
\label{eqn: ordinary case O.2 last}
\end{align}
where inequality~\eqref{eqn: ordinary case O.2 third} uses~\eqref{eq:advStepCost ordinary},
inequality \eqref{eqn: ordinary case O.2 penultimate} holds because $w_{\fstar} \le w_\firstlightpacket$ and $w_g\le w_p$,
and the last inequality \eqref{eqn: ordinary case O.2 last} follows from $w_p\ge w_\firstlightpacket$.
This concludes the analysis of an ordinary step.   

\subsubsection{Leap Step: a Roadmap}
\label{subsubsec: leap step: a road map} 

In the remainder of the analysis, we focus on a leap step of Algorithm~$\PlanMonotonicity(t)$, when it transmits a packet
$p \in \planPaftArrivals(\alpha,\timehorizon]$, where, as defined earlier,
$\planPaftArrivals$ represents the current optimal plan, already updated to take into account packet arrivals.    
As in this case some packet weights change, we use notation $w^t_a$ and $w^{t+1}_a$ 
(or  $w^t(a)$ and $w^{t+1}(a)$) for the weights of a packet $a$ before and after $p$ is transmitted, respectively.
For packets whose weight does not change, we may occasionally omit the superscript $t$. 
We apply the same convention to deadlines.

By Lemma~\ref{lem:leapStep}(a), the new optimal plan (starting at time $t+1$) is
\begin{equation*}
\planPbefNextStep \;=\; \planPaftArrivals \setminus \braced{p, \firstlightpacket} \cup \braced{\rho}\,,
\end{equation*}
where, as usual, $\rho=\substpacket^t(p)$ and $\firstlightpacket$ is the lightest packet in $\planPaftArrivals[t,\alpha]$.

All changes in the optimal plan resulting from transmitting $p$ are within two intervals of the plan: the initial segment $[t,\alpha]$
and the interval $(\beta,\gamma]$, where $\beta = \prevtightslot^t(d^t_p)$ and $\gamma = \nexttightslot^t(d^t_\rho)$
(possibly, $\alpha = \beta$). In segment $[t,\alpha]$, packet $\firstlightpacket$ is ousted.
In interval $(\beta,\gamma]$ (that could consist of several segments), packet $p$ gets replaced by $\rho$, and Algorithm~$\PlanMonotonicity$
increases the weight of $\rho$ to $\mu = \minwt(\planPaftArrivals,d^t_\rho)$.
Furthermore, if this is an iterated leap step (i.e., $k > 1$ in the algorithm), $\PlanMonotonicity$ also modifies weights and deadlines of some packets $h_i$. 
These changes in the optimal plan  may change some values of $\packetslack(\tau)$ and may force changes in $\varADV$ or $\varbackupplan$,
which might in turn trigger additional changes in these sets to restore invariants~{\InvariantA} and {\InvariantB}.

We divide the analysis of a leap step into two \emph{substeps}, namely (i) processing the initial segment $[t,\alpha]$ and (ii) processing interval $(\beta,\gamma]$.
These are outlined later in this section and described in detail in Sections~\ref{subsubsec: leap step: processing S1},
\ref{subsubsec: processing (beta,gamma] in a simple leap step}, and~\ref{subsubsec: processing (beta,gamma] in an iterated leap step}.
Before proceeding to that, we introduce two key inequalities bounding the potential change in these substeps, and we show how to derive from them
the packet-transmission inequality~\eqref{eq:oneStep}.



\myparagraph{Two key inequalities.}
We now define several quantities that we will use in our estimates and in the two key inequalities below: 
\begin{description}
\item{$\Delta_{[t,\alpha]} \Psi$\,:}    
	The change of the potential due to the modifications in $\varADV$ and $\varbackupplan$ implemented when processing interval $[t, \alpha]$.
	These modifications include the removal of $\firstlightpacket$ from $\planPaftArrivals[t,\alpha]$, as well as other modifications that are triggered by it.
\item{$\Delta_{(\beta,\gamma]} \Psi$\,:} 
	The change of the potential due to the modifications in $\varADV$ and $\varbackupplan$ implemented when processing interval $(\beta, \gamma]$.
	These modifications include the replacement of $p$ by $\rho$ in $(\beta,\gamma]$, the increases of weights and possible decreases of deadlines of some packets, 
	as well as other modifications that are triggered by these changes.
\item{$\advcredit^t_{(\beta,\gamma]}$\,:} The adversary adjustment credit 
for modifications of packets in $\varADV$ implemented when processing interval $(\beta,\gamma]$,
namely for removing them or replacing them by lighter packets.
\item{$\Delta^t \weights$\,:}
	As already defined earlier, this is the total adjustment of weights of packets in $\planPaftArrivals$ 
	that is done by the algorithm in step $t$
	(always strictly positive as the algorithm increases the weight of $\rho$ by a non-zero amount).
\end{description}

Note that we do not define the adversary adjustment credit for processing the initial segment $[t, \alpha]$,
as in that substep of the analysis the adversary does not need to be compensated for any change in $\varADV$.

We will derive the proof of the packet-transmission inequality~\eqref{eq:oneStep} from the two
key inequalities below, that bound the changes of the potential resulting from processing intervals $[t,\alpha]$ and $(\beta,\gamma]$,
respectively: 
\begin{align}
\Delta_{[t,\alpha]} \Psi &\ge  - \phinegonef\razy {w_\firstlightpacket^t}
	\label{eq:leap-S1changes}
\\
\Delta_{(\beta,\gamma]} \Psi - \phi\razy (\Delta^t \weights) - \advcredit^t_{(\beta,\gamma]}
 				&\ge -  {w_p^t} + {w_\rho^t}  
	\label{eq:leap-laterSegments}
\end{align}


\myparagraph{Deriving the packet-transmission inequality.}
We now derive~\eqref{eq:oneStep} from~\eqref{eq:leap-S1changes} and~\eqref{eq:leap-laterSegments}.
We start with two simple but useful bounds. First, using inequality~\eqref{eq:advStepCost}  and the definition of $\rho$, we have
\begin{equation}\label{eq:advStepCost with rho}
\DeltaPsiADV - \advcredit^t_t
			\;\ge\; -  \phinegtwof\razy w_p^t - \phinegonef\razy w_\rho^t\,.
\end{equation}
Also,  for any $\tau\ge t$, we have
\begin{equation}\label{eq:leapStepVsOmega}
\phinegtwof\razy {w_p^t} + \phinegonef\razy {w_\rho^t} 
			\;\ge\; w_\firstlightpacket^t 
			\;\ge\; \minwt(\planPaftArrivals,\tau)\,,
\end{equation}   
where the first inequality follows from the choice of $p$
in line~\ref{algLn:transmit} of the algorithm (specifically, because the  
algorithm chose $p$ over $\firstlightpacket$ and because $\firstlightpacket = \substpacket(\planPaftArrivals,\firstlightpacket)$),
and the second one follows from $w_\firstlightpacket^t = \minwt(\planPaftArrivals, t)\ge \minwt(\planPaftArrivals, \tau)$, that is
the monotonicity of $\minwt(\planPaftArrivals, \tau)$ with respect to $\tau$.

So assume now that \eqref{eq:leap-S1changes} and \eqref{eq:leap-laterSegments} hold. The total potential change in this step is
\begin{equation*}
\potentialBefNextStep - \potentialAftArrivals \;=\; 
 	\DeltaPsiADV  + \Delta_{[t,\alpha]} \Psi +  \Delta_{(\beta,\gamma]} \Psi\,,
\end{equation*}
because all changes in the optimal plan and in sets $\varADV$ and $\varbackupplan$ are accounted for (uniquely)
in the terms on the right-hand side. (The changes during the adversary step, accounted for in $\DeltaPsiADV$,
were discussed in Section~\ref{subsubsec: adversary step}. The changes associated with the algorithm's step, that contribute to 
$\Delta_{[t,\alpha]} \Psi$ and $\Delta_{(\beta,\gamma]} \Psi$, will be detailed in
Sections~\ref{subsubsec: leap step: processing S1}-\ref{subsubsec: processing (beta,gamma] in an iterated leap step}.)
The total adversary credit is the sum of the transmission credit, that was analyzed in Section~\ref{subsubsec: adversary step},
and the adjustment credits for decreasing some packet weights in $\varADV$, so
\begin{equation*}
	\advcredit^t \;=\; \advcredit^t_t + \advcredit^t_{(\beta,\gamma]}  \,,
\end{equation*}
as the changes in $[t,\alpha]$ will not involve any weight decreases in the adversary stash $\varADV$.
Combining it all together, we have
\begin{align}
\phi\razy &w^t(\textsf{ALG}[t]) - \phi\razy (\Delta^t \weights) + (\potentialBefNextStep - \potentialAftArrivals) - \advcredit^t
\nonumber 
\\
&=\; 
	\phi\razy w_p^t -  \phi\razy (\Delta^t \weights)  + 
		 [\,  \DeltaPsiADV + \Delta_{[t,\alpha]} \Psi  + \Delta_{(\beta,\gamma]} \Psi \,]	 
	 -  [\, \advcredit^t_t  + \advcredit^t_{(\beta,\gamma]} \,]
\nonumber\\
&= 
	\phi\razy {w_p^t} + [\, \DeltaPsiADV - \advcredit^t_t \,] + \Delta_{[t,\alpha]} \Psi 
			+ [\, \Delta_{(\beta,\gamma]} \Psi - \phi\razy (\Delta^t \weights) - \advcredit^t_{(\beta,\gamma]}  \,]
\nonumber\\
&\ge\; 
	\phi\razy w_p^t  + \left[\,  -\phinegtwof\razy  {w_p^t} -  \phinegonef\razy {w_\rho^t} \,\right] 
	 + \left[\, -  \phinegonef\razy {w_\firstlightpacket^t} \,\right]  
	 + \left[\, - {w_p^t} + {w_\rho^t} \,\right]
\label{eqn: roadmap ineq 3rd step}
\\
&=  \;  
 	\phinegthreef\razy w_p^t + \phinegtwof\razy w_\rho^t - \phinegonef\razy {w_\firstlightpacket^t}    
	\;\ge\; 0 \,, 
\label{eqn: roadmap ineq last step}
\end{align}
where in inequality \eqref{eqn: roadmap ineq 3rd step} we use, in this order, 
inequalities~\eqref{eq:advStepCost with rho}, \eqref{eq:leap-S1changes}, and~\eqref{eq:leap-laterSegments},
and in line~\eqref{eqn: roadmap ineq last step} the equality is obtained by
repeatedly applying the definition of $\phi$, while the last inequality follows from~\eqref{eq:leapStepVsOmega}.

\smallskip

Therefore, to complete the analysis, it is now sufficient to show 
that the two key inequalities~\eqref{eq:leap-S1changes} and~\eqref{eq:leap-laterSegments} hold,
and that invariants~{\InvariantA} and~{\InvariantB} are preserved after the step.
As mentioned above, we divide the proof into two substeps, more precisely defined as follows:
\begin{description}
	\item{\emph{Processing the initial segment $[t,\alpha]$:}} 
	In this substep, presented in Section~\ref{subsubsec: leap step: processing S1} below,
	we assume that the changes described in the adversary step
			have already been implemented. We describe changes in $\varADV$ and $\varbackupplan$
			triggered by the removal of $\firstlightpacket$ from the optimal plan and by incrementing the time to $t+1$.
			We show that inequality~\eqref{eq:leap-S1changes} holds and that
			invariants~{\InvariantA} and~{\InvariantB} are preserved after these changes. More precisely,	
			let $\planPaftInitSeg = \planPaftArrivals\setminus \braced{\firstlightpacket}$ be plan $\varplanP$
			after removing $\firstlightpacket$ and incrementing the current time to $t+1$,
			and let $\ADVaftInitSeg$ and $\backupplanAftInitSeg$ be the snapshots
			of sets $\varADV$  and $\varbackupplan$ after implementing necessary changes resulting from modifying $\varplanP$. 
			We prove that invariants~{\InvariantA} and~{\InvariantB} hold for plan $\planPaftInitSeg$
			and sets $\ADVaftInitSeg$, $\backupplanAftInitSeg$. (As indicated by superscript $t+1$, these
			sets are defined with respect to time $t+1$.)

	\item{\emph{Processing interval $(\beta,\gamma]$:}}
				In this substep, we assume that the changes described in the adversary step and in the processing of $[t,\alpha]$
				have already been implemented. 
				We describe changes in $\varADV$ and $\varbackupplan$ triggered by the replacement of $p$ by $\rho$ and,
				for an iterated leap step, by modifications of packets $h_i$.
				We prove that inequality~\eqref{eq:leap-laterSegments} holds and that invariants~{\InvariantA} and~{\InvariantB}
				are preserved after these changes.  
				In other words, we show that both invariants hold after the transmission event, i.e., for plan $\planPbefNextStep$ and sets 
				$\ADVbefNextStep$ and $\backupplanBefNextStep$.	
								
				The proof will be divided into two cases, depending on whether it is a simple
				or an iterated leap step (see Sections~\ref{subsubsec: processing (beta,gamma] in a simple leap step} 
				and~\ref{subsubsec: processing (beta,gamma] in an iterated leap step}, respectively).
				The proof for an iterated leap step is further divided into a number of smaller stages.
				Each such stage consists of modifications of $\varplanP$, $\varbackupplan$ and
				$\varADV$ that involve packets with deadlines in some time interval $(\zeta, \eta] \subseteq (\beta,\gamma]$,
				where  $\zeta$ and $\eta$ are tight slots of $\varplanP$.

\end{description}
Throughout the process, while we gradually make changes in sets $\varplanP$, $\varbackupplan$ and
$\varADV$, we will ensure that the following two properties hold:
				
\begin{description}

\item{\emph{Locality of changes:}} For a substep or stage in which an interval $(\zeta, \eta]$ of slots is processed,
let $\dot{\planP}$ be the snapshot of $\varplanP$ just before processing interval $(\zeta, \eta]$
and let $\ddot{\planP}$ be the optimal plan after this subset or stage; $\zeta$ and $\eta$ will always be tight slots of $\dot{\planP}$.
We will maintain the property that the changes of $\varplanP$ are confined to the set $\varplanP(\zeta, \eta]$ of packets,
or more precisely,
\begin{equation}\label{eqn: roadmap locality}
\dot{\planP} \setminus \dot{\planP}(\zeta, \eta] = \ddot{\planP} \setminus \ddot{\planP}(\zeta, \eta]\,.
\end{equation}
We also do not change the deadlines and weights of packets in $\dot{\planP} \setminus \dot{\planP}(\zeta, \eta]$.
Equation~\ref{eqn: roadmap locality} implies that $|\dot{\planP}(\zeta, \eta]| = |\ddot{\planP}(\zeta, \eta]|$, since all plans are full.
Therefore, $\eta$ and $\zeta$ are tight slots of $\ddot{\planP}$ and all values of  $\pslack(\varplanP, \tau)$ remain unaffected for $\tau\notin (\zeta, \eta)$.
We remark that substeps and stages in the analysis of each case will have pairwise disjoint intervals $(\zeta, \eta]$.

\item{\emph{Monotonicity of $\varbackupplan \setminus \varplanP$:}}
No packets will be added to $\varbackupplan(\alpha,\timehorizon] \setminus \varplanP$ throughout the analysis.
(In fact, the only packet that may be added to $\varbackupplan$ is $\firstlightpacket$. This could happen
when segment $[t,\alpha]$ is processed.) 

\end{description}

\smallskip

We remark that some ``intermediate snapshots'' of $\varplanP$ in the analysis of a leap step, such as $\planPaftInitSeg$, mentioned above, 
are only auxiliary concepts used to capture the changes of $\varplanP$ after each substep or stage of the analysis.
These snapshots satisfy the definition of plans and are full, but they may not be optimal for the current set of pending packets.
As we do not rely on these intermediate plans being optimal, this presents no issue for the analysis.

\subsubsection{Leap Step: Processing the Initial Segment}
\label{subsubsec: leap step: processing S1} 

We now focus on the first substep of the analysis of the leap step, namely on
processing segment $[t,\alpha]$ of the optimal plan $\varplanP$.
As explained earlier in Section~\ref{subsubsec: leap step: a road map},  
we assume that the changes in sets $\varADV$ and $\varbackupplan$ described in the adversary step
(Section~\ref{subsubsec: adversary step}) have already been implemented. 
Recall that $\ADVaftAdvMove$ and $\backupplanAftAdvMove$ represent the snapshots of sets
$\varADV$ and $\varbackupplan$, respectively, resulting from the adversary step.
As the adversary step does not affect the optimal plan, the current plan is still $\planPaftArrivals$.

In this substep, we advance the time to $t+1$ and remove $\firstlightpacket$ from $\varplanP$;
that is, plan $\planPaftArrivals$ is changed to $\planPaftInitSeg = \planPaftArrivals\setminus\braced{\firstlightpacket}$.
$\planPaftInitSeg$ is a plan for the set of packets pending at time $t+1$, that consists of
the packets that were pending at time $t$, but excluding $\firstlightpacket$ and packets that expired at time $t$.
Since $d_\firstlightpacket\le \alpha = \nexttightslot(\planPaftArrivals,t)$,
the values of $\pslack(\varplanP, \tau)$ may only change for $\tau < \alpha$,
i.e., $\pslack(\planPaftInitSeg, \tau) = \pslack(\planPaftArrivals, \tau)$ for $\tau \ge \alpha$.
We may also need to make changes in sets $\varADV$ and $\varbackupplan$, in order to preserve the invariants.
We refer to this substep as ``processing the initial segment'', although some modifications of $\varbackupplan$
may involve pending packets with deadlines after $\alpha$. We divide this substep into two parts: 
\begin{description}
	\setlength{\itemsep}{0in}
	\item{(i)} first we show that if $\firstlightpacket \in \ADVaftAdvMove$
		then we can remove it from $\varADV$, preserving the invariants and not decreasing the potential, 
		and then 
	\item{(ii)} assuming that $\firstlightpacket \notin \varADV$, we analyze the effect of 
removing $\firstlightpacket$ from $\varplanP$ and incrementing the time.
\end{description}
Let $\ADVaftPartI$ denote the intermediate adversary stash, after 
the change in (i) and before the change in (ii), where we let 
$\ADVaftPartI = \ADVaftAdvMove$ if $\firstlightpacket\notin \ADVaftAdvMove$,
that is when change~(i) does not apply. We adopt the same notation for set $\varbackupplan$.
Note that $\varplanP$ does not change in (i), i.e., $\planPaftPartI = \planPaftArrivals$.


\indentemparagraph{(i) Dealing with the case $\firstlightpacket \in \ADVaftAdvMove$.}
We now consider the case when $\firstlightpacket\in\ADVaftAdvMove$. Since $\firstlightpacket\notin \planPaftInitSeg$,
we need to remove $\firstlightpacket$ from $\varADV$ to preserve invariant~{\InvariantA}.
In order to satisfy invariant~{\InvariantB}, we are then forced to add $\firstlightpacket$ to $\varbackupplan$,
which in turn requires us to remove some packet from $\varbackupplan\setminus\varplanP$ to restore feasibility of $\varbackupplan$.
To identify this packet we use Lemma~\ref{lem: restore feasibility of B} (with $g = \firstlightpacket$),
which shows that there is a packet $f_\firstlightpacket \in \backupplanAftAdvMove\setminus\planPaftArrivals$ for which the
set $\backupplanAftAdvMove \setminus \braced{f_\firstlightpacket} \cup \braced{\firstlightpacket}$
is feasible and satisfies $w(\backupplanAftAdvMove \setminus \braced{f_\firstlightpacket} \cup \braced{\firstlightpacket})\ge w(\backupplanAftAdvMove)$.

We thus set $\ADVaftPartI = \ADVaftAdvMove \setminus \braced{\firstlightpacket}$
and $\backupplanAftPartI = \backupplanAftAdvMove \setminus \braced{f_\firstlightpacket}\cup \braced{\firstlightpacket}$,
and this preserves both invariants~{\InvariantA} and~{\InvariantB}.
Note that the adversary does not get any adjustment credit for the removal of $\firstlightpacket$,
as $w_\firstlightpacket=\minwt(\planPaftArrivals, d_\firstlightpacket)$, by the definition of $\firstlightpacket$.
Further, since the contribution of this change of $\varbackupplan$ to $\Delta_{[t,\alpha]} \Psi$ is positive, when estimating $\Delta_{[t,\alpha]} \Psi$
we will account for this contribution without an explicit reference.


\indentemparagraph{(ii) Removing $\firstlightpacket$ from $\varplanP$ and incrementing the time.}
By (i) we have $\firstlightpacket\notin \ADVaftPartI$, and thus also $\firstlightpacket\in \backupplanAftPartI$, by invariant~{\InvariantB}(ii).
Now, the plan is changed from $\planPaftArrivals$ to $\planPaftInitSeg = \planPaftArrivals\setminus\braced{\firstlightpacket}$.
As we do not further change $\varADV$, i.e., $\ADVaftInitSeg = \ADVaftPartI$, invariant~{\InvariantA} holds. 
(Note that after the adversary step, analyzed in Section~\ref{subsubsec: adversary step},
$\ADVaftPartI$ cannot contain a packet that expires at time $t$.)  However, we need to preserve invariant~{\InvariantB}.
Denoting by $f_1$ the earliest-deadline packet in $\backupplanAftPartI\setminus\planPaftArrivals$, we consider two cases.

\smallskip
\noindent
\mycase{L.InSeg.1} $d_{f_1}\ge d_\firstlightpacket$.
We let $\backupplanAftInitSeg = \backupplanAftPartI\setminus \braced{\firstlightpacket}$. Inequality~\eqref{eq:leap-S1changes} holds as
$\Delta_{[t,\alpha]} \Psi \ge  - \phinegonef\razy {w_\firstlightpacket^t}$.
(This inequality could be strict, due to the contribution from replacing $f_\firstlightpacket$ by $\firstlightpacket$, as discussed in part~(i) above.)
Invariant~{\InvariantB}(ii) holds because $\firstlightpacket\notin \planPaftInitSeg$ and $\varADV$ does not change in~(ii).
To show that {\InvariantB}(i), i.e., the feasibility of $\varbackupplan$, is preserved, 
note that $\packetslack(\backupplanAftInitSeg, \tau) = \packetslack(\backupplanAftPartI, \tau)\ge 0$ for $\tau\ge d_\firstlightpacket$.
Furthermore, $\packetslack(\backupplanAftPartI, \tau)\ge \packetslack(\planPaftArrivals, \tau)\ge 1$ 
for $\tau\in [t,d_\firstlightpacket)$, where the first inequality follows from the definition of $f_1$,
$d_{f_1}\ge d_\firstlightpacket > \tau$, and Lemma~\ref{lem:pslackRelation}(a) with $\zeta = t-1$,
and the second inequality follows from the fact that $\tau$ is before the first tight slot $\alpha$ of $\planPaftArrivals$.
Hence, for $\tau\in [t+1,d_\firstlightpacket)$ we have $\packetslack(\backupplanAftInitSeg, \tau) = \packetslack(\backupplanAftPartI, \tau) - 1\ge 0$,
concluding the proof that invariant~{\InvariantB} is maintained.

\smallskip
\noindent
\mycase{L.InSeg.2} $d_{f_1} < d_\firstlightpacket$. In this case, we 
let $\backupplanAftInitSeg = \backupplanAftPartI\setminus \braced{f_1}$.
Then $\Delta_{[t,\alpha]} \Psi \ge - \phinegonef\razy {w_{f_1}^t}\ge - \phinegonef\razy {w_\firstlightpacket^t}$,
where the second inequality follows from $\firstlightpacket\in \planPaftArrivals[t,\alpha]$ and $f_1\notin \planPaftArrivals$.
Thus inequality~\eqref{eq:leap-S1changes} holds. Regarding invariant {\InvariantB},
its part~(ii) follows from the changes of sets $\varADV$ and $\varbackupplan$. 
To show invariant {\InvariantB}(i),
note that $\packetslack(\backupplanAftInitSeg, \tau) = \packetslack(\backupplanAftPartI, \tau)\ge 0$ for $\tau\ge d_{f_1}$.
For $\tau\in [t,d_{f_1})$, we first get $\packetslack(\backupplanAftPartI, \tau)\ge \packetslack(\planPaftArrivals, \tau)\ge 1$ 
by Lemma~\ref{lem:pslackRelation}(a) with $\zeta = t-1$, using the fact that $\tau$ is before the first tight slot $\alpha$ of $\planPaftArrivals$,
as $\tau < d_{f_1} < d_\firstlightpacket \le \alpha$ (here, we use the case condition in the second inequality).
Then, for $\tau\in [t+1,d_{f_1})$, we have $\packetslack(\backupplanAftInitSeg, \tau)\ge \packetslack(\backupplanAftPartI, \tau) - 1\ge 0$.
Thus, invariant~{\InvariantB} is preserved after processing the initial segment $[t,\alpha]$.

\subsubsection{Processing $(\beta,\gamma]$ in a Simple Leap Step}
\label{subsubsec: processing (beta,gamma] in a simple leap step} 

We now address the second substep of analyzing a leap step, namely processing the interval $(\beta,\gamma]$, focusing on
the case of a simple leap step. (The case of an iterated leap step will be handled separately,
in Section~\ref{subsubsec: processing (beta,gamma] in an iterated leap step}.) Recall that
the algorithm replaces $p$ by $\rho$ in $\varplanP$, and that in a simple leap step we have $k=0$,
that is $(\beta,\gamma]$ is a segment of $\planPaftArrivals$ and $d_p,d_\rho \in (\beta,\gamma]$ (see Section~\ref{sec: online algorithm}).

We assume that the changes in  $\varplanP$ and sets $\varADV$ and $\varbackupplan$ described in Sections~\ref{subsubsec: adversary step}
and~\ref{subsubsec: leap step: processing S1} have already been implemented. 
(We note that these changes might have involved some packets considered in this section;  for example, it is possible that
$\rho$ is the same as packet $f_j$ removed from $\varbackupplan$ when processing the adversary step in Section~\ref{subsubsec: adversary step}.)
In particular, the current time is already advanced to $t+1$, and
we assume that both invariants~{\InvariantA} and~{\InvariantB} hold for sets $\ADVaftInitSeg$,
$\backupplanAftInitSeg$, and $\planPaftInitSeg$, all three defined with respect to time $t+1$.
It now remains to describe how we process the interval $(\beta,\gamma]$. In this substep, the plan changes from 
$\planPaftInitSeg$ to $\planPbefNextStep = \planPaftInitSeg \setminus \braced{p}\cup \braced{\rho}$ and the weight of $\rho$ is increased.
The replacement of $p$ by $\rho$ may trigger modifications in $\varADV$ and $\varbackupplan$, in order to restore the invariants. 

$\planPaftInitSeg$ only differs from $\planPaftArrivals$ in that in $\planPaftInitSeg$ the current time is
incremented to $t+1$, packet $\firstlightpacket$ is removed, and the interval $[t+1,\alpha]$ of $\planPaftArrivals$ may be split into multiple segments.
Slot $\alpha = \nexttightslot(\planPaftArrivals,t)$ is also tight in $\planPaftInitSeg$, and there were no changes in $\varplanP$ after slot $\alpha$;
that is $\planPaftArrivals(\alpha,\timehorizon] = \planPaftInitSeg(\alpha,\timehorizon]$. 
Therefore $d_\rho$ and $d_p$ are still in the same segment $(\beta, \gamma]$ of $\planPaftInitSeg$,
so $\prevtightslot(\planPaftInitSeg, d_\rho) = \prevtightslot(\planPaftInitSeg, d_p) = \beta$ and
$\nexttightslot(\planPaftInitSeg, d_\rho) = \nexttightslot(\planPaftInitSeg, d_p) = \gamma$. 

\smallskip

In order to build some intuition behind our modifications of $\varADV$ and $\varbackupplan$, before we give a formal argument,
let us consider a few illustrative scenarios, focusing on maintaining invariant~{\InvariantB}.
The simplest case is when $p\notin \backupplanAftInitSeg$ and $\rho \in \backupplanAftInitSeg$,
as then we can simply let $\backupplanBefNextStep = \backupplanAftInitSeg$.

Suppose next that $p\in \backupplanAftInitSeg$ and $\rho \notin \backupplanAftInitSeg$.
In this case we are forced to remove $p$ from $\varbackupplan$, and to add $\rho$ to $\varbackupplan$.
If $d_\rho\ge d_p$, then we don't need to make any more changes,
that is we can let $\backupplanBefNextStep = \backupplanAftInitSeg \setminus\braced{p}\cup \braced{\rho}$, and
$\backupplanBefNextStep$ will be feasible.
In fact, even when $d_\rho < d_p$, using the properties of $\varbackupplan$
established in Section~\ref{subsec: backup plan and the potential function},
we can take $\backupplanBefNextStep = \backupplanAftInitSeg \setminus\braced{p}\cup \braced{\rho}$ as long as $\ADVaftInitSeg(\beta, \gamma]= \emptyset$.

In the remaining situations, some additional
modifications of $\varbackupplan$ are needed. In essence, what will happen is this: If $p\in \backupplanAftInitSeg$, we will
replace it in $\varbackupplan$ by the latest-deadline packet $\gstar$ in $\ADVaftInitSeg[t+1, \gamma]$, that will
be removed from $\varADV$.
Similarly, if $\rho \notin \backupplanAftInitSeg$ then $\rho$ will replace in $\varbackupplan$
the earliest-deadline packet $\fstar$ in $\backupplanAftInitSeg(\beta, T]\setminus\planPaftInitSeg$.
Using the results from Section~\ref{subsec: backup plan and the potential function}, we show that
such packets exist and that these changes will indeed preserve the feasibility of $\varbackupplan$.

\smallskip

We now give a formal argument. The proof is organized into
two subcases, depending on whether some changes in $\varADV$ are needed or not.


\medskip
\noindent
\mycase{L.$(\beta,\gamma]$.S.1}
$\rho\notin \backupplanAftInitSeg$ and $\ADVaftInitSeg(\beta, \gamma]= \emptyset$.
(In particular, the case condition implies that $p\not\in \ADVaftInitSeg$, as $p\in \planPaftInitSeg(\beta, \gamma]$,
and that $\planPaftInitSeg(\beta, \gamma]\subseteq \backupplanAftInitSeg$; it is possible though that this
inclusion is strict.) Then we do not further change set $\varADV$, i.e., $\ADVbefNextStep = \ADVaftInitSeg$. 
Since $p\not\in \ADVbefNextStep$ and deadlines of all packets remain unchanged, invariant~{\InvariantA} holds.

Regarding $\varbackupplan$, note that $\rho\notin \ADVaftInitSeg$ by invariant~{\InvariantA}, and 
$p\in \backupplanAftInitSeg$, by invariant~{\InvariantB}(ii) and $p\notin \ADVaftInitSeg$.
Thus, we let $\backupplanBefNextStep = \backupplanAftInitSeg \setminus\braced{p}\cup \braced{\rho}$,
which preserves invariant~{\InvariantB}(ii).
Next, we need to show the feasibility of $\backupplanBefNextStep$, i.e., invariant~{\InvariantB}(i).
If $d_\rho \ge d_p$, then replacing $p$ by $\rho$ in $\varbackupplan$
preserves feasibility and thus, as $\backupplanAftInitSeg$ is feasible,
the new backup plan $\backupplanBefNextStep$ is feasible as well.
Otherwise, $d_\rho < d_p$. For $\tau\notin [d_\rho, d_p)$, it holds that
$\packetslack(\backupplanBefNextStep, \tau) = \packetslack(\backupplanAftInitSeg, \tau)\ge 0$,
by the feasibility of $\backupplanAftInitSeg$. Next, consider some $\tau\in [d_\rho, d_p) \subseteq (\beta, \gamma)$.
By the case condition, $\ADVaftInitSeg(\beta, \gamma] = \emptyset$,
thus using Lemma~\ref{lem:pslackRelation}(b) with $\zeta = \gamma$ we get that 
$\packetslack(\backupplanAftInitSeg, \tau)\ge \packetslack(\planPaftInitSeg, \tau)\ge 1$,
where the second inequality holds because $\tau\in (\beta, \gamma)$ is not a tight slot in $\planPaftInitSeg$.
Hence, $\packetslack(\backupplanBefNextStep, \tau) = \packetslack(\backupplanAftInitSeg, \tau) - 1\ge 0$
for $\tau\in [d_\rho, d_p)$, and we can conclude that invariant~{\InvariantB} holds after the step.

Finally, we prove inequality~(\ref{eq:leap-laterSegments}). According to the algorithm, we have $w^{t+1}_\rho = \minwt(\planPaftArrivals,d^t_\rho)$.
Then the terms on the left-hand side of~(\ref{eq:leap-laterSegments}) are
\begin{equation*} 
	\Delta_{(\beta,\gamma]}\Psi = \phinegonef \razy ( - w^t_p + w^{t+1}_\rho) \;,\quad
 	\Delta^t \weights = w^{t+1}_\rho - w^t_\rho \;, \quad\textrm{and}\quad
	\advcredit^t_{(\beta,\gamma]} = 0\;. 
\end{equation*}
Plugging these in, we obtain
\begin{align}
\Delta_{(\beta,\gamma]} \Psi - \phi\razy (\Delta^t \weights) - \advcredit^t_{(\beta,\gamma]}
		\;&=\; \phinegonef \razy ( - w^t_p + w^{t+1}_\rho) - \phi\razy ( w^{t+1}_\rho - w^t_\rho)  - 0
		\nonumber\\
		&=\; - \phinegonef \razy w^t_p - w^{t+1}_\rho + \phi\razy w^t_\rho
		\nonumber\\
		&\ge\; - \phinegonef\razy w^t_p - ( \phinegtwof\razy w^t_p + \phinegonef\razy w^t_\rho ) + \phi\razy w^t_\rho
		\label{eqn:processing beta-gamma S.1 step 3}
		\\
		&=\; - w^t_p + w^t_\rho\,, \nonumber
\end{align}
where inequality~\eqref{eqn:processing beta-gamma S.1 step 3} follows from
$w^{t+1}_\rho = \minwt(\planPaftArrivals,d^t_\rho) \le \phinegtwof w^t_p + \phinegonef w^t_\rho$,
by \eqref{eq:leapStepVsOmega}. This shows that~\eqref{eq:leap-laterSegments} holds.


\medskip
\noindent
\mycase{L.$(\beta,\gamma]$.S.2}
$\rho\in \backupplanAftInitSeg$ or $\ADVaftInitSeg(\beta, \gamma] \neq \emptyset$.
(Thus, there is a packet in $(\backupplanAftInitSeg\setminus\planPaftInitSeg)\cup \ADVaftInitSeg$
which has deadline in $(\beta, \gamma]$.)

\smallskip


\indentemparagraph{Changing sets $\varADV$ and $\varbackupplan$.} 
Let $\gstar$ be the latest-deadline packet in $\ADVaftInitSeg[t+1, \gamma]$.  
Packet $\gstar$ is well-defined: this is trivially true if the second condition of the case is
satisfied, and otherwise we have $\rho\in \backupplanAftInitSeg\setminus\planPaftInitSeg$,
in which case Observation~\ref{obs: A not empty vs B-P not empty}(b) with $\eta = \gamma$ implies that
$\ADVaftInitSeg[t+1, \gamma] \neq \emptyset$.
(We remark that $d_{\gstar}\le \beta$ if $\ADVaftInitSeg(\beta, \gamma] = \emptyset$.) 
Similarly, let $\fstar$ be the earliest-deadline packet in $\backupplanAftInitSeg(\beta, T]\setminus\planPaftInitSeg$.
We show that packet $\fstar$ is well-defined. This is trivially true if $\rho\in \backupplanAftInitSeg$, 
while otherwise it holds that $\ADVaftInitSeg(\beta, \gamma] \neq \emptyset$, and then Observation~\ref{obs: A not empty vs B-P not empty}(a)
with $\eta = \beta$ gives us that $\backupplanAftInitSeg(\beta, T]\setminus\planPaftInitSeg \neq \emptyset$.
(It is possible that $d_\fstar > \gamma$ if $\rho\notin \backupplanAftInitSeg$.)

We now define packets $g$ and $f$, and we modify $\varbackupplan$ and $\varADV$ as follows.
If $p\in \ADVaftInitSeg$, let $g = p$; otherwise let $g = \gstar$.
If $\rho\in \backupplanAftInitSeg$, let $f = \rho$; otherwise let $f = \fstar$.
We remove $g$ from $\varADV$, i.e., we set $\ADVbefNextStep = \ADVaftInitSeg\setminus\braced{g}$.
To preserve invariant~{\InvariantB}(ii), we must add $\rho$ to $\varbackupplan$ (if it's not there already) because $\rho\in \planPbefNextStep\setminus\ADVbefNextStep$. Hence, we let
$\backupplanBefNextStep = \backupplanAftInitSeg \cup\braced{g} \setminus \braced{f,p} \cup \braced{\rho}$;
this definition applies no matter whether $g=p$ or $g=\gstar$ and whether $f=\rho$ or $f=\fstar$.
(In particular, if $g=p$ and $f = \rho$ then $\varbackupplan$ will not change.)

\smallskip


\indentemparagraph{Deriving inequality~(\ref{eq:leap-laterSegments}).}
We first claim that $w_f^t\le w_\rho^t$. As $f\in\braced{\fstar,\rho}$, it is enough to prove that $w_\fstar^t\le w_\rho^t$. 
We have $d_\fstar > \beta \ge \alpha$ and, by the monotonicity property of $\varbackupplan \setminus \varplanP$
(see Section~\ref{subsubsec: leap step: a road map}), no packets were added to $\varbackupplan(\alpha,\timehorizon] \setminus \varplanP$
when processing segment $[t,\alpha]$, implying that $\fstar\in \backupplanAftArrivals(\beta,T]\setminus\planPaftArrivals$. 
Then the inequality $w_\fstar^t\le w_\rho^t$ follows from the choice of
$\rho = \substpacket^t(p)$ as the heaviest pending packet not in $\planPaftArrivals$ with deadline after slot $\beta$.

We also claim that $w_g^t\le w_p^t$. As $g\in\braced{\gstar,p}$, it is enough to prove that $w_\gstar^t\le w_p^t$.
Indeed, as we do not add any packet to $\varADV$ when processing segment $[t,\alpha]$,
it holds that $\gstar\in \ADVaftArrivals$ and thus, $\gstar\in \planPaftArrivals$ by invariant~{\InvariantA}.
The choice of $\gstar$ implies that  $\prevtightslot(\planPaftArrivals,d^t_{\gstar}) \le \prevtightslot(\planPaftArrivals,d^t_p)$, which
in turn implies that $w(\substpacket(\planPaftArrivals,\gstar)) \ge w_\rho^t$. Therefore, if we had $w_\gstar^t > w_p^t$ then
the algorithm would transmit $\gstar$ instead of $p$. 
(Strictly speaking, here were rely on the ``locality of changes'' invariant from Section~\ref{subsubsec: leap step: a road map}, namely that $\varplanP(\alpha,T]$ does not change when processing segment $[t,\alpha]$.)

According to the algorithm, $w^{t+1}_\rho = \minwt(\planPaftArrivals, d^t_\rho)$. 
Then the terms on the left-hand side of~(\ref{eq:leap-laterSegments}) are
\begin{align}
\Delta_{(\beta,\gamma]}\Psi \;&=\; \phinegonef \razy ( w_g^t - w^t_p - w_f^t + w^{t+1}_\rho) \,,
	\nonumber
\\
\Delta^t \weights \;&=\; w^{t+1}_\rho - w^t_\rho \,, \; \textrm{and}
			\nonumber
\\
\advcredit^t_{(\beta,\gamma]} \;&=\; w^t_g - \minwt(\planPaftArrivals, d^t_g) 
 \nonumber
\\
	 \;&\le\; w^t_g - \minwt(\planPaftArrivals, d^t_\rho) \;=\; w^t_g - w^{t+1}_\rho,
	\label{eqn: simple leap beta gamma 1}
\end{align}
where the inequality in line~\eqref{eqn: simple leap beta gamma 1} follows 
from $\nexttightslot(\planPaftArrivals, d^t_g) \le \gamma = \nexttightslot(\planPaftArrivals, d^t_\rho)$.
(For simplicity, the estimate for $\advcredit^t_{(\beta,\gamma]}$ is given above
with respect to $\planPaftArrivals$, as described in Section~\ref{subsec: adversary stash}.
To be completely formal, in this substep this credit would be with respect to the intermediate plan $\planPaftInitSeg$.
However, using $\planPaftArrivals$ instead of $\planPaftInitSeg$ cannot make the adversary credit smaller 
since $\minwt(\planPaftInitSeg, d^t_g)\ge \minwt(\planPaftArrivals, d^t_g)$, which holds as $\planPaftInitSeg = \planPaftArrivals\setminus\braced{\firstlightpacket}$
and as any tight slot of $\planPaftArrivals$ is also tight in $\planPaftInitSeg$.) Thus the changes described above give us that
\begin{align}
\Delta_{(\beta,\gamma]} \Psi &- \phi\razy (\Delta^t \weights) - \advcredit^t_{(\beta,\gamma]}
		\nonumber \\
	\;&\ge\; \phinegonef \razy ( w_g^t - w^t_p - w_f^t + w^{t+1}_\rho) - \phi\razy ( w^{t+1}_\rho - w^t_\rho) - ( w^t_g - w^{t+1}_\rho )
		\nonumber \\ 
	&=\; - (\, \phinegtwof\razy {w^t_g} + \phinegonef\razy {w^t_p} \,)
	 			+ (\, - \phinegonef\razy {w^t_f} + \phi\razy {w^t_\rho} \,)
		\nonumber \\
	&\ge\; - {w^t_p} + {w^t_\rho}\,,
	\label{eqn: simple leap beta gamma 2}
\end{align}
using bounds $w_f^t\le w_\rho^t$ and $w_g^t\le w_p^t$ in step~\eqref{eqn: simple leap beta gamma 2}.
This shows \eqref{eq:leap-laterSegments}.


\indentemparagraph{Preserving the invariants.}
Since after the changes it holds that $p\notin\ADVbefNextStep$, invariant~{\InvariantA} is maintained.
We now show that invariant~{\InvariantB} holds. That condition~{\InvariantB}(ii) is true follows
directly from the way we update $\varADV$ and $\varbackupplan$. Therefore, we will focus on
condition~{\InvariantB}(i), that is the feasibility of $\varbackupplan$.
If we have both $f=\rho$ and $g=p$ then $\backupplanBefNextStep = \backupplanAftInitSeg$, and
the feasibility of $\backupplanBefNextStep$ is trivial. Hence, from now on we assume that
$f = \fstar\neq\rho$ or  $g = \gstar \neq p$ (or both). 

\smallskip

We split the changes in $\varbackupplan$ into two parts (see Figure~\ref{fig: 5-5-5 two parts}):
\begin{enumerate}[nosep,label=(\roman*)]
\item First, we show that invariant~{\InvariantB} holds for the intermediate
backup plan $\backupplanLeapAftPartI = \backupplanAftInitSeg \cup\braced{g} \setminus \braced{f}$,
which we consider with respect to adversary stash $\ADVleapAftPartI = \ADVbefNextStep = \ADVaftInitSeg \setminus\braced{g}$
and plan $\planPleapAftPartI = \planPaftInitSeg$.
In other words, we first implement the removal of $g$ from $\varADV$
together with the corresponding changes in $\varbackupplan\setminus\varplanP$, but postpone the change in the plan.
\item Second, we replace $p$ by $\rho$ in $\varplanP$, so $\planPbefNextStep = \planPleapAftPartI \setminus\braced{p}\cup \braced{\rho}$
and $\backupplanBefNextStep = \backupplanLeapAftPartI \setminus \braced{p} \cup \braced{\rho}$,
and we show that $\backupplanBefNextStep$ is feasible as well.
\end{enumerate}


\begin{figure}[!ht]
\centering
\medskip
\includegraphics[width=5.5in]{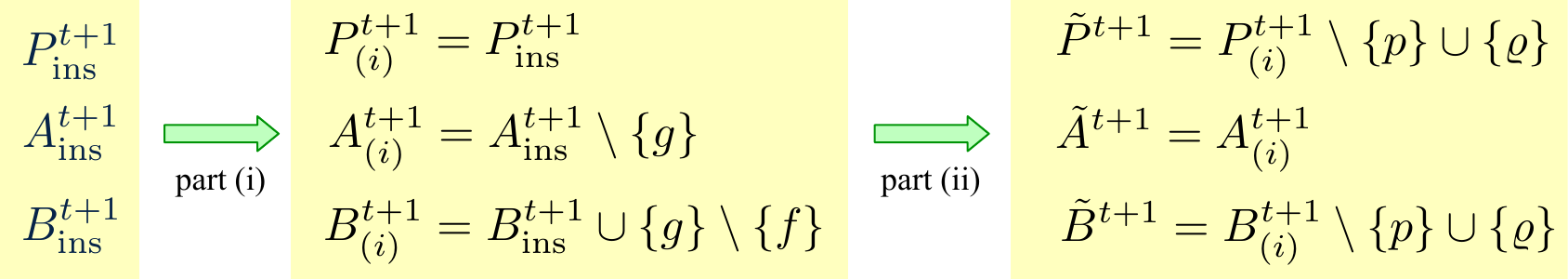}
\caption{Changing $\varbackupplan$ in two parts.
}
\label{fig: 5-5-5 two parts}
\end{figure}



\begin{claim}\label{clm:caseL1B-feasibilityOfC}
Assume that $f=\fstar\neq\rho$ or $g=\gstar\neq p$.
Set $\backupplanLeapAftPartI = \backupplanAftInitSeg \cup\braced{g} \setminus \braced{f}$ is a feasible
set of packets with respect to time $t+1$.
\end{claim}

\begin{proof} 
The claim clearly holds if $d_f\le d_g$, so for the rest of the proof we assume that $d_g < d_f$.
Recall that, according to the assumptions of our Case~{L.$(\beta,\gamma]$.S.2}, we have that either
$\rho\in \backupplanAftInitSeg$ or $\ADVaftInitSeg(\beta, \gamma] \neq \emptyset$ holds.

It is sufficient to show that $\pslack(\backupplanAftInitSeg, \tau)\ge 1$ for any $\tau\in [d_g, d_f)$,
as this will imply that $d_f \le \nexttightslot(\backupplanAftInitSeg, d_g)$, which in turn implies
the feasibility of $\backupplanLeapAftPartI$. We consider two cases.

First, assume that $\rho\notin \backupplanAftInitSeg$. This implies $f = \fstar$ and $\ADVaftInitSeg(\beta, \gamma] \neq \emptyset$,
which means that $d_g\in (\beta, \gamma]$, no matter whether $g = p$ or $g = \gstar$.
Then, for $\tau\in [d_g, d_f)$, Lemma~\ref{lem:pslackRelation}(a) with $\zeta = \beta$ gives us
\begin{equation*} 
\pslack(\backupplanAftInitSeg, \tau) \;\ge\; \pslack(\planPaftInitSeg, \tau) + |\ADVaftInitSeg(\beta,\tau]|
										\;\ge\; 1,
\end{equation*} 
since $g\in \ADVaftInitSeg(\beta,\tau]$ and $\pslack(\planPaftInitSeg, \tau)\ge 0$.

In the second case we have $\rho\in \backupplanAftInitSeg\setminus\planPaftInitSeg$, so $f = \rho$.
We then have $g = \gstar$ (as we assume that $f \neq\rho$ or $g \neq p$).
For $\tau\in [d_g, d_f)$, we use Lemma~\ref{lem:pslackRelation}(b) with $\zeta = \gamma$ to obtain
\begin{equation*}
\pslack(\backupplanAftInitSeg, \tau) \;\ge\; \pslack(\planPaftInitSeg, \tau) + |\backupplanAftInitSeg(\tau,\gamma]\setminus\planPaftInitSeg| 
						\;\ge\; 1,
\end{equation*}
where we use $f = \rho\in \backupplanAftInitSeg(\tau,\gamma]\setminus\planPaftInitSeg$ in the second inequality.
\end{proof}

From the definitions of $\planPleapAftPartI$, $\ADVleapAftPartI$, and $\backupplanLeapAftPartI$, and using
Claim~\ref{clm:caseL1B-feasibilityOfC}, invariants~{\InvariantA} and~{\InvariantB} hold for these three
intermediate sets.

Before replacing $p$ by $\rho$ in the optimal plan and in $\varbackupplan$, we prove the following claim:


\begin{claim}\label{clm:caseL1B-slackInC}
Assume that $f=\fstar\neq\rho$ or $g=\gstar\neq p$.
If $d_\rho < d_p$, then $\pslack(\backupplanLeapAftPartI, \tau)\ge 1$ for any $\tau\in [d_\rho, d_p)$.
\end{claim}

\begin{proof}
Let $\hatf$ be the earliest-deadline packet in $\backupplanLeapAftPartI(\beta, T] \setminus \planPleapAftPartI$ and
let $\hatg$ be the latest-deadline packet in $\ADVleapAftPartI[t+1, \gamma]$.
These definitions are similar to the definitions of $\fstar$ and $\gstar$, but they are not equivalent
because in part~(i) we could have removed $\fstar$ from $\varbackupplan$ and $\gstar$ from $\varADV$.
For this reason,
packets $\hatf$ and $\hatg$ may not actually exist, in which case, abusing notation, we will introduce artificial notations
for their deadlines:
if $\hatf$ does not exist, we let $d_\hatf = T$ and, similarly, if $\hatg$ does not exist, we let $d_\hatg = t+1$. 
We now observe that
\begin{equation}
\pslack(\backupplanLeapAftPartI,\tau)\ge \pslack(\planPleapAftPartI,\tau)\ge 1
\quad\textrm{for}\; \tau\in (\beta, \min(d_\hatf,\gamma))\cup [\max(d_\hatg,\beta+1),\gamma).
\label{eqn: obs 1 of caseL1B-slackInC}
\end{equation}
Indeed, the first inequality follows from Lemma~\ref{lem:pslackRelation}, which we can apply here 
because both invariants~{\InvariantA} and~{\InvariantB} are true for $\planPleapAftPartI$, $\ADVleapAftPartI$,
and $\backupplanLeapAftPartI$. Specifically, apply Lemma~\ref{lem:pslackRelation}(a) with $\zeta = \beta$,
obtaining the first inequality for $\tau\in (\beta, d_\hatf)$, and 
Lemma~\ref{lem:pslackRelation}(b) with $\zeta = \gamma$, to obtain the first inequality
for  $\tau\in (d_\hatg,\gamma)$. To justify the second inequality in~\eqref{eqn: obs 1 of caseL1B-slackInC},
note that, as $(\beta, \gamma]$ is a segment of $\planPleapAftPartI = \planPaftInitSeg$, we have
$\pslack(\planPleapAftPartI,\tau)\ge 1$ for $\tau\in (\beta, \gamma)$.

Since $[d_\rho, d_p)\subseteq (\beta,\gamma)$,
inequality~\eqref{eqn: obs 1 of caseL1B-slackInC} directly implies Claim~\ref{clm:caseL1B-slackInC}
when $d_\hatf \ge \gamma$ or $d_\hatg \le \beta$.
Therefore, for the rest of the argument, we will assume that $d_\hatf < \gamma$ and $d_\hatg > \beta$.

Next, observe that, since $\backupplanLeapAftPartI$ is obtained from $\backupplanAftInitSeg$ by replacing
$f$ by $g$, we have
\begin{equation}
\pslack(\backupplanLeapAftPartI,\tau) = \pslack(\backupplanAftInitSeg,\tau) + 1 \ge 1
\quad\textrm{for}\;  \tau\in [d_f, d_g),
\label{eqn: obs 2 of caseL1B-slackInC}
\end{equation}
where the inequality follows from the feasibility of $\backupplanAftInitSeg$. (It is possible that $d_g\le d_f$, in which case this interval is empty.)

From these two observations, \eqref{eqn: obs 1 of caseL1B-slackInC} and \eqref{eqn: obs 2 of caseL1B-slackInC},
it is now sufficient to show that
\begin{equation}
[d_\rho, d_p) \;\subseteq\; (\beta, d_\hatf)\cup  [d_f, d_g)  \cup  [d_\hatg,\gamma).
\label{eqn: last eqn of caseL1B-slackInC}
\end{equation}
Note that $d_\hatg \leq d_\gstar$
and $d_\hatf \geq d_\fstar$, by the definitions of these packets and of $\backupplanLeapAftPartI$ and $\ADVleapAftPartI$
(note that $g$ is not a candidate for $\hatf$ as $g\in \planPleapAftPartI$).
Now we consider three cases.

If $f = \rho$, then it holds that $g = \gstar$ by the claim's assumption, and interval $[d_\rho, d_p)$ is covered by
$[d_f, d_g) = [d_\rho, d_\gstar)$ and $[d_\hatg,\gamma)$, because $d_\hatg \le d_\gstar$ and $d_p\le \gamma$.

Similarly, if $g = p$, then we have $f = \fstar$, and interval $[d_\rho, d_p)$ is covered by
$[d_f, d_g) = [d_\fstar, d_p)$ and $(\beta, d_\hatf)$, because $d_\hatf \ge d_\fstar$.

The last case is when $f = \fstar$ and $g = \gstar$. Then interval $[d_\rho, d_p)\subseteq (\beta, \gamma)$ is covered by
$(\beta, d_\fstar)\cup [d_\fstar, d_\gstar)\cup [d_\gstar,\gamma)$,
which implies~\eqref{eqn: last eqn of caseL1B-slackInC}
because $d_\hatf \ge d_\fstar$ and $d_\hatg \le d_\gstar$.
\end{proof}

Finally, we show that the new backup plan $\backupplanBefNextStep = \backupplanLeapAftPartI \setminus \braced{p} \cup \braced{\rho}$
is feasible. If $d_\rho \ge d_p$, then the feasibility of $\varbackupplan$ is clearly preserved
after replacing $p$ by $\rho$. Otherwise, $d_\rho < d_p$, and for $\tau\in [d_\rho, d_p)$
we have $\pslack(\backupplanBefNextStep,\tau) = \pslack(\backupplanLeapAftPartI,\tau) - 1 \ge 0$ by Claim~\ref{clm:caseL1B-slackInC},
while $\pslack(\backupplanBefNextStep,\tau) = \pslack(\backupplanLeapAftPartI,\tau)\ge 0$ for $\tau\notin [d_\rho, d_p)$
by Claim~\ref{clm:caseL1B-feasibilityOfC}. Hence, invariant~{\InvariantB} holds after the step.

\subsubsection{Processing $(\beta,\gamma]$ in an Iterated Leap Step}
\label{subsubsec: processing (beta,gamma] in an iterated leap step} 

Here we address the last (and by far most involved) part of our argument, namely
the analysis of an iterated leap step. As a reminder, this step occurs when $k \ge 1$ in the algorithm, or,
equivalently, when $d_\rho$ is in a later segment of $\planPaftArrivals$ than $d_p$. 
The initial comments in Section~\ref{subsubsec: processing (beta,gamma] in a simple leap step} 
apply here as well. To recap: We assume that $\varADV$ and $\varbackupplan$ have already been modified as  
described in Sections~\ref{subsubsec: adversary step} (the adversary step)
and~\ref{subsubsec: leap step: processing S1} (processing segment $[t,\alpha]$). 
In particular, the time step has already been incremented to $t+1$ and the
current snapshots of ${\varplanP}$, ${\varbackupplan}$ and ${\varADV}$
are denoted $\planPaftInitSeg$, $\backupplanAftInitSeg$ and $\ADVaftInitSeg$, respectively. 
Recall that $\planPaftInitSeg= \planPaftArrivals \setminus \braced{\firstlightpacket}$.
We assume that both invariants~{\InvariantA} and~{\InvariantB} hold for these sets (with respect to time step $t+1$).

As before, $\beta = \prevtightslot^t(d_p)$ and $\gamma = \nexttightslot^t(d_\rho)$.  
As this is an iterated leap step, we have $d_\rho > \nexttightslot^t(d_p)$, so the
interval $(\beta,\gamma]$ of $\planPaftInitSeg$ is a union of two or more consecutive segments of $\planPaftInitSeg$. 
(The content of this interval is identical to the same interval in $\planPaftArrivals$; in fact, $\planPaftInitSeg$ and $\planPaftArrivals$ are
identical starting right after $\alpha = \nexttightslot^t(t)$, the first tight slot of $\planPaftArrivals$.)
Let $h_0 = p, h_1, \dots, h_k$ be the packets from Algorithm~{\PlanMonotonicity} (line~\ref{algLn:choosing_h_i})
and let $h_{k+1} = \rho$. All packets $h_0, h_1, \dots, h_k$ are in different segments of $\planPaftInitSeg$, not necessarily
consecutive, while $h_{k+1} = \rho$ has deadline in the same segment as $h_k$, but $\rho\notin \planPaftInitSeg$.
The deadlines and weights of packets $\rho$ and $h_1, \dots, h_k$ are still unchanged before processing $(\beta,\gamma]$
(that is, they are the same as in step $t$). 

We need to estimate the potential change due to weight increases and the changes
triggered by the replacement of $p$ by $\rho$ and by
the ``shifting'' of $h_i$'s. We also need to prove key inequality~\eqref{eq:leap-laterSegments}
and that invariants~{\InvariantA} and~{\InvariantB} hold after the step.
 
The overall structure of the argument in this case is similar to 
Case~L.$(\beta,\gamma]$.S (simple leap step), with the analysis split according to
whether changes in $\varADV$ are needed (Case~L.$(\beta,\gamma]$.I.2) or not (Case~L.$(\beta,\gamma]$.I.1).
In Case~L.$(\beta,\gamma]$.I.2, the analysis will be divided into several stages corresponding to
processing of different groups of segments.

For $i=0,\dots,k$, let $\mu_i = \minwt(\planPaftArrivals, d^t_{h_i})$; note that 
$w^t_\firstlightpacket \ge \mu_0 \ge \mu_1\ge \dots \ge \mu_k = \minwt(\planPaftArrivals, d_\rho)$ and
the algorithm ensures that $w^{t+1}_{h_i} \ge \mu_{i-1}$ for $i=1,\dots,k+1$.
In both cases below, we use the following simple bound on the change of weights of $h_i$'s.


\begin{lemma}\label{lem:changeOfWeightsOfHs}
For any $1\le a' \le b' \le k$, let $\Delta^t w({h_{a'}}, \dots, {h_{b'}})$
be the total amount by which the algorithm increases the weights of packets $h_{a'}, \dots, h_{b'}$ in step $t$.
Suppose that there exists $i\in [a',b']$ such that $w^t_{h_i} < \mu_{i-1}$, i.e.,
the algorithm increases the weight of $h_i$ in Line~11 (and thus $\Delta^t w({h_{a'}}, \dots, {h_{b'}}) > 0$).
Then $\Delta^t w({h_{a'}}, \dots, {h_{b'}}) \le \mu_{a'-1} - w^t_{h_{b'}}$.
\end{lemma}

\begin{proof}
Let $c \in [a', b']$ be the maximum index such that $w^t_{h_{c}} < \mu_{c-1}$;
such $c$ exists by the assumption of the lemma. 
We show the claim as follows:
\begin{align}
\Delta^t w({h_{a'}}, \dots, {h_{b'}}) &=   \sum_{i = a'}^{c} \max(\mu_{i-1}-w^t_{h_i}, 0) 
\nonumber                         
\\           
&\le   \sum_{i = a'}^{c - 1} \max(\mu_{i-1}-\mu_i, 0) + \max(\mu_{c-1}- w^t_{h_c}, 0)
\label{eqn: changeOfWeightsOfHs, ineq 1}                         
\\
&= \sum_{i = a'}^{c - 1} (\mu_{i-1} - \mu_i ) + \mu_{c-1} - w^t_{h_c}
\label{eqn: changeOfWeightsOfHs, ineq 2} 
\\ 
&=  \mu_{a'-1} - w^t_{h_c} 
\nonumber
\\
&\le \mu_{a'-1} - w^t_{h_{b'}}\,,
\label{eqn: changeOfWeightsOfHs, ineq 3} 
\end{align}  
where inequality~\eqref{eqn: changeOfWeightsOfHs, ineq 1} follows from $w^t_{h_i} \ge \mu_i$,
equality~\eqref{eqn: changeOfWeightsOfHs, ineq 2} from $\mu_{i-1}\ge \mu_i$ and
from $\mu_{c-1} > w^t_{h_c}$ (by the choice of $c$),
and inequality~\eqref{eqn: changeOfWeightsOfHs, ineq 3} from $w^t_{h_c} \ge w^t_{h_{b'}}$ by Lemma~\ref{lem:leapStep}(b) and $c\le b'$.
\end{proof}


\medskip
\noindent
\mycase{L.$(\beta,\gamma]$.I.1}  
$\rho\notin \backupplanAftInitSeg$ and $\ADVaftInitSeg(\beta, \gamma] = \emptyset$.
In this case, we do not make any changes in $\varADV$, i.e. $\ADVbefNextStep = \ADVaftInitSeg$.
From the case condition, $p = h_0,h_1,...,h_k\notin \ADVaftInitSeg$, because each $h_i$ is in $\planPaftInitSeg(\beta, \gamma]$. 
It follows that $\varADV$ is indeed not affected by the decrease of the deadlines of $h_i$'s.
Thus invariant~{\InvariantA} holds. Note that $\rho\notin \ADVaftInitSeg$ by invariant~{\InvariantA}.
From the case condition and $p\in \planPaftInitSeg(\beta,\gamma]$, we have $p\notin \ADVaftInitSeg$;
therefore  $p\in \backupplanAftInitSeg$ by invariant~{\InvariantB}(ii). This means that in order to preserve invariant~{\InvariantB}(ii) we
need to remove $p$ from $\varbackupplan$ and add $\varrho$ to it.
Thus we set $\backupplanBefNextStep = \backupplanAftInitSeg \setminus\braced{p}\cup \braced{\rho}$.


\indentemparagraph{Deriving inequality~(\ref{eq:leap-laterSegments}).}
To prove this inequality, we first show  that 
\begin{equation}
	\Delta^t w({h_1}, \dots, {h_k}) \;\le\; \mu_0 - \mu_k.
	\label{eqn: case L.betagamma.I.1 Delta weights bound}
\end{equation}
(The weight increase $\Delta^t w({h_1}, \dots, {h_k})$ was defined in Lemma~\ref{lem:changeOfWeightsOfHs}.)
Indeed, if there is no $i\in [1,k]$ such that $w^t_{h_i} < \mu_{i-1}$,
then $\Delta^t w({h_1}, \dots, {h_k}) = 0 \le \mu_0 - \mu_k$ as $\mu_0\ge \mu_k$.
Otherwise, we use Lemma~\ref{lem:changeOfWeightsOfHs} with $a'=1$ and $b'=k$ to get 
$\Delta^t w({h_1}, \dots, {h_k}) \le \mu_0 - w^t_{h_k} \le   \mu_0 - \mu_k$, where the last inequality follows from $w^t_{h_k} \ge \mu_k$.

In this case, the terms on the left-hand side of inequality~(\ref{eq:leap-laterSegments}) are
\begin{align*}
	\Delta_{(\beta,\gamma]} \Psi \;&=\; \phinegonef \razy \big( - w^t_p + w^{t+1}_\rho + \Delta^t w({h_1}, \dots, {h_k})\big)
	\\
	\Delta^t \weights \;&=\; w^{t+1}_\rho - w^t_\rho + \Delta^t w({h_1}, \dots, {h_k})
	\\
	\advcredit^t_{(\beta,\gamma]} \;&=\; 0
\end{align*}
The equation for $\Delta_{(\beta,\gamma]} \Psi$ is true because
all packets $h_1,...,h_k$ are both in $\planPaftInitSeg$ and in $\planPbefNextStep$, by Lemma~\ref{lem:leapStep}(a),
and none of them is in $\ADVaftInitSeg = \ADVbefNextStep$, so we keep all of them in $\varbackupplan$ to preserve invariant~{\InvariantB}(ii).
The last equation holds because  we have not changed the adversary stash. 
Plugging in these formulas, taking into account that $w_\rho^{t+1} = \mu_k$, and using inequality~\eqref{eqn: case L.betagamma.I.1 Delta weights bound}, we can obtain inequality~(\ref{eq:leap-laterSegments}) with an easy calculation:
\begin{align}
\Delta_{(\beta,\gamma]} \Psi &- \phi\razy (\Delta^t \weights) - \advcredit^t_{(\beta,\gamma]}
\nonumber
\\
&= \; \phinegonef \razy ( - w^t_p + w^{t+1}_\rho + \Delta^t w({h_1}, \dots, {h_k}))
\nonumber
\\
&\qquad\qquad
	- \phi\razy ( w^{t+1}_\rho - w^t_\rho + \Delta^t w({h_1}, \dots, {h_k}))  - 0
\nonumber
\\
&= \;- \phinegonef \razy w^t_p - \mu_k + \phi \razy w^t_\rho - \Delta^t w({h_1}, \dots, {h_k})
\nonumber
\\
&\ge \; - \phinegonef \razy w^t_p - \mu_k + \phi \razy w^t_\rho - (\mu_0 - \mu_k)
\nonumber
\\
&= \; - \phinegonef \razy w^t_p + \phi w^t_\rho - \mu_0
\nonumber
\\
&\ge\;  - \phinegonef \razy w^t_p + \phi \razy w^t_\rho - (\phinegtwof\razy w_p^t +  \phinegonef\razy w_\rho^t)
\label{eqn: case L.2.A, ineq 3} 
\\
&=\; - w^t_p + w^t_\rho\,,
\nonumber
\end{align}  
where inequality~\eqref{eqn: case L.2.A, ineq 3} follows from 
${\mu_0} = \minwt(\planPaftArrivals, d_p) \le \phinegtwof\razy w_p^t +  \phinegonef\razy w_\rho^t$, 
which is inequality~\eqref{eq:leapStepVsOmega} with $\tau = d_p$.


\indentemparagraph{Preserving invariant~{\InvariantB}.}
That invariant~{\InvariantB}(ii) is preserved follows immediately from the definitions of $\planPbefNextStep$,
$\ADVbefNextStep$, and $\backupplanBefNextStep$.
Next, we show invariant~{\InvariantB}(i), that is the feasibility of $\backupplanBefNextStep$,
and we split the argument into two parts. First, we 
show that the invariant holds for the intermediate backup plan resulting from replacing $p$ by $\rho$,
and then later we show that modifying the deadlines of the $h_i$'s preserves the invariant as well.

So consider an intermediate plan $\planPleapAftPartI = \planPaftInitSeg \setminus\braced{p}\cup \braced{\rho}$
and backup plan $\backupplanLeapAftPartI =  \backupplanAftInitSeg \setminus\braced{p}\cup \braced{\rho}$,
both with respect to time $t+1$ but with the deadlines of packets $h_1,h_2,...,h_k$ still unchanged.
Since this is an iterated leap step, we have that $d_\rho > d_p$, so
replacing $p$ by $\rho$ in $\varbackupplan$ preserves feasibility. Therefore $\backupplanLeapAftPartI$ is feasible.

Now, we consider $\backupplanBefNextStep$, that is the backup plan at step $t+1$, which
is obtained from $\backupplanLeapAftPartI$ by decreasing the deadlines of $h_i$'s, as in the algorithm.
As slots outside $(\beta,\gamma)$ are not affected, we have
$\pslack(\backupplanBefNextStep,\tau) = \packetslack(\backupplanLeapAftPartI, \tau)\ge 0$ for $\tau\notin (\beta, \gamma)$.
So it remains to consider slots $\tau\in (\beta, \gamma)$.

Recall that packets $h_1, \dots, h_k$ belong to both $\backupplanAftInitSeg$ and $\backupplanBefNextStep$.
We also have that $\pslack(\planPleapAftPartI, \tau)\ge 1$ for any $\tau\in (\beta, \gamma)$,
which holds as replacing $p$ by $\rho$ merges all segments between $\beta = \prevtightslot^t(d_p)$
and $\gamma = \nexttightslot^t(d_\rho)$. Since $\ADVleapAftPartI(\beta, \gamma] = \ADVaftInitSeg(\beta, \gamma] = \emptyset$,
Lemma~\ref{lem:pslackRelation}(b) with $\zeta = \gamma$ gives us that for any $\tau\in (\beta, \gamma)$,
\begin{equation}\label{eqn: case L.2.I.1, invB, ineq 1}
\packetslack(\backupplanLeapAftPartI, \tau)\ge \packetslack(\planPleapAftPartI, \tau)\ge 1\,.
\end{equation}

Next, we claim that for $\tau\in (\beta, \gamma)$,
\begin{equation}\label{eqn: case L.2.I.1, invB, ineq 2}
\pslack(\backupplanBefNextStep,\tau) \ge \packetslack(\backupplanLeapAftPartI, \tau) - 1\,.
\end{equation}
To justify this inequality, we observe that  $d^{t+1}_{h_i} = \tau_{i-1}\ge d^t_{h_{i-1}}$ for $i=1, \dots, k$,
which gives us that intervals $[d^{t+1}_{h_i}, d^t_{h_i})\subset (\beta, \gamma)$ are disjoint. This implies in turn
that, for $\tau\in (\beta, \gamma)$, decreasing the deadlines of all $h_i$'s can decrease $\packetslack(\calB,\tau)$ at most by one,
which yields inequality~\eqref{eqn: case L.2.I.1, invB, ineq 2}.

Inequalities~\eqref{eqn: case L.2.I.1, invB, ineq 1} and~\eqref{eqn: case L.2.I.1, invB, ineq 2} yield
$\pslack(\backupplanBefNextStep,\tau) \ge \packetslack(\backupplanLeapAftPartI, \tau) - 1 \ge 0$ for $\tau\in (\beta, \gamma)$,
completing the proof that $\backupplanBefNextStep$ satisfies invariant~{\InvariantB}(i).


\bigskip
\noindent
\mycase{L.$(\beta,\gamma]$.I.2}
$\rho\in \backupplanAftInitSeg$ or $\ADVaftInitSeg(\beta, \gamma] \neq \emptyset$.
In this case, both sets $\varADV$ and $\varbackupplan$ will be changed. We focus on the segments of $\planPaftInitSeg$ which
contain the packets $h_0 = p, h_1,...,h_k$. These are the segments of the current optimal plan that are modified by the algorithm
when $p$ is transmitted.
Specifically, for $i=0,\dots,k$, let $S_i$ be the segment of $\planPaftInitSeg$ that ends at $\tau_i = \nexttightslot^t(d^t_{h_i})$,
that is the segment containing $d^t_{h_i}$. (See Figure~\ref{fig:shift}.) 
Recall that $\prevtightslot^t(d_p)= \beta$, $\tau_k = \gamma$, and that we defined $h_{k+1} = \rho$.
We first organize segments $S_i$ into groups and then split the sequence of changes in this case into stages, having one stage for each group.


\paragraph{Groups.}
We start by defining a packet $g \in\ADVaftInitSeg$.
Let $\gstar$ be the latest-deadline packet in $\ADVaftInitSeg[t+1,\gamma]$.
Observe that packet $\gstar$ is well-defined. This is trivially true
if the second part of the case condition holds, while in the case when $\rho\in \backupplanAftInitSeg\setminus\planPaftInitSeg$,
we use Observation~\ref{obs: A not empty vs B-P not empty}(b) with $\eta = \gamma$ to get that
$\ADVaftInitSeg[t+1,\gamma]\neq \emptyset$. (It is possible that $d_{\gstar}\le \beta$ in the second case.) 

We now define $g$ as follows: If $d^t_{\gstar}$ is in a segment $S_i$ for which $h_i\in \ADVaftInitSeg$, then let $g = h_i$;
otherwise let $g= \gstar$. Observe that, in particular, if $h_k\in \ADVaftInitSeg$ then $g = h_k$.

We will process segments $S_i$ in \emph{groups}, where each group is specified by some non-empty
interval of indices $[a,b]\subseteq \braced{ 0,\dots,k}$ of segments $S_i$.
We will denote such a group by $\seggroup{a,b}$ (to avoid confusion with the notation for segments and time
intervals).
These intervals will be chosen so that they partition the set of all indices $\braced{ 0,\dots,k}$.
Roughly, we have a group for each $h_i\in \ADVaftInitSeg$ (that needs to be 
replaced in or removed from $\varADV$ because its deadline needs to be decreased), a special last group,
and possibly a special group at the beginning.
Let $i_1 < i_2 < \cdots < i_l$ be the indices of those packets $h_0 = p, h_1,...,h_k$ that are in $\ADVaftInitSeg$.   
Note that, if $l>0$, then $d^t_g \in [d^t(h_{i_l}),\gamma]$, because $h_{i_l}\in\ADVaftInitSeg$ is a
candidate for $\gstar$. In particular, since $g\in\ADVaftInitSeg$, we have that
$g\notin \braced{h_0, ... , h_k} \setminus \braced{h_{i_l}}$; that is, among 
\emph{all} packets $h_0,...,h_k$, $g$ may be possibly equal only to $h_{i_l}$. 
The definition of these groups depends on whether some packet $h_i$ is in $\ADVaftInitSeg$ (that is $l > 0$)
or none of them is in  $\ADVaftInitSeg$ (that is $l = 0$):

\begin{figure}
\begin{center}
	\includegraphics[width = 6in]{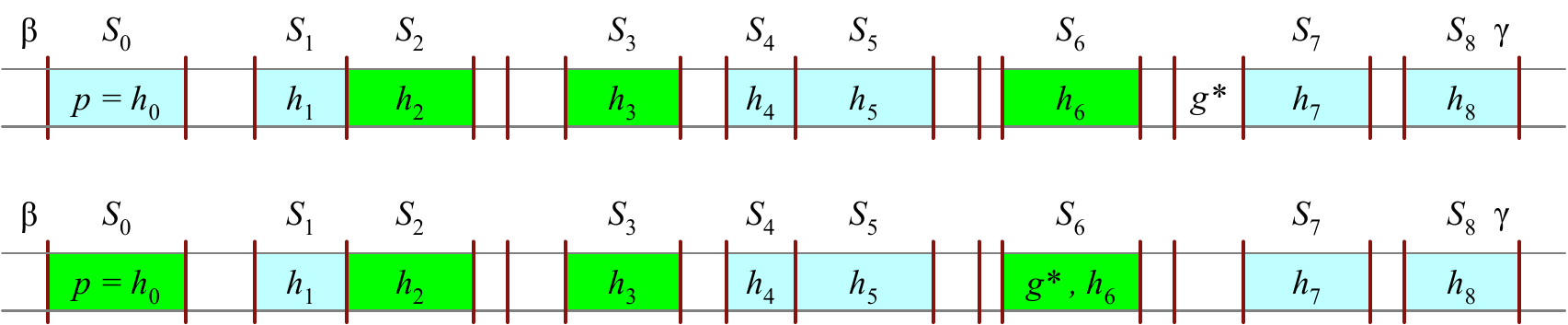}
	\caption{Two examples that illustrate assigning segments into groups. In both examples we have $k=8$.
	Segments are separated by vertical bars.
	Segments $S_i$ are colored (shaded). Green (or dark shaded) segments $S_i$ are those for which packet $h_i$ is in $\ADVaftInitSeg$.
	In the first example, we have $l = 3$, $(i_1,i_2,i_3) = (2,3,6)$, $g=\gstar$ (as $\gstar$ is not in $S_i$ for any $i$),
	and $\terminala = 7$. We have 5 groups: the initial group  $\seggroup{0,1}$, middle groups 
	$\seggroup{2,2}$, $\seggroup{3,5}$, and $\seggroup{6,6}$,
	and the terminal group $\seggroup{7,8}$.
	In the second example we have $l = 4$, $(i_1,i_2,i_3,i_4) = (0,2,3,6)$, and $g = h_6$ (because $\gstar$ is in $S_6$ and
	$h_6$ is in $\ADVaftInitSeg$). Since $i_1= 0$ there is no
	initial group. We have 4 groups: middle groups $\seggroup{0,1}$, $\seggroup{2,2}$, $\seggroup{3,5}$, and the terminal group $\seggroup{6,8}$.
	}
	\label{fig:groups}
\end{center}
\end{figure}

\begin{description}

\item{\mycase{$l > 0$}} If $l>1$, then for each $c =1, \dots, l-1$, the interval $\seggroup{i_c, i_{c+1} - 1}$ is a \emph{middle group}.
If $i_1 > 0$ (meaning that $h_0 = p\not\in \ADVaftInitSeg$),
then there is a special \emph{initial group} $\seggroup{0, i_1-1}$. If $i_1 = 0$ then there is no initial group.
Next, we assign the indices in $[i_l, k]$ to one or two groups and define index $\terminala$.
If $g = h_{i_l}$ (in particular, if $i_l=k$), then $\seggroup{i_l,k}$ is the \emph{terminal group} and we let $\terminala = i_l$. 
Otherwise, $g \neq h_{i_l}$ and
we let $\terminala$ be the smallest index in $0,...,k$ for which $\tau_\terminala \ge d^t_g$.   
The assumption that $g\neq h_{i_l}$ implies that $\terminala > i_l$.
Then $\seggroup{\terminala, k}$ is the terminal group and $\seggroup{i_l, \terminala - 1}$ is a new middle group. 
(See Figure~\ref{fig:groups} for two examples that illustrate the case when $l > 0$.)

\item{\mycase{$l = 0$}} We define one or two groups only.
There is the terminal group $\seggroup{\terminala, k}$, where $\terminala$ is again the smallest index 
in $0,...,k$ with $\tau_\terminala\ge d^t_g$. If $\terminala > 0$, we also have the initial group $\seggroup{0, \terminala - 1}$ .

\end{description}

Observe that we always define $\terminala$ so that $\seggroup{\terminala, k}$ is the terminal group.
We further define index $\initialb$ so that the first group is $\seggroup{0, \initialb}$;
for example, if $l>0$ and $i_1 > 0$, then $\initialb = i_1-1$.
Note that in almost every group $\seggroup{a,b}$ packet $h_a$ is in $\ADVaftInitSeg$; the only
two possible exceptions are (i) the initial group, and (ii) the terminal group in case when
either $l = 0$ or $l\ge 1$ and $\terminala \neq i_l$. 
On the other hand, for any group $\seggroup{a, b}$ packets $h_{a+1}, \dots, h_b$ are never in $\ADVaftInitSeg$
and thus these packets are in $\backupplanAftInitSeg$ by invariant~{\InvariantB}(ii).

%

 
\paragraph{Stages.}
The changes implemented during the processing of interval $(\beta, \gamma]$ involve the
replacement of $p$ by $\rho$ in $\varplanP$, weight increases of $\rho$ and $h_i$'s, and
deadline decreases of $h_i$'s, as described in the algorithm.
These changes will also trigger changes of $\varbackupplan$ and $\varADV$ necessary to preserve the invariants.
We will organize all these changes (of $\varplanP$, $\varbackupplan$, $\varADV$, and the instance itself)
into a sequence of \emph{stages}, with one stage for each group $\seggroup{a,b}$. 
The stages will be implemented in the reverse order of the groups,
that is, starting from the stage for the terminal group (which is defined in all cases),
continuing with middle groups in the reverse order of their indices,
and ending by the group containing index $0$, which may be of any type.

To understand better the partition into stages, as described below,
it helps to think of the substitution of $p$ by $\rho = h_{k+1}$ in the optimal
plan as a chain of substitutions: first replace $h_k$ by $\rho$, then replace $h_{k-1}$ by $h_k$, and so on, and
eventually replace $p = h_0$ by $h_1$. In our process here, rather than processing these substitutions individually,
we process the substitutions within each group $\seggroup{a,b}$ simultaneously as one stage, in which $h_a$ is replaced by $h_{b+1}$.
(See Figure~\ref{fig:group shift} for illustration.)

\begin{figure}[ht]
\begin{center}
	\includegraphics[width = 5.5in]{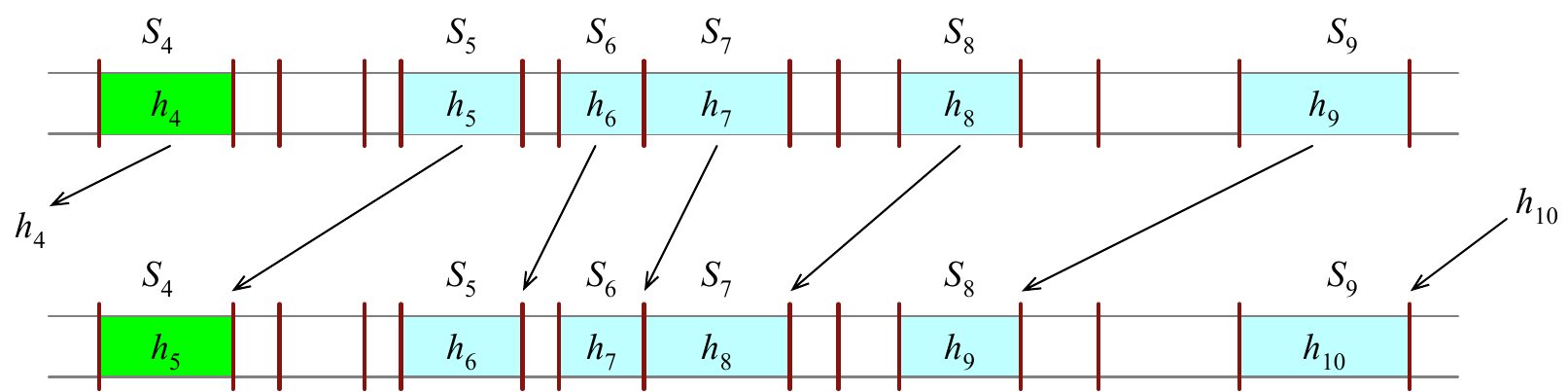}
	\caption{An illustration of the changes in $\varplanP$ in the stage when a middle group $\seggroup{a,b} = \seggroup{4,9}$
	is processed. 
	}
	\label{fig:group shift}
\end{center}
\end{figure}

Formally, let $\planPbefGroup{\seggroup{a,b}}$ be the snapshot of the intermediate plan $\varplanP$ just before the stage
for group $\seggroup{a,b}$ and let $\planPaftGroup{\seggroup{a,b}}$ be the snapshot of the intermediate plan $\varplanP$ just after the stage
for group $\seggroup{a,b}$.  (We remark that $\planPbefGroup{\seggroup{a,b}}$ is not necessarily an optimal plan. It is only meant to be
an ``intermediate plan'' used to capture gradual changes in the optimal plan as stages are executed. 
It is, however, feasible and it will satisfy
appropriate invariants.)  Naturally, if $\seggroup{a,b}$ and $\seggroup{b+1,c}$ are two consecutive groups, then 
$\planPaftGroup{\seggroup{b+1,c}} =\planPbefGroup{\seggroup{a,b}}$. 
We adopt the same conventions for sets $\varbackupplan$ and $\varADV$. In this notation
	(see Figure~\ref{fig:notation for stages} for illustration):
\begin{itemize}
\item
In the first stage, that corresponds to the terminal group $\seggroup{\terminala,k}$, we have 
$\planPbefGroup{\seggroup{\terminala,k}} = \planPaftInitSeg$,
$\backupplanBefGroup{\seggroup{\terminala,k}} = \backupplanAftInitSeg$ and $\AdvBefGroup{\seggroup{\terminala,k}} = \ADVaftInitSeg$. 
Note that in this stage, it holds that $h_{k+1} = \rho\notin \planPbefGroup{\seggroup{\terminala,k}}$.
\item
When describing the stage for each group $\seggroup{a,b}$, we assume that the changes in stages
for groups $\seggroup{a',b'}$ with $a' > b$ have already been implemented.
As an invariant we will ensure that $h_{b+1}\notin  \planPbefGroup{\seggroup{a,b}}$
(for $b = k$, this holds by $h_{k+1} = \rho \notin \planPbefGroup{\seggroup{\terminala,k}} = \planPaftInitSeg$).
In this stage, the plan changes from $\planPbefGroup{\seggroup{a,b}}$ to
$\planPaftGroup{\seggroup{a,b}} = \planPbefGroup{\seggroup{a,b}} \setminus \braced{h_a} \cup \braced{h_{b+1}}$, which
will imply that $h_a\notin \planPaftGroup{\seggroup{a,b}}$ just after this stage is processed.
Further, the weights and deadlines of packets $h_{a+1},...,h_{b+1}$ are adjusted according to the algorithm
(for $b=k$, the deadline of $h_{k+1} = \rho$ remains the same).
This stage also involves changes in $\varADV$ and $\varbackupplan$ triggered by the changes in $\varplanP$ and by the
adjustments of weights and deadlines.  These changes will convert
$\backupplanBefGroup{\seggroup{a,b}}$ and $\AdvBefGroup{\seggroup{a,b}}$ to $\backupplanAftGroup{\seggroup{a,b}}$ and $\AdvAftGroup{\seggroup{a,b}}$,
respectively.
\item
After the last stage, corresponding to the first group $\seggroup{0,\initialb}$ (which could be an initial, middle, or terminal group),
we have $\planPaftGroup{\seggroup{0,\initialb}} = \planPbefNextStep$, $\backupplanAftGroup{\seggroup{0,\initialb}} = \backupplanBefNextStep$,
and $\AdvAftGroup{\seggroup{0,\initialb}} = \ADVbefNextStep$.
\end{itemize}

\begin{figure}[ht]
\begin{center}
	\includegraphics[width = 5.7in]{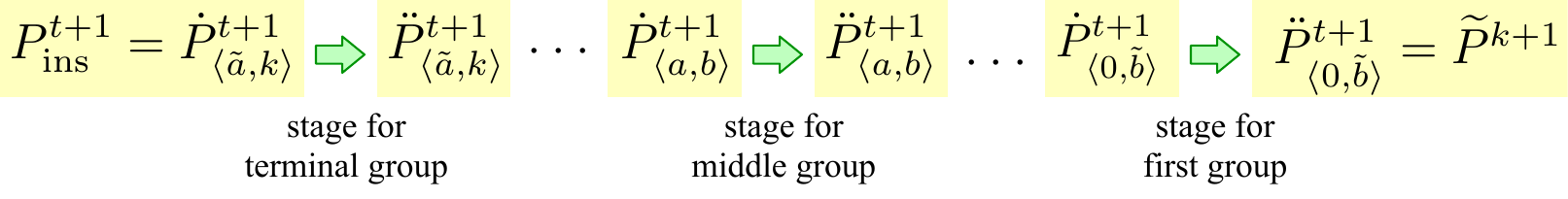}
	\caption{Notation for stages. 
	}
	\label{fig:notation for stages}
\end{center}
\end{figure}

By the results of Section~\ref{subsubsec: leap step: processing S1}, invariants~{\InvariantA} and {\InvariantB} hold for 
$\planPbefGroup{\seggroup{\terminala,k}}$, $\backupplanBefGroup{\seggroup{\terminala,k}}$, and $\AdvBefGroup{\seggroup{\terminala,k}}$.
In the stage for each group $\seggroup{a,b}$, we will show that these invariants are
preserved and that an appropriate bound on amortized profits holds. This will give us
that inequality~(\ref{eq:leap-laterSegments}) holds when processing interval $(\beta,\gamma]$
and that invariants~{\InvariantA} and {\InvariantB} hold afterwards, that is for
sets $\planPbefNextStep$, $\backupplanBefNextStep$, and $\ADVbefNextStep$.

Note that the changes of $\varplanP$ implemented by the algorithm
in the stage for group $\seggroup{a,b}$ satisfy
$$ \planPbefGroup{\seggroup{a,b}} \setminus \planPbefGroup{\seggroup{a,b}}(\eta, \tau_b]
= \planPaftGroup{\seggroup{a,b}} \setminus \planPaftGroup{\seggroup{a,b}}(\eta, \tau_b]$$
for $\eta = \prevtightslot(\planPbefGroup{\seggroup{a,b}}, d^t_{h_a})$,
which is the ``locality of changes'' invariant~\eqref{eqn: roadmap locality} from Section~\ref{subsubsec: leap step: a road map}. The $\pslack(\varplanP, \tau)$
values are actually affected only for $\tau\in [d^t_{h_a}, \tau_b)$.
(We remark that changes of $\varplanP$ involve packet $h_{b+1}$ with $d^t_{h_{b+1}} > \tau_b$ (if $b < k$),
however, $h_{b+1}\notin \planPbefGroup{\seggroup{a,b}}$ and we change the deadline of $h_{b+1}$ to $\tau_b$
during the stage for group $\seggroup{a,b}$.)

Finally, we will maintain the property that
after the stage for each group $\seggroup{a,b}$, packet $h_a$ will not be in $\AdvAftGroup{\seggroup{a,b}}$ ---
even though it may have been in $\varADV$ earlier. This is necessary for invariant~{\InvariantA}, because
$h_a\notin \planPaftGroup{\seggroup{a,b}}$. As also $\rho = h_{k+1}\notin \AdvBefGroup{\seggroup{\terminala,k}} = \ADVaftInitSeg$, this will give us that
$h_{b+1}\notin \AdvBefGroup{\seggroup{a,b}}$ for any group $\seggroup{a,b}$.
Recall that, importantly, packets $h_{a+1}, \dots, h_b$ are also not in $\AdvBefGroup{\seggroup{a,b}}$.


\indentemparagraph{Deriving inequality~(\ref{eq:leap-laterSegments}).}
We split the potential changes and the adversary credit for replacing packets in $\varADV$ among groups in a natural way.
Let $\Delta_{\seggroup{a,b}} \Psi$ be the total change of the potential due to the modifications
implemented in the stage for group $\seggroup{a,b}$, and let $\advcredit^t_{\seggroup{a,b}}$ be the adversary adjustment credit
for the changes of $\varADV$ implemented in the stage for group $\seggroup{a,b}$
(that is, for removing a packet or for replacing $h_a$ by a lighter packet).
Recall that $\Delta^t w({h_{a+1}}, \dots, {h_{b+1}})$
is the total amount by which the algorithm increases the weights of $h_{a+1}, \dots, h_{b+1}$.
Our goal is to prove for each group $\seggroup{a,b}$ that
\begin{equation}\label{eq:leap-group}
\Delta_{\seggroup{a,b}} \Psi - \phi\razy \Delta^t w({h_{a+1}}, \dots, {h_{b+1}}) - \advcredit^t_{\seggroup{a,b}} 
				\ge - w^t_{h_a} + w^t_{h_{b+1}}\,.
\end{equation}
(Note that the right-hand side of \eqref{eq:leap-group} is negative, by Lemma~\ref{lem:leapStep}(b).)

It remains to observe that the sum of~\eqref{eq:leap-group} over all groups gives us exactly the
key inequality~\eqref{eq:leap-laterSegments}. Indeed, the right-hand sides sum to
$- w^t_{h_0} + w^t_{h_{k+1}} = - w^t_p + w^t_{\rho}$.
Regarding the left-hand side of~\eqref{eq:leap-group}, $\Delta_{(\beta,\gamma]} \Psi$ equals the sum of $\Delta_{\seggroup{a,b}} \Psi$
over all groups $\seggroup{a,b}$, and similarly for $\Delta^t \weights$ and $\advcredit^t_{(\beta,\gamma]}$.
(Note that the increase of the weight of $\rho = h_{k+1}$ is accounted for in inequality~\eqref{eq:leap-group}
for the terminal group $\seggroup{\terminala,k}$.)


\paragraph{Stage for the terminal group.}
Let $\seggroup{\terminala,k}$ be the terminal group of segments.
In this stage, we change plan $\varplanP$ by removing $h_\terminala$ and adding $\rho$,
we increase the weights of packets $h_{\terminala+1}, \dots, h_{k+1} = \rho$, and we decrease the deadlines of 
packets $h_{\terminala+1}, \dots, h_k$.
We thus have $\planPaftGroup{\seggroup{\terminala,k}} = \planPbefGroup{\seggroup{\terminala,k}} \setminus \braced{h_\terminala} \cup \braced{\rho}$.

Let $\fstar$ be the earliest-deadline packet in $\backupplanBefGroup{\seggroup{\terminala,k}}(\beta,\timehorizon]\setminus\planPbefGroup{\seggroup{\terminala,k}}$. 
We argue that packet $\fstar$ is well-defined, which is trivially true if
$\rho\in \backupplanAftInitSeg(\beta,\timehorizon] = \backupplanBefGroup{\seggroup{\terminala,k}}(\beta,\timehorizon]$.
Otherwise, $\ADVaftInitSeg(\beta, \gamma] = \AdvBefGroup{\seggroup{\terminala,k}}(\beta,\gamma] \neq \emptyset$
and then Observation~\ref{obs: A not empty vs B-P not empty}(a) with $\eta = \beta$ implies that 
$\backupplanBefGroup{\seggroup{\terminala,k}}(\beta,\timehorizon]\setminus\planPbefGroup{\seggroup{\terminala,k}}\neq \emptyset$.
(Possibly, $d_\fstar > \gamma$ in the second case.)

We then define a packet $f\in \backupplanBefGroup{\seggroup{\terminala,k}}\setminus\planPbefGroup{\seggroup{\terminala,k}}$, letting $f = \rho$ if
$\rho\in \backupplanBefGroup{\seggroup{\terminala,k}}$, and  $f = \fstar$ otherwise.  

Next, we modify sets $\varADV$ and $\varbackupplan$.
We remove $g$ from $\varADV$, i.e., $\AdvAftGroup{\seggroup{\terminala,k}} = \AdvBefGroup{\seggroup{\terminala,k}} \setminus \braced{g}$. This
ensures that invariant~{\InvariantA} is preserved,
because $h_\terminala\notin \AdvAftGroup{\seggroup{\terminala,k}}$ and
none of packets $h_{\terminala+1}, \dots, h_k$ is in $\AdvAftGroup{\seggroup{\terminala,k}}$ (by the definition of groups), thus
the decrease of their deadlines, implemented in this stage, does not affect $\varADV$.

Regarding $\varbackupplan$, note that $\rho\in \planPaftGroup{\seggroup{\terminala,k}}\setminus \AdvAftGroup{\seggroup{\terminala,k}}$, so we need
to include $\rho$ in $\backupplanAftGroup{\seggroup{\terminala,k}}$ to preserve invariant~{\InvariantB}(ii). Hence, we let
$\backupplanAftGroup{\seggroup{\terminala,k}} = \backupplanAftInitSeg \cup\braced{g} \setminus \braced{f,h_\terminala} \cup \braced{\rho}$;
this definition applies no matter whether $g=h_\terminala$ or $g=\gstar$ and whether $f=\rho$ or $f=\fstar$.  

\smallskip


\indentemparagraph{Deriving~\eqref{eq:leap-group} for the terminal group.}
Apart from the changes in the paragraph above,
the algorithm changes the weights of packets $h_{\terminala+1}, \dots, h_{k+1} = \rho$, which results in the
increase of the weight of the optimal plan.  Taking all these changes into account, the terms
on the left-hand side of~\eqref{eq:leap-group} are:
\begin{align}
\Delta_{\seggroup{\terminala,k}}\Psi 
			\;&=\; \phinegonef\razy \left( w^t_{g} - w^t_f - w^t_{h_\terminala} + w^{t+1}_\rho + \Delta^t w({h_{\terminala+1}}, \dots, {h_k}) \right)
			\nonumber
	 \\
	    	&=\; \phinegonef\razy \left( w^t_{g} - w^t_f - w^t_{h_\terminala} + \mu_k + \Delta^t w({h_{\terminala+1}}, \dots, {h_k}) \right)
			\label{eqn: 5.5.6 terminal delta psi}
	 \\
	\Delta^t w({h_{\terminala+1}}, \dots, {h_{k+1}}) 
			\;&=\; \Delta^t w({h_{\terminala+1}}, \dots, {h_k}) + w^{t+1}_\rho - w^t_\rho
			\nonumber
	\\
			\;&=\; \Delta^t w({h_{\terminala+1}}, \dots, {h_k}) + \mu_k - w^t_\rho
			\label{eqn: 5.5.6 terminal delta w}
	\\
			\advcredit^t_{\seggroup{\terminala,k}} \;&=\; w^t_g - \minwt(\planPaftArrivals, d^t_g)
			\label{eqn: 5.5.6 terminal advgain}
\end{align}
Equation~\eqref{eqn: 5.5.6 terminal delta psi} follows from the definition of $\backupplanAftGroup{\seggroup{\terminala,k}}$,
as $\rho$ is added to $\varbackupplan$ with its new weight $w^{k+1}_\rho = \mu_k$ and 
packets ${h_{\terminala+1}}, \dots, {h_k}$ remain in $\varbackupplan$, although their weights are changed.
(Recall that for each group $\seggroup{a,b}$, packets $h_{a+1},...,h_b$ are always in $\varbackupplan$.)
Equation~\eqref{eqn: 5.5.6 terminal delta w} holds because $w^{k+1}_\rho = \mu_k$. 
In $\varADV$, we just removed $g$, implying~\eqref{eqn: 5.5.6 terminal advgain}.
(See Section~\ref{subsec: adversary stash} for the definition of the adversary adjustment credit associated with the
removal of packets from  $\varADV$.)

We further claim that the following inequalities hold:
\begin{align}
	\Delta^t w({h_{\terminala+1}}, \dots, {h_k}) \;&\le\; \mu_\terminala - \mu_k
		\label{eqn: 5.5.6 terminal delta w inequality}
	\\
	 w^t_g \;&\le\; w^t_{h_\terminala}
	 \label{eqn: 5.5.6 terminal w(g) inequality}
	 \\
	 w^t_f\; &\le\; w^t_\rho
	  \label{eqn: 5.5.6 terminal w(f) inequality}
	  \\
	  \minwt(\planPaftArrivals, d^t_g) \;&\ge\; \mu_\terminala
		\label{eqn: 5.5.6 terminal minwt inequality}
\end{align}
To show~\eqref{eqn: 5.5.6 terminal delta w inequality}, we consider two
simple cases. If there is no $i\in \seggroup{\terminala+1,k}$ such that $w^t_{h_i} < \mu_{i-1}$,
then $\Delta^t w({h_{\terminala+1}}, \dots, {h_k}) = 0 \le \mu_\terminala - \mu_k$ as $\mu_\terminala\ge \mu_k$.
Otherwise, we use Lemma~\ref{lem:changeOfWeightsOfHs} with $a'=\terminala+1$ and $b'=k$ to get 
$\Delta^t w({h_{\terminala+1}}, \dots, {h_k}) \le \mu_\terminala - w^t_{h_k} \le \mu_\terminala - \mu_k$,
where the last inequality follows from $w^t_{h_k} \ge \mu_k$. 

Inequality~\eqref{eqn: 5.5.6 terminal w(g) inequality} is trivial if $g=h_{i_l}$, in which case also $\terminala = i_l$.
Otherwise, $\terminala > i_l$. By the definition of the terminal group,
$\terminala$ is the smallest index with $\tau_\terminala \ge d^t_g$,  and thus $d^t_g > \tau_{\terminala-1}$.
We also have that $h_\terminala$ is the heaviest packet in plan $\planPbefGroup{\seggroup{\terminala,k}} = \planPaftInitSeg$ with deadline in $(\tau_{\terminala-1},\gamma]$.
(This property holds for plan $\planPaftArrivals$, by the algorithm. However,
the changes in Section~\ref{subsubsec: leap step: processing S1}, where we process $[t,\alpha]$,
do not affect packets in the optimal plan after the first tight slot $\alpha\le \beta$;
see also the ``locality of changes'' invariant in Section~\ref{subsubsec: leap step: a road map}.)
As $g\in \planPbefGroup{\seggroup{\terminala,k}}$ and $d^t_g\in (\tau_{\terminala-1}, \gamma]$ 
(by the definition of $\terminala$), we conclude that $w^t_g\le w^t_{h_\terminala}$. 

Next, we justify~\eqref{eqn: 5.5.6 terminal w(f) inequality}. Indeed, since $d^t_f > \beta$
and the only packet possibly added to $\varbackupplan$ in the analysis of the transmission step
is $\firstlightpacket$ with $d^t_\firstlightpacket \le \alpha \le \beta$,
we have that $f\notin \planPaftArrivals$ and then the inequality holds by the definition of $\rho$.

Finally, inequality~\eqref{eqn: 5.5.6 terminal minwt inequality} follows from $\tau_\terminala\ge d^t_g$
and the monotonicity of the $\minwt()$ function.

Using the above equations and inequalities, we derive inequality~\eqref{eq:leap-group} as follows:
\begin{align}
\Delta_{\seggroup{\terminala,k}}\Psi \;&- \phi\razy \Delta^t w({h_{\terminala+1}}, \dots, h_{k+1}) - \advcredit^t_{\seggroup{\terminala,k}}
\nonumber\\
& =\; \phinegonef\razy \left( w^t_{g} - w^t_f - w^t_{h_\terminala} + \mu_k + \Delta^t w({h_{\terminala+1}}, \dots, {h_k}) \right) 
\nonumber\\
&\quad\quad	 - \phi\razy (\, \Delta^t w({h_{\terminala+1}}, \dots, {h_k}) + \mu_k - w^t_\rho \,) - (\, w^t_g - \minwt(\planPaftArrivals, d^t_g) \,)
\nonumber\\
& =\; - \phinegtwof\razy {w^t_g} - \phinegonef\razy {w^t_f} - \phinegonef\razy w^t_{h_\terminala}
	- \Delta^t w({h_{\terminala+1}}, \dots, {h_k}) 
\nonumber\\
&\qquad\qquad
	 - \mu_k + \phi\razy w^t_\rho + \minwt(\planPaftArrivals, d^t_g)
\nonumber\\
& \ge\;  - \phinegtwof\razy {w^t_{h_\terminala}} - \phinegonef\razy {w^t_\rho} - \phinegonef\razy w^t_{h_\terminala}
	 - (\mu_\terminala - \mu_k) - \mu_k + \phi\razy w^t_\rho + \mu_\terminala
\label{eqn: beta gamma terminal 1}
\\
& =\; - w^t_{h_\terminala} + w^t_{h_{k+1}}\,,
\label{eqn: beta gamma terminal 2}
\end{align}
where in inequality~\eqref{eqn: beta gamma terminal 1} we use~\eqref{eqn: 5.5.6 terminal delta w inequality},
\eqref{eqn: 5.5.6 terminal w(g) inequality}, \eqref{eqn: 5.5.6 terminal w(f) inequality}, 
and~\eqref{eqn: 5.5.6 terminal minwt inequality}, and in step~\eqref{eqn: beta gamma terminal 2}  we substituted $w^t_\rho = w^t_{h_{k+1}}$.

\smallskip

 
\subparagraph{Invariant~{\InvariantB} after the stage for the terminal group}
We claim that after the stage for the terminal group, invariant~{\InvariantB} holds for backup plan $\backupplanAftGroup{\seggroup{\terminala,k}}$.
(The proof given here is an extension of the one for the case~L.$[\beta,\gamma)$.S.2
in Section~\ref{subsubsec: processing (beta,gamma] in a simple leap step}.) 
To show this, we split all changes of this stage into three parts (see Figure~\ref{fig: 5.5.6 terminal three parts}):
\begin{enumerate}[nosep,label=(\roman*)]
\item First, we  implement the changes in $\varADV$ and $\varbackupplan\setminus\varplanP$, but postpone 
changes in $\varplanP$, that is, we consider sets
$\planPleapAftPartI = \planPbefGroup{\seggroup{\terminala,k}}$,
$\backupplanLeapAftPartI = \backupplanBefGroup{\seggroup{\terminala,k}} \cup\braced{g} \setminus \braced{f}$, and 
$\ADVleapAftPartI  = \AdvBefGroup{\seggroup{\terminala,k}} \setminus \braced{g} = \AdvAftGroup{\seggroup{\terminala,k}}$.
The deadlines of packets $h_{\terminala+1}, \dots, h_k$ are yet unchanged.
\item Second, we replace $h_\terminala$ by $\rho$ in $\varplanP$, but still not implement the deadline changes,
so $\planPleapAftPartII = \planPleapAftPartI \setminus\braced{h_\terminala}\cup \braced{\rho}$,
$\backupplanLeapAftPartII = \backupplanLeapAftPartI \setminus \braced{h_\terminala} \cup \braced{\rho}$, 
and  $\ADVleapAftPartII = \ADVleapAftPartI$.
\item Third, we decrease the deadlines of packets $h_{\terminala+1}, \dots, h_k$, ending up with plan $\planPaftGroup{\seggroup{\terminala,k}}$
and backup plan $\backupplanAftGroup{\seggroup{\terminala,k}}$. 
(Note that $\varADV$ and $\varbackupplan\setminus\varplanP$ are not affected by the deadline decreases,
since packets $h_{\terminala+1}, \dots, h_k$ are in $\planPleapAftPartII\setminus \ADVleapAftPartII$.)
\end{enumerate}
%

\begin{figure}[ht]
\begin{center}
	\includegraphics[width = 6.4in]{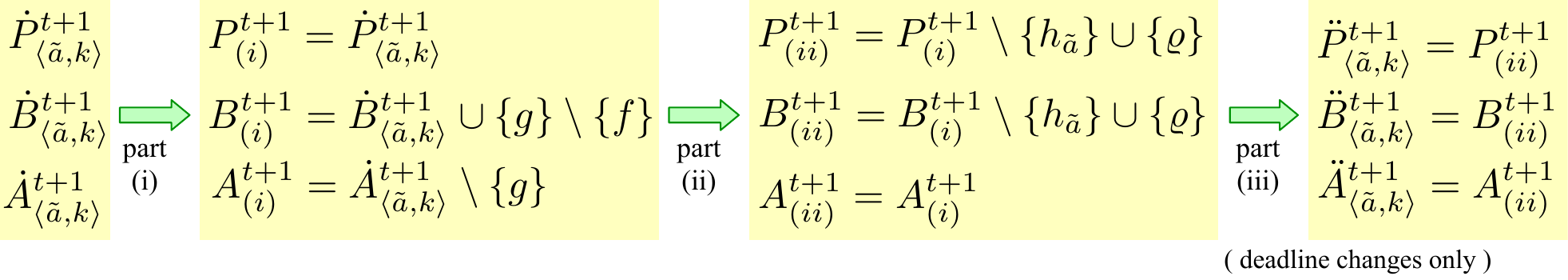}
	\caption{Partition of the stage for the terminal group into three parts in the
		proof that invariant~{\InvariantB} is preserved.
	}
	\label{fig: 5.5.6 terminal three parts}
\end{center}
\end{figure}

In the argument for the first two parts, we can assume that $f=\fstar\neq \rho$ or $g=\gstar\neq h_\terminala$ (or both),
since if both $f = \rho$ and $g = h_\terminala$ then 
$\backupplanLeapAftPartII= \backupplanBefGroup{\seggroup{\terminala,k}} = \backupplanAftInitSeg$,
in which case invariant~{\InvariantB} holds after part~(ii), as~{\InvariantB}(i) is trivial and~{\InvariantB}(ii)
follows from the way $\varplanP$ and $\varADV$ are modified in parts~(i) and~(ii).
(Recall that $g \in\braced{ \gstar, h_{\terminala} }$ and that $g = h_{\terminala}$ iff $\terminala = i_l$.)


\begin{claim}\label{clm:caseL2B-feasibilityOfC}
Assume that $f=\fstar\neq \rho$ or $g=\gstar\neq h_\terminala$.
Then set $\backupplanLeapAftPartI = \backupplanBefGroup{\seggroup{\terminala,k}} \cup\braced{g} \setminus \braced{f}$ is a feasible
set of packets pending at time $t+1$.
\end{claim}

\begin{proof} 
Recall that $\planPbefGroup{\seggroup{\terminala,k}} = \planPaftInitSeg$ and that the same relations hold for
sets $\varbackupplan$ and $\varADV$. Also, we are currently in Case~L.$(\beta,\gamma]$.I.2, in which
either $\rho\in \backupplanAftInitSeg$ or $\ADVaftInitSeg(\beta, \gamma] \neq \emptyset$.

The claim clearly holds if $d^t_f\le d^t_g$.
Otherwise, by the definition of tight slots, it suffices to prove $d^t_f \le \nexttightslot(\backupplanAftInitSeg, d^t_g)$, from which
the feasibility of $\backupplanLeapAftPartI$ follows.
In particular, we show that $\pslack(\backupplanAftInitSeg, \tau)\ge 1$ for any $\tau\in [d^t_g, d^t_f)$.

We now consider two cases.
First, assume that $\rho\notin \backupplanAftInitSeg$.
This implies $f = \fstar$ and $\ADVaftInitSeg(\beta, \gamma] \neq \emptyset$,
which means that $d^t_g\in (\beta, \gamma]$, no matter whether $g = h_\terminala$ or $g = \gstar$.
Then, for $\tau\in [d^t_g, d^t_f)$, Lemma~\ref{lem:pslackRelation}(a) with $\zeta = \beta$ gives us 
\begin{equation*}
\pslack(\backupplanAftInitSeg, \tau) 
		\;\ge\; \pslack(\planPaftInitSeg, \tau) + |\ADVaftInitSeg(\beta,\tau]|
		\;\ge\; 1,
\end{equation*}
since $\ADVaftInitSeg(\beta,\tau]$ contains $g$ and $\pslack(\planPaftInitSeg, \tau)\ge 0$.

Otherwise, $\rho\in \backupplanAftInitSeg\setminus\planPaftInitSeg$, so $f = \rho$.
In this case, $g = \gstar$ holds (as we assume in this part that $f = \fstar\neq \rho$ or $g = \gstar\neq h_\terminala$).
For $\tau\in [d^t_g, d^t_f)$, we use Lemma~\ref{lem:pslackRelation}(b) with $\zeta = \gamma$ to obtain
\begin{equation*}
\pslack(\backupplanAftInitSeg, \tau) 
			\;\ge\; \pslack(\planPaftInitSeg, \tau) + |\backupplanAftInitSeg(\tau,\gamma]\setminus\planPaftInitSeg| 
			\;\ge\; 1,  
\end{equation*}
where we use $f = \rho\in \backupplanAftInitSeg(\tau,\gamma]\setminus\planPaftInitSeg$ in the second inequality.
\end{proof}

From the definitions of $\planPleapAftPartI$, $\ADVleapAftPartI$, and $\backupplanLeapAftPartI$, and using
Claim~\ref{clm:caseL2B-feasibilityOfC}, invariants~{\InvariantA} and~{\InvariantB} hold for these three
intermediate sets.

Before considering part~(ii) (that is, replacing $h_\terminala$ by $\rho$ in $\varplanP$ and $\varbackupplan$),
we prove the following claim:


\begin{claim}\label{clm:caseL2B-slackInC}
Assume that $f=\fstar\neq \rho$ or $g=\gstar\neq h_\terminala$. If $d^t_\rho < d^t_{h_\terminala}$, then $\terminala = k$ and
$\pslack(\backupplanLeapAftPartI, \tau)\ge 1$ for any $\tau\in [d^t_\rho, d^t_{h_\terminala})$.
\end{claim}

\begin{proof}
Let $\chi = \prevtightslot^t(d^t_\rho)$.
By the definitions of $\gamma$ and $\rho$ and since the tight slots of $\planPaftInitSeg$ are the same as those of $\planPaftArrivals$
starting with slot $\alpha$, $d^t_\rho$ is in the segment $(\chi,\gamma]$ of $\planPleapAftPartI = \planPaftInitSeg$; 
consequently, $d^t_{h_\terminala}$ also belongs to this segment. So $\terminala = k$, proving the first part of the claim.
Moreover, it follows that $\pslack(\planPleapAftPartI,\tau)\ge 1$ for $\tau\in [d^t_\rho, d^t_{h_k})$.

We now make two observations that will be used later to prove Claim~\ref{clm:caseL2B-slackInC}.
First, observe that replacing $f$ by $g$ in $\varbackupplan$ results in
\begin{equation}
\pslack(\backupplanLeapAftPartI,\tau) = \pslack(\backupplanAftInitSeg,\tau) + 1 \ge 1 \quad\textrm{for}\; \tau\in [d^t_f, d^t_g) \,,
\label{eqn: claim caseL2B-slackInC 1}
\end{equation}
where the inequality follows from the feasibility of $\backupplanAftInitSeg$.   
(It is possible that $d^t_g \le d^t_f$, in which case this interval is empty.)

Next, let $\hatf$ be the earliest-deadline packet in $\backupplanLeapAftPartI(\chi, T] \setminus \planPleapAftPartI $ and
let $\hatg$ be the latest-deadline packet in $\ADVleapAftPartI[t+1, \gamma]$;
if $\hatf$ is undefined, we let $d^t_\hatf = T$, and similarly if $\hatg$ is undefined, we let $d^t_\hatg = t+1$.  
By Lemma~\ref{lem:pslackRelation}
(which we can apply because invariants~{\InvariantA} and~{\InvariantB} are true for $\planPleapAftPartI$, $\ADVleapAftPartI$, and $\backupplanLeapAftPartI$),
it holds that
\begin{equation}
\pslack(\backupplanLeapAftPartI,\tau)\ge  \pslack(\planPleapAftPartI,\tau)\ge 1
\quad\textrm{for}\; \tau\in (\chi, \min(d^t_\hatf,\gamma))\cup [\max(d^t_\hatg,\chi+1),\gamma).
\label{eqn: claim caseL2B-slackInC 2}
\end{equation}
Specifically, 
we apply Lemma~\ref{lem:pslackRelation}(a) with $\zeta = \chi$ to derive the first inequality for $\tau\in (\chi, d^t_\hatf)$,
and Lemma~\ref{lem:pslackRelation}(b) with $\zeta = \gamma$ to derive the first inequality for $\tau\in [d^t_\hatg,\gamma)$.
The second inequality holds because $(\chi,\gamma]$ is a segment of $\planPleapAftPartI$.

If $d^t_\hatf \ge \gamma$ or $d^t_\hatg \le \chi$, then inequality~\eqref{eqn: claim caseL2B-slackInC 2} directly
implies the lemma, because $[d^t_\rho,d^t_{h_k}) \subseteq (\chi,\gamma)$.
Therefore, for the rest of the proof, we will assume that $d^t_\hatf < \gamma$ and $d^t_\hatg > \chi$.

To complete the proof, 
it is sufficient to show that inequalities~\eqref{eqn: claim caseL2B-slackInC 1} and~\eqref{eqn: claim caseL2B-slackInC 2} cover the whole 
range of $\tau\in [d^t_\rho, d^t_{h_k})$, that is
\begin{equation}
[d^t_\rho, d^t_{h_k}) \subseteq (\chi, d^t_\hatf )\cup  [d^t_f, d^t_g) \cup [ d^t_\hatg,\gamma).
\label{eqn: claim caseL2B-slackInC 3}
\end{equation}

Note that $d^t_\hatg \leq d^t_\gstar$
and $d^t_\hatf \geq d^t_\fstar$ by the definitions of these packets and of $\backupplanLeapAftPartI$ and $\ADVleapAftPartI$ (note that $g$ is not a candidate for $\hatf$ as $g\in \planPleapAftPartI$).
If $f = \rho$, then it holds that $g = \gstar$ by the assumption, and the interval $[d^t_\rho, d^t_{h_k})$ is covered by
$[d^t_f, d^t_g) = [d^t_\rho, d^t_\gstar)$ and $[d^t_\hatg,\gamma)$, because $d^t_\hatg \le d^t_\gstar$ and $d^t_{h_k}\le \gamma$.
Similarly, if $g = h_k$, we have $f = \fstar$ and $[d^t_\rho, d^t_{h_k})$ is covered by
$[d^t_f, d^t_g) = [d^t_\fstar, d^t_{h_k})$ and $(\chi, d^t_\hatf)$, since $d^t_\hatf \ge d^t_\fstar$.
Finally, consider the case of $f = \fstar$ and $g = \gstar$. Then
the interval $[d^t_\rho, d^t_{h_k})\subseteq (\chi, \gamma)$ is covered by
$(\chi, d^t_\fstar)\cup [d^t_\fstar, d^t_\gstar)\cup [d^t_\gstar,\gamma)$,
which implies~\eqref{eqn: claim caseL2B-slackInC 3}, because
$d^t_\hatf \ge d^t_\fstar$ and $d^t_\hatg \le d^t_\gstar$.
\end{proof}


\smallskip

We now address the second part, where $\backupplanLeapAftPartI$ is changed into  
$\backupplanLeapAftPartII = \backupplanLeapAftPartI \setminus \braced{h_\terminala} \cup \braced{\rho}$,
and we show that $\backupplanLeapAftPartII$ satisfies invariant~{\InvariantB}.
We consider the feasibility of $\backupplanLeapAftPartII$ first.
Recall that in this part we assume $f = \fstar$ or $g = \gstar$.
If $d^t_\rho \ge d^t_{h_\terminala}$, then the feasibility of $\varbackupplan$ is clearly preserved after replacing $h_\terminala$ by $\rho$. 
Otherwise, $d^t_\rho < d^t_{h_\terminala}$, in which case, by Claim~\ref{clm:caseL2B-slackInC},
we have $\terminala=k$ and $\pslack(\backupplanLeapAftPartII,\tau) = \pslack(\backupplanLeapAftPartI,\tau) - 1 \ge 0$ for $\tau\in [d^t_\rho, d^t_{h_\terminala})$;
while $\pslack(\backupplanLeapAftPartII,\tau) = \pslack(\backupplanLeapAftPartI,\tau)\ge 0$ for $\tau\notin [d^t_\rho, d^t_{h_\terminala})$,
by Claim~\ref{clm:caseL2B-feasibilityOfC}. Therefore, $\backupplanLeapAftPartII$ is indeed feasible.
In this part, we replace $h_\tildea$ by $\varrho$ in both $\varplanP$ and $\varbackupplan$ and $\varADV$ is unchanged,
so invariant~{\InvariantB}(ii) is preserved as well.

\smallskip

Finally, we consider the third part of this stage
where the deadline decreases of $h_{\terminala+1}, \dots, h_k$ are implemented,
converting $\backupplanLeapAftPartII$ into $\backupplanAftGroup{\seggroup{\terminala,k}}$.
(Unlike in the first two parts, we perform this part also in the case $f=\rho$ and $g=h_\terminala$.)
We only need to show that $\backupplanAftGroup{\seggroup{\terminala,k}}$ remains feasible, since
invariant~{\InvariantB}(ii) is not affected.

There is no deadline decrease in this part if $\terminala=k$, so it remains to deal with the case when $\terminala < k$.
Recall that the algorithm will change the deadline $d^t_{h_i}$ to $d^{t+1}_{h_i} = \tau_{i-1}$, for each $i = \terminala+1,...,k$.
These changes result in decreasing $\pslack(\varbackupplan)$
by one in each interval $[\tau_\terminala, d^t_{h_{\terminala+1}}), \dots, [\tau_{k-1}, d^t_{h_k})$. Other slots in
$\varbackupplan$ are not affected.
These intervals are disjoint, because $d^t_{h_i}\le \tau_i$ for $i=0,\dots,k$. Thus, it suffices to show the following claim.


\begin{claim}\label{clm:caseL2B-part(iii) slack in B}
	Assume that $\terminala < k$. Then
$\pslack(\backupplanLeapAftPartII, \tau)\ge 1$ for any $\tau\in [\tau_\terminala, \gamma)$.
\end{claim}


\begin{proof}
We first show that $\pslack(\planPleapAftPartII, \tau)\ge 1$ for any $\tau\in [\tau_\terminala, \gamma)$.
To this end, we consider two cases.

The first case is when $\tau \in [\tau_\terminala,d^t_\rho)$. For such $\tau$ we have
$\pslack(\planPleapAftPartII, \tau) \ge \pslack(\planPleapAftPartI, \tau) + 1\ge 1$,
because replacing $h_\terminala$ by $\rho$ causes  $\pslack(\varplanP, \tau)$ to increase by 1 for 
$\tau\in [d^t_{h_\terminala}, d^t_\rho)\supseteq [\tau_\terminala, d^t_\rho)$, and because $\planPleapAftPartI = \planPaftInitSeg$ is feasible.

The second case is when $\tau \in [d^t_\rho,\gamma)$. In this case we argue as follows.
Since for $\terminala < k$ replacing $h_\terminala$ by $\rho$ in $\varplanP$
does not decrease any value $\pslack(\varplanP, \tau)$ (i.e., there are no new tight slots in $\varplanP$)
and each slot $\tau\in [d^t_\rho,\gamma)$ belongs to the segment of $\planPleapAftPartII$ ending at $\gamma$,
we indeed have $\pslack(\planPleapAftPartII, \tau)\ge 1$ for any $\tau \in [d^t_\rho,\gamma)$. 

We thus have $\pslack(\planPleapAftPartII, \tau)\ge 1$ for all $\tau\in [\tau_\terminala, \gamma)$, as claimed.
To finish the proof,
note that $d^t_\gstar \le \tau_\terminala$ by the definitions of $\gstar$ and the terminal group $\seggroup{\terminala,k}$,
and that $\ADVleapAftPartII(\tau_\terminala, \gamma] = \emptyset$. Thus
for $\tau\in [\tau_\terminala, \gamma)$, we can apply Lemma~\ref{lem:pslackRelation}(b) with $\zeta = \gamma$
to obtain $\pslack(\backupplanLeapAftPartII, \tau)\ge \pslack(\planPleapAftPartII, \tau) \ge 1$,
completing the proof.
\end{proof}

Putting it all together and using Claim~\ref{clm:caseL2B-part(iii) slack in B}, we obtain that
$\pslack(\backupplanAftGroup{\seggroup{\terminala,k}}, \tau)\ge \pslack(\backupplanLeapAftPartII, \tau) - 1\ge 0$
for any $\tau\in [\tau_\terminala, \gamma)$, while
$\pslack(\backupplanAftGroup{\seggroup{\terminala,k}}, \tau) = \pslack(\backupplanLeapAftPartII, \tau)\ge 0$
for $\tau\notin [\tau_\terminala, \gamma)$. This concludes the proof that invariant~{\InvariantB} holds after
the stage for the terminal group.

\paragraph{Stage for a middle group.}
We now analyze the stage in which the segments represented by a middle group $\seggroup{a,b}$ of indices are processed.
The current intermediate plan, backup plan, and adversary stash are $\planPbefGroup{\seggroup{a,b}}$,
$\backupplanBefGroup{\seggroup{a,b}}$, and $\AdvBefGroup{\seggroup{a,b}}$, respectively, and we assume that they satisfy
invariants~{\InvariantA} and~{\InvariantB}. At this point, it holds that:
\begin{enumerate}[label=(\alph*),itemsep=2pt] 
\item $h_a\in \AdvBefGroup{\seggroup{a,b}}$ by the definition of the middle group,
\item $h_{b+1}\notin  \planPbefGroup{\seggroup{a,b}}$ as a result of the previous stage (recall that $b < k$ for a middle group),
which implies that $h_{b+1}\notin \AdvBefGroup{\seggroup{a,b}}$ by invariant~{\InvariantA}, and
\item $h_{b+1}\notin  \backupplanBefGroup{\seggroup{a,b}}$, because $h_{b+1}$ was removed from $\calP$ in the previous stage
(that is, we had $h_{b+1}\in \planPaftArrivals(\alpha, \timehorizon]$, but now $h_{b+1}\notin \planPbefGroup{\seggroup{a,b}}$),
and because we do not add any packets to $\varbackupplan(\alpha, \timehorizon]\setminus\varplanP$ while processing a leap step 
(see the ``monotonicity of $\varbackupplan\setminus\varplanP$'' invariant in Section~\ref{subsubsec: leap step: a road map}).
\end{enumerate}

The changes of plan $\varplanP$ in this stage involve removing $h_a$ and adding back $h_{b+1}$;
that is $\planPbefGroup{\seggroup{a,b}}$ is changed to $\planPaftGroup{\seggroup{a,b}} = \planPbefGroup{\seggroup{a,b}} \setminus \braced{h_a} \cup \braced{h_{b+1}}$.
We also decrease the deadlines of packets $h_{a+1}, \dots, h_{b+1}$ and
increase the weights of a subset of packets $h_{a+1}, \dots, h_{b+1}$, according to the algorithm.
Since $h_a$ is removed from $\varplanP$, we also need to remove it from  $\varADV$ to preserve invariant~{\InvariantA}.  
Similarly, as $h_{b+1}$ is added back to $\varplanP$ (and is not in $\varADV$), it needs to be added
to $\varbackupplan$ or to $\varADV$ so that invariant~{\InvariantB}(ii) continues to hold.
The changes of $\varADV$ and $\varbackupplan$ depend on two cases, which are described below.


\medskip
\noindent
\mycase{L.$(\beta,\gamma]$.M.1}
There is an index $i\in [a+1,b+1]$ such that $w^t_{h_{i}} < \mu_{i-1}$, i.e., the algorithm (in line~11) increases the weight of $h_{i}$.
We will remove $h_a$ from $\varADV$ and add $h_{b+1}$ to $\varbackupplan$.
To restore the feasibility of $\varbackupplan$,
we apply Lemma~\ref{lem: restore feasibility of B} for packet $g = h_a$
to identify a packet $f_a\in \backupplanBefGroup{\seggroup{a,b}} \setminus \planPbefGroup{\seggroup{a,b}}$ for which set 
$\backupplanBefGroup{\seggroup{a,b}}\setminus \braced{f_a}\cup \braced{h_a}$ is feasible
and $d^t_{f_a} > \prevtightslot(\planPbefGroup{\seggroup{a,b}}, d^t_{h_a})$. We then let
\begin{equation*}
\AdvAftGroup{\seggroup{a,b}} \;=\; \AdvBefGroup{\seggroup{a,b}}\setminus\braced{h_a} 
\quad\textrm{and}\quad
\backupplanAftGroup{\seggroup{a,b}} \;=\; \backupplanBefGroup{\seggroup{a,b}}\setminus \braced{f_a}\cup \braced{h_{b+1}}    \,.
\end{equation*}
With these changes, {\InvariantA} is preserved, since also none of packets $h_{a+1}, \dots, h_{b+1}$,
whose deadlines are decreased, is in $\AdvAftGroup{\seggroup{a,b}}$.
We will show below that $\backupplanAftGroup{\seggroup{a,b}}$ is feasible,
even with the deadlines of $h_{a+1}, \dots, h_{b+1}$ decreased by the algorithm.

\smallskip


\indentemparagraph{Deriving \eqref{eq:leap-group} in Case~L.$(\beta,\gamma]$.M.1.}
We need to take into account the addition of $h_{b+1}$ in $\varbackupplan$ and
possible change of weights of some packets $h_{a+1}, \dots, h_{b+1}$. Taking these changes into account, the terms
on the left-hand side of~\eqref{eq:leap-group} can be expressed or estimated as follows:
\begin{align}
	\Delta_{\seggroup{a,b}}\Psi 
			\;&=\; \phinegonef \razy \left( \Delta^t w( {h_{a+1}}, \dots, {h_{b+1}} ) + w^t_{h_{b+1}}  - w^t_{f_a} \right)
			\label{eqn: 5.5.6 middle 1 psi}
	 \\
	\Delta^t w({h_{a+1}}, \dots, {h_{b+1}}) \;&\le\; \mu_a - w^t_{h_{b+1}}
			\label{eqn: 5.5.6 middle 1 delta w}
	\\
	\advcredit^t_{\seggroup{a,b}} \;&=\; w^t_{h_a} - \minwt(\planPbefGroup{\seggroup{a,b}}, d^t_{h_a}) \;\le\; w^t_{h_a} - \mu_a
			\label{eqn: 5.5.6 middle 1 advgain}
\end{align}
To justify~(\ref{eqn: 5.5.6 middle 1 psi}), note that while
$h_{b+1}\notin \backupplanBefGroup{\seggroup{a,b}}$ is added to $\varbackupplan$ with its new weight $w^{t+1}_{h_{b+1}}$, its
weight increase of $w^{t+1}_{h_{b+1}} - w^t_{h_{b+1}}$ is already accounted for in $\Delta^t w( {h_{a+1}}, \dots, {h_{b+1}} )$. 

Inequality~(\ref{eqn: 5.5.6 middle 1 delta w}) can be derived from the case condition,
which says that there is $i\in [a+1,b+1]$ such that $w^t_{h_{i}} < \mu_{i-1}$, so we just
need to apply Lemma~\ref{lem:changeOfWeightsOfHs} with $a'=a+1$ and $b'=b+1$
to get~(\ref{eqn: 5.5.6 middle 1 delta w}).

Inequality (\ref{eqn: 5.5.6 middle 1 advgain}) follows from the definition of $\mu_a$ as
$\mu_a = \minwt(\planP^t, d^t_{h_a})$ and from the observation that
$\minwt(\planPbefGroup{\seggroup{a,b}}, d^t_{h_a}) \ge  \minwt(\planP^t, d^t_{h_a})$, which in turn follows from
$\nexttightslot(\planPbefGroup{\seggroup{a,b}}, d^t_{h_a}) = \nexttightslot(\planP^t, d^t_{h_a})$
and $\planPbefGroup{\seggroup{a,b}}[t+1, \eta] = \planP^t[t, \eta]\setminus\braced{\firstlightpacket}$,
where $\eta = \nexttightslot(\planPbefGroup{\seggroup{a,b}}, d^t_{h_a})$.
(The equality $\nexttightslot(\planPbefGroup{\seggroup{a,b}}, d^t_{h_a}) = \nexttightslot(\planP^t, d^t_{h_a})$ holds
because of the ``locality of changes'' invariant from Section~\ref{subsubsec: leap step: a road map}).

Next, we claim that 
\begin{equation}
		w^t_{f_a} \;\le\; w^t_{h_{b+1}}. 
		\label{eqn: 5.5.6 bound on w-f-a}
\end{equation}
Indeed, $f_a$ is not in $\planP^t$ since, by its definition, $f_a\notin \planPbefGroup{\seggroup{a,b}}$ and
$d^t_{f_a} > \prevtightslot(\planPbefGroup{\seggroup{a,b}} d^t_{h_a}) \ge \beta$, so $f_a$ could not have been
removed from $\varplanP$ in this step. (Note that $h_{b+1}$ is temporarily removed from $\varplanP$, but it
cannot be equal $f_a$ because $h_{b+1}\notin \backupplanBefGroup{\seggroup{a,b}}$).
Thus $w^t_{f_a} \le w^t_\rho$, because $\rho$ is the heaviest pending packet not in $\planPaftArrivals$ with deadline after $\beta$
and $f_a\notin \planPaftArrivals$ is pending in step $t$.
By Lemma~\ref{lem:leapStep}(b) we have $w^t_\rho \le w^t_{h_{b+1}}$, which implies that
$w^t_{f_a} \le w^t_{h_{b+1}}$. (As an additional clarification, the argument in this paragraph relies tacitly on two properties
of earlier changes in processing the transmission event: One, no packets with deadlines after
$\alpha$ are added to $\varbackupplan\setminus\varplanP$, and two, 
the tight slots of $\planPbefGroup{\seggroup{a,b}}$ in the interval $(\beta, \tau_b]$ are the same as tight slots in $\planPaftArrivals$;
cf.\ Section~\ref{subsubsec: leap step: a road map}.)

Then we prove \eqref{eq:leap-group} as follows:
\begin{align}
\Delta_{\seggroup{a,b}}\Psi &- \phi\razy \Delta^t w({h_{a+1}}, \dots, {h_{b+1}}) - \advcredit^t_{\seggroup{a,b}}
\nonumber\\
& \ge\; \phinegonef \razy \left( \Delta^t w( {h_{a+1}}, \dots, {h_{b+1}} ) + w^t_{h_{b+1}}  - w^t_{f_a} \right)
\nonumber\\
&\qquad\qquad
	 - \phi\razy \Delta^t w({h_{a+1}}, \dots, {h_{b+1}}) - (w^t_{h_a} - \mu_a)
\nonumber\\
& =\; \phinegonef \razy \left( w^t_{h_{b+1}} -  w^t_{f_a} \right) - \Delta^t w({h_{a+1}}, \dots, {h_{b+1}}) - w^t_{h_a} + \mu_a
\nonumber\\
& \ge \; - (\mu_a - w^t_{h_{b+1}}) - w^t_{h_a} + \mu_a
\label{eqn: iterated middle M1}
\\
&=\;  - w^t_{h_a} + w^t_{h_{b+1}} \,,
\nonumber
\end{align}
where inequality~\eqref{eqn: iterated middle M1} follows from~(\ref{eqn: 5.5.6 bound on w-f-a}) and~(\ref{eqn: 5.5.6 middle 1 delta w}).

\smallskip


\indentemparagraph{Invariant~{\InvariantB} in Case~L.$(\beta,\gamma]$.M.1.}
We claim that after the stage for the middle group $\seggroup{a,b}$, invariant~{\InvariantB} holds,
which we show by partitioning this stage into three parts:
\begin{description}
	\item{(i)} We first let $\backupplanLeapAftPartI = \backupplanBefGroup{\seggroup{a,b}}\setminus \braced{f_a}\cup \braced{h_a}$.
	\item{(ii)} Second, we modify $\backupplanLeapAftPartI$ to obtain set $\backupplanLeapAftPartII = \backupplanLeapAftPartI\setminus \braced{h_a}\cup\braced{h_{b+1}}$.
	\item{(iii)} Finally, we implement the decrease of the deadlines of packets $h_{a+1}, \dots, h_{b+1}$,
changing $\backupplanLeapAftPartII$ to $\backupplanAftGroup{\seggroup{a,b}}$. 
\end{description}
In each part, we also appropriately modify the sets $\varplanP$ and $\varADV$, as illustrated in
Figure~\ref{fig: 5.5.6 middle three parts}.

\begin{figure}[ht]
\begin{center}
	\includegraphics[width = 6.35in]{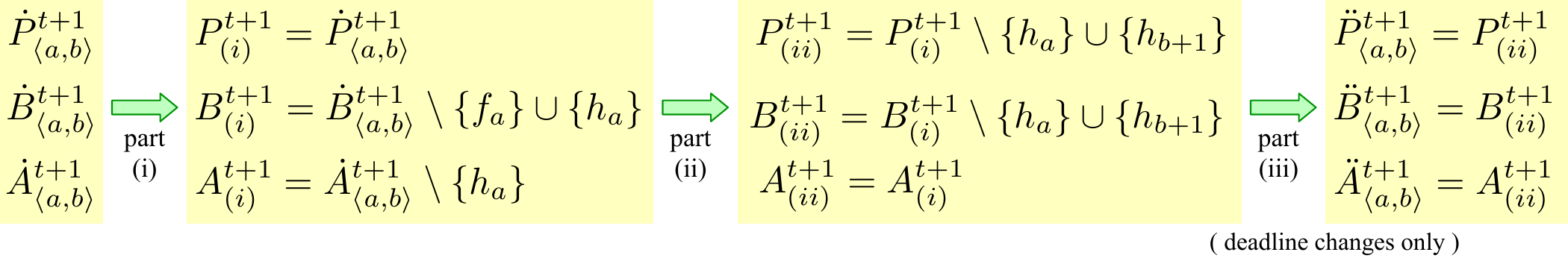}
	\caption{Partition of the stage for the middle group in Case~L.$(\beta,\gamma]$.M.1 into three parts.
	}
	\label{fig: 5.5.6 middle three parts}
\end{center}
\end{figure}

That $\backupplanLeapAftPartI$ is feasible follows from the choice of $f_a$ and Lemma~\ref{lem: restore feasibility of B}.
Since $d^t_{h_a} < d^t_{h_{b+1}}$, $\backupplanLeapAftPartII$ is also a feasible set of packets,
and moreover, $\pslack(\backupplanLeapAftPartII, \tau) =  \pslack(\backupplanLeapAftPartI, \tau) + 1\ge 1$ for any $\tau\in [d^t_{h_a}, d^t_{h_{b+1}})$, by the feasibility of $\backupplanLeapAftPartI$.
Finally, we argue that $\backupplanAftGroup{\seggroup{a,b}}$ remains feasible. The deadline decreases for packets $h_{a+1}, \dots, h_{b+1}$
decrease $\pslack(\varbackupplan)$ by one for slots
in each of the disjoint intervals $[\tau_a, d^t_{h_{a+1}}), \dots, [\tau_b, d^t_{h_{b+1}})$.
As $d^t_{h_a}\le \tau_a$, all these intervals are contained in $[d^t_{h_a}, d^t_{h_{b+1}})$; hence, by the earlier observation,
it holds that $\pslack(\backupplanAftGroup{\seggroup{a,b}}, \tau) \ge \pslack(\backupplanLeapAftPartII, \tau) - 1\ge 0$ 
for any $\tau\in [d^t_{h_a}, d^t_{h_{b+1}})$.   For $\tau\notin [d^t_{h_a}, d^t_{h_{b+1}})$, the values of $\pslack(\varbackupplan)$ do not change.
It follows that $\backupplanAftGroup{\seggroup{a,b}}$ is feasible; that is, invariant~{\InvariantB}(i) holds after the stage.
That invariant~{\InvariantB}(ii) holds as well follows directly from the definitions of
$\planPaftGroup{\seggroup{a,b}}$, $\backupplanAftGroup{\seggroup{a,b}}$, and $\AdvAftGroup{\seggroup{a,b}}$.


\medskip
\noindent
\mycase{L.$(\beta,\gamma]$.M.2}
For all indices $i\in [a+1,b+1]$ we have $w^t_{h_{i}} \ge \mu_{i-1}$. Thus in this case
the algorithm does not increase the weights of packets $h_{a+1},...,h_{b+1}$, i.e.,
$\Delta^t w({h_{a+1}}, \dots, {h_{b+1}}) = 0$.

Before describing the changes in $\varADV$, we argue that $h_{a+1}\notin \AdvBefGroup{\seggroup{a,b}}$.
This is trivial if $b > a$, because then packets $h_{a+1}, \dots, h_b$ are not in $\varADV$
before processing the groups. Otherwise, we have $a=b$. Since in the stages, we
process the groups backwards with respect to time slots,
the stage for the group $\seggroup{b+1,b'}$ containing index $b+1$ is already completed, and
it guarantees that $h_{b+1}$ is not in $\planPaftGroup{\seggroup{b+1,b'}} = \planPbefGroup{\seggroup{a,b}}$
and thus not in $\AdvAftGroup{\seggroup{b+1,b'}} = \AdvBefGroup{\seggroup{a,b}}$ by invariant~{\InvariantA}, as claimed. 
(See also the overview of the partition into stages, earlier in the description of Case~L.$(\beta,\gamma]$.I.2.)

We now change $\varADV$ as follows:
we remove $h_a$ from $\varADV$, and we put $h_{a+1}$ in the former slot of $h_a$ in $\AdvBefGroup{\seggroup{a,b}}$,
that is  $\AdvAftGroup{\seggroup{a,b}} = \AdvBefGroup{\seggroup{a,b}}\setminus\braced{h_a}\cup \braced{h_{a+1}}$.
This is valid because, as we showed above, $h_{a+1}\notin \AdvBefGroup{\seggroup{a,b}}$, and
the new deadline of $h_{a+1}$ is $\tau_a$.
Then $\advcredit^t_{\seggroup{a,b}} = w^t_{h_a} - w^t_{h_{a+1}}$, as the weight of $h_{a+1}$ does not change. 
Since packets $h_{a+2},\dots,h_{b+1}$ are not in $\AdvAftGroup{\seggroup{a,b}}$,
the decreases of their deadlines do not affect $\varADV$.
After these changes, invariant~{\InvariantA} will be preserved.

Regarding $\varbackupplan$, recall that $h_{b+1}\notin \backupplanBefGroup{\seggroup{a,b}}$,
$h_a\in \AdvBefGroup{\seggroup{a,b}}$, and $h_a\notin \planPaftGroup{\seggroup{a,b}}$. Also,
if $b > a$ then $h_{a+1}\in \backupplanBefGroup{\seggroup{a,b}}$.
Then, since $h_{a+1}$ is newly added to $\varADV$, 
we let $\backupplanAftGroup{\seggroup{a,b}} = \backupplanBefGroup{\seggroup{a,b}} \cup \braced{h_{b+1}} \setminus \braced{h_{a+1}}$;
it is possible here that $h_{a+1} = h_{b+1}$, in which case $\varbackupplan$ does not change.

\smallskip


\indentemparagraph{Deriving~\eqref{eq:leap-group} in Case~L.$(\beta,\gamma]$.M.2.}
In this case, the terms on the left-hand side of~\eqref{eq:leap-group} are
\begin{align*}
	\Delta_{\seggroup{a,b}}\Psi 
			\;&=\; \phinegonef \razy \left( - w^t_{h_{a+1}} + w^t_{h_{b+1}} \right)
	 \\
	\Delta^t w({h_{a+1}}, \dots, {h_{b+1}}) \;&=\; 0
	\\
	\advcredit^t_{\seggroup{a,b}} \;&=\; w^t_{h_a} - w^t_{h_{a+1}}
\end{align*}
with all formulas following directly from the definition and the fact that 
the weights remain unchanged. Plugging these in, we obtain:
\begin{align}
\Delta_{\seggroup{a,b}}\Psi &- \phi \razy \Delta^t w({h_{a+1}}, \dots, {h_{b+1}}) - \advcredit^t_{\seggroup{a,b}}
	\nonumber\\
	& =  \; \phinegonef \razy \left( - w^t_{h_{a+1}} + w^t_{h_{b+1}} \right) - \phi\cdot 0 - ( w^t_{h_a} - w^t_{h_{a+1}} )
	\nonumber\\
	& = \; - w^t_{h_a}  + \phinegtwof \razy w^t_{h_{a+1}} + \phinegonef \razy w^t_{h_{b+1}}
	\nonumber\\
	& \ge \; - w^t_{h_a}  + \phinegtwof \razy w^t_{h_{b+1}} + \phinegonef \razy w^t_{h_{b+1}}
	\label{eqn: iterated middle M2}
	\\
	& = \; -  w^t_{h_a}  + w^t_{h_{b+1}} \,,
	\nonumber
\end{align}
where inequality~\eqref{eqn: iterated middle M2} follows from $w^t_{h_{a+1}} \ge w^t_{h_{b+1}}$, by Lemma~\ref{lem:leapStep}(b) and $a\le b$.

\smallskip


\subparagraph{Invariant~{\InvariantB} in Case~L.$(\beta,\gamma]$.M.2}
We show that $\backupplanAftGroup{\seggroup{a,b}}$ is feasible  by considering first an intermediate backup plan
$\backupplanLeapAftPartI = \backupplanBefGroup{\seggroup{a,b}}\cup \braced{h_{b+1}}\setminus\braced{h_{a+1}}$,
computed with packets $h_{a+1}, \dots, h_{b+1}$ having their deadlines unchanged. Since
$\backupplanLeapAftPartI$ is obtained by replacing $h_{a+1}$ by $h_{b+1}$ in $\varbackupplan$,
and $d^t_{h_{a+1}} \le d^t_{h_{b+1}}$, $\backupplanLeapAftPartI$ is feasible;
moreover, it satisfies $\pslack(\backupplanLeapAftPartI, \tau) = \pslack(\backupplanBefGroup{\seggroup{a,b}}, \tau) + 1\ge 1$
for any $\tau\in [d^t_{h_{a+1}}, d^t_{h_{b+1}})$. Next, to prove the feasibility of $\backupplanAftGroup{\seggroup{a,b}}$,
we decrease the deadlines of packets $h_{a+1}, \dots, h_{b+1}$.  
This decreases $\pslack(\varbackupplan)$ by one for slots in disjoint intervals
$[\tau_{a+1}, d^t_{h_{a+2}}), \dots, [\tau_b, d^t_{h_{b+1}})$.
(We do not consider $h_{a+1}$ here because this packet was already removed from $\varbackupplan$.)
These intervals are contained in $[d^t_{h_{a+1}}, d^t_{h_{b+1}})$,
and we thus have $\pslack(\backupplanAftGroup{\seggroup{a,b}}, \tau)\ge \pslack(\backupplanLeapAftPartI,\tau) -1 \ge 0$
for $\tau\in [d^t_{h_{a+1}}, d^t_{h_{b+1}})$. For   $\tau\notin [d^t_{h_{a+1}}, d^t_{h_{b+1}})$, the values of
$\pslack(\varbackupplan)$ do not change. Therefore $\backupplanAftGroup{\seggroup{a,b}}$ is feasible, showing {\InvariantB}(i).
Part (ii) of invariant~{\InvariantB} is implied by the changes of sets $\varplanP, \varADV$, and $\varbackupplan$.

\smallskip

This completes the proof that invariants~{\InvariantA} and~{\InvariantB} are
preserved after the stage for the middle group $\seggroup{a,b}$, and that
inequality~\eqref{eq:leap-group} holds for this stage.

\paragraph{Stage for the initial group.}
The initial group $\seggroup{0,\initialb}$ is defined when either $l>0$ and $i_1 > 0$ or if $l = 0$ and $\terminala > 0$. In either case, 
we have $h_0, ..., h_\initialb \not\in \AdvBefGroup{\seggroup{0,\initialb}}$.  

In the stage in which this group is processed, we change plan $\varplanP$ by removing $h_0 = p$ and adding $h_{\initialb+1}$,
so $\planPbefNextStep = \planPbefGroup{\seggroup{0,\initialb}} \setminus \braced{p} \cup \braced{h_{\initialb+1}}$.
(Recall that $\planPbefNextStep = \planPaftGroup{\seggroup{0,\initialb}}$ and the corresponding relation holds
for $\varADV$ and $\varbackupplan$, since we process the stage for the initial group after all other groups.)
We also decrease the deadlines of packets $h_{1}, \dots, h_{\initialb+1}$, and for some
we increase their weights, according to the algorithm. We do not change $\varADV$;
thus $\ADVbefNextStep = \AdvBefGroup{\seggroup{0,\initialb}}$ and $\advcredit^t_{\seggroup{0,\initialb}} = 0$.
Since packets $h_1, \dots, h_{\initialb+1}$ are not in $\AdvBefGroup{\seggroup{0,\initialb}}$, it follows that 
invariant~{\InvariantA} is preserved. To update $\varbackupplan$, we let
$\backupplanBefNextStep = \backupplanBefGroup{\seggroup{0,\initialb}}\setminus \braced{p} \cup \braced{h_{\initialb+1}}$
(recall that $h_{\initialb+1}\notin \backupplanBefGroup{\seggroup{0,\initialb}}$).

\smallskip


\indentemparagraph{Deriving~\eqref{eq:leap-group} for the initial group.}
The terms on the left-hand side of~\eqref{eq:leap-group} can be expressed or estimated as follows:
\begin{align}
	\Delta_{\seggroup{0,\initialb}}\Psi 
			\;&=\; \phinegonef \razy \left( \Delta^t w({h_1}, \dots, {h_{\initialb+1}}) - w^t_p + w^t_{h_{\initialb+1}} \right)
			\label{eqn: 5.5.6 initial psi}
	 \\
	\Delta^t w({h_1}, \dots, {h_{\initialb+1}}) \;&\le\; \phinegtwof \razy ( w^t_p -  w^t_{h_{\initialb+1}} )
			\label{eqn: 5.5.6 initial delta w}
	\\
	\advcredit^t_{\seggroup{0,\initialb}} \;&=\; 0
			\label{eqn: 5.5.6 initial advgain}
\end{align}
Equations~\eqref{eqn: 5.5.6 initial psi} and~\eqref{eqn: 5.5.6 initial advgain} follow directly from the
definitions and the case assumption.

To justify~\eqref{eqn: 5.5.6 initial delta w}, we consider two cases.
If the algorithm does not increase the weight of any of the packets $h_1, \dots, h_{\initialb+1}$,
then $\Delta^t w({h_1}, \dots, {h_{\initialb+1}}) = 0\le \phinegtwof \razy ( w^t_p - \razy w^t_{h_{\initialb+1}})$,
by Lemma~\ref{lem:leapStep}(b). Otherwise, there is $i\in [1,\initialb+1]$ such that $w^t_{h_{i}} < \mu_{i-1}$,
and applying Lemma~\ref{lem:changeOfWeightsOfHs} with $a'=1$ and $b'=\initialb+1$, we obtain
\begin{align}
\Delta^t w({h_1}, \dots, {h_{\initialb+1}}) 
	\;&\le\; \mu_0 - w^t_{h_{\initialb+1}}
	\nonumber\\
	&\le\; \phinegtwof \razy w^t_p  + \phinegonef \razy w^t_\rho - w^t_{h_{\initialb+1}}
	\label{eqn: iterated initial delta weights 1}\\
	&\le\; \phinegtwof \razy ( w^t_p -  w^t_{h_{\initialb+1}} ) \,,
	\label{eqn: iterated initial delta weights 2}
\end{align} 
where inequality~\eqref{eqn: iterated initial delta weights 1} uses the bound
$\mu_0\le  \phinegtwof \razy w^t_p  + \phinegonef \razy w^t_\rho$, that follows from (\ref{eq:leapStepVsOmega}),
and inequality~\eqref{eqn: iterated initial delta weights 2}  uses $w^t_\rho \le w^t_{h_{\initialb+1}}$, that follows from Lemma~\ref{lem:leapStep}(b).
Thus~\eqref{eqn: 5.5.6 initial delta w} holds.

The calculation showing \eqref{eq:leap-group} is now quite simple:
\begin{align}
\Delta_{\seggroup{0,\initialb}}\Psi &- \phi\razy \Delta^t w({h_1}, \dots, {h_{\initialb+1}}) - \advcredit^t_{\seggroup{0,\initialb}}
	\nonumber\\
	\;&=\; \phinegonef \razy \left( \Delta^t w({h_1}, \dots, {h_{\initialb+1}}) - w^t_p + w^t_{h_{\initialb+1}} \right)
			 		- \phi\razy \Delta^t w({h_1}, \dots, {h_{\initialb+1}}) - 0 
	\nonumber\\
	&=\; \phinegonef \razy (- w^t_p + w^t_{h_{\initialb+1}}) -\Delta^t w({h_1}, \dots, {h_{\initialb+1}})
	\nonumber\\
	&\ge\; \phinegonef \razy (- w^t_p + w^t_{h_{\initialb+1}}) - \phinegtwof \razy ( w^t_p -  w^t_{h_{\initialb+1}})
	\label{eqn: iterated initial delta phi}
	\\
	&=\; - w^t_p + w^t_{h_{\initialb+1}} \,,
	\nonumber
\end{align} 
where inequality\eqref{eqn: iterated initial delta phi} holds by~\eqref{eqn: 5.5.6 initial delta w}.

\smallskip


\indentemparagraph{Invariant~{\InvariantB} after the stage for the initial group.}
That invariant~{\InvariantB}(ii) is preserved follows directly from the formulas for $\planPbefNextStep$,
$\backupplanBefNextStep$, and $\ADVbefNextStep$.

To show invariant~{\InvariantB}(i), that is the feasibility of $\backupplanBefNextStep$,
we consider an intermediate backup plan $\backupplanLeapAftPartI = \backupplanBefGroup{\seggroup{0,\initialb}}\setminus\braced{p}\cup \braced{h_{\initialb+1}}$,
where the deadlines of packets  $h_1, \dots, h_{\initialb+1}$ remain unchanged.
Replacing $p$ by $h_{\initialb+1}$ in $\varbackupplan$ preserves feasibility as $d^t_p < d^t_{h_{\initialb+1}}$.
So $\backupplanLeapAftPartI$ is feasible; moreover, it satisfies 
$\pslack(\backupplanLeapAftPartI, \tau) = \pslack(\backupplanBefGroup{\seggroup{0,\initialb}}, \tau) + 1\ge 1$
for any $\tau\in [d^t_p, d^t_{h_{\initialb+1}})$.  

We now decrease the deadlines of packets $h_1, \dots, h_{\initialb+1}$, as in the algorithm.
This decreases $\pslack(\varbackupplan, \tau)$ by one for slots $\tau$ in disjoint intervals
$[\tau_0, d^t_{h_1}), \dots, [\tau_\initialb, d^t_{h_{\initialb+1}})$. These intervals are contained in $[d^t_p, d^t_{h_{\initialb+1}})$,
and we thus have $\pslack(\backupplanBefNextStep, \tau)\ge \pslack(\backupplanLeapAftPartI,\tau) -1 \ge 0$ for $\tau\in [d^t_p, d^t_{h_{\initialb+1}})$. 
The values of $\pslack(\varbackupplan)$ are not affected for slots $\tau\notin [d^t_p, d^t_{h_{\initialb+1}})$.
Therefore  $\backupplanBefNextStep$ is feasible;
that is, invariant~{\InvariantB}(i) holds after the stage for the initial group and after the algorithm's step as well.
 
\paragraph{Summary of Case~L.$\boldsymbol{(\beta,\gamma]}$.I.2.}
In this case, we processed changes in the interval $(\beta,\gamma]$ when  
$\rho\in \backupplanAftInitSeg$ or $\ADVaftInitSeg(\beta, \gamma] \neq \emptyset$.
We divided the segments in this interval that contain packets $h_i$ into groups,
and each group was processed in a separate stage, in the reverse order (i.e., backwards with respect to time slots).
For each stage, we proved that (i) invariants~{\InvariantA} and~{\InvariantB} are preserved,
and that (ii) inequality~\eqref{eq:leap-group} holds for this stage.
The last stage corresponds to a group $\seggroup{0,\initialb}$ (that may be of any type),
and we proved that after this last stage the backup plan
$\backupplanBefNextStep = \backupplanAftGroup{\seggroup{0,\initialb}}$  satisfies invariant~{\InvariantB},
and that the adversary stash  $\ADVbefNextStep = \AdvAftGroup{\seggroup{0,\initialb}}$  satisfies invariant~{\InvariantA}.
All inequalities~\eqref{eq:leap-group} (one for each group) imply the key
inequality~\eqref{eq:leap-laterSegments}, needed for amortized analysis.

\bigskip
  
This concludes the analysis of an iterated leap step and thus also the analysis of an
event of packet transmission. Together with the analysis of packet arrivals given in Section~\ref{subsec: arrival of a packet},
this completes the proof of $\phi$-competitiveness of Algorithm~$\PlanMonotonicity$, as described in Section~\ref{subsec: overview of the analysis}.

\section{Final Comments}
\label{sec: final comments}

Our result establishes a tight bound of $\phi$ on the competitive ratio of
{\BDPS} in the deterministic case, settling a long-standing open problem. 
There is still a number of intriguing open problems about online 
packet scheduling that deserve further study.

Of these open problems, the most prominent one is to establish tight bounds for randomized algorithms 
for {\BDPS}. The best known upper bound to date is 
$\e/(\e-1) \approx 1.582$~\cite{bartal_online_competitive_04,chin_weighted_throughput_06,bienkowski_randomized_algorithms_11,Jez13}. 
This ratio is achieved by a memoryless algorithm and it holds
even against an adaptive adversary. No better upper bound for the oblivious adversary is known.
(In fact, against the oblivious adversary, the same ratio can be attained for a more general problem of
online vertex-weighted bipartite matching~\cite{AggarwalGKM11,DevanurJK13}.)
The best lower bounds, to our knowledge,
are $4/3\approx 1.333$~\cite{bienkowski_randomized_algorithms_11} against the
adaptive adversary, and $1.25$~\cite{chin_partial_job_values_03} against the oblivious one, 
respectively. (Both lower bounds use $2$-bounded instances and are in fact 
tight for $2$-bounded instances.)
Closing these gaps would provide insight into the power of randomization in packet scheduling.

In practical settings, the choice of the packet to transmit needs to be made at speed matching the link's rate,
so the running time and simplicity of the scheduling algorithm are important factors.
This motivates the study of \emph{memoryless} algorithms for {\BDPS}, as those algorithms
tend to be easy to implement and fast. All known upper bounds for competitive
randomized algorithms we are aware of are achieved by memoryless algorithms (see~\cite{goldwasser_survey_10}).   
For deterministic algorithms, the only memoryless one that beats ratio $2$ is the 
$1.893$-competitive algorithm in~\cite{englert_suppressed_packets_12}. 
The main question here is whether the ratio of $\phi$ can be achieved by a memoryless algorithm.

Among other models for packet scheduling, optimal competitiveness for the FIFO model is
still wide open, both in the deterministic and randomized cases. We refer the reader to
the (still mostly current) survey of Goldwasser~\cite{goldwasser_survey_10}, who provides a thorough
discussion of various models for online packet scheduling and related open problems.


\subsection*{Acknowledgements}
Work done in part while P.~Vesel\'{y} was at the University of Warwick.
P.~Vesel\'{y} and J.~Sgall were partially supported by GA \v{C}R project
19-27871X.
P.~Vesel\'{y} was also partially supported by European Research Council grant ERC-2014-CoG 647557 and by Charles University project UNCE/SCI/004.
M.~Chrobak was partially supported by NSF grant CCF-1536026 and CCF-2153723.
Ł.~Jeż was partially supported by NCN grants 2016/21/D/ST6/02402, 2016/22/E/ST6/00499, and 2020/39/B/ST6/01679.
We are all grateful to Martin B\"ohm for useful discussions.

\bibliographystyle{plainurl}
\bibliography{packets_references}

\vfill
\eject

\appendix

\section{Appendix: Properties of Optimal Plans}
\label{app: more on properties of plans}

A computation of an online algorithm can be thought of as a sequence 
of events, with each event being either a packet arrival or transmitting a packet (which includes 
incrementing the current time). In this appendix, we give formal statements of ``plan-update lemmas'', 
that explain how the structure of the optimal plan changes in response to these events.

The proofs are based on analyzing Algorithm~{\ComputePlan},   given in Section~\ref{sec: plans and their properties},
for computing the optimal plan in a greedy fashion. In the proofs, we will refer to 
the packets $j$ added to the plan in line~5 of this algorithm as \emph{admitted} to the plan, 
and to the remaining packets as \emph{rejected}.

We will use notation $\pendpackets$ for the current set of pending packets,
which is the default argument for Algorithm~{\ComputePlan}.
For a set $Z \subseteq \pendpackets$ of pending packets and a
packet $j$, by $Z_{\gtrweight j}$ we will denote the set of packets in $Z$
that are (strictly) heavier than $j$ and by $Z_{\geqweight j}$ the set of packets in $Z$
that are at least as heavy as $j$. (In most cases we will use this notation when
$Z$ is the optimal plan and $j\in\pendpackets$, but occasionally we apply it in other contexts.)
If $\planP$ is an optimal plan, then $\planP_{\gtrweight j}$ can be thought of, equivalently, as the set of
all packets that are admitted to set $X$ in Algorithm~{\ComputePlan}
when all packets heavier than $j$ have already been considered (and similarly for $\planP_{\geqweight j}$).

In the condition of admitting $j\in\pendpackets$ in Algorithm~{\ComputePlan} (line~4) it is
sufficient to check if $\packetslack(X\cup\braced{j}, \tau) \ge 0$ only for $\tau\ge d_j$, because
for smaller values of $\tau$ the addition of $j$ does not change the value of $\packetslack(\tau)$.
This condition can be equivalently stated as ``$\packetslack(X, \tau) \ge 1$ for each $\tau\ge d_j$''.
Note also that, letting $\planP$ be the computed optimal plan, right before processing $j$ we have $X = P_{\gtrweight j}$,
and right after processing $j$ we have $X = P_{\geqweight j}$.

\smallskip

The following observation summarizes some simple but useful properties of plans.


\begin{observation}\label{obs: plan properties}
Let $\pendpackets$ be the set of packets pending at time $t$ and let $\planP$ be the optimal plan at time $t$. Then
\begin{enumerate}[label=(\alph*),itemsep=2pt]
	\item $\planP$ depends only on the ordering of packets in $\pendpackets$ with respect to weights
			(not on the actual values of the weights).
	\item $\planP$ does not change if the weight of any packet in $\planP$ is increased.
			It also does not change if the weight of any packet in $\pendpackets\setminus\planP$
			is decreased.
	\item For a packet $j\in \planP$, let $\beta = \prevtightslot(\planP,d_j)$. 
			Then for any slot $\xi\in (\beta,d_j]$ there is a schedule of $\planP$ in 
			which $j$ is scheduled in slot $\xi$.
	\item If $j\in\pendpackets$ is a packet and $\zeta$ a slot such that $\packetslack({\planP}_{\geqweight j},\zeta) = 0$,
		then ${\planP}_{\geqweight j}[t,\zeta] = {\planP}[t,\zeta]$, and consequently
		$\packetslack({\planP}_{\geqweight j},\tau) = \packetslack({\planP},\tau)$ for all $\tau \in [t,\zeta]$.
\end{enumerate}
\end{observation}

\begin{proof}
Part~(a) follows directly from the correctness of Algorithm~{\ComputePlan} that computes $\planP$.
Part~(b) is also straightforward: $\planP$ is heavier than any other plan (feasible set of pending packets) for $\pendpackets$.
If the weight of $j\in \planP$  is increased by 
$\delta > 0$, the weight of $\planP$ will increase by $\delta$, 
and as the weight of any other plan $X$ increases by at most $\delta$, plan $\planP$
remains to be optimal.
The proof for the second part of~(b) is similar.

To show (c), compute the desired schedule as follows. First, schedule packets in $\planP$
in the canonical order. In this schedule, $j$ will be scheduled at some slot $\xi'\in (\beta,d_j]$.
If $\xi' = \xi$, we are done.
If $\xi' \in (\beta,\xi)$, shift the packets in $[\xi'+1,\xi]$ to the left by one slot and schedule $j$ at $\xi$.
The last case is when $\xi'\in (\xi,d_j]$. In this case, we use the fact that 
$\packetslack(\planP,\tau)\ge 1$ for all $\tau\in (\beta,d_j)$, which means that for any $\tau\in (\beta,d_j)$
the number of packets in $\planP$ with deadline at most $\tau$ is strictly less than $\tau-t+1$,
the number of slots in $[t,\tau]$. Then, by the canonical ordering,  we have that
none of the packets scheduled in $[\xi,\xi'-1]$ is scheduled at its deadline.
Therefore, we can shift these packets to the right by one slot, which allows us to schedule $j$ at slot $\xi$.

Part~(d) also follows from the correctness of {\ComputePlan}: Consider its run that produces $\planP$.
Immediately after packet $j$ is processed we have $\packetslack({X}_{\geqweight j},\zeta) = 0$,
i.e., slot $\zeta$ is already tight.
Then no packet lighter than $j$ with deadline at most $\zeta$  will be later admitted to $X$.  
Thus, $\planP[t,\zeta] = X[t,\zeta]$ is already fixed at that point, and so are its slack values.
\end{proof}

\smallskip

In the rest of this section  we now deal with plan-update lemmas.
First, in Section~\ref{sec: updating a plan after a packet arrival} we explain how the optimal
plan changes in response to packet arrivals. 
In Section~\ref{sec: updating the plan after transmitting a packet} we detail how
to update the optimal plan after some arbitrary packet from the plan is transmitted. The
plan properties in these two sections are not specific to our algorithm -- they address general 
properties of optimal plans. Later in Section~\ref{app: leap step of algorithm planmonotonicity},
we give the plan-update lemma specifically for a leap step of Algorithm~{\PlanMonotonicity}.


\subsection{Updating the Optimal Plan after a Packet Arrival}
\label{sec: updating a plan after a packet arrival}

In this section, we address the question of how the optimal plan changes when a new packet is 
released. So consider some step $t$, and let $\pendpackets$ be the current set of pending
packets and $\planP$ the optimal plan for $\pendpackets$. Suppose that now
another packet $s$ is released, and let $\pendpackets' = \pendpackets\cup\braced{s}$ be
the new set of pending packets. Denote by $\planQ$ the optimal plan for $\pendpackets'$.
The lemma below explains how to determine $\planQ$ from $\planP$ and 
it gives its important properties. (See Figure~\ref{fig:packets-plan-update-arrival} for illustration.)


\begin{lemma}\label{lem:plan_update-arrival}
(The Plan-Update Lemma for Packet Arrival.)
Consider the scenario above, involving the arrival of a new packet $s$.
Let $\gamma =\nexttightslot(\planP,d_s)$ and define $f$ to be the packet in $\planP[t,\gamma]$ with $w_f = \minwt(\planP,d_s)$;
in other words, the minimum-weight packet in $\planP[t,\gamma]$. Also, let $\beta = \prevtightslot(\planP,d_f)$.
By definition, we have $\beta < \gamma$ and $d_f,d_s\in (\beta,\gamma]$.

If $w_s< w_f$, then $\planQ = \planP$, i.e., $s$ is not added to the optimal plan. 

If $w_s > w_f$, then $\planQ = \planP\cup\braced{s} \setminus \braced{f}$, i.e.,
	$s$ is added to the optimal plan replacing $f$. In this case the following properties hold:  
\begin{enumerate}[label=(\alph*),itemsep=2pt]
	\item $\packetslack(\planQ,\tau) = \packetslack(\planP,\tau)$ for
		$\tau \in [t, \min(d_f, d_s)) \cup [\max(d_f, d_s),\timehorizon]$. 
	
	\item For slots $\tau \in [ \min(d_f, d_s) , \max(d_f, d_s))$, there are two cases:      
		\begin{enumerate}[label=(b\arabic*),itemsep=1pt]  
			
			\item If $d_s > d_f$, then
			 $\packetslack(\planQ,\tau) = \packetslack(\planP,\tau)+1$ for $\tau\in [d_f,d_s)$.			 
			 It follows that all segments of $\planP$ in the interval $(\beta,\gamma]$
			get merged into one segment $(\beta,\gamma]$ of $\planQ$.

			\item If $d_s \le d_f$,
			then $\packetslack(\planQ,\tau) = \packetslack(\planP,\tau)-1$ for $\tau\in [d_s,d_f)$,
			and thus there might be new tight slots in $[d_s,d_f)$.
			This case happens only when $(\beta,\gamma]$ is itself a segment of $\planP$ and it contains
			 both $d_f$ and $d_s$. This segment may get split into multiple segments of $\planQ$.
			
		\end{enumerate}
		
	\item For any slot $\tau \ge t$, $\minwt(\planQ,\tau) \ge \minwt(\planP,\tau)$.
\end{enumerate}
\end{lemma}

We remark that packet $f$ satisfies $w_f\neq w_s$, by assumption~\ref{ass:diffWeights}.

\begin{figure}[!ht]
\centering
\includegraphics[width=6in]{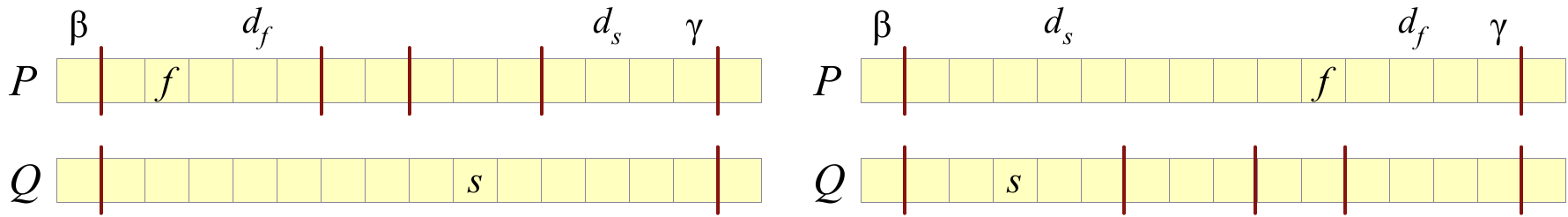}
\caption[Illustration of changes of tight slots and segments after arrival of a packet]{An
illustration of changes of tight slots and segments
in Case~(b1) of Lemma~\ref{lem:plan_update-arrival} on the left and in Case~(b2) on the right. 
The segments are separated by vertical dark-red line segments.
}
\label{fig:packets-plan-update-arrival}
\end{figure}

\begin{proof}
We start by proving how the set of packets in the optimal plan changes.
Let $\pendpackets$ be the set of packets pending before $s$ is released and 
$\pendpackets' = \pendpackets\cup \braced{s}$ be the set of pending packets after the release of $s$.
We consider two ``parallel'' runs of Algorithm~{\ComputePlan}, one on $\pendpackets$ to compute $\planP$
and the other on $\pendpackets'$ to compute $\planQ$.

Suppose first that $w_s < w_f$. Then both runs will be identical until right after the time
when $f$ is considered. Thus $\planP_{\geqweight f} = \planQ_{\geqweight f}$.
Moreover, as $f$ is the minimum-weight packet in $\planP$
with $d_f\le \gamma =\nexttightslot(\planP,d_s)$, we have $\planP[t,\gamma]\subseteq \planP_{\geqweight f}$, that is,
all packets in $\planP$ with deadline at most $\gamma$ are in $\planP_{\geqweight f}$.
Thus $\pslack(\planQ_{\geqweight f},\gamma) = \pslack(\planP_{\geqweight f},\gamma) = \pslack(\planP,\gamma) = 0$.
Therefore, as $d_s\le \gamma$ and as $s$ is considered
after $f$, packet $s$ will not be admitted, which implies that both runs produce the same plan $\planP=\planQ$.

For the rest of the proof, we assume that $w_s > w_f$. To prove that $\planQ = \planP\cup\braced{s} \setminus \braced{f}$, 
based on Observation~\ref{obs: plan properties}(b), without loss of generality we can assume that 
$w_s = w_f + \varepsilon$ for a tiny $\varepsilon>0$,
so that $s$ immediately precedes $f$ in the ordering of $\pendpackets'$ by decreasing weights.
Specifically, since $f \in \planP$, $\planP$ does not change if $f$'s weight is increased to
$w_s - \varepsilon$.  We will show that $\planQ = \planP\cup\braced{s} \setminus \braced{f}$
with this increased weight of $f$.  But as then $f \notin \planQ$, decreasing $f$'s weight
back to its original value does not change $Q$.

With the assumption about $w_f$, the run on  $\pendpackets'$ before $s$ is considered will be identical to the run  on $\pendpackets$
before $f$ is considered, thus $\planQ_{\gtrweight s} = \planP_{\gtrweight f}$. In the next step of these
runs, we consider packet $f$ for $\planP$ and $s$ for $\planQ$. We now have two cases.

Suppose that $d_s > d_f$. Then, 
since $f$ is admitted to $\planP$, we have  $\pslack(\planP_{\gtrweight f},\tau) \ge 1$ for $\tau\ge d_f$,
which, by the paragraph above, implies that  $\pslack(\planQ_{\gtrweight s},\tau) \ge  1$ for $\tau \ge d_s$; and
thus $s$ will be admitted to $\planQ$. As $d_f,d_s\le \gamma$, we also have
$\pslack(\planQ_{\geqweight s},\gamma) = \pslack(\planP_{\geqweight f},\gamma) = 0$. This means that in the
run on $\pendpackets'$ packet $f$, which is considered next, will be rejected.
From this point on, both runs are again the same, because no packet with deadline at most $\gamma$
is admitted, and because for $\tau > \gamma$ and any packet $g$ lighter than $f$ we have
$\pslack(\planP_{\gtrweight g},\tau) = \pslack(\planQ_{\gtrweight g},\tau)$, so
$g$ gets admitted to $\planQ$ if and only if $g$ is admitted to $\planP$.

Next, consider the case when $d_s \le d_f$. In this case, since also $d_f\le \gamma$, we have
$\nexttightslot(\planP,d_f) = \gamma$, so the interval $(\beta,\gamma]$ is in fact a single segment of $\planP$.
(This also proves the second claim in Case~(b2), namely that $(\beta,\gamma]$ is a segment of $\planP$.)
This gives us that $\pslack(\planP_{\gtrweight f},\tau) \ge \pslack(\planP,\tau) \ge 1$ for all $\tau\in (\beta,\gamma)$.
Since $\planP$ admits $f$, we also have $\pslack(\planP_{\gtrweight f},\tau) \ge 1$ for all $\tau \ge d_f$.
Therefore $\pslack(\planQ_{\gtrweight s},\tau) =  \pslack(\planP_{\gtrweight f},\tau) \ge 1$ for all $\tau \ge d_s$,
which implies that $\planQ$ will admit $s$. As in the previous case, we will then have
$\pslack(\planQ_{\geqweight s},\gamma) = \pslack(\planP_{\geqweight f},\gamma) = 0$, so the run for $\planQ$
will not admit $f$ (which is next in the ordering), and the remainders of both runs will be identical.

\smallskip

Next, to prove parts~(a)--(c) in the case $w_s > w_f$, we
analyze the changes of $\packetslack()$ and $\minwt()$ caused by the arrival of $s$.

\smallskip\noindent
(a) This part is straightforward, as replacing $f$ by $s$ in the plan does not change any value of 
$|\planP[t,\tau]|$ for slots $\tau\in [t, \min(d_f, d_s)) \cup [\max(d_f, d_s),\timehorizon]$.
In particular, $\beta$ and $\gamma$ remain tight in $\planQ$, and
slots in $(\beta,\min(d_f, d_s))\cup [\max(d_f, d_s),\gamma)$ are not tight in $\planQ$.

\smallskip\noindent
(b) 
Recall that $\planQ = \planP\cup\braced{s} \setminus \braced{f}$.
Suppose that $d_s > d_f$.  Then $\packetslack(\planQ,\tau) = \packetslack(\planP,\tau)+1$
for $\tau \in [ d_f,d_s)$ as $|\planQ[t,\tau]| = |\planP[t,\tau]|-1$ for such $\tau$.
So there are no tight slots in the interval $[d_f,d_s)$ in $\planQ$. This shows (b1).

The argument for (b2) is similar. Suppose that $d_s \le d_f$. As already justified
earlier, $(\beta,\gamma]$ is a segment of $\planP$. Hence, as $\beta < d_s \le d_f \le \gamma$,
the slack values for $\tau\in [d_s,d_f)$ are strictly positive in $\planP$.
Further, $\packetslack(\planQ,\tau) = \packetslack(\planP,\tau) - 1 \geq 0$ for $\tau \in [d_s,d_f)$.
This means that in $\planQ$ there may appear additional tight slots in $[d_s,d_f)$.

\smallskip\noindent
(c) For $\tau \in [t,\beta]$, the set of packets with deadline at most $\tau$
does not change and tight slots up to $\beta$ remain the same.
Hence, $\minwt(\planQ,\tau) = \minwt(\planP,\tau)$ for such $\tau$.
For $\tau\in (\beta, \gamma]$, it holds that $\minwt(\planP,\tau) = w_f$
and, as $s$ replaces $f$ and $\gamma$ remains a tight slot,
all packets in $\planQ$ with deadlines at most $\nexttightslot(\planQ,\tau)\le\gamma$ are heavier than $f$, so
we actually have $\minwt(\planQ,\tau) > \minwt(\planP,\tau)$.
Finally, consider $\tau > \gamma$. Since tight slots after $\gamma$ are unchanged,
the set of packets considered in the definition of
$\minwt(\tau)$ changes only by replacing $f$ by $s$ with $w_s > w_f$, which implies $\minwt(\planQ,\tau) \ge \minwt(\planP,\tau)$.
\end{proof}

The observation below provides a characterization of optimal plans in terms of the  $\minwt()$ function.
Its $(\Rightarrow)$ implication follows from the first claim of Lemma~\ref{lem:plan_update-arrival}.
The $(\Leftarrow)$ implication follows by noting that if a plan $\planP$ satisfies the condition
in the observation then running Algorithm~{\ComputePlan} on $\pendpackets$ will produce
exactly the set of packets in $\planP$. (The formal proof proceeds by induction, showing
that packets from $\planP$ will be admitted to $X$ because $\planP$ is feasible, but
packets from $\pendpackets\setminus\planP$ will not, as the condition in line~(4) will be violated
at the time they are considered.)


\begin{observation}\label{ref: obs: characterization of optimal plans}
A plan $\planP$ is optimal for a set $\pendpackets$ of pending packets if and only if
$w_j < \minwt(\planP,d_j)$ for all packets $j\in \pendpackets\setminus\planP$.
\end{observation}



\subsection{Updating the Optimal Plan after Transmitting a Packet}
\label{sec: updating the plan after transmitting a packet}

Next, we analyze how the optimal plan evolves after an algorithm transmits a packet from the current
optimal plan. (Only such transmissions are relevant to our analysis as in our algorithm \PlanMonotonicity{}
packets from outside the optimal plan are not considered for transmission.)
We stress that in this section we only analyze the effects of the event of a packet being transmitted, and the
properties we establish are independent of the algorithm. 
In Appendix~\ref{app: leap step of algorithm planmonotonicity} we will deal with 
the changes in the plan specific to  {\PlanMonotonicity{}}, namely with packet transmissions that
also involve adjustments of the instance.

So let $\pendpackets$ be the set of packets pending at a time $t$ and denote by
$\planP = \planP^t$ the optimal plan for $\pendpackets$. Suppose that we choose some
packet $p\in \planP$, transmit it at time $t$, and advance the time to $t+1$. 
The set $\pendpackets'$ of packets pending at time $t+1$ is obtained from 
$\pendpackets$ by removing $p$ and all packets that expire at time $t$.
In this section, $\planQ$ will denote the new optimal plan for $\pendpackets'$ and time step $t+1$;
that is $\planQ = {\widetilde{\planP}}^{t+1}$ in the notation from Section~\ref{subsec: transmitting a packet}.

We now explain how to obtain $\planQ$ from $\planP$.
There are two cases, depending on whether or not $p$ is in the initial segment $[t,\alpha]$ of $\planP$,
where $\alpha = \nexttightslot(\planP,t)$.
We start with the case when $p$ is in the initial segment, which is quite straightforward.
See Figure~\ref{fig:packets-plan-update-sch1} for an illustration.


\begin{lemma}\label{lem:plan-update_transmission-1stSegment}
(The Plan-Update Lemma for Transmitting $p\in \planP[t,\alpha]$.)
Suppose that at time $t$ we transmit a packet $p\in \planP[t,\alpha]$, 
where $\alpha = \nexttightslot(\planP, t)$. Then:
\begin{enumerate}[label=(\alph*)] 
\item $\planQ = \planP \setminus \braced{p}$. Thus the new optimal plan $\planQ$ is obtained from $\planP$
		simply by removing $p$.
\item  $\packetslack(\planQ, \tau) = \packetslack(\planP, \tau) - 1$ for $\tau \in [t+1, d_p)$
   and $\packetslack(\planQ, \tau) = \packetslack(\planP, \tau)$ for $\tau \ge d_p$.
   Thus new tight slots may appear in $[t+1,d_p)$.
\item For any slot $\tau\ge t+1$, $\minwt(\planQ,\tau) \ge \minwt(\planP,\tau)$.
\end{enumerate}
\end{lemma}

\begin{figure}[!ht]
\centering
\includegraphics[width =3in]{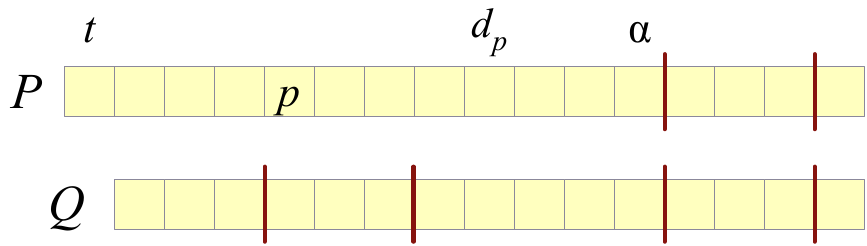}
\caption[Illustration of changes of tight slots and segments after transmitting a packet
from the initial segment]{An illustration of changes of tight slots
and segments in Lemma~\ref{lem:plan-update_transmission-1stSegment}.}
\label{fig:packets-plan-update-sch1}
\end{figure}

Note that part~(b) is meaningful in the special case when $d_p = t$, as according to our definition 
of $\pslack()$, $\pslack(\planQ,t)$ is defined and equal to $0$, that is slot $t$ is considered tight at time $t+1$.

\begin{proof}
Observe first that $\planP\setminus \braced{p}\subseteq \pendpackets'$. Indeed, if some
$q\in \planP\setminus \braced{p}$ had $d_q = t$, it would mean that the first tight slot $\alpha$ is $t$,
so $q$ would have to be equal to $p$.

\smallskip\noindent
(a) According to Observation~\ref{obs: plan properties}(c), $\planP$ has a schedule where $p$ is scheduled in slot $t$. 
This implies directly that $\planP\setminus \braced{p}$ is a feasible set of packets at time $t+1$. If $\planQ' \neq \planP\setminus \braced{p}$ is any other 
feasible subset of $\pendpackets'$ then $\planQ'$ would have a schedule starting at time $t+1$, so
$\planP'= \planQ'\cup\braced{p}$ would have a schedule starting at time $t$, and thus would be feasible at time $t$.
Since $\planP$ is heavier than $\planP'$, $\planP\setminus \braced{p}$ is heavier than $\planQ'$; thus proving that the optimal plan
at time $t+1$ is indeed $\planQ = \planP\setminus \braced{p}$.

\smallskip\noindent
(b) 
For $\tau \in [t+1,d_p)$, $\packetslack(\planQ, \tau) = \packetslack(\planP, \tau) - 1$ as in $\planQ$ the time is incremented
and $|\planQ[t+1,\tau]| = |\planP[t+1,\tau]|$. 
For $\tau \ge d_p$, incrementing the time for $\planQ$ is compensated by $p$ not contributing to $\planQ[t+1,\tau]$
(that is $|\planQ[t+1,\tau]| = |\planP[t+1,\tau]|-1$), and thus $\packetslack(\planQ, \tau) = \packetslack(\planP, \tau)$.

\smallskip\noindent
(c) By~(b), the value of $\nexttightslot(\planP, \tau)$ may only decrease.
Since also no packet is added to the plan, in the definition of $\minwt(\planQ,\tau)$ 
we consider a subset of packets used to define $\minwt(\planP,\tau)$, which shows~(c).
\end{proof}

For the case $p\in \planP(\alpha,\timehorizon]$, recall that the substitute packet $\substpacket(\planP, p)$ for $p$
is the heaviest pending packet $\rho\in \pendpackets\setminus \planP$ satisfying $d_\rho > \prevtightslot(\planP, d_p)$.
We prove that $\substpacket(\planP, p)$ really appears in the plan when $p$ is transmitted.
See Figure~\ref{fig:packets-plan-update-schLater} for an illustration.


\begin{lemma}\label{lem:plan-update_transmission-LaterSegment}
(The Plan-Update Lemma for Transmitting $p\ \in \planP(\alpha,\timehorizon]$.)
Suppose that at time $t$ we transmit a packet $p\in \planP(\alpha,\timehorizon]$, where $\alpha = \nexttightslot(\planP, t)$.
Let $\rho = \substpacket(\planP, p)$, $\beta = \prevtightslot(\planP,d_p)$, $\gamma =\nexttightslot(\planP,d_\rho)$,
and denote by $\firstlightpacket$ the lightest packet in $\planP[t,\alpha]$, i.e., in the initial segment of $\planP$. Then:
\begin{enumerate}[label=(\alph*)]	
	\item $\planQ = \planP\setminus \braced{p,\firstlightpacket}\cup \braced{\rho}$. 
		Thus the new optimal plan $\planQ$ is obtained from $\planP$ by removing $\firstlightpacket$ and
		replacing $p$ by $\rho$.
	\item $\packetslack(\planQ, \tau) = \packetslack(\planP, \tau) - 1$ for $\tau\in [t+1, d_\firstlightpacket)$.
	\item $\packetslack(\planQ, \tau) = \packetslack(\planP, \tau)$ 
				for $\tau\in [d_\firstlightpacket, \min(d_p, d_\rho))\cup [\max(d_p, d_\rho) , \timehorizon]$.
	\item For $\tau\in [ \min(d_p, d_\rho),\max(d_p, d_\rho) )$ there are two cases:
		\begin{enumerate}[label=(d\arabic*)]
			\item If $d_\rho > d_p$, then	
				$\packetslack(\planQ, \tau) = \packetslack(\planP, \tau) + 1$ for $\tau \in [d_p, d_\rho)$.
				This means that all segments of $\planP$ in $(\beta,\gamma]$ get merged into one segment $(\beta,\gamma]$ of $\planQ$.
			\item If $d_\rho < d_p$, then
				$\packetslack(\planQ, \tau) = \packetslack(\planP, \tau) - 1$ for $\tau\in [d_\rho, d_p)$.
				In this case $(\beta,\gamma]$ is itself a segment of $\planP$, and in $\planQ$ 
				there might be new tight slots in $[d_\rho, d_p)$ that partition $(\beta,\gamma]$ into smaller segments.
		\end{enumerate}
		
\end{enumerate}
\end{lemma}

\begin{figure}[!ht]
\centering
\includegraphics[width=5in]{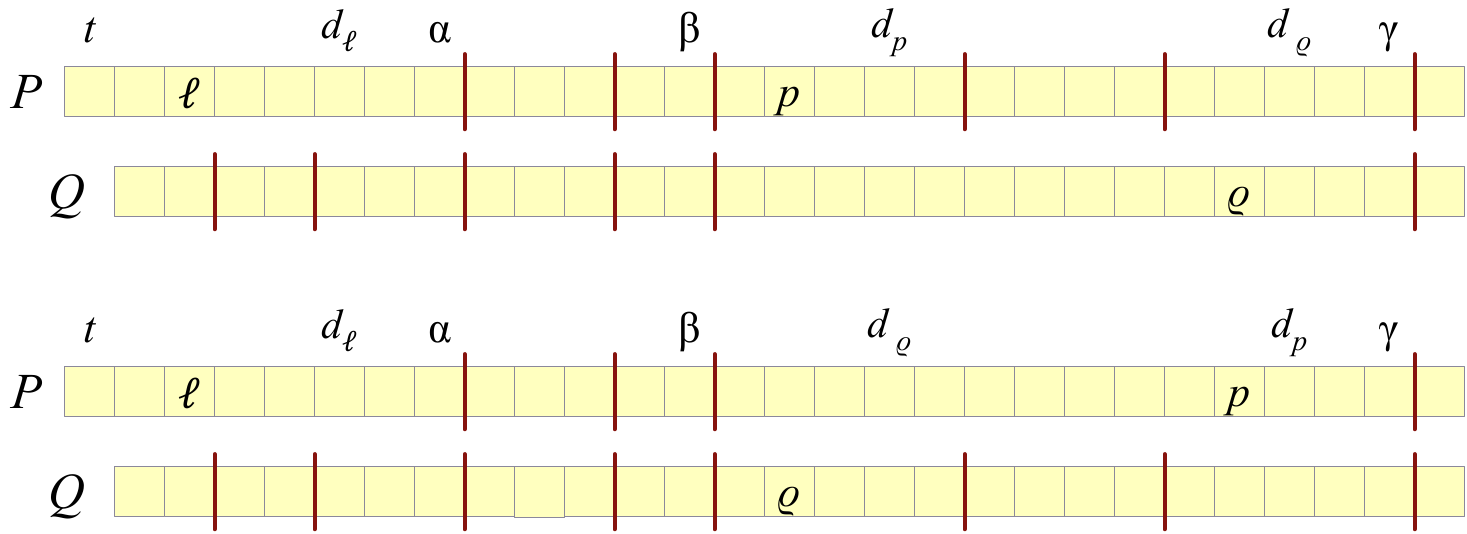}
\caption[Illustration of changes of tight slots and segments after transmitting a packet
from a later segment]{An illustration of changes of tight slots
and segments, in Case (d1) of Lemma~\ref{lem:plan-update_transmission-LaterSegment} on the top
and in Case (d2) on the bottom.}
\label{fig:packets-plan-update-schLater}
\end{figure}

\begin{proof}
(a) Let $\pendpackets$ be the set of packets pending at time $t$ before transmitting $p$ and
$\pendpackets'$ be the packets that remain pending at time $t+1$ after $p$ is transmitted;
that is, $\pendpackets'$ is obtained from $\pendpackets$ by removing $p$ and all packets that expire at time $t$.
We observe first that $\planP\setminus\braced{p,\firstlightpacket}\subseteq \pendpackets'$.
Indeed, if there was some $q \in \planP\setminus\braced{p,\firstlightpacket}$ with $d_q = t$, this would
imply that $\alpha = t$, so $d_\firstlightpacket = t$ and thus $\firstlightpacket = q$, contradicting the choice of $q$.
Obviously, $p\notin\pendpackets'$.
As for $\firstlightpacket$, it may or may not be in $\pendpackets'$, depending on whether or not $d_\firstlightpacket > t$.
We also have $\rho\in \pendpackets'$, because $d_\rho> \beta \ge t$.

We prove that  $\planQ=\planP\setminus \braced{p,\firstlightpacket}\cup \braced{\rho}$ in two substeps:
\begin{itemize}
	\item First, we increment the time to $t+1$ but without transmitting $p$. After this substep the set
of pending packets is $\barpendpackets = \braced{a\in \pendpackets \suchthat d_a \ge t+1}$, i.e.,
we exclude packets that expire at time $t$. We show that the optimal plan for $\barpendpackets$ is
$\barplanQ = \planP\setminus\braced{\firstlightpacket}$. 
	\item In the second substep we remove
$p$ from $\barpendpackets$ (and transmit it in time slot $t$), so the new set of pending packets
will be  $\pendpackets' = \barpendpackets\setminus\braced{p}$, and we prove that the new optimal plan
will be $\planQ = \barplanQ\setminus\braced{p}\cup\braced{\rho}$.
(Recall that we have $p\neq\firstlightpacket$
as $p\not\in \planP[t,\alpha]$ and $\firstlightpacket\in \planP[t,\alpha]$.)
\end{itemize}

The intuition is that these two substeps are largely independent: incrementing $t$ to $t+1$ will
squeeze out $\firstlightpacket$ from $\planP$, no matter what packets with deadlines after $\alpha$
are removed from $\pendpackets$. 


\smallskip
\emph{First substep.}
For the first substep, we use an argument similar to the one in Lemma~\ref{lem:plan-update_transmission-1stSegment}.
Observe that $\planP\setminus\braced{\firstlightpacket}$ is a feasible set of packets at time $t+1$,
since $\planP$ has a schedule where $\firstlightpacket$ is in slot $t$, by Observation~\ref{obs: plan properties}(c).
Let $\planQ' \neq \planP\setminus \braced{\firstlightpacket}$ be any other subset of $\barpendpackets$ that is feasible at time $t+1$;
we show that $w(\planQ') < w(\planP\setminus \braced{\firstlightpacket})$.
Since $\planP$ contains exactly $\alpha - t + 1$ packets with deadline at most $\alpha$,
there is a packet $q'\in \planP[t,\alpha]$ such that $q'\notin \planQ'$
(possibly $q' = \firstlightpacket$). As $\planQ'$ has a schedule starting at time $t+1$, 
$\planP'= \planQ'\cup\braced{q'}$ has a schedule starting at time $t$ and thus is feasible at time $t$.
Since $\planP$ is at least as heavy as $\planP'$ and $\firstlightpacket$ is the lightest packet in $\planP[t,\alpha]$, 
implying $w_\firstlightpacket\le w_{q'}$, we conclude that indeed $w(\planP\setminus \braced{\firstlightpacket}) \ge w(\planQ')$.
Hence, the plan after incrementing the time is indeed $\barplanQ = \planP\setminus \braced{\firstlightpacket}$.

Thus we can think of $\barplanQ$ as being obtained from  $\planP$ by advancing the time to $t+1$ and removing $\firstlightpacket$. 
It follows that $\pslack(\barplanQ,\tau) = \pslack(\planP,\tau)$ for  $\tau \geq d_\firstlightpacket$. 
In particular, $\beta$ is a tight slot of $\barplanQ$ and $\beta = \prevtightslot(\barplanQ,d_p)$.
Further, as $\pendpackets$ and $\barpendpackets$ contain the same packets with deadline after $\beta$
and as $\beta \ge d_\firstlightpacket$, we get that
$\rho$ is the heaviest packet in $\barpendpackets\setminus\barplanQ$ with deadline greater than $\beta$. 
Summarizing, we have
\begin{align}
    \beta \;&=\;   \prevtightslot(\barplanQ,d_p)
	 \label{eq: QP prop 1}
	\\
	\rho \;&=\; \argmax_j
			\braced{ w_j  \suchthat j\in \barpendpackets\setminus\barplanQ \;\&\; d_j > \beta }
	 \label{eq: QP prop 2} 
	\\
    \gamma \;&=\;   \nexttightslot(\barplanQ,d_\rho) 
	 \label{eq: QP prop 3} 	   
\end{align}	


\smallskip
\emph{Second substep.}
In the second substep we establish a relationship between $\barplanQ$ and $\planQ$.
To this end, we consider two ``parallel'' runs of the greedy algorithm to compute plans: one run
of $\ComputePlan(\barpendpackets, t+1)$ which yields $\barplanQ$ and another run 
of $\ComputePlan(\pendpackets', t+1)$ which yields $\planQ$. (Recall that $\pendpackets' = \barpendpackets \setminus\braced{p}$.)    
The contents of the set $X$ in these two runs after processing packets heavier than some packet $j$ is equal to
$\barplanQ_{\gtrweight j}$ and ${\planQ}_{\gtrweight j}$, respectively. To avoid ambiguity, we will
be using these notations when referring to the contents of $X$ during these runs.

We divide these two runs into three phases. In the first phase we consider the two runs for packets $j$ with $w_j > w_p$
(that is, before $p$ gets admitted to $\barplanQ$). Since $\barpendpackets_{\gtrweight p} = \pendpackets'_{\gtrweight p}$, 
in this phase the two runs are exactly the same. This gives us that
\begin{equation*}
\barplanQ_{\gtrweight p} \;=\; \planQ_{\gtrweight p}.   
\end{equation*}

\smallskip

The second phase involves packets $j$ with $w_p \geq w_j > w_\rho$.
We claim that during this phase the following invariants will be maintained: 
\begin{align}
	\planQ_{\geqweight j} \;&=\; {\barplanQ}_{\geqweight j} \setminus\braced{p}  
		\label{eq: QQ phase 2 inv 1}
	\\
\packetslack(\planQ_{\geqweight j},\tau) \;&=\;
						\begin{cases} 
							\packetslack(\barplanQ_{\geqweight j},\tau) & \quad \text{for } \tau \in [t+1,\beta] \\
							\packetslack(\barplanQ_{\geqweight j},\tau) \ge 1 & \quad\text{for } \tau \in (\beta,d_p) \\        
							 \packetslack(\barplanQ_{\geqweight j},\tau)+1 & \quad\text{for } \tau\ge d_p  \\  
						\end{cases} 
		\label{eq: QQ phase 2 inv 2} 
\end{align}
We first show that these properties hold for $j=p$. We have $p\in \barplanQ$ and $\pendpackets' = \barpendpackets \setminus\braced{p}$;
this means that $p$ will be admitted to $\barplanQ_{\geqweight p}$ in the run for $\barpendpackets$
but it will not be considered in the run for  $\pendpackets'$. Thus Property~(\ref{eq: QQ phase 2 inv 1}) follows.
Invariant~(\ref{eq: QQ phase 2 inv 1}) and the definition of $\packetslack()$ imply all relations
between $\packetslack(\planQ_{\geqweight p},\tau)$ and $\packetslack(\barplanQ_{\geqweight p},\tau)$
in invariant~(\ref{eq: QQ phase 2 inv 2}), and the inequality  
$\packetslack(\barplanQ_{\geqweight p},\tau) \ge 1$ for  $\tau \in (\beta,d_p)$ follows immediately from~(\ref{eq: QP prop 1}).
(Note that $\packetslack(\barplanQ_{\geqweight a},\tau) \ge \packetslack(\barplanQ,\tau)$ holds for any packet $a$ and slot $\tau$.)

Next, consider packets $j$ with $w_p > w_j > w_\rho$. If $j$ gets admitted to
$\barplanQ_{\geqweight j}$ then invariant~(\ref{eq: QQ phase 2 inv 2}) implies that 
it also gets admitted to ${\planQ}_{\geqweight j}$, preserving~(\ref{eq: QQ phase 2 inv 1}).
That~(\ref{eq: QQ phase 2 inv 2}) is preserved follows from the same argument as for $j=p$ above.

So consider the second case, when $j$ does not get admitted to $\barplanQ_{\geqweight j}$.    
We claim that $j$ will also not be admitted to ${\planQ}_{\geqweight j}$. To justify this, we argue as follows.
By~(\ref{eq: QP prop 2}), $\rho$ is the heaviest packet in $\barpendpackets$ with deadline after $\beta$
that will not be admitted to $\barplanQ$. Therefore, since $w_j > w_\rho$, we have $d_j \le \beta$.
Suppose for a contradiction that $j$ is admitted to ${\planQ}_{\geqweight j}$.
As $j$ is not admitted to $\barplanQ_{\geqweight j}$, there must be $\tau\ge d_j$ such that
$\packetslack(\barplanQ_{\gtrweight j}, \tau) = 0$, although $\packetslack(\planQ_{\gtrweight j}, \tau) > 0$ as $j$ is admitted to ${\planQ}_{\geqweight j}$.
Due to invariant~(\ref{eq: QQ phase 2 inv 2}) for $\tau\in [t+1,\beta]$,
we must have $\tau>\beta$. As $\tau$ is already tight, by Observation~\ref{obs: plan properties}(d)
we get that $\packetslack(\barplanQ_{\gtrweight j}, \beta) = \packetslack(\barplanQ, \beta) = 0$,
where the second step holds by~\eqref{eq: QP prop 1}.
Using invariant~(\ref{eq: QQ phase 2 inv 2}) for $\tau\in [t+1,\beta]$, 
we obtain $\packetslack(\planQ_{\gtrweight j}, \beta) = \packetslack(\barplanQ_{\gtrweight j}, \beta) = 0$,
contradicting that we must have $\packetslack(\planQ_{\gtrweight j}, \beta) > 0$
as $j$ is admitted to ${\planQ}_{\geqweight j}$ and $d_j \le \beta$.
Thus we indeed have that $j$ will not be admitted to ${\planQ}_{\geqweight j}$.
As in this case $j$ is not admitted to both $\barplanQ_{\geqweight j}$ and ${\planQ}_{\geqweight j}$,
both invariants~(\ref{eq: QQ phase 2 inv 1}) and~(\ref{eq: QQ phase 2 inv 2}) are obviously preserved.

\smallskip

In the third phase we analyze both runs for packets $j$ with $w_j \leq w_\rho$. We claim that for these
packets the following invariants hold:
\begin{align}
	\planQ_{\geqweight j} \;&=\; {\barplanQ}_{\geqweight j} \setminus\braced{p} \cup \braced{\rho} 
		\label{eq: QQ phase 3 inv 1}
	\\
\packetslack(\planQ_{\geqweight j},\tau) \;&=\;
							\packetslack(\barplanQ_{\geqweight j},\tau) \quad \text{for } \tau \ge \gamma 
		\label{eq: QQ phase 3 inv 2}
		\\     
\packetslack(\planQ_{\geqweight j},\gamma) \;&=\;
							\packetslack(\barplanQ_{\geqweight j},\gamma) \;=\; 0
		\label{eq: QQ phase 3 inv 3} 
\end{align}
Note that property~(\ref{eq: QQ phase 3 inv 2}) follows directly from~(\ref{eq: QQ phase 3 inv 1}), because
$d_p, d_\rho \le \gamma$. Thus we only need to show that~(\ref{eq: QQ phase 3 inv 1}) and~(\ref{eq: QQ phase 3 inv 3})
hold for all $j$ in this phase. 
	
We first consider $j=\rho$. As $\rho\notin \barplanQ$, $\rho$ will not be admitted to $\barplanQ_{\geqweight\rho}$.
From~(\ref{eq: QQ phase 2 inv 2}) and~(\ref{eq: QP prop 2}) we have that $\packetslack(\planQ_{\gtrweight\rho},\tau) \ge 1$ for
$\tau\ge d_\rho$, and thus $\rho$ will be admitted to $\planQ_{\geqweight\rho}$. This gives us~(\ref{eq: QQ phase 3 inv 1}) and~(\ref{eq: QQ phase 3 inv 2}). 
Since $\rho\notin \barplanQ$, all packets in $\barplanQ[t+1, \gamma]$ are heavier than $\rho$; in other words,
we have $\barplanQ[t+1, \gamma]  = \barplanQ_{\geqweight\rho}[t+1, \gamma]$, which together with~(\ref{eq: QP prop 3}) implies that
$\gamma =   \nexttightslot(\barplanQ_{\geqweight \rho},d_\rho)$. Then,
from~(\ref{eq: QQ phase 3 inv 2}) (that we already established for $\rho$), we obtain that 
$\packetslack(\planQ_{\geqweight \rho},\gamma) =  \packetslack(\barplanQ_{\geqweight \rho},\gamma) = 0$. 

In the rest of the third phase, for each packet $j$ with $w_j < w_\rho$, condition~(\ref{eq: QQ phase 3 inv 3}) implies
that if $d_j \le \gamma$ then $j$ will not be admitted to  $\barplanQ_{\geqweight j}$ and also to $\planQ_{\geqweight j}$.
On the other hand, if $d_j > \gamma$, then condition~(\ref{eq: QQ phase 3 inv 2}) implies that 
$j$ will be either admitted to  both $\barplanQ_{\geqweight j}$ and $\planQ_{\geqweight j}$ or to none.
This completes the proof that the above invariants are preserved, and the proof of~(a) as well.

\smallskip\noindent
(b)
We now analyze the changes in the values of $\packetslack()$. If $d_\firstlightpacket > t$, it holds that 
$\packetslack(\planQ, \tau) = \packetslack(\planP, \tau)-1$ for $\tau \in [t+1,d_\firstlightpacket)$, as in the transmission step
the time was incremented and there is no change in the set of packets taken into account.

\smallskip\noindent
(c)
For $\tau \in [d_\firstlightpacket, \min(d_\rho, d_p))$ we have $\packetslack(\planQ, \tau) = \packetslack(\planP, \tau)$
as in the transmission step the time was incremented, but $\firstlightpacket$ was forced out. (In particular,
since $d_\firstlightpacket\le \alpha < \min(d_\rho, d_p)$, $\alpha$ is a tight slot in $\planQ$, including the special case when $\alpha = t$.)
Similarly, for $\tau\ge \max(d_\rho, d_p)$, it holds that
$\packetslack(\planQ, \tau) = \packetslack(\planP, \tau)$, since in the transmission step the time was incremented, $\firstlightpacket$ was forced out,
$p$ was transmitted, and $\rho$ appeared in $\planQ$.

\smallskip\noindent
(d)
In Case~(d1), for $\tau \in [d_p,d_\rho)$, we have
$\packetslack(\planQ, \tau) = \packetslack(\planP, \tau) + 1$ as in the transmission step the time was increased,
but $\firstlightpacket$ was forced out and $p$ was transmitted.
It follows that $\planQ$ has no tight slots in $[d_p, d_\rho)$, so
$\prevtightslot(\planQ,d_\rho) = \beta$ and $\nexttightslot(\planQ,d_p) = \gamma$, which means that
in $\planQ$ the interval $(\beta,\gamma]$ forms one segment.

In Case~(d2), the definitions of $\beta = \prevtightslot(\planP,d_p)$ and $\gamma =  \nexttightslot(\planP,d_\rho)$, together with
$d_\rho < d_p$, imply that $(\beta,\gamma]$ is indeed a segment of $\planP$. In particular, this means that 
$\packetslack(\planP, \tau) \ge 1$ for $\tau\in [d_\rho, d_p)$. For such slots $\tau$ we thus obtain that
$\packetslack(\planQ, \tau) = \packetslack(\planP, \tau) - 1$ because in the transmission step the time was increased, $\firstlightpacket$ was forced out,
and $\rho$ was added to $\planQ$. 
\end{proof}

Note that after transmitting $p \in \planP(\alpha,\timehorizon]$ (without adjusting packet weights or deadlines appropriately),
the slot monotonicity property does not hold. This is because $\minwt(d_\rho)$ decreases, 
as $\minwt(\planP, d_\rho) > w_\rho$ but $\minwt(\planQ, d_\rho) = w_\rho$.

\section{Appendix: Leap Step of Algorithm~{\PlanMonotonicity}}
\label{app: leap step of algorithm planmonotonicity}

Algorithm~$\PlanMonotonicity$ artificially adjusts weights and deadlines of some pending packets in leap steps,
and therefore the results from the previous section are not sufficient
to fully characterize how the optimal plan evolves during its computation.
Lemmas~\ref{lem:plan_update-arrival} and~\ref{lem:plan-update_transmission-1stSegment} 
remain valid for Algorithm~$\PlanMonotonicity$, but we still need a variant of 
Lemma~\ref{lem:plan-update_transmission-LaterSegment} that will characterize the changes in
the optimal plan after a leap step. We provide such a characterization below, see Lemma~\ref{lem:leapStep}. 
Then we use this lemma to prove Lemma~\ref{lem:monotonicity-LeapStep}, the slot-monotonicity property for leap steps.

In the following, similar to the previous section, we use $\planP$ for $\planP^t$ (the optimal
plan after all packets at time $t$ have been released) and by $\planQ = {\widetilde{\planP}}^{t+1}$
we denote the optimal plan just after the event of transmitting $p$, incrementing the time,
and changing weights and deadlines (but before new packets at time $t+1$ are released).


\begin{lemma}\label{lem:leapStep}
Suppose that $t$ is a leap step of Algorithm~\PlanMonotonicity{} in which $p$ was transmitted,
and let $\rho = \substpacket(\planP,p)$ be the substitute packet of $p$.
Let packets $h_0 = p, h_1, \dots, h_k$ and tight slots
$\tau_0, \dots, \tau_k = \gamma = \nexttightslot(\planP,d^t_\rho)$ be as defined in the algorithm.
Furthermore, let $\beta = \tau_{-1} =\prevtightslot(\planP,d^t_p)$.
Then:  
\begin{enumerate}[label=(\alph*)]
\item $\planQ = \planP \setminus \braced{p, \firstlightpacket} \cup \braced{\rho}$; in particular, if $k\ge 1$, all packets $h_1, h_2, \dots, h_k$ are in $\planQ$.
\item $w^t_p = w^t_{h_0} > w^t_{h_1} > w^t_{h_2} > \cdots > w^t_{h_k} > w^t_\rho$.
\item $h_k$'s deadline is in the segment of $\planP$ ending at $\gamma$, 
		that is, $\prevtightslot(\planP,d^t_\rho) < d^t_{h_k} \le \gamma$.
\item The $\pslack()$ values change as follows (see Figure~\ref{fig:leapStep-alg-pslack}):

\begin{enumerate}[label=(d.\roman*)]
\item $\packetslack(\planQ, \tau) = \packetslack(\planP, \tau) - 1$ for $\tau\in [t+1, d^t_\firstlightpacket)$. 
\item  $\packetslack(\planQ, \tau) = \packetslack(\planP, \tau)$ for $\tau\in [ d^t_\firstlightpacket, \beta]$.
\item If $k \ge 1$, 
then for $i=0,\dots, k-1$ we have the following changes
in $(\tau_{i-1}, \tau_i]$:
   \begin{enumerate}[label=(d.iii.\arabic*)]
   \item $\pslack(\planQ, \tau) = \pslack(\planP, \tau)$ for $\tau\in (\tau_{i-1},d^t_{h_i}) \cup\braced{\tau_i}$.
   \item $\pslack(\planQ, \tau) = \pslack(\planP, \tau) + 1$ for $\tau\in [d^t_{h_i}, \tau_i)$.
\end{enumerate}
\item $\packetslack(\planQ, \tau) = \packetslack(\planP, \tau)$ for $\tau\in (\tau_{k-1}, \min(d^t_{h_k}, d^t_\rho))$.
\item For $\tau\in [\min(d^t_{h_k}, d^t_\rho), \max(d^t_{h_k}, d^t_\rho))$, there are two cases:
	\begin{enumerate}[label=(d.v.\arabic*)]
	\item If $d^t_{h_k} < d^t_\rho$, then $\pslack(\planQ, \tau) = \pslack(\planP, \tau) + 1$
	for $\tau\in [d^t_{h_k} , d^t_\rho)$.
	\item If $d^t_\rho < d^t_{h_k}$, then $\pslack(\planQ, \tau) = \pslack(\planP, \tau) - 1$
	for $\tau\in [d^t_\rho , d^t_{h_k})$.
\end{enumerate}
\item $\pslack(\planQ, \tau) = \pslack(\planP, \tau)$ for $\tau \ge \max(d^t_{h_k}, d^t_\rho)$.
\end{enumerate}
\item Any tight slot of $\planP$ is a tight slot of $\planQ$, but there might be new tight slots of $\planQ$   
in $[t+1,d^t_\firstlightpacket)$ and, in Case~(d.v.2), also in $[d^t_\rho, d^t_{h_k})$. Thus, in general, we have
$\nexttightslot(\planQ,\tau)\le \nexttightslot(\planP,\tau)$ for all $\tau\ge t+1$.
\item $w^{t+1}_{h_i}\ge \minwt(\planP, d^{t+1}_{h_i})$ for any $i=1, \dots, k$, and $w^{t+1}_\rho\ge \minwt(\planP, d^{t+1}_\rho)$.
\end{enumerate}
\end{lemma}

\begin{proof}
(a) The claim clearly holds for a simple leap step by Lemma~\ref{lem:plan-update_transmission-LaterSegment},
i.e., when $k=0$, because in this case the algorithm does not decrease deadlines, and 
increasing the weight of $\rho$ does not change the new optimal plan $\planQ$, by Observation~\ref{obs: plan properties}(b).

Thus consider an iterated leap step.
Let $\overline{\planQ}$ be the optimal plan after $p$ is transmitted and the time is increased,
but before the adjustment of weights and deadlines in line~\ref{algLn:h_i-setWeightAndDdln}
of the algorithm is taken into account. From Lemma~\ref{lem:plan-update_transmission-LaterSegment} we have
$\overline{\planQ}=\planP\setminus \braced{p,\firstlightpacket}\cup \braced{\rho}$;
in particular, packets $h_1, h_2, \dots, h_k$ are in $\overline{\planQ}$.
Our goal is to prove that $\overline{\planQ} = \planQ$.

By Observation~\ref{obs: plan properties}(b), increasing the weights of packets in the optimal
plan cannot change the plan, so it is sufficient to show that the feasibility and optimality of the
plan is not affected by the decrease of the deadlines of packets $h_1, h_2, \dots, h_k$.

By Lemma~\ref{lem:plan-update_transmission-LaterSegment},  all segments of $\overline{\planQ}$ 
between $\beta = \prevtightslot(\planP,d^t_p)$ and $\gamma = \nexttightslot(\planP,d^t_\rho)$
get merged into one, which means that $\pslack(\overline{\planQ}, \tau)\ge 1$ for any $\tau\in (\beta, \gamma)$.
For $i=1,\dots,k$, decreasing the deadline of $h_i$ from $d^t_{h_i}$ to $d^{t+1}_{h_i} = \tau_{i-1}$
decreases $\pslack(\overline{\planQ},\tau)$ by $1$ for $\tau\in [\tau_{i-1}, d^t_{h_i})$.
All $k$ intervals where these decreases occur
are contained in $(\beta, \gamma)$ and, since $d^t_{h_i}\le \tau_i$ for all $i$, these intervals do not overlap.
Thus, after decreasing the deadlines of the $h_i$'s, 
all values of $\pslack()$ will remain non-negative for $\overline{\planQ}$, so $\overline{\planQ}$
remains a feasible set of packets.
Decreasing these deadlines cannot make previously unfeasible sets of packets feasible,
implying that $\planQ$ remains optimal. Thus $\overline{\planQ} = \planQ$, as claimed.
In particular, packets $h_1, h_2, \dots, h_k$ remain in $\planQ$.

\smallskip\noindent
(b) By the choice of $p$ in the algorithm,
$p$ is the heaviest packet in the segments of $\planP$ in $(\beta, \gamma]$, because $\rho$ is the substitute packet
for any packet in $\planP$ with deadline in $(\beta, \gamma]$. Thus for each $i=1,\dots,k$,
since $d^t_{h_i} \in (\beta , \gamma]$, we get $w^t_p > w^t_{h_i}$. 
The ordering of weights of $h_i$'s follows from the definition of $h_i$'s in
line~\ref{algLn:choosing_h_i} of the algorithm's description.
Finally, for any $i=0,\dots,k$ inequality $w^t_{h_i} > w^t_\rho$ holds, because $\rho\notin \planP$
and $d^t_{h_i} < \gamma = \nexttightslot(\planP,d^t_\rho)$.

\smallskip\noindent
(c) This holds by the definition of $h_i$'s in line~\ref{algLn:choosing_h_i}
and by the condition of the while loop in line~\ref{algLn:whileCycle}, which stops the loop
when $\tau_i = \nexttightslot(\planP,d^t_{h_i}) = \gamma$.

\smallskip\noindent
(d) For any slot $\tau \in [t+1, \beta]\cup (\gamma,\timehorizon]$ 
the value of $\pslack(\tau)$ is not affected by the decrease of the deadlines of $h_i$'s,
since $h_i$'s are both in $\planP$ and in $\planQ$ and since their old and new deadlines are in $(\beta, \gamma]$.
We thus get exactly the same changes of the $\pslack()$ values outside $(\beta,\gamma]$ as 
in Lemma~\ref{lem:plan-update_transmission-LaterSegment}.
Regarding slots in $(\beta, \gamma]$, Lemma~\ref{lem:plan-update_transmission-LaterSegment}(d) shows that $\pslack(\tau)$
decreases by 1 for $\tau\in [d^t_\rho, d^t_p)$ if $d^t_\rho < d^t_p$,
and increases by 1 for $\tau\in [d^t_p, d^t_\rho)$ if $d^t_p < d^t_\rho$.
In the former case, we have $\tau_0 = \gamma$ as $d^t_\rho$ and $d^t_p$ are in the same segment, that is $k=0$.
In the latter case, we sum the increase of values $\pslack(\tau)$ for  $\tau\in [d^t_p, d^t_\rho)$
with the changes of the $\pslack()$ values due to decreasing the deadlines of $h_i$'s, analyzed in~(a),
and we get the changes summarized in~(d); see Figure~\ref{fig:leapStep-alg-pslack}.

\smallskip\noindent
Part~(e) follows from~(c) and~(d), which imply 
that the value of $\pslack(\planP,\tau)$ does not increase for any tight slot $\tau$ in $\planP$.

\smallskip\noindent
Part~(f) follows from the way the weights are changed in lines~\ref{algLn:incrWeightRho} and~\ref{algLn:h_i-setWeightAndDdln}
of the algorithm.
\end{proof}

\begin{figure}[!ht]
\centering
\includegraphics[width=\textwidth]{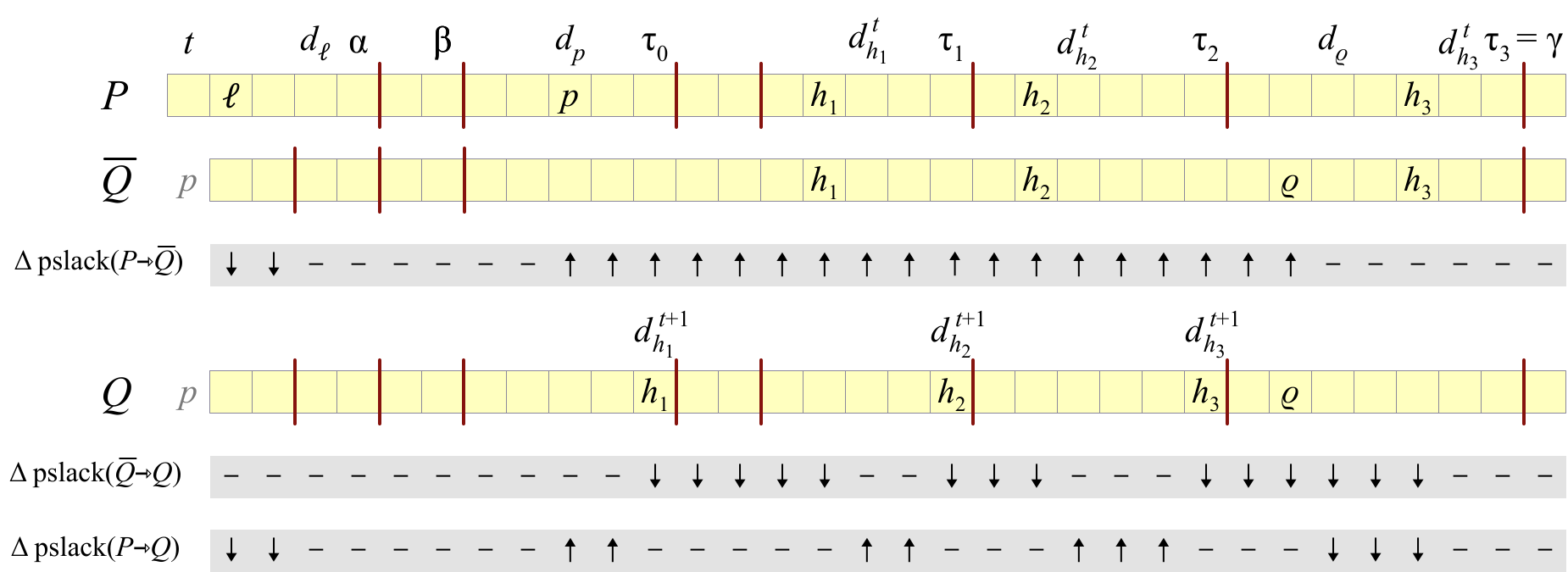}
\caption{An example with $k=3$ illustrating the proof of Lemma~\ref{lem:leapStep}(d), by
showing changes of the $\packetslack()$ values in an iterated leap step.
Here, notations $\Delta \packetslack (\planP \rightarrow \overline{\planQ})$,
$\Delta \packetslack (\overline{\planQ} \rightarrow {\planQ})$, and
$\Delta \packetslack (\planP \rightarrow \planQ)$ represent the changes of the $\packetslack()$
in the two substeps (from $\planP$ to $\overline{\planQ}$ and from $\overline{\planQ}$ to ${\planQ}$), 
and the overall change from $\planP$ to $\planQ$, respectively.
Up and down arrows represent increases and decreases of the values of $\packetslack()$ in individual slots. } 
\label{fig:leapStep-alg-pslack}
\end{figure}

Next, we show that, thanks to the deadline adjustments in Algorithm~\PlanMonotonicity{},
the slot monotonicity property for optimal plans is preserved in leap steps.


\begin{lemma}\label{lem:monotonicity-LeapStep}
If step $t$ is a leap step of Algorithm~\PlanMonotonicity{}, then
$\minwt(\planQ, \tau)\ge \minwt(\planP, \tau)$ for any $\tau \ge t+1$.
\end{lemma}

\begin{proof}
We use notation from Lemma~\ref{lem:leapStep}. By Lemma~\ref{lem:leapStep}(e) all tight slots of $\planP$
are tight slots of $\planQ$ (in particular, $\beta$ and $\gamma$ remain tight slots);
thus $\nexttightslot(\planQ,\tau)\le \nexttightslot(\planP,\tau)$ for all $\tau\ge t+1$.
Fix some $\tau\ge t+1$, and let $a$ be the packet that realizes $\minwt(\planQ, \tau)$, that is
the minimum-weight packet in $\planQ$ with $d^{t+1}_a \le \nexttightslot(\planQ,\tau)$.
We need to show that $w^{t+1}_a \ge \minwt(\planP, \tau)$. We have three cases.

\smallskip
\noindent
\mycase{1}
 $\tau \le \beta$.
Lemma~\ref{lem:leapStep}(a) shows that $\firstlightpacket$ is forced out of the optimal plan,
thus not in $\planQ$, and otherwise the set of packets in the optimal plan with deadline at most $\beta$
does not change, that is, $\planQ[t+1,\beta] = \planP[t,\beta] \setminus \braced{\firstlightpacket}$.
It follows that $a \in \planP$ and its weight and deadline remain the same.
We thus have $d^t_a = d^{t+1}_a \le \nexttightslot(\planQ,\tau)\le \nexttightslot(\planP,\tau)$,
so $a$ is considered in the definition of $\minwt(\planP, \tau)$, which implies
$w^{t+1}_a = w^t_a \ge \minwt(\planP, \tau)$.

\smallskip
\noindent
\mycase{2}
$\tau \in (\beta, \gamma]$.
If $a\notin \braced{h_1,...,h_{k+1}}$ (in particular, this means that $a\neq\rho=h_{k+1}$),
then $a$ is also in $\planP$ with the same deadline and weight. 
As $\nexttightslot(\planQ,\tau)\le \nexttightslot(\planP,\tau)$, we get $d^t_a\le \nexttightslot(\planP,\tau)$,
hence $w^{t+1}_a \ge \minwt(\planP,\tau)$ follows as in Case~1.

Next, suppose that $a=h_i$ for some $i\in\braced{ 1,\dots,k}$ (excluding the case $a=h_{k+1}=\rho$). 
Recall that, as a result of changes in line~11 in the algorithm, we have
$w^{t+1}_a \ge \minwt(\planP,\tau_{i-1})$ and $d^{t+1}_a = \tau_{i-1}$. 
The definition of $\minwt(\planQ,\tau)$ gives us that $\tau_{i-1} = d^{t+1}_a = \nexttightslot(\planQ,\tau)$,
so $\tau > \prevtightslot(\planQ,\tau_{i-1})$. Thus,
since tight slots of $\planP$ are tight also in $\planQ$, implying $\prevtightslot(\planQ,\tau_{i-1})\ge \prevtightslot(\planP,\tau_{i-1})$, we get that $\tau > \prevtightslot(\planP,\tau_{i-1})$.
So $\tau$ is in the segment $(\prevtightslot(\planP,\tau_{i-1}),\tau_{i-1}]$ of $P$.
By definition, for each segment of $\planP$ the value of $\minwt(\planP,\tau')$ is constant 
for all $\tau'$ in this segment. Therefore $\minwt(\planP,\tau_{i-1}) = \minwt(\planP,\tau)$,
and we conclude that $w^{t+1}_a \ge \minwt(\planP,\tau_{i-1}) = \minwt(\planP,\tau)$.

Finally, consider the case when $a=\rho$. Recall that $w^{t+1}_\rho = \minwt(\planP,d^t_\rho)$ and $d^t_\rho = d^{t+1}_\rho$.
We have $\tau > \prevtightslot(\planQ,d^t_\rho)$, which implies that $\tau > \prevtightslot(\planP,d^t_\rho) = \prevtightslot(\planP,\gamma)$.
Thus $\tau$ is in the segment $(\prevtightslot(\planP,\gamma),\gamma]$ of $P$, just like $\rho$.
It follows that $w^{t+1}_\rho = \minwt(\planP,d^t_\rho) = \minwt(\planP,\tau)$.

\smallskip
\noindent
\mycase{3}
$\tau > \gamma$. The set of packets in the plan with deadline strictly after $\gamma$ does not change
and also their weights and deadlines remain the same.
Thus if $d^t_a > \gamma$, then using  $\nexttightslot(\planQ,\tau)\le \nexttightslot(\planP,\tau)$ again, 
we obtain that $w^{t+1}_a \ge \minwt(\planP,\tau)$.
Otherwise, $d^t_a \le \gamma$, thus  $w^{t+1}_a\ge \minwt(\planQ,\gamma)\ge \minwt(\planP,\gamma)\ge \minwt(\planP,\tau)$,
where the second inequality follows from Case~2 and the third one from $\gamma < \tau$.
\end{proof}


\section{Appendix: Simpler Variants of Algorithm \PlanMonotonicity}
\label{app: simpler_algs}


\newcommand{\AlgSim}{\textsf{PlanM-Simpler}\xspace}
\newcommand{\AlgEvSim}{\textsf{PlanM-EvenSimpler}\xspace}

In this appendix, we provide some intuition on why the iterative choice of packets $h_i$ in Algorithm~\PlanMonotonicity{}
(line~\ref{algLn:choosing_h_i}) is essentially needed in an iterated leap step. 
We first give an example of an instance showing that a memoryless variant of our algorithm does not achieve ratio $\phi$;
this variant simply transmits packet $p$ according to Line~1 of Algorithm~\PlanMonotonicity{} (and does not make any
changes of packet weights and deadlines). 

This memoryless algorithm does not maintain the slot-monotonicity property, so this still leaves open
the possibility that perhaps some simpler way of maintaining this property is sufficient to
achieve $\phi$-competitiveness. To address this, we further give an instance which 
shows a lower bound higher than $\phi$ on the competitive ratio 
of two simpler variants of Algorithm~\PlanMonotonicity{} that maintain the slot-monotonicity property.


\subsection{Counterexample for the Memoryless Variant of~\PlanMonotonicity{}} 
\label{subsec: counterexample for memoryless}

The memoryless variant of Algorithm~\PlanMonotonicity{}, denoted \textsf{PlanM-Memoryless}, is very simple:
In each step $t$, transmit packet $p\in P^t$ that maximizes $w_p + \phi\cdot w(\substpacket^t(p))$.
That is, it chooses the packet for transmission in the same way as Algorithm~\PlanMonotonicity{}, but it does
not make any modifications of the instance. We now show that \textsf{PlanM-Memoryless} is not $\phi$-competitive.

Recall that slot-monotonicity is preserved if Algorithm~\PlanMonotonicity{} 
executes a packet from the initial segment in each step (see Lemma~\ref{lem:plan-update_transmission-1stSegment}).
For such instances, \textsf{PlanM-Memoryless} is equivalent to Algorithm~\PlanMonotonicity{} and thus it achieves ratio $\phi$.
Hence, a hard example for \textsf{PlanM-Memoryless} needs to involve some leap steps.

We give an example of an instance for which \textsf{PlanM-Memoryless} is not $\phi$-competitive. 
Our instance uses three interleaved infinite sequences of packets with weights increasing geometrically by
factor $\phi^2$, similar to the examples in Section~\ref{sec: online algorithm}. These packets will now have span $4$. 
The basic idea is that
the competitive ratio in each step (i.e., the ratio of the weights of packets
transmitted by the optimum and by the algorithm) will be $\phi^2$ on these geometric sequences.
In the example in Section~\ref{sec: online algorithm}, the online algorithm's ratio on the geometric sequence is also $\phi^2$, 
but it can catch up with the optimum at the end, as in the last slots
it still has some heavy packets pending that are not pending for the optimum. The trick behind our construction is to
add several packets near the end that will force \textsf{PlanM-Memoryless} to make some leap steps in which these new
packets will be transmitted, while the older packets will expire.


\begin{figure}[!ht]
\begin{center}
	\renewcommand{\arraystretch}{1.2}
	{\small
	\begin{tabular}{|c|c|} \hline
	packet $q$ & weight $w_q$ \\ \hline
	$\rho$  		& $1-\epsilon$ \\ \hline
	$\rho_i$  		& $\phi^{-2i}w_\rho$ \\ \hline
	$x$ 			& 	$1$ \\ \hline
	$x_i$ 			&  $\phi^{-2i}w_x$ \\ \hline
	 $z$ 			& $\phi$ \\ \hline
	$z_i$   		& $\phi^{-2i}w_z$ \\ \hline
	 $\dotz$ 		&  $\phi-\epsilon$ \\ \hline
	 $s'$ 			& $2+2\epsilon$ \\ \hline
	  $s$    		& $\phi^2+2\epsilon$ \\ \hline
	$\dotx$ 		&   $1-\epsilon$    \\ \hline
	$\rho'$ 		&  $\phi^{-2}$ \\ \hline
	$e$				& $\epsilon$ \\ \hline
	\end{tabular}
	}
	\hspace{0.15in}
	\includegraphics[valign = c,width=4.75in]{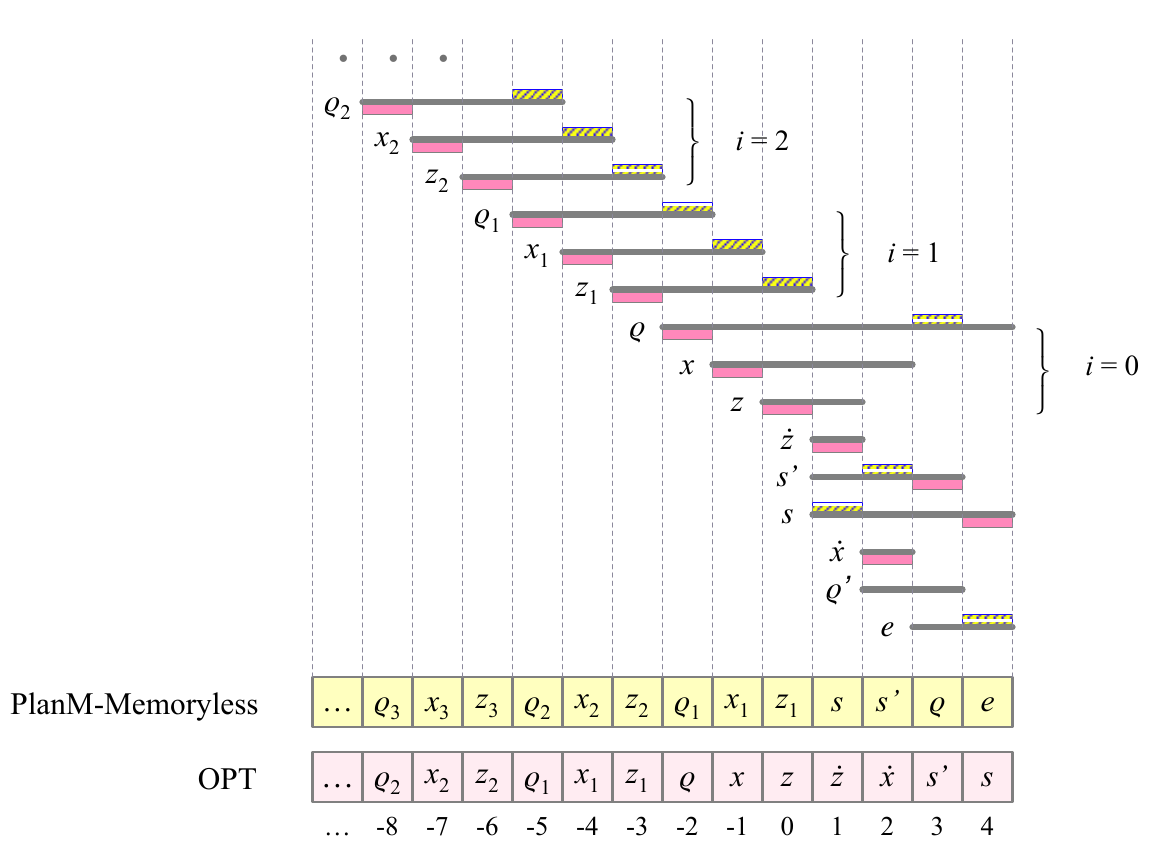}
	\caption{
	An instance on which \textsf{PlanM-Memoryless} has ratio larger than $\phi$. The table on the left shows
	the weights, the figure on the right shows release times and deadlines.
	The figure also shows the schedules of \textsf{PlanM-Memoryless} and the optimal schedule.
	The schedules are outlined at the bottom, and they are also depicted using unit-width rectangles (pink for the optimum and yellow and striped for
	the algorithm) attached to the line segments representing packet spans.
	}
	\label{fig:memoryless-counterEx}
\end{center}
\end{figure}


The instance is defined in  Figure~\ref{fig:memoryless-counterEx}.
To ensure consistent tie-breaking,
we use a small parameter $\eps > 0$ in packet weights.
The three sequences of packets with span $4$ and geometrically increasing weights are
$\braced{\rho_i}_{i=1}^\infty$, $\braced{x_i}_{i=1}^\infty$ and $\braced{z_i}_{i=1}^\infty$.
When estimating profits,  we will also include packets $\rho_0 = \rho$, $x_0 = x$, and $z_0 = z$ in these sequences, even
though the spans of $\rho$ and $z$ are not equal to $4$.
This figure also shows the schedule computed by \textsf{PlanM-Memoryless}
and the optimal schedule. The steps of \textsf{PlanM-Memoryless} are explained below:

\begin{description} 

\item{\emph{Steps $t = ...,-2,-1,0$:}}
As in the example in Section~\ref{sec: online algorithm}, when scheduling
the packets in steps $t = ...,-2,-1,0$, the algorithm will always transmit the packet that is tight at the given step.
This is because at each step $t = ...,-4,-3$ there are four packets pending $a,b,c,a'$, with respective deadlines $t$, $t+1$, $t+2$ and $t+3$,
and with weights that satisfy $w_a < w_b < w_c < w_{a'} = \phi^2 w_a$. Thus, breaking the tie between $w_a$ and $w_{a'}$, the
algorithm will transmit $a$. For times $t=-2,-1,0$ the deadlines do not satisfy the above condition; nevertheless, by inspection, the
algorithm will still schedule the tight packet $w_a$ in these steps. In other words, the three interleaved packet sequences do not interfere with each other.

\item{\emph{Step~$1$:}}
The pending packets are $\rho,x, z, \dotz, s', s$ and the
optimal plan is  $P^1= \braced{z, x, s', s}$ (with each packet in its own segment); see Figure~\ref{fig:memoryless-counterEx plans}.
We also have $\substpacket^1(s) = \substpacket^1(s') = \substpacket^1(x) = \rho$,  $\substpacket^1(z) = z$,
$w_s> w_{s'},w_x$, and 
\begin{equation*}
w_s + \phi\razy w_{\rho} \;=\; (\phi^2+ 2\epsilon) + \phi(1-\epsilon) 
						\;=\;  (1+\phi) w_z + (2 - \phi)\epsilon
						\;>\; w_z + \phi\razy w_z,
\end{equation*}						
so the algorithm will transmit $p=s$.


\begin{figure}[!ht]
\begin{center}
	\renewcommand{\arraystretch}{1.1}
	{\small
	\begin{tabular}{|c|c|} \hline
	packet $q$ & weight $w_q$ \\ \hline
	$\rho$  		& $1-\epsilon$ \\ \hline
	$x$ 			& 	$1$ \\ \hline
	 $z$ 			& $\phi$ \\ \hline
	$\dotz$ 		&  $\phi-\epsilon$ \\ \hline
	 $s'$ 			& $2+2\epsilon$ \\ \hline
	  $s$    		& $\phi^2+2\epsilon$ \\ \hline
	$\dotx$ 		&   $1-\epsilon$    \\ \hline
	$\rho'$ 		&  $\phi^{-2}$ \\ \hline
	$e$				& $\epsilon$ \\ \hline
	\end{tabular}
	}
	\hspace{0.25in}
	\includegraphics[valign = c,width=4in]{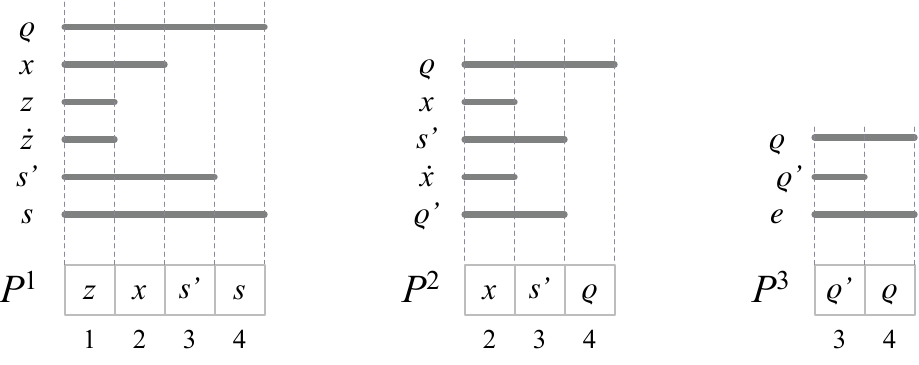}
	\caption{
	Optimal plans of \textsf{PlanM-Memoryless} in steps $1$, $2$ and $3$. 
	}
	\label{fig:memoryless-counterEx plans}
\end{center}
\end{figure}


\item{\emph{Step~$2$:}}
The pending packets are $\rho, x, s',\dotx,\rho'$ and the optimal plan is
$P^2 = \braced{x, s', \rho}$. Also, $\substpacket^2(s') = \rho'$, $\substpacket^2(x) = x$, 
$\substpacket^2(\rho)$ is a packet of weight $0$, and 
\begin{align*}
w_{s'} + \phi\razy w_{\rho'} \;&=\; 2+ 2\epsilon + \phi\cdot \phi^{-2} 
	 					\;=\; (1+\phi) w_x   + 2\epsilon
						\;>\; w_x + \phi\razy w_x,
						\\
w_{s'} + \phi\razy w_{\rho'} \;&>\; w_\rho + \phi\razy 0,
\end{align*}
so the algorithm will transmit $p=s'$.

\item{\emph{Step~$3$:}}
The pending packets are $\rho', \rho,e$, the optimal plan is $P^3 = \braced{\rho', \rho}$,
$\substpacket^3(\rho) = e$, $\substpacket^3(\rho') = \rho'$, and 
\begin{equation*}
w_{\rho} + \phi\razy w_e \;=\; (1+\phi)w_{\rho'} + (\phi-1) \epsilon 
			\;>\; w_{\rho'} + \phi\razy w_{\rho'},
\end{equation*}
so the algorithm will transmit $p = \rho$.

\item{\emph{Step~$4$:}}
The only (non-zero weight) pending packet for the algorithm in this step is $e$, so this is the packet transmitted.

\end{description}

It remains to calculate the profits of \textsf{PlanM-Memoryless} and of \OPT. As $\eps$  can be made
arbitrarily small, in these calculations we assume that  $\eps = 0$. 
We also assume for simplicity that the three interleaved geometric sequences are backwards-infinite.
Recall that $\sum_{i = 0}^\infty \phi^{-2i} = \phi$.

For \OPT, the profit in steps $...-2,-1,0$ from the three interleaved geoemtric sequences,
is $\phi (w_{\rho} + w_x + w_z) = \phi(1 + 1 + \phi) = 3\phi + 1$.
Its profit in steps $1,\dots,4$ is $w_{\dotz} + w_{\dotx} + w_{s'} + w_s = \phi + 1 + 2 + \phi^2 = 2\phi + 4$.
In total, $w(\OPT) = 5\phi + 5$.

The profit of the algorithm in steps $...,-2,-,1,0$ is $\phi^{-2}$ times that of $\OPT$,
that is, $(3\phi + 1)/\phi^2 = \phi + \phi^{-1}$.
In steps $1,\dots,4$, its profit is $w_s + w_{s'} + w_{\rho}+ w_e = \phi^2 + 2 + 1 + 0 = \phi+4$.
Thus, $w(\ALG) = \phi + \phi^{-1} + \phi + 4 = 3\phi + 3$, and thus the ratio
equals $w(\OPT) / w(\ALG) = 5 / 3 > \phi$.
(In fact, a slightly higher ratio can be achieved by having four interleaved geometric sequences,
which would allow the optimal schedule to also include the packet corresponding to $\rho'$.)

Finally, we remark that even if we parameterize \textsf{PlanM-Memoryless} by $\alpha$
and select a packet $p$ maximizing $w_p + \alpha\cdot w(\substpacket^t(p))$, 
we do not obtain a $\phi$-competitive algorithm. It would be interesting to know the competitive ratio
of \textsf{PlanM-Memoryless} for the best choice of $\alpha$. 
The currently best memoryless deterministic algorithm is $1.893$-competitive~\cite{englert_suppressed_packets_12}.


\medskip
\myparagraph{Comparison to Algorithm~{\PlanMonotonicity}.}
For comparison, consider what will our Algorithm~{\PlanMonotonicity} do on this instance.
The computation up to and including step~0 will be identical to that of \textsf{PlanM-Memoryless}.
In step~1 the same packet $s$ will be transmitted, but the weight of its substitute packet $\rho$ will be increased to $1$ in Line~4 of the algorithm.
Similarly, in step~2 $s'$ will be transmitted, but the weight of its substitute packet $\rho'$ will be increased to $1$.
Then, in step~3 {\PlanMonotonicity} will transmit $\rho'$ (whose weight is now $1$) and in step~4 it will transmit $\rho$.
This will increase the profit of the algorithm by $w^0_{\rho'} = 1/\phi^2$ (the original weight of $\rho'$), and its competitive ratio on
this instance will be $\frac{5\phi + 5}{3\phi+3+1/\phi^2} < \phi$.


\subsection{Counterexample for Simpler Variants of~\PlanMonotonicity{}} 

As explained in Appendix~\ref{app: more on properties of plans}, the value of $\minwt^t(\tau)$ may decrease only
if a packet $p$ from a non-initial segment of $P^t$ is transmitted. This is because in such a case
$p$ is replaced in the optimum plan by a packet $\rho$ with $w_\rho < \minwt^t(d_\rho)$.
\PlanMonotonicity{} maintains the slot-monotonicity property by adjusting the weights and deadlines of some pending packets.
First, it increases the weight of the substitute packet $\rho$ to $\minwt^t(d_\rho)$. The example in the
previous section showed how this weight increase of $\rho$ helps in reducing the overall ratio on some instances.

Second, if $\nexttightslot^t(d_p) < \nexttightslot^t(d_\rho)$ (i.e., in an iterated leap step),
then \PlanMonotonicity{} also makes an intricate shift of iteratively chosen packets $h_1, \dots, h_k$.
This shift is implemented by decreasing the deadlines and it prevents the algorithm from merging segments of $P^t$.
In this section we address the question whether this shift process is necessary. To this end,
we consider two simpler variants of \PlanMonotonicity{} that maintain
the slot-monotonicity property as well\footnote{The slot-monotonicity property for both variants of \PlanMonotonicity{} follows             
from arguments similar to the proofs in Appendices~\ref{app: more on properties of plans} and~\ref{app: leap step of algorithm planmonotonicity}.
We do not include the proofs here as our only aim is to show that these algorithms are not $\phi$-competitive.},
but do not perform an iterative shift 
of packets $h_1, \dots, h_k$. Both variants only differ from \PlanMonotonicity{}
in iterated leap steps, i.e., they use the same rule for choosing a packet to transmit and
also increase the weight of $\rho$ to $\minwt^t(d_\rho)$ in each leap step.

In the first variant, called~\AlgSim and defined in Algorithm~\ref{alg:planmonotonicity-simpler},
instead of choosing packets $h_1, \dots, h_k$, we identify the heaviest packet $\hstar$
in the segment of $P^t$ containing $d^t_\rho$, we decrease its deadline to
$\nexttightslot^t(d_p)$, and, if its weight is too small to preserve slot monotonicity, then we appropriately increase it.
Recall that $\alpha = \nexttightslot^t(t)$ is the first tight slot of $P^t$.

\begin{algorithm}[ht]
\caption{Algorithm~\AlgSim($t$)}
\label{alg:planmonotonicity-simpler}
\begin{algorithmic}[1]
\State{transmit packet $p\in \planPaftArrivals$ that maximizes  $w^t_p + \phi\cdot w^t(\substpacket^t(p))$} 
\If{$d^t_p > \alpha$} 
\Comment{``leap step''}    
	\State{$\rho \assign \substpacket^t(p)$} 
	\State{$w^{t+1}_{\rho} \assign \minwt^t(d^t_\rho)$} 
	\Comment{increase $w_\rho$}
    \State{$\eta \assign \nexttightslot^t(d^t_p)$ and $\gamma \assign \nexttightslot^t(d^t_\rho)$}
	\If{$\gamma > \eta$}         
	      \State{$\hstar \assign$ heaviest packet in $P^t(\prevtightslot^t(d^t_\rho), \gamma]$} 
		  \State{$d^{t+1}_{\hstar} \assign \eta$ and $w^{t+1}_{\hstar} \assign \max\big(w^t_{\hstar}, \minwt^t(\eta)\big)$} 
	\EndIf
\EndIf
\end{algorithmic}
\end{algorithm}

The second variant is even simpler as we just decrease the deadline of $\rho$
if it is after $\nexttightslot^t(d^t_p)$ and increase $w_\rho$ to match the minimum of $\minwt()$
in its new segment. We call this algorithm~\AlgEvSim and summarize it in Algorithm~\ref{alg:planmonotonicity-toosimple}.

\begin{algorithm}[ht]
\caption{Algorithm~\AlgEvSim($t$)}
\label{alg:planmonotonicity-toosimple}
\begin{algorithmic}[1]
\State{transmit packet $p\in \planPaftArrivals$ that maximizes  $w^t_p + \phi\cdot w^t(\substpacket^t(p))$} 
\If{$d^t_p > \alpha$}
\Comment{``leap step''}    
	\State{$\rho \assign \substpacket^t(p)$}
	\State{$d^{t+1}_{\rho} \assign \min\big(\,d^t_\rho, \nexttightslot^t(d^t_p)\,\big)$} 
	\State{$w^{t+1}_{\rho} \assign \minwt^t(d^{t+1}_\rho)$} 
		\Comment{increase $w_\rho$}
\EndIf
\end{algorithmic}
\end{algorithm}


\paragraph{Instance with ratio above $\phi$.}
We now present a hard instance, for which the ratio of both algorithms exceeds $\phi$.
The weights of some packets in this instance are parametrized by a parameter $\epsilon >0$, that can be made arbitrarily small.

Our instance is an extension of the lower-bound instance for \textsf{PlanM-Memoryless}
described in Section~\ref{subsec: counterexample for memoryless}. In this instance 
we will use \emph{five} (instead of just three) backwards-infinite interleaving sequences of packets whose weights increase geometrically 
by factor $\phi^2$. These packets will have span $6$. The principle is the same as in Section~\ref{subsec: counterexample for memoryless}:
At each step $t$ the algorithm will have six packets $a_0,a_1,...,a_5$ pending, with each $a_i$ having its deadline at $t+i$.
The weights will satisfy $w_{a_0}\le w_{a_i} \le \phi^2 w_{a_0}$ for $i=2,3,4$ and $w_{a_5} = \phi^2 w_{a_0}$.
Thus, breaking the tie, the algorithm will transmit $a_0$ at each such step.
This way, at each step the current ratio between the optimum profit and the algorithm's profit will be $\phi^2$,
although the algorithm will always have some remaining heavy packets pending that are not pending in the optimum schedule.
At the end of the instance we release several new packets that will be transmitted by the algorithm instead of
these older pending packets (which will expire).


Figure~\ref{fig:simpler-counterEx} shows the complete instance. It also shows the schedules of \AlgSim
(also represented by striped yellow rectangles), of \AlgEvSim, and of an optimal solution (pink rectangles).
Note that the schedules of \AlgSim and \AlgEvSim are essentially the same, as the only difference is that slots of $h$ and $\varrho$ are swapped.

We now explain the behavior of \AlgSim on instance in Figure~\ref{fig:simpler-counterEx}
(for \AlgEvSim, it is nearly the same). As explained earlier, the instance starts with five interleaved sequences of packets
$\braced{\rho_i}_{i=1}^\infty$, $\braced{x_i}_{i=1}^\infty$, $\braced{h_i}_{i=1}^\infty$, $\braced{r_i}_{i=1}^\infty$, 
and $\braced{z_i}_{i=1}^\infty$, all of span $6$, with weights forming a geometric sequence with factor $\phi^2$.
We think about these as being grouped into batches of $5$.
Both algorithms \AlgSim and \AlgEvSim will transmit tight packets in each step, while the adversary will always choose the
heaviest packet to schedule. 
Later we have one more batch of five packets. These packets have different spans but their weights extend these geometric sequences,
so we call them $\rho = \rho_0$, $x=x_0$, $h = h_0$, $r = r_0$ and $z = z_0$. For these packets the same basic principle applies, 
and in steps $-5,-4,-3,-2,-1$ the algorithm's schedule will follow the same pattern.
In step~0, for the same reason again, the tight packet $r$ will be transmitted.


\begin{figure}[!ht]
\begin{center}
	\renewcommand{\arraystretch}{1.1}
	{\small
	\begin{tabular}{|c|c|} \hline
	\bigcell{c}{packet\\ $q$} & \bigcell{c}{initial\\ weight $w_q$} \\ \hline
	$\rho$  		& $1-\epsilon$ \\ \hline
	$\rho_i$  		& $\phi^{-2i}w_\rho$ \\ \hline
	$x$ 			& 	$1$ \\ \hline
	$x_i$ 			&  $\phi^{-2i}w_x$ \\ \hline
	$h$ 			& 	$1$ \\ \hline
	$h_i$ 			&  $\phi^{-2i}w_h$ \\ \hline
	 $r$ 			& $1+2\epsilon$ \\ \hline
	$r_i$   		& $\phi^{-2i}w_r$ \\ \hline
	 $z$ 			& $\phi$ \\ \hline
	$z_i$   		& $\phi^{-2i}w_z$ \\ \hline
	  $s$    		& $\phi^2+2\epsilon$ \\ \hline	
	  $g$    		& $\phi^2$ \\ \hline	  
	 $z'$ 		&  $\phi-\epsilon$ \\ \hline
	 	 $y$    		& $\phi - \epsilon$ \\ \hline			 
	 $g'$ 			& $\phi^2-\epsilon$ \\ \hline
	$\rho'$ 		&  $1-\epsilon$ \\ \hline
	\end{tabular}
	}
	\hspace{0.05in}
	\includegraphics[valign = c,width=4.8in]{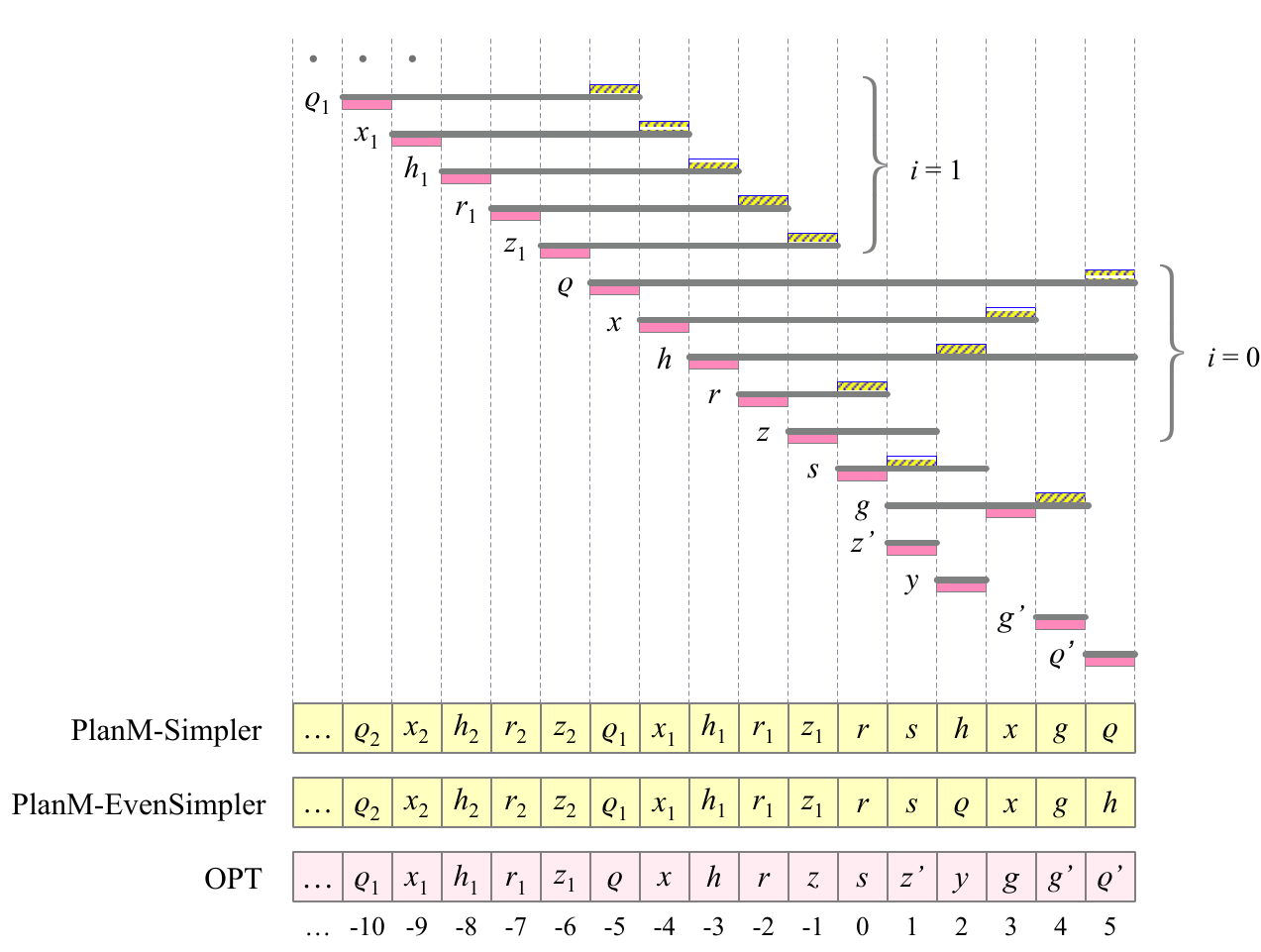}
	\caption{
	An example for algorithms \AlgSim and \AlgEvSim. The original packet weights are in the table on the left;
	packet spans are depicted in the figure, represented by a thick line segment.
	At the bottom, the figure shows the schedules of algorithms \AlgSim and \AlgEvSim, as well as the 
	optimum schedule. 
	The schedules of \AlgSim are also represented by unit-length thin rectangles (pink for the optimum and yellow and striped for
		\AlgSim) attached to the line segments representing packet spans.
	}
	\label{fig:simpler-counterEx}
\end{center}
\end{figure}


The crucial step is step~1. In this step \AlgSim transmits packet $p = s$, since $s$ is the heaviest pending packet,
$\rho$ is the only substitute packet
available, and $w_s + \phi w_\rho > \phi^2 w_z$, which holds for all $\epsilon>0$.
\AlgSim increases the weight of $\rho$ to $\minwt(5) = 1$,
decreases the deadline of packet $h$ to $2$, and increases its weight of $h$ to $\minwt(2) = \phi$.
This is illustrated in Figure~\ref{fig:simpler-counterEx plans}, that shows the optimal plans of \AlgSim
in step~1 and step~2.
(\AlgEvSim also transmits $p=s$, decreases the deadline of $\rho$ to $2$, and increases its weight to $\phi$.)


\begin{figure}[!ht]
\begin{center}
	{\small
	\renewcommand{\arraystretch}{1.1}
	\begin{tabular}{|c|c|} \hline
	packet & weight $w^1$ \\ \hline
	$\rho$  		& $1-\epsilon$ \\ \hline
	$x$ 			& 	$1$ \\ \hline
	 $h$ 			& $1$ \\ \hline
	 $z$ 			& $\phi$ \\ \hline
	 $s$ 			& $\phi^2+2\epsilon$ \\ \hline
 	 $g$ 			& $\phi^2$ \\ \hline
 	 $z'$ 			& $\phi-\epsilon$ \\ \hline
	\end{tabular}
	}
	\hspace{0.1in}
	\includegraphics[valign = c,width=1.45in]{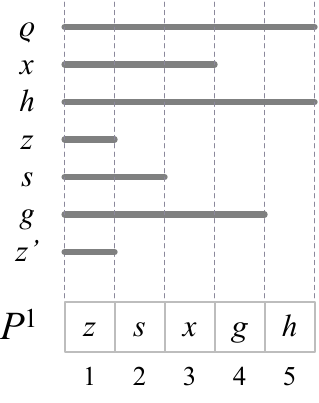}
	\hspace{0.3in}
	{\small
	\renewcommand{\arraystretch}{1.1}
	\begin{tabular}{|c|c|} \hline
	packet & weight $w^2$ \\ \hline
	$\rho$  		& $1$ \\ \hline
	$x$ 			& 	$1$ \\ \hline
	 $h$ 			& $\phi$ \\ \hline
	$g$ 		&  $\phi^2$ \\ \hline
	 $y$ 			& $\phi-\epsilon$ \\ \hline
	\end{tabular}
	}
	\hspace{0.1in}
	\includegraphics[valign = c,width=1.45in]{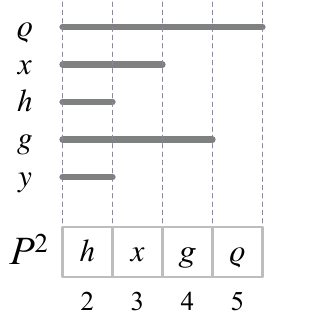}
	\caption{
	Optimal plans of \AlgSim in steps $1$ and $2$. The tables show packet weights.
	Note that the weight of $\rho$ changed, and for $h$ both the deadline and weight changed.
	}
	\label{fig:simpler-counterEx plans}	
\end{center}
\end{figure}


After step~1, there are no leap steps in the remaining four steps as there won't be any non-expiring packet
outside the optimal plan (that could serve as a non-zero weight substitute packet) and at each step the weights of pending packets are within
a factor of $\phi^2$. (Note that packet $y$ will not appear in the plan in step 2, even though it is
heavier than the \emph{original} weight of $h$.)

Next, we calculate the profits. Since $\epsilon$ can be made arbitrarily small, in the calculations we can assume that $\epsilon = 0$.
For the optimal schedule, adding up the weights of the five interleaved packet sequences, up to step -1,
their total weight is $\phi(w_\rho + w_x + w_h + w_r+ w_z) = 5\phi+1$. 
(Recall that $\sum_{i=0}^\infty \phi^{-2i} = \phi$.)
The optimal profit in steps $0,1,...,5$ is $w_s + w_{z'} + w_y + w_g + w_{g'} + w_{\rho'}
						=  \phi^2 + \phi + \phi + \phi^2 + \phi^2 + 1 = 5\phi+4$.
Thus $w(\OPT) =10\phi+5$.

The profit of either algorithm before step~$0$ (i.e., from the five interleaved sequences) is smaller by factor of $\phi^2$ than the
corresponding profit in the optimum schedule, so it equals $(5\phi+1)/\phi^2 = 4\phi -3$.
The algorithm's profit in steps $0,1,...,5$ is
$w_r + w_s + w_h + w_x + w_g + w_\rho = 1 +  \phi^2 + 1 + 1 + \phi^2 + 1 = 2\phi+6$.
(Here, of course, we take the original weights into account).
Thus its total profit is $w(\ALG) = 6\phi+3$, giving us ratio $w(\OPT) / w(\ALG) = 5/3 \approx 1.666$.


\medskip

\paragraph{What goes wrong in the analysis.}
We explain where the analysis in the paper fails for Algorithm \AlgSim (the reason is basically the same for \AlgEvSim). 
Consider step~1, and assume that $\epsilon = 0$.
In this step, using the profit of $\phi\cdot w_s = \phi^3$, the algorithm needs to ``pay'' for all of the following:
(i) the adversary profit from $z'$, equal to $\phi$, (ii) removing $z$ and $s$ from $\varbackupplan$, which costs $1 + \phi$, 
(iii) increasing the weight of $\rho$, which is negligible, 
and (iv) increasing the weight of $h$, which costs $\phi^{-1}$, as the potential increases as well.
The total cost thus sums to $\phi^3 + \phi^{-1}$;
it is no coincidence that the excess of $\phi^{-1}$ over the profit equals $w(\OPT) - \phi\cdot w(\ALG)$,
since all other steps are essentially tight (in terms of calculations in the analysis).


\medskip
\myparagraph{Comparison to Algorithm~{\PlanMonotonicity}.}
For comparison, consider what our Algorithm~{\PlanMonotonicity} does on this instance.
The computation up to and including step~0 remains identical to that of \AlgSim.
We focus on step~1. In step~1 the same packet $s$ gets transmitted, and the weight of $\rho$
is also increased to $1$.
However, {\PlanMonotonicity} further performs its iterative shift, with $k=2$,
$h_1 = g$, and $h_2 = h$. Figure~\ref{fig:algM plans} shows the effect of this change on the
plan in step~2, in which {\PlanMonotonicity} transmits $g$, unlike \AlgSim which transmits
$h$ in this step.


\begin{figure}[!ht]
\begin{center}
	{\small
	\renewcommand{\arraystretch}{1.1}
	\begin{tabular}{|c|c|} \hline
	packet & weight $w^1$ \\ \hline
	$\rho$  		& $1-\epsilon$ \\ \hline
	$x$ 			& 	$1$ \\ \hline
	 $h$ 			& $1$ \\ \hline
	 $z$ 			& $\phi$ \\ \hline
	 $s$ 			& $\phi^2+2\epsilon$ \\ \hline
 	 $g$ 			& $\phi^2$ \\ \hline
 	 $z'$ 			& $\phi-\epsilon$ \\ \hline
	\end{tabular}
	}
	\hspace{0.1in}
	\includegraphics[valign = c,width=1.45in]{9-C-2_example_simpler_algs_plan1.pdf}
	\hspace{0.3in}
	{\small
	\renewcommand{\arraystretch}{1.1}
	\begin{tabular}{|c|c|} \hline
	packet & weight $w^2$ \\ \hline
	$\rho$  		& $1$ \\ \hline
	$x$ 			& 	$1$ \\ \hline
	 $h$ 			& $1$ \\ \hline
	$g$ 		&  $\phi^2$ \\ \hline
	 $y$ 			& $\phi-\epsilon$ \\ \hline
	\end{tabular}
	}
	\hspace{0.1in}
	\includegraphics[valign = c,width=1.45in]{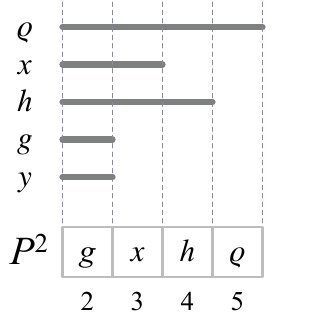}
	\caption{
	Optimal plans of Algorithm~{\PlanMonotonicity} in steps $1$ and $2$. The tables show packet weights before the
	corresponding steps.
	Note the the weight of packet $\rho$ was changed, and the deadlines of packets $g$ and $h$ were reduced.
	}
	\label{fig:algM plans}
\end{center}
\end{figure}



\end{document}